\documentclass[12pt]{iopart}

\usepackage{float}  

\usepackage{graphicx}%
\usepackage{multirow}%
\usepackage{physics}
\usepackage{multirow}
\usepackage{colortbl}
\usepackage{xcolor}
\usepackage{booktabs}
\usepackage{multirow}
\usepackage{mathtools} 
\usepackage{amsfonts}
\usepackage{amsthm}%
\usepackage{lipsum}
\usepackage{mathrsfs}%
\usepackage[title]{appendix}%
\usepackage{xcolor}%
\usepackage{textcomp}%
\usepackage{manyfoot}%
\usepackage{booktabs}%
\usepackage{algorithm}%
\usepackage{algorithmicx}%
\usepackage{algpseudocode}%
\usepackage{listings}%
\usepackage{braket}
\usepackage{geometry}
\usepackage{graphicx}
\usepackage{subcaption}
\usepackage{dsfont}
\usepackage{hyperref}
\usepackage{soul}
\usepackage{enumitem}
\usepackage[nice]{nicefrac}
\usepackage[normalem]{ulem}
\usepackage{changes}
\usepackage{todonotes}
\usepackage{esvect}
\usepackage[T1]{fontenc}
\usepackage{tikz}
\usepackage{stackengine}
\usepackage{cite}
\usepackage{array} 
\newcolumntype{P}[1]{>{\raggedright\arraybackslash}p{#1}} 

\newtheorem{theorem}{Theorem}
\newtheorem{lemma}{Lemma}

\newcommand{\floor}[1]{\left\lfloor{#1}\right\rfloor}
\newcommand{\ceil}[1]{\left\lceil #1 \right\rceil}

\usepackage{accents}

\begin{document}

\title[]{Requirements for Teleportation in an Intercity Quantum Network}

\author{Soubhadra Maiti$^{1,2,3}$, Guus Avis$^{1,2,3,4}$, Sounak Kar$^{1,2,3}$, and Stephanie Wehner$^{1,2,3}$}

\address{$^1$QuTech, Delft University of Technology}
\address{$^2$Kavli Institute of Nanoscience, Delft University of Technology}
\address{$^3$Quantum Computer Science, Department of Electrical Engineering, Mathematics and Computer Science, Delft University of Technology}
\address{$^4$Manning College of Information and Computer Sciences, University of Massachusetts Amherst}

\ead{s.maiti-1@tudelft.nl}

\vspace{10pt}

\begin{abstract}
We investigate the hardware requirements for quantum teleportation in an intercity-scale network topology consisting of two metropolitan-scale networks connected via a long-distance backbone link.
Specifically, we identify the minimal improvements required beyond the state-of-the-art to achieve an end-to-end expected teleportation fidelity of $2/3$, which represents the classical limit.
To this end, we formulate the hardware requirements computation as optimisation problems, where the hardware parameters representing the underlying device capabilities serve as decision variables.
Assuming a simplified noise model, we derive closed-form analytical expressions for the teleportation fidelity and rate when the network is realised using heterogeneous quantum hardware, including a quantum repeater chain with a memory cut-off.
Our derivations are based on events defined by the order statistics of link generation durations in both the metropolitan networks and the backbone, and the resulting expressions are validated through simulations on the NetSquid platform.
The analytical expressions facilitate efficient exploration of the optimisation parameter space without resorting to computationally intensive simulations.
We then apply this framework to a representative realisation in which the metropolitan nodes are based on trapped-ion processors and the backbone is composed of ensemble-based quantum memories.
Our results suggest that teleportation across metropolitan distances is already achievable with state-of-the-art hardware when the data qubit is prepared after end-to-end entanglement has already been established, whereas extending teleportation to intercity scales requires additional, though plausibly achievable, improvements in hardware performance.

\vspace{-10pt}
\end{abstract}
%
%
%
%
%

\section{Introduction}
\label{sec:introduction}
\everypar{\looseness=-1}
Quantum networks are anticipated to facilitate tasks beyond the reach of the classical internet~\cite{wehner2018quantum}.
In such networks, remote entanglement serves as the primary resource for performing a wide range of applications, including quantum key distribution~\cite{ekert1991quantum, bennett1992quantum},
enhanced sensing~\cite{giovannetti2004quantum}, secure remote computation~\cite{arrighi2006blind, broadbent2009universal}, clock synchronisation~\cite{jozsa2000quantum}, and many other applications.
Furthermore, entanglement shared across spatially separated locations enables foundational tests of quantum mechanics, most notably Bell inequality violation~\cite{bell1964einstein}, as well as recent proposals for probing quantum gravity effects~\cite{borregaard2025testing}, among others.
A remote entangled link consists of a pair of entangled qubits shared between end nodes, i.e., quantum devices directly accessible to users.
The quality of service for any such application depends on the quality of these links, typically represented by the entanglement fidelity, and the rate at which they are generated.
Furthermore, the required thresholds for these performance metrics vary depending on the underlying application, motivating a careful analysis of these metrics in realistic network architectures.

\begin{figure}[t]
    \centering
    \includegraphics[width=0.70\linewidth]{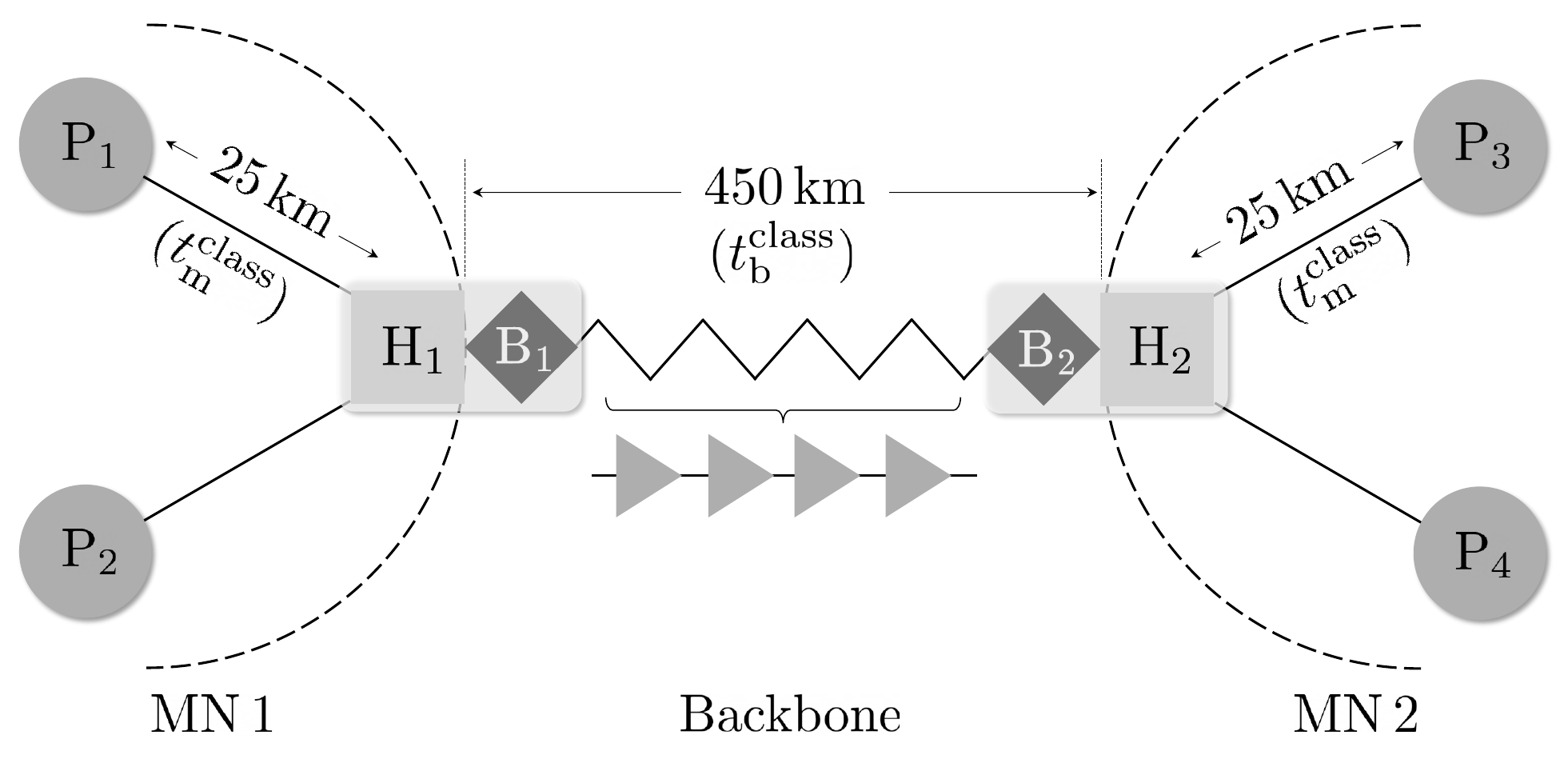}
    \caption{Schematic of an intercity quantum network architecture comprising four different components.
    User-controlled end nodes are denoted by circles $\text{P}_1$--$\text{P}_4$ and metropolitan hubs by squares $\text{H}_1$--$\text{H}_2$.
    The border nodes, shown as diamonds $\text{B}_1$--$\text{B}_2$, form the backbone together with the zig-zag line, which can be realised using either a space-based quantum communication channel or a terrestrial linear quantum repeater chain, with individual repeater nodes depicted as triangles.
    Each end node is connected to its nearest hub $25\,$km away and forms a metropolitan network (MN $1$ or $2$).
    The $450\,$km backbone connects the metropolitan regions via the border nodes at both ends.
    Together, the MNs and the backbone form the full IN, enabling long-distance quantum communication between multiple end nodes.
    }
    \vspace{-10pt}
    \label{fig:intercity_network_diagram}
\end{figure}

Photon transmission through optical fibre decays exponentially with propagation distance due to attenuation losses, which limits the feasibility of entanglement generation over long distances via direct photon transmission.
Moreover, the no-cloning theorem~\cite{wootters1982single} forbids the creation of identical copies of an unknown quantum state, making classical amplification techniques infeasible for quantum states.
To overcome this, quantum repeaters~\cite{briegel1998quantum} 
introduce intermediate nodes that enable entanglement distribution over extended distances.
Entanglement generation has been experimentally demonstrated across various physical platforms, including trapped ions~\cite{maunz2007quantum, stephenson2020high, Krutyanskiy2023Entanglement230Meters, Krutyanskiy2019Light, Krutyanskiy2023TelecomWavelength, Krutyanskiy2024Entanglement101km, Schupp2021InterfaceBetweenTrapped-Ion}, colour centres in solids~\cite{bernien2013heralded, 
humphreys2018deterministic}, rare-earth ions~\cite{lago2021telecom, liu2021heralded, ruskuc2025multiplexed}, and atomic ensemble memories~\cite{chou2007functional, yuan2008experimental, yu2020entanglement}, where matter-photon entanglement is used to generate entanglement between distant matter qubits located at remote nodes.
These platforms exhibit distinct advantages and limitations, and future large-scale quantum networks are therefore envisioned to adopt a heterogeneous architecture that integrates multiple hardware platforms~\cite{tissot2025hybrid, tissot2025single, sun2025hybrid}.

One of the central challenges in realising a scalable quantum network lies in identifying the precise hardware requirements to support efficient quantum communication.
Experimentally, Pompili \textit{et al.}~\cite{pompili2021realization} demonstrated a lab-scale three-node quantum network based on nitrogen-vacancy centres in diamond, while Krutyanskiy \textit{et al.}~\cite{Krutyanskiy2023TelecomWavelength} presented a quantum repeater node based on trapped ions capable of distributing entanglement over distances up to $50\,\text{km}$, marking significant experimental milestones.
On the theoretical side, the performance of single-repeater setups and homogeneous repeater chains has been studied extensively, both analytically~\cite{brand2020efficient, khatri2019practical, hartmann2007role, goodenough2025noise, de2024analysis, li2021efficient, rozpkedek2018parameter} and via simulation~\cite{ Ferreira_da_Silva_2021, Avis2023DelftEindhoven}.
However, to the best of our knowledge, performance analyses of heterogeneous repeater chains have so far relied predominantly on simulation-based approaches~\cite{Avis2023DelftEindhoven}.

In the context of quantum networks, quantum teleportation serves both as a fundamental primitive and as a high-level performance benchmark.
Applications such as blind quantum computing impose stringent requirements on teleportation fidelity and rate to ensure correctness and security~\cite{chia2012quantum}.
Beyond its direct use for communication, teleportation constitutes a key resource for implementing non-local gates, thereby enabling the preparation of multipartite entangled states and distributed quantum information processing~\cite{gottesman1999demonstrating}.
At the network level, the achievable end-to-end teleportation fidelity
provides a quantitative measure of overall network performance~\cite{mylavarapu2024teleportation}.
Moreover, attaining a teleportation fidelity exceeding the classical limit of $2/3$ constitutes a clear demonstration of quantum advantage~\cite{horodecki1999general}, 
analogous to Bell-inequality tests~\cite{hensen2015loophole}.

In this work, we investigate requirements for quantum teleportation in an \emph{intercity-scale network} (IN) shown in Fig.~\ref{fig:intercity_network_diagram},
consisting of two \emph{metropolitan-scale networks} (MNs) and a long-distance \emph{backbone} realised as a quantum repeater chain.
A detailed description of the network model is provided in Sec.~\ref{sec:setup}.
We adopt standard assumptions, including heralded entanglement generation (HEG), depolarising noise in memory, instantaneous local operations, swap-as-soon-as-possible (swap-ASAP) policy~\cite{kamin2023exact}, and a cut-off strategy together with the standard deterministic teleportation protocol in which teleportation succeeds with unit probability~\cite{bennett1993teleporting};
see~\ref{assumption:werner_state}–\ref{assumption:instantaneous_qubit_preparation_local_operation} for details.
Specifically, we study the requirement for attaining a teleportation fidelity of $2/3$ across metropolitan and intercity-scale distances and carry out our analysis in a hardware-agnostic manner,
i.e., we characterise the network using a set of parameters capturing the essential hardware properties.
We also evaluate these requirements for a heterogeneous platform comprising trapped-ion processors in MNs and an ensemble-based repeater chain in the backbone, thereby illustrating their relevance to realistic quantum network implementations.

For studying the requirements, we focus on two teleportation scenarios that capture different operational regimes:
\begin{itemize}
    \item \label{def:entanglement_ready_teleportation} \textbf{Entanglement-ready (ER)}:
    in ER teleportation, the data qubit (qubit to be teleported) is prepared only after end-to-end entanglement has been established.
    This approach allows the system to avoid the decoherence accumulated during the probabilistic entanglement generation process, thereby significantly improving teleportation fidelity
    ~\cite{pfaff2014unconditional}.
    \item \label{def:qubit_ready_teleportation} \textbf{Qubit-ready (QR)}:
    here, the data qubit is prepared in advance and stored in quantum memory until end-to-end entanglement is established.
    This scenario naturally arises when the data qubit preparation rate and entanglement generation rate are mismatched~\cite{
    chandra2022scheduling}, 
    a condition anticipated in large-scale network architectures.
    The resulting memory storage leads to additional decoherence.
    Although we assume instantaneous location operation and qubit preparation (see ~\ref{assumption:instantaneous_qubit_preparation_local_operation}), we include the QR scenario for completeness, as it represents an extreme yet practically relevant operating regime where the entanglement generation begins only after the data qubit is ready.
\end{itemize}

In order to determine the hardware requirements, we formulate the task as an optimisation problem, in which the relevant cost function is designed along the lines of~\cite{Avis2023DelftEindhoven}.
To solve the resulting formulation, we first derive closed-form analytical expressions for the expected fidelities and generation rates of both the end-to-end entangled link and teleported qubit.
We then validate the accuracy of the expressions by comparing them with empirical estimates from NetSquid-based simulations~\cite{Coopmans2021Netsquid}.
The analytical expressions enable us to efficiently identify the hardware requirements without resorting to computationally expensive simulations~\cite{Ferreira_da_Silva_2021, Avis2023DelftEindhoven, da2024requirements}.
Also, to quantify the requirements, we introduce two reference values for each network parameter: the \emph{baseline} value corresponds to the state-of-the-art hardware capability, while the \emph{optimistic} value represents capability anticipated in the near future based on ongoing experimental advances and theoretical predictions.
We define the hardware requirements as the \textit{minimal improvements} over the baseline parameters needed to achieve the target teleportation fidelity threshold of $2/3$.
Although our framework itself is hardware-agnostic, we evaluate the requirements using a realistic network architecture comprising trapped-ion processors for the MN and a repeater chain based on ensemble-based memories as the backbone~\cite{tissot2025hybrid, tissot2025single}, along with corresponding baseline and optimistic parameter sets.

\noindent \textit{Summary of results:}
We specifically address the following questions with respect to hardware requirements:
\begin{enumerate}[label={\textbf{Q\arabic*}},nolistsep,leftmargin=*]
    \item \label{Q:1}
    Do the baseline parameters enable quantum teleportation in an MN with a fidelity exceeding the classical limit of $2/3$?
    If not, what are the minimal parameter improvements required to reach this threshold?
    \begin{itemize}[label={-}, leftmargin=1em, itemsep=0pt]
        \item We find that current trapped-ion technology (baseline) already supports ER teleportation at metropolitan scales, whereas QR teleportation requires further improvements and becomes feasible with near-term experimental advances; see Sec.~\ref{sec:results_teleportation_in_metropolitan_network} for details.
    \end{itemize}
    \item \label{Q:2a}
    Assuming that the backbone attains optimistic performance levels, what minimal improvements to the MN parameters are required to achieve quantum teleportation with an expected fidelity $\geq 2/3$ across the IN?
    This analysis reveals the trade-offs among metropolitan parameters and provides a lower bound on the required improvements for metropolitan hardware.
    \begin{itemize}[label={-}, leftmargin=1em, itemsep=0pt]
        \item We find that although ER teleportation across IN is possible, QR teleportation requires significant improvements beyond the baseline, which nonetheless remains within the reach of optimistic estimates.
        We also identify the feasible regions in parameter space and determine the optimal parameter configurations that minimise the hardware cost; see Sec.~\ref{sec:results_teleportation_in_intercity_network} for details.
    \end{itemize}

    \item \label{Q:2b}
    Conversely, assuming optimistic MN performance, what are the corresponding backbone requirements to enable teleportation with an expected fidelity $\geq 2/3$ in the IN?
    This highlights the trade-offs between backbone parameters and establishes a lower bound on their required improvements.
    \begin{itemize}[label={-}, leftmargin=1em, itemsep=0pt]
        \item Similar to~\ref{Q:2a}, we find that ER teleportation across IN is possible and QR teleportation requires significant improvements beyond the baseline but is still within the reach of optimistic estimates;
        see Sec.~\ref{sec:results_teleportation_in_intercity_network} for details.
    \end{itemize}

    \item \label{Q:2c}
    Given baseline parameters representative of current experimental capabilities for both the MN and backbone, what minimal joint improvements are needed to achieve an expected teleportation fidelity $\geq 2/3$ in the IN?
    \begin{itemize}[label={-}, leftmargin=1em, itemsep=0pt]
        \item In this case, neither ER nor QR teleportation is possible, and we identify the minimal joint improvements required to attain the teleportation fidelity threshold; see Sec.~\ref{sec:results_teleportation_in_intercity_network} for details.
    \end{itemize}
\end{enumerate}
Beyond their role in the optimisation procedure, the analytical expressions and the underlying methodology can also be used for evaluating performance metrics for heterogeneous networks and are therefore of independent interest.

The rest of this paper is organised as follows.
First, we review relevant literature in Sec.~\ref{sec:related_work}, followed by a description of the network setup and assumptions in Sec.~\ref{sec:setup}.
Sec.~\ref{sec:objectives_and_methods} formally presents the objectives and methodology. 
In Sec.~\ref{sec:derivation_rate_an_fidelity_IN}, we derive analytical expressions for the teleportation rate and expected fidelity for the IN.
Evaluations addressing questions \textbf{Q1}--\textbf{Q4} are presented in Sec.~\ref{sec:results} and the conclusions are drawn in Sec.~\ref{sec:conclusion}.
\vspace{-8pt}

\section{Related Work}
\label{sec:related_work}
\everypar{\looseness=-1}
We begin by surveying prior studies that are closely related to our work.
This review is not intended to be an extensive review of performance analyses of quantum repeater chains; instead, we focus on key contributions that are directly relevant to our setting.
In this context, commonly considered performance metrics include the end-to-end entanglement generation time, end-to-end entanglement fidelity, and secret key rate~\cite{shor2000simple}.
Regarding the performance of a quantum repeater, Rozp\k{e}dek \textit{et al.}~\cite{rozpkedek2018parameter} identified realistic parameter regimes where a single sequential repeater outperforms direct transmission.
Analytical models for entanglement generation in a homogeneous repeater chain under various swapping and decoherence assumptions have been developed in~\cite{bernardes2011rate, brand2020efficient, khatri2019practical, hartmann2007role, kamin2023exact, goodenough2025noise, de2024analysis}, while Li \textit{et al.}~\cite{li2021efficient} optimised cutoff strategies to improve the resulting entanglement quality.
Goodenough \textit{et al.}~\cite{goodenough2025noise} derived exact formulas and tight approximations for expected fidelity under a global cut-off policy and swap-asap policy, while Andrade \textit{et al.}~\cite{de2024analysis} provided closed-form expressions for expected end-to-end fidelity in homogeneous chains with sequential entanglement swapping.
As an alternative strategy, Collins \textit{et al.}~\cite{collins2007multiplexed} proposed multiplexed quantum repeater architectures that significantly reduce sensitivity to memory coherence times.
However, these frameworks do not directly extend to inhomogeneous repeater chains where the distances between nodes are different.
In our work, we analyse a four-node repeater chain comprising two memory-equipped repeaters, where the two outer links are identical while the central link differs, thereby enabling heterogeneous hardware implementations.
We derive exact analytical expressions for expected end-to-end fidelity and generation rate of both entanglement and teleportation with a cutoff strategy, while fully accounting for communication delays to closely reflect real-world scenarios.

Several experimental approaches to quantum repeaters are based on atomic ensemble memories and linear optics following the DLCZ protocol~\cite{duan2001long}. 
Notably, Simon \textit{et al.}~\cite{simon2007quantum} extended this scheme by combining photon-pair sources with multimode quantum memories in rare-earth–doped solids, enabling faster and robust entanglement generation while retaining the simplicity of linear optics and single-photon detection.
Jiang \textit{et al.}~\cite{jiang2007fast} proposed an improved ensemble-based repeater that allows active purification, suppresses multi-excitation noise, and achieves polynomial scaling with realistic inefficiencies.
Sangouard \textit{et al.}~\cite{Sangouard2011} provided extensive surveys of such theoretical proposals and compared their entanglement generation rates against direct photon transmission.
More recently, Wu \textit{et al.}~\cite{wu2020near} incorporated time-dependent memory decay into the analysis and obtained analytical expressions for entanglement rate, showing feasibility with current technology while highlighting the importance of multiplexing.
Implementations of quantum repeaters based on other hardware platforms involve trapped ions~\cite{sangouard2009quantum, Krutyanskiy2023TelecomWavelength, zwerger2017quantum},
trapped neutral atoms~\cite{lloyd2001long, razavi2006long}, rare-earth ions doped crystals~\cite{asadi2018quantum}, and colour centres~\cite{childress2006fault, rozpkedek2019near}.
In particular, Sangouard \textit{et al.}~\cite{sangouard2009quantum} demonstrated that trapped-ion–based repeaters employing deterministic entanglement swapping and temporal multiplexing can achieve substantially higher entanglement rates than ensemble-based schemes.
Experimentally, Krutyanskiy \textit{et al.}~\cite{Krutyanskiy2023TelecomWavelength} demonstrated a trapped-ion quantum repeater node that generates and swaps entanglement over two $25\,\text{km}$ fibres, extending it to $50\,\text{km}$., and outlined near-term improvements to scale such nodes to a repeater chain of $800\,\text{km}$.
For this chain, however, they assume that the memory coherence time significantly exceeds the elementary link creation time between adjacent nodes, thereby ignoring memory decoherence.
In this context, Zwerger \textit{et al.}~\cite{zwerger2017quantum} proposed decoherence-free subspace encoding to mitigate collective dephasing.
Lloyd \textit{et al.}~\cite{lloyd2001long} proposed a cavity-based quantum network architecture in which long-distance entanglement is stored in trapped atoms, enabling unconditional teleportation via full Bell-state measurements.
Razavi \textit{et al.}~\cite{razavi2006long} compared this scheme with the DLCZ protocol, noting that while DLCZ allows faster entanglement distribution, it supports only conditional teleportation, whereas the trapped-atom approach enables unconditional teleportation.
Recent developments include hybrid architectures combining ensemble-based repeaters with single-atom or ion-based processors~\cite{tissot2025hybrid, tissot2025single, sun2025hybrid}, integration of repeaters with satellite-based links for global-scale entanglement distribution~\cite{liorni2021quantum}, and all-photonic repeaters that operate without matter-based memories~\cite{azuma2015all}.
In contrast to these works, we derive the analytics in a platform-agnostic manner while adopting parameter values for evaluations from particular experimental implementations with trapped ions~\cite{Krutyanskiy2023Entanglement230Meters, Krutyanskiy2019Light, Krutyanskiy2023TelecomWavelength, Krutyanskiy2024Entanglement101km} 
and predictions from~\cite{tissot2025hybrid,tissot2025single} for the backbone; see Sec.~\ref{sec:baseline_and_optimistic_parameters} for details.

Finally, we review prior studies most relevant to our work that address the hardware specifications for efficient quantum communication.
A critical challenge in realising quantum repeaters is that their hardware requirements remain largely unknown.
Silva \textit{et al.}~\cite{Ferreira_da_Silva_2021} introduced a systematic approach using genetic algorithms and NetSquid simulations to optimise entanglement generation in repeater chains, identifying minimal hardware requirements under realistic conditions.
In a related work, Avis \textit{et al.}~\cite{Avis2023DelftEindhoven} analysed entanglement generation between two end nodes connected via a single repeater, determining the requirements for verifiable blind quantum computing (VBQC).
Their analysis considers the restrictions imposed by real-world fibre grids and employs hardware-specific models of colour centres and trapped ions using NetSquid.
Conversely, Silva \textit{et al.}~\cite{da2024requirements} extended this to a multi-repeater network spanning $900\,\text{km}$, determining the specifications for repeaters enabling VBQC and quantum key distribution (QKD). 
Similarly, van Dam \textit{et al.}~\cite{vanDam2024hardware} identified the requirements for a trapped-ion server and measurement-only client to perform VBQC across $50\,~\text{km}$.
While these works rely on simulations in NetSquid, we derive exact analytical expressions for teleportation rates and expected fidelities to carry out the analysis.
In particular, we identify the minimal requirements in terms of the coherence time, elementary link generation probability, and link fidelity needed to surpass the classical fidelity threshold of $2/3$ for teleportation in an MN over $50\,\text{km}$ and a IN over $500\,\text{km}$~\ref{fig:intercity_network_diagram}.
\vspace{8pt}
\section{Setup and Assumptions}
\label{sec:setup}
\everypar{\looseness=-1}

In this section, we briefly introduce the quantum network setup from~\cite{beauchamp2025modularquantumnetworkarchitecture} and outline the assumptions underlying our analysis.
We specify the relevant hardware parameters that characterise the network, enabling us to formally define~\ref{Q:1}--\ref{Q:2c}.
In Fig.~\ref{fig:intercity_network_diagram}, we provide a schematic of the network,
which comprises four fundamental components: end nodes, metropolitan hubs, a backbone, and border nodes, each serving distinct functions to enable scalable quantum communication.
Given that we analyse requirements for teleportation over metropolitan (up to $50\,\text{km}$) and intercity-scale (up to $500\,\text{km}$) distances, the network components must be equipped with capabilities beyond the minimal specifications described in~\cite{beauchamp2025modularquantumnetworkarchitecture}.
We now briefly describe the functions of these components.

The \emph{end nodes}, labelled $\text{P}_1$--$\text{P}_4$ in Fig.~\ref{fig:intercity_network_diagram}, serve as primary access points for users, supporting quantum applications under user-defined control.
These devices can generate photons entangled with matter qubits and possess quantum processing capabilities, such as measurement and gate operations on matter qubits.
For deterministic teleportation, we require the end nodes to possess two quantum memories and execute classical operations, including classical protocols for coordination and computation.
Recent experiments~\cite{stephenson2020high, Krutyanskiy2023Entanglement230Meters, pompili2021realization, covey2023quantum, uphoff2016integrated, Ruf2021quantum} have demonstrated promising platforms for such end nodes, achieving long coherence times and high quantum gate fidelities using colour centres, trapped ions, and neutral atoms.

End nodes are connected to their respective \emph{metropolitan hubs}, labelled $\text{H}_1$--$\text{H}_2$ in Fig.~\ref{fig:intercity_network_diagram}, which serve as central points for facilitating HEG between node pairs within an MN, spanning city-scale distances up to $50\,\text{km}$.
End nodes $\text{P}_1$ and $\text{P}_2$ together with hub $\text{H}_1$ constitute MN $1$, while end nodes $\text{P}_3$ and $\text{P}_4$ with hub $\text{H}_2$ constitute MN $2$.
To generate and route entanglement between different pairs of nodes, these hubs may comprise entanglement generation switches~\cite{gauthier2023architecture}
or entanglement distribution switches~\cite{herbauts2013demonstration}.
For our purpose of teleportation, we only require that the hubs contain beam splitters and photon detectors to perform Bell state measurements (BSMs) on incoming photons from two nodes.
End nodes are assumed to be connected to their respective hub via standard telecommunication fibres supporting photonic qubit transmission.

To overcome the distance limitations imposed by fibre loss in direct photon transmission, the architecture incorporates a long-range \emph{backbone}, denoted by the zigzag line in Fig.~\ref{fig:intercity_network_diagram}, together with {\emph{border nodes}} $\text{B}_1$ and $\text{B}_2$.
These border nodes (referred to as junction nodes in~\cite{beauchamp2025modularquantumnetworkarchitecture}) are located adjacent to the metropolitan hubs and serve as interfaces between the backbone and MNs.
Each border node is assumed to contain at least two memories and to possess the same processing and storage capabilities as the end nodes,  allowing it to store entanglement both with neighbouring end nodes and across the backbone.
A key feature of this architecture is its ability to integrate heterogeneous quantum hardware platforms, thereby leveraging their distinct advantages, such as high-fidelity quantum gates, long memory coherence time, and multiplexing.
The backbone itself may be realised using different technological platforms.
One approach relies on space-based quantum communication~\cite{ursin2007entanglement, yin2012quantum, yin2017satellite} 
which offers low-loss photon transmission over long distances.
Alternatively, a fibre-based implementation using a multiplexed quantum repeater chain~\cite{tissot2025hybrid, simon2007quantum, Sangouard2011} can be employed, offering a promising route to preserve quantum coherence while extending the range of entanglement distribution.
We further define the IN as the composite architecture in which two distant MNs are connected via a backbone, enabling quantum communication between them.

Throughout this paper, we use the terms \emph{entanglement} and \emph{link} interchangeably to denote an entangled quantum state shared between two nodes.
Furthermore, we model entanglement generation in the IN as a single-shot process: once a border node performs an entanglement swap, it remains idle and does not initiate further generation until the ongoing end-to-end entanglement generation process is completed.
We refer to a complete trial to establish an end-to-end link as a \emph{round}, while individual trials for elementary link generation are termed as \emph{attempts}.
Under this framework, we analyse the performance of this architecture in a platform-agnostic manner, where the network is fully characterised by key hardware parameters.
In particular, for an MN, the expected teleportation fidelity and corresponding rate are primarily determined by three parameters: the elementary link generation probability between end nodes $p_{\text{m}'}$, the memory coherence time of the end nodes $t_{\text{coh}}$, and the fidelity of a freshly generated link $f_{\text{m}'}$.
Here, we assume a symmetric architecture where the two MNs are identical, end nodes are equidistant from the metropolitan hub, and border nodes have the same properties as end nodes.
Thus, the performance of the IN is governed by five key hardware parameters: the entanglement-generation probability between an

\begin{center}
\begin{tabular}{p{0.1\textwidth}p{0.83\textwidth}}
\multicolumn{2}{c}{\textbf{Definitions}}\\[1pt]
\hline
  $\lfloor x \rfloor$ & $\max\{n \in \mathbb{N}: n \le x\}$ for $x \in \mathbb{R}$ \\
  $\lceil x \rceil$ & $\min\{n \in \mathbb{N}: n \ge x\}$ for $x \in \mathbb{R}$ \\
  $m^*$ & $t_{\text{b}} / \text{gcd}(t_{\text{m}}, t_{\text{b}})$ \\ [5pt]
  \multicolumn{2}{c}{\textbf{Parameters}}\\[1pt]
\hline
  $p_{\text{m}}^0$ & base efficiency, i.e., the success probability of an entanglement generation attempt between two end nodes (or an end and border node) at zero separation \\
  $d_{\text{m}'}$ & distance between an end node and a nearby metropolitan hub \\
  $p_{\text{m}'}$ & success probability of an entanglement generation attempt between two end nodes in an MN, i.e., ${p_{\text{m}'} \!=\! p_\text{m}^0 10^{-2 \alpha d_{\text{m}'}/10}}$; see~\eqref{eq:p_m_metro} \\
  $p_{\text{m}}$ & success probability of an entanglement generation attempt between an end node and a neighbouring border node, i.e., ${p_{\text{m}} \!=\! p_\text{m}^0 10^{-\alpha d_{\text{m}'}/10}}$; see~\eqref{eq:p_m_int} \\
  $t_{\text{m}'}$ & duration of an entanglement generation attempt between two end nodes in an MN \\
  $t_{\text{m}}$ & duration of an entanglement generation attempt between an end node and a neighbouring border node \\
  $t_{\text{prep}}$ & (constant) average duration of a photon generation from an end node or border node \\
  $f_{\text{m}'}$ & fidelity of a freshly generated entangled link between two end nodes in an MN, $f_{\text{m}'} = (1+3w_{\text{m}'})/4$, where $w_{\text{m}'}$ is the corresponding Werner parameter
  \\
  $f_{\text{m}}$ & fidelity of a freshly generated entangled link between an end and a neighbouring border node, $f_{\text{m}}=(1+3w_\text{m})/4$ where $w_\text{m}$ is the corresponding Werner parameter
  \\
  $p_{\text{b}}$ & success probability of an entanglement generation attempt in the backbone \\
  $t_{\text{b}}$ & (constant) duration of an entanglement generation attempt in the backbone \\
  $f_{\text{b}}$ & fidelity of a freshly generated entangled link in the backbone, equals $(1+3w_\text{b})/4$ where $w_\text{b}$ is the corresponding Werner parameter 
  \\
  $t_\text{coh}$ & memory coherence time of an end or border node\\
  $t_\text{cut}$ & cut-off time, i.e., the maximum allowable time between the earliest and latest generated entangled links \\
  $t_\text{msg}$ & communication time between an end node and the farthest border node \\
  \multicolumn{2}{c}{\textbf{Variables}}\\[1pt]
\hline
  $M_1$ & number of attempts until successful entanglement generation between an end node in MN $1$ and neighbouring border node, $M_1 \sim \text{Geo}(p_{\text{m}})$ \\
  $X_1$ & time until successful entanglement generation between an end node in MN $1$ and neighbouring border node, i.e., $X_1 = t_{\text{m}} M_1$ \\
  $M_2$ & number of attempts until successful entanglement generation between an end node in MN $2$ and neighbouring border node, $M_2 \sim \text{Geo}(p_{\text{m}})$ \\
  $X_2$ & time until successful entanglement generation between an end node in MN $2$ and neighbouring border node, i.e., $X_2 = t_{\text{m}} M_2$ \\
  $M_\text{b}$ & number of attempts until successful entanglement generation in the backbone, $M_b \sim \text{Geo}(p_{\text{b}})$ \\
  $X_\text{b}$ & time until successful entanglement generation in the backbone, i.e., $X_\text{b} = t_{\text{b}} M_\text{b}$ \\
  \hline
\end{tabular}
\captionsetup{hypcap=false}
\captionof{table}{List of notations.}
\vspace{-0pt}
\label{tab:parameter_definitions}
\end{center}
end node and a border node $p_\text{m}$, the memory coherence time of nodes $t_{\text{coh}}$, the fidelity of a freshly generated link between an end and a border node $f_\text{m}$, entanglement-generation probability in the backbone $p_{\text{b}}$, and the fidelity of a freshly generated link in backbone $f_\text{b}$.
The definitions of all parameters are summarised in Tab.~\ref{tab:parameter_definitions}.
Note that while we focus on this symmetric setting for clarity, the framework naturally extends to more general asymmetric networks with heterogeneous nodes and link distances, at the cost of introducing additional parameters.
Throughout the analysis, we assume time is slotted and derive all performance metrics accordingly.

We now state the assumptions necessary to model the processes that influence teleportation in the intercity network.
This helps us derive the performance metrics, i.e., the expected fidelity of teleportation and corresponding rate.
\begin{enumerate}[label={\textbf{A\arabic*}},nolistsep,leftmargin=*]
\item \label{assumption:werner_state} \textbf{Entangled state description}:
We model all entangled links in the MN and backbone as Werner states~\cite{werner1989quantum}:
\begin{equation}
\label{eq:werner_state}
    \rho = w\ket{\Phi^+}\bra{\Phi^+} + \left(1-w\right)\frac{\mathbb{I}_4}{4}~,
\end{equation}
where ${\ket{\Phi^+} \!=\! (\ket{00} \!+\! \ket{11})/\sqrt{2}}$ is a maximally entangled Bell state, $w\in[0,1]$ is the corresponding Werner parameter, and $\mathbb{I}_4$ is the $4 \times 4$ identity matrix.
The fidelity of $\rho$ with respect to $\ket{\Phi^+}$ can be seen as ${{(1+3w)}/{4}}$.

While physically realised remote entangled states are generally not of the Werner form~\cite{simon2007quantum, zwerger2017quantum}, we model link-level entanglements as Werner states due to the following reasons.
First, when two Werner states are swapped, the resulting state is also a Werner state with the corresponding Werner parameter given by the product of those of the initial states~\cite{Munro2015}, making analysis easier.
Further, any bipartite state can be transformed into a Werner state by twirling, i.e., by applying transformations uniformly at random from a set of operations that involve identical rotations on each qubit~\cite{bennett1996mixed}. 

\item \label{assumption:depolarising_noise} \textbf{Noise in memory}: 
During entanglement generation across multiple links or in the presence of classical communication delays, the qubits stored in the memory undergo decoherence.
We model this noise as a depolarising channel acting on the stored qubit.
Specifically, the evolution of a quantum state ${\rho}_A$ undergoing a depolarising channel $\mathcal{E}_t$ over a storage time $t$ in memory $A$ is given by
\vspace{-8pt}
\begin{equation}
\label{eq:depolarising_channel}
    \mathcal{E}_t({\rho}_A) = e^{-t/t_{\text{coh}}}{\rho}_A + \left(1-e^{-t/t_{\text{coh}}}\right) \frac{\mathbb{I}_2}{2}~,
\end{equation}
where $t_{\text{coh}}$ denotes the memory coherence time.
For a Werner state ${\rho}_{AB}$ which has maximally mixed marginals, when both qubits undergo depolarising noise for a duration $t$, the resulting state is given by
\begin{equation}
     \mathcal{E}_t \otimes \mathcal{E}_t({\rho}_{AB}) = e^{-2t/t_{\text{coh}}}{\rho}_{AB} + \left(1-e^{-2t/t_{\text{coh}}}\right) \frac{\mathbb{I}_4}{4}~.
\end{equation}
\looseness=-1
It is worth noting that in certain platforms, such as trapped-ion systems, memory noise can be described more accurately by a correlated dephasing process with a Gaussian temporal profile characterised by the coherence time~\cite{Krutyanskiy2023TelecomWavelength}.
In contrast, we model the memory noise as an exponential depolarising channel as defined in~\eqref{eq:depolarising_channel}.
The choice is primamily motivated by analytical tractability and platform independence.
Specifically, depolarising noise preserves the Werner form of states~\eqref{eq:werner_state}, thereby simplifying the analysis~\cite{goodenough2025noise}.
Moreover, while the depolarising channel does not accurately capture the correlated and non-Markovian nature of Gaussian dephasing, it provides a convenient and broadly comparable benchmark that facilitates general conclusions across architectures.
A more detailed investigation of how correlated Gaussian dephasing influences network performance would be valuable but lies beyond the scope of this work.
\item \label{assumption:entangling_metro} \looseness = -1 \textbf{Entanglement generation between end nodes in a metropolitan network}:
We adopt the double-click HEG protocol~\cite{Barrett2005DoubleClick} in our model.
Specifically, to entangle two end nodes in an MN, e.g., $\text{P}_1$ and $\text{P}_2$ in MN $1$ (Fig.~\ref{fig:intercity_network_diagram}), both nodes first generate matter-photon entanglement.
Photons from both nodes then travel to the midpoint metropolitan hub, where a photonic BSM is performed.
The BSM succeeds with a certain probability, and the result travels back to both end nodes, heralding the signal whether entanglement is successfully generated.
We further assume that entanglement generation is attempted at fixed intervals defined as the cycle time.
The cycle time and success probability depend on the fibre length between the nodes.

We denote the distance between an end node and its nearest hub as $d_{\text{m}'}$, such that the fibre length between two end nodes in an MN is given by $2d_{\text{m}'}$.
The entanglement-generation probability is attenuated by fibre loss and is modelled as an exponential factor $\eta(\cdot)$~\cite{da2024requirements}.
Consequently, the entanglement-generation probability between two nodes is given by:
\begingroup
    \setlength{\abovedisplayskip}{6pt}   
    \setlength{\belowdisplayskip}{6pt}   
    \setlength{\abovedisplayshortskip}{2pt}
    \setlength{\belowdisplayshortskip}{2pt}
    \begin{align}
    \label{eq:p_m_metro}
        & p_{\text{m}'} = p_\text{m}^0 ~\eta(2d_{\text{m}'}) = p_\text{m}^0 10^{-2 \alpha d_{\text{m}'}/10}~,
    \end{align}
\endgroup
where $\alpha \!=\! 0.2\,\text{km}^{-1}$ is the attenuation coefficient~\cite{sangouard2009quantum} for typical telecommunication optical fibres in the optical wavelength range around $1550\,\text{nm}$ and $p_\text{m}^0$, named base efficiency, captures all distance-independent factors, including photon source brightness, coupling losses, and visibility.

To derive the cycle time, i.e., the time for each entanglement generation attempt between the end nodes, we observe that the photon propagation time from an end node to the hub and the heralding signal propagation time from the hub to the end node are each given by $t_\text{m}^\text{class} \!=\! d_{\text{m}'}/c$, where $c$ is the speed of light in fibre.
To account for delays apart from the propagation times, we introduce $t_{\text{prep}}$ as the average duration of all local (i.e., node-level) experimental overheads, including photon generation attempts. 
This can, for example, represent intrinsic delays between entanglement generation attempts in trapped-ion systems, which arise from necessary procedures such as laser excitation, qubit initialisation, cooling, optical pumping, and system control latencies~\cite{Krutyanskiy2023TelecomWavelength}.
Since the photon generation attempts at the end nodes are synchronised, the cycle time is given by
\begin{align} 
\label{eq:t_m_metro}
    t_{\text{m}'} = t_{\text{prep}} + 2t_{\text{m}}^{\text{class}}~.
\end{align}

Note that we express the terms $t_{\text{m}'}$, $t_\text{prep}$, and $t_\text{m}^\text{class}$ in units chosen such that their numerical values correspond to integers.  
Finally, in our discrete-time analysis, the number of cycles $M'$ to successfully establish a link is geometrically distributed:
\begingroup
    \setlength{\abovedisplayskip}{5pt}   
    \setlength{\belowdisplayskip}{5pt}   
    \setlength{\abovedisplayshortskip}{2pt}
    \setlength{\belowdisplayshortskip}{2pt}
    \begin{align}
    \label{eq:attempts_unitl_success_metro}
        M' \sim \text{Geo}({p}_{\text{m}'})~,
    \end{align}
\endgroup
whereas a freshly generated link is assumed to be in the Werner form~\eqref{eq:werner_state} with fidelity $f_{\text{m}'}$ and corresponding Werner parameter $w_{\text{m}'}$ such that $f_{\text{m'}}\!=\!(1\!+\!3w_{\text{m}'})/4$ is the fidelity.
Note that imperfections such as reduced two-photon interference visibility, detector inefficiencies, or dark counts contribute to noise and consequently lower the fidelity.
We do not model these effects explicitly and instead incorporate them into the effective fidelity $f_\text{m'}$.
This approach similarly applies to the assumptions outlined in~\ref{assumption:entangling_metro_int}--\ref{assumption:entangling_backbone}.

\item \label{assumption:entangling_metro_int} \textbf{Entanglement generation between an end node and a neighbouring border node}:
Similar to~\ref{assumption:entangling_metro}, we adopt a discrete-time model for entanglement generation, now between an end node and a neighbouring border node, e.g., between $\text{P}_1$ and $\text{B}_1$.
Since $\text{B}_1$ is located right next to the hub $\text{H}_1$, only the photon from $\text{P}_1$ traverses the distance $d_{\text{m}'}$ to $\text{H}_1$, while the photon from $\text{B}_1$ has negligible travel time.
Since the border nodes are assumed to have the same properties as the end nodes, the entanglement-generation probability between an end node and a border node can be obtained as 
\begingroup
    \setlength{\abovedisplayskip}{6pt}   
    \setlength{\belowdisplayskip}{6pt}   
    \setlength{\abovedisplayshortskip}{2pt}
    \setlength{\belowdisplayshortskip}{2pt}
    \begin{align}
    \label{eq:p_m_int}
        p_\text{m} = p_\text{m}^0\eta(d_{\text{m}'}) = p_\text{m}^0 10^{-\alpha d_{\text{m}'}/10}~.
    \end{align}
\endgroup
Furthermore, the photon travel time from an end node to the hub and the heralded signal propagation time from the hub to the end node are each given by $t_\text{m}^\text{class} \!=\! d_{\text{m}'}/c$.
In comparison to this, the duration for both the photon from the border node to the hub and the return heralding signal from the hub to the border node is negligible.
For properly synchronised entanglement generation attempts, considering that the average duration of all local experimental overheads is as in~\ref{assumption:entangling_metro}, the cycle time is given by
\begingroup
    \setlength{\abovedisplayskip}{6pt}   
    \setlength{\belowdisplayskip}{6pt}   
    \setlength{\abovedisplayshortskip}{2pt}
    \setlength{\belowdisplayshortskip}{2pt}
    \begin{align}
    \label{eq:t_m_int}
        t_{\text{m}} = t_{\text{prep}} + 2t_{\text{m}}^{\text{class}}~. 
    \end{align}
\endgroup
Note that the quantities $t_{\text{m}}$, $t_\text{prep}$, and $t_\text{m}^\text{class}$ are expressed in units chosen such that their numerical values are integers.  
Finally, in an entanglement generation attempt, the number of trials $M_i$ required to establish an entangled link $i$ follows a geometric distribution:
\begingroup
    \setlength{\abovedisplayskip}{6pt}   
    \setlength{\belowdisplayskip}{6pt}   
    \setlength{\abovedisplayshortskip}{2pt}
    \setlength{\belowdisplayshortskip}{2pt}
    \begin{align}
        M_i \sim \text{Geo}({p}_{\text{m}})~, \quad i \in \{1,2\}~.
    \end{align}
\endgroup
Once an attempt succeeds, we assume that the generated link is of the Werner form~\eqref{eq:werner_state} with fidelity $f_\text{m}$ and corresponding Werner parameter $w_\text{m}$ such that $f_\text{m}\!=\!(1\!+\!3w_\text{m})/4$.

\item \label{assumption:entangling_backbone} \textbf{Entanglement generation in the backbone}:
We model entanglement generation in the backbone as a monolithic process, where each attempt takes a fixed duration $t_{\text{b}}$ and succeeds with probability $p_{\text{b}}$.
This abstraction captures the behaviour of quantum systems that can store spin qubits in memory, such as trapped ions, colour centres, as well as systems capable of storing photonic qubits in memory for a finite time window, such as atomic ensemble-based memories.
Thus, the number of attempts $M_b$ required for successful entanglement generation in the backbone and the corresponding generation time $X_b$ are given by
\begingroup
    \setlength{\abovedisplayskip}{6pt}   
    \setlength{\belowdisplayskip}{5pt}   
    \setlength{\abovedisplayshortskip}{2pt}
    \setlength{\belowdisplayshortskip}{2pt}
    \begin{align}
        M_\text{b} \sim \text{Geo}(p_{\text{b}})~, \quad X_\text{b}:=t_{\text{b}} M_\text{b}~,
    \end{align}
\endgroup
\vspace{-4pt}
and the resulting entanglement generation rate is 
\begin{align}
\label{eq:backbone_rate_definition}
    R_\text{b} = 1/\mathbb{E}(X_\text{b}) = p_{\text{b}}/t_{\text{b}}~.
\end{align}
Since we consider that the border nodes share the same properties as the end nodes, we again assign $t_{\text{prep}}$ as the average duration of a single photon generation attempt at a border node, including local experimental delays (see~\ref{assumption:entangling_metro}).
Thus, we define the cycle time $t_\text{b}$ as the sum of the time it takes to create a single photon from the border node, i.e., $t_{\text{prep}}$, added with the one-way classical communication time between the two border nodes across the backbone  $t_{\text{b}}^{\text{class}}$, i.e.,
\begingroup
    \setlength{\abovedisplayskip}{6pt}   
    \setlength{\belowdisplayskip}{6pt}   
    \setlength{\abovedisplayshortskip}{2pt}
    \setlength{\belowdisplayshortskip}{2pt}
    \begin{equation} \label{eq:t_b_total}
        t_{\text{b}} = t_{\text{prep}} + t_{\text{b}}^{\text{class}}~.
    \end{equation}
\endgroup
For evaluation, we use an estimate of the entanglement generation rate in the backbone from a model of trapped-ion nodes connected via repeaters composed of ensemble-based memory~\cite{tissot2025hybrid}.
Thus, we can obtain the estimate for $p_\text{b}$ by using~\eqref{eq:backbone_rate_definition} and~\eqref{eq:t_b_total}.
Note that the quantities $t_{\text{b}}$, $t_\text{prep}$, and $t_\text{b}^\text{class}$ are expressed in units chosen such that their numerical values are integers.  

We also assume that the entangled link produced in the backbone can be transferred to the border node memories without loss.
Consequently, the resulting link between the border nodes can also be described by a Werner state~\eqref{eq:werner_state} with a Werner parameter $w_\text{b}$ and corresponding fidelity $f_\text{b}$ such that $f_\text{b} \!=\! (1\!+\!3w_\text{b})/4$.

\item \label{assumption:entanglement_swap_asap} \textbf{Entanglement swapping}:
End-to-end entanglement across the IN is created in the following steps: first, entangled links between end and border nodes and in the backbone are created in parallel, which we refer to as elementary links.
We adopt the swap-ASAP policy~\cite{kamin2023exact}, i.e., as soon as a border node holds entangled links on both sides, it performs an entanglement swap~\cite{bennett1993teleporting}.
Two such swaps on two border nodes complete the creation of an end-to-end entanglement in the IN.
We assume this swap is realised deterministically, since we assume that the border nodes can implement swap using quantum gates and measurements on the qubits in memory.
Furthermore, we assume that the swap process is noiseless.
This is justified in the context of long-range quantum teleportation, where qubit decoherence due to waiting times during remote entanglement generation is the dominant source of noise.
\item \label{assumption:node_communication} \textbf{Classical communication time between nodes}:
At the beginning of each end-to-end entanglement generation round, all participating nodes in the network are synchronised and prepared to initiate entanglement generation with their immediate neighbours.
In an elementary link generation attempt, whether between adjacent end nodes or between an end node and a border node, a BSM is performed on the incoming photons at the hub.
Based on the outcome, a classical heralding signal is issued to both nodes to indicate the success or failure of the attempt.
Similarly, upon completion of the final swap, a classical message is sent to both end nodes, confirming successful end-to-end entanglement generation.
Since the border node is located right next to the metropolitan hub, the time to send a message from a border node to a neighbouring end node is $t_{\text{m}}^{\text{class}}$.
Denoting by $t_{\text{b}}^{\text{class}}$ the time for propagation of a message across the backbone~\ref{assumption:entangling_backbone}, the total communication delay associated with the final notification is given by
\begingroup
    \setlength{\abovedisplayskip}{0pt}   
    \setlength{\belowdisplayskip}{6pt}   
    \setlength{\abovedisplayshortskip}{2pt}
    \setlength{\belowdisplayshortskip}{2pt}
    \begin{align}
    \label{eq:t_msg}
        t_\text{msg} = t_\text{m}^\text{class} + t_\text{b}^\text{class}~.
    \end{align}
\endgroup
This communication time is shown in Fig.~\ref{fig:intercity_network_time}.
Consequently, the communication time across the IN between two end nodes belonging to different metropolitan networks, e.g., $\text{P}_1$ and $\text{P}_3$ (see Fig.~\ref{fig:intercity_network_time}), is given by
\begingroup
    \setlength{\abovedisplayskip}{0pt}   
    \setlength{\belowdisplayskip}{6pt}   
    \setlength{\abovedisplayshortskip}{2pt}
    \setlength{\belowdisplayshortskip}{2pt}
    \begin{align}
    \label{eq:t_class_int}
        t_\text{int}^\text{class} = 2 t_\text{m}^\text{class} + t_\text{b}^\text{class}~.
    \end{align}
\endgroup

\begin{figure}[t]
    \centering
    \includegraphics[width=0.70\linewidth]{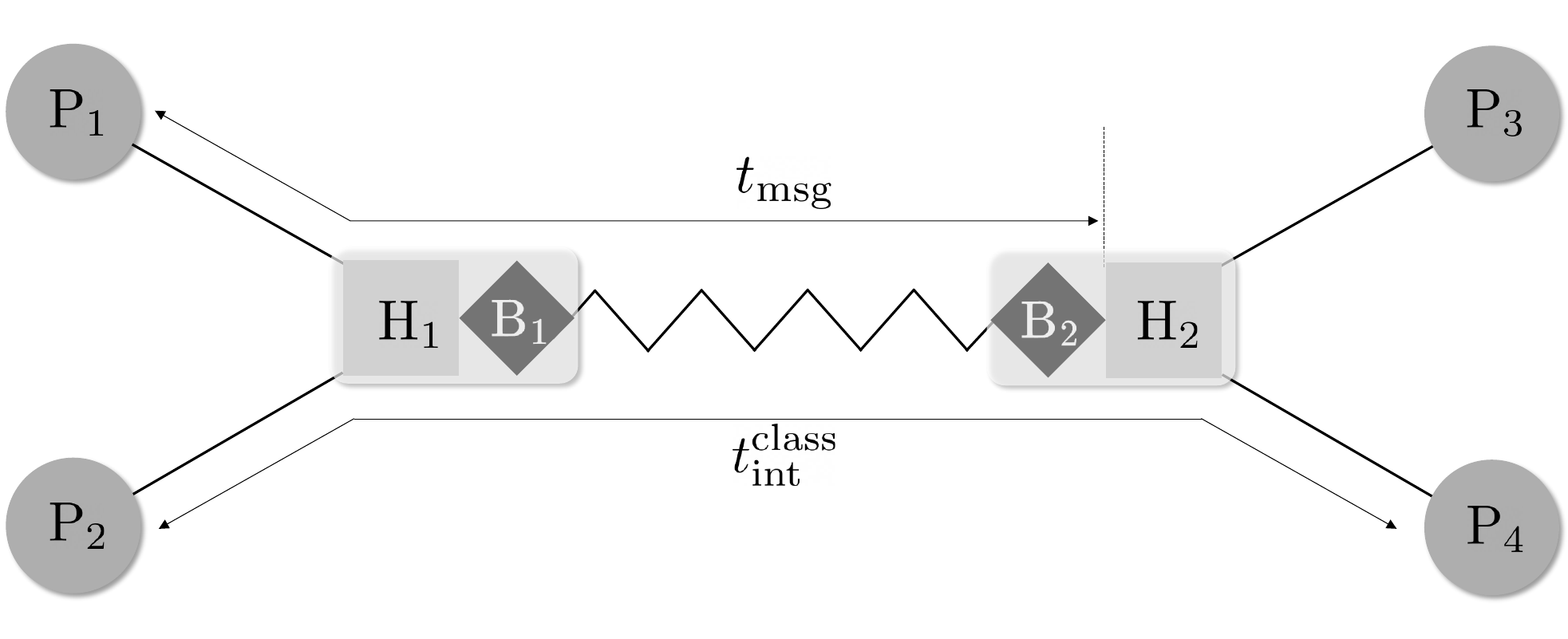}
    \caption{Classical communication time between nodes in an intercity quantum network.
    }
    \vspace{-10pt}
    \label{fig:intercity_network_time}
\end{figure}

\item \label{assumption:cutoff_time} \textbf{Cut-off time}:
Due to the probabilistic nature of elementary link generation, 
links that are generated earlier must wait until a neighbouring link is formed to start entanglement swapping.
During this period, links decay in the memory, which degrades the overall fidelity of the final link.
To mitigate the effects of decoherence, a common strategy is to employ a cut-off mechanism~\cite{collins2007multiplexed} 
where old degraded links are discarded and replaced with freshly prepared links.
In our model for single-shot entanglement generation across the intercity network, we implement this strategy as follows: initiate entanglement generation trials by simultaneously starting generation attempts across all elementary links.
A timer activates upon the successful establishment of the first link.
If the remaining links are established within a predetermined cut-off time $t_\text{cut}$, we perform entanglement swaps between adjacent links as they become available, ultimately yielding an end-to-end link after the final swap.
Conversely, if all required links are not created within $t_\text{cut}$, all active links are discarded, and the entanglement generation process restarts across the network.
This entire approach is commonly referred to as a \emph{global cut-off} strategy in literature, e.g.~\cite{goodenough2025noise}.
As mentioned earlier, we define a \emph{round} as a whole trial to establish an end-to-end link, composed of multiple attempts for elementary link generation.
A round is considered successful if the time interval between the earliest and latest generated elementary links remains below a specified cut-off time.

While a lower cut-off time $t_\text{cut}$ improves the fidelity of the end-to-end link by reducing the storage time of qubits, it leads to a reduction in the corresponding entanglement generation rate, as more links are discarded. 
Following~\cite{Avis2023DelftEindhoven}, we vary $t_\text{cut}$ within the range
\begingroup
    \setlength{\abovedisplayskip}{6pt}   
    \setlength{\belowdisplayskip}{6pt}   
    \setlength{\abovedisplayshortskip}{2pt}
    \setlength{\belowdisplayshortskip}{2pt}
    \begin{align}
    \label{eq:t_cut_range_naive}
        t_\text{cut} / t_\text{coh} \in [0.01, 1]~,
    \end{align}
\endgroup
where $t_\text{coh}$ denotes the memory coherence time.
Moreover, the cut-off time must be sufficiently large to accommodate both the time required for entanglement generation attempts on individual links ($t_\text{m}, t_\text{b}$) and the classical communication delays ($t_\text{msg}$).
In this work, time is treated as a discrete variable, taking values in the set of natural numbers.
Thus, we obtain the following range for the cut-off time:
\begingroup
    \setlength{\abovedisplayskip}{6pt}   
    \setlength{\belowdisplayskip}{6pt}   
    \setlength{\abovedisplayshortskip}{2pt}
    \setlength{\belowdisplayshortskip}{2pt}
    \begin{align} 
    \label{eq:t_cut_range}
    t_\text{cut} \in \left( \max\left\{t_\text{m}, t_\text{b}, t_\text{msg}, 0.01\,t_\text{coh}\right\},\ t_\text{coh} \right] \cap \mathbb{N} =: T_\text{cut}~.
    \end{align}
\endgroup
Note that our analysis in Sec.~\ref{sec:derivation_rate_IN} only requires the lower bound of $t_\text{cut}$, but not the upper bound.
Also, cut-off times exclusively apply to HEG between end nodes involving multiple links and the data qubits are never discarded.

Note that, as per our cut-off strategy for single-shot entanglement generation in the intercity network, the decision to restart is determined independently at each link based on its generation attempt.
However, it is possible to enhance the end-to-end entanglement generation rate by adopting more flexible cut-off strategies, 
for instance, one may discard and restart only those links whose generation time exceeds $t_\text{cut}$, while retaining other successfully generated links. 
This process continues until all three links are generated within a $t_\text{cut}$ time window.
Analysing this cut-off policy is more involved and is left as future work.
\item \label{assumption:instantaneous_qubit_preparation_local_operation} \textbf{Data qubit preparation, local gate operation, and qubit measurement}:
Although the time scales associated with qubit initialisation, gate operations, and measurements vary across different physical platforms, it is orders of magnitude shorter than the entanglement generation times and classical message propagation times in networks where nodes are separated by hundreds of kilometres.
For example, durations for local operations such as qubit preparation, gate operation, and measurement in trapped-ion systems typically lie in the range of hundreds of nanoseconds to a few microseconds~\cite{Avis2023DelftEindhoven, Bruzewicz2019}.
For colour centres, these operations can range from a few nanoseconds to a few microseconds ~\cite{humphreys2018deterministic, pfaff2014unconditional, taminiau2014universal, reiserer2016robust, kalb2017entanglement}.
In contrast, classical communication between spatially separated nodes over hundreds of kilometres incurs delays on the order of milliseconds, constrained by the speed of light in optical fibre (approx. $5\,\mu\text{s/km}$).
Furthermore, entanglement generation between distant nodes typically requires multiple attempts due to the probabilistic photon generation process, photon loss through fibre, detector inefficiencies, and the probabilistic nature of heralding, which significantly increases the overall time required for entanglement generation.
Due to this pronounced separation of time scales, we assume that the local operations are instantaneous relative to the long-distance entanglement generation time.
Note that in the HEG process, we account for the local experimental overheads separately in~\ref{assumption:entangling_metro} and~\ref{assumption:entangling_metro_int} via~\eqref{eq:t_m_metro} and~\eqref{eq:t_m_int}, which can range from several hundred microseconds to milliseconds per attempt~\cite{Krutyanskiy2023TelecomWavelength}. 

\end{enumerate}

\noindent Equipped with these assumptions, we now formally describe~\ref{Q:1}--\ref{Q:2c}.
\vspace{-10pt}

\section{Objectives and Methods}
\label{sec:objectives_and_methods}
\everypar{\looseness=-1}
In this section, we present the methodology for determining the necessary requirements to achieve teleportation with a desired level of quality in the network shown in Fig.~\ref{fig:intercity_network_diagram}.
Specifically, we address \ref{Q:1}-\ref{Q:2c}.
We observe that the teleportation rate and the expected fidelities of ER and QR teleportation in the MNs and IN depend on the hardware parameters as shown in Tab.~\ref{tab:quantities_and_parameters}.
In addition, we include the \textit{non-hardware} parameter cut-off time $t_\text{cut}$ (in parentheses) to show its influence on the performance-relevant metrics.
\begin{table}[ht]
\centering
\caption{Description of quantities of interest and the influencing parameters.}
\label{tab:quantities_and_parameters}
\begin{tabular}{ll}
\hline
\textbf{Quantities of interest} & \textbf{Influencing parameters} \\
\hline
\textit{Metropolitan network (MN)}: \\
Teleportation rate $R_\text{m}$         & $p_\text{m}^0$ \\
Expected ER teleportation fidelity $\mathbb{E}(F_\text{m}^\text{ER})$      & $p_\text{m}^0$, $t_\text{coh}$, $f_{\text{m}'}$ \\
Expected QR teleportation fidelity $\mathbb{E}(F_\text{m}^\text{QR})$      & $p_\text{m}^0$, $t_\text{coh}$, $f_{\text{m}'}$ \\
\hline
\textit{Intercity network (IN)}: \\
Teleportation rate $R_\text{int}$      & $p_\text{m}^0$, $p_\text{b}$, ($t_\text{cut}$) \\
Expected ER teleportation fidelity $\mathbb{E}(F_\text{int}^\text{ER})$      & $p_\text{m}^0$, $t_\text{coh}$,  $f_\text{m}$, $p_\text{b}$, $f_\text{b}$, ($t_\text{cut}$) \\
Expected QR teleportation fidelity $\mathbb{E}(F_\text{int}^\text{QR})$      & $p_\text{m}^0$, $t_\text{coh}$, $f_\text{m}$, $p_\text{b}$, $f_\text{b}$, ($t_\text{cut}$) \\
\hline
\end{tabular}
\vspace{-7pt}
\end{table}
Note that under our assumption of instantaneous local operations, the teleportation rates are identical for both teleportation types in the MN and IN.
Our primary objective is to identify the combination of hardware parameters required to achieve the target teleportation fidelity $2/3$ in the MN and IN.
Within the feasible parameter space, we seek to determine the minimal hardware improvement required over the state-of-the-art parameters, with respect to the cost function introduced in the next section.

\subsection{Optimisation Framework and Hardware Cost}
\label{sec:optimization_and_cost_function}

We recall from~\ref{Q:1}--\ref{Q:2c} that our first objective is to investigate whether the state-of-the-art hardware can achieve the target teleportation fidelity $f_\text{target} \!:=\! 2/3$.
To that end, we characterise the feasible parameter range that helps attain this threshold.
Let $\vec\lambda$ denote the generic hardware parameter vector determining the expected fidelity of teleportation $\mathbb{E}(F)$, where ${F \in \{F_\text{m}^\text{ER}, F_\text{m}^\text{QR}, F_\text{int}^\text{ER}, F_\text{int}^\text{QR}\}}$; see Tab.~\ref{tab:quantities_and_parameters} for the explicit forms of $\vec\lambda$ and $F$.
In addition, Tab.~\ref{tab:quantities_and_parameters} shows the dependence of $F$ on the \textit{optional} non-hardware parameter $t_\text{cut}$.
Denoting the feasible hardware parameter space as ${\Lambda}$ and the range of $t_\text{cut}$ as $T_\text{cut}$, the desired range of hardware parameters can be written as
\begingroup
    \setlength{\abovedisplayskip}{6pt}   
    \setlength{\belowdisplayskip}{6pt}   
    \setlength{\abovedisplayshortskip}{2pt}
    \setlength{\belowdisplayshortskip}{2pt}
    \begin{align}\label{eq:Lambda+}
        \Lambda_{+} &:= \{\vec\lambda \in {\Lambda} \!: \max_{t_\text{cut} \in T_\text{cut}} \mathbb{E}\big(F(\vec\lambda \mid t_\text{cut})\big) \geq f_\text{target}\}\,,~\text{where}~{\Lambda}:= \times_{i=1}^{m} [\underaccent{=}{\lambda}^{(i)}, \overline{\lambda}^{(i)}]~,
    \end{align}
\endgroup
with $\underaccent{=}{\lambda}^{(i)}$, $\overline{\lambda}^{(i)}$, and $m$ respectively denoting the 
lowest possible value of the $i$th hardware parameter, its optimistic value, and the total number of hardware parameters influencing the expected teleportation fidelity $\mathbb{E}(F)$.
The optimistic values represent projections by experimental groups for achievable performance in the near future; see Sec.~\ref{sec:baseline_and_optimistic_parameters}.
Since these projections are based on current experimental status and possible room for improvements, it remains uncertain a priori whether these optimistic values are sufficient to achieve the target teleportation performance.
Therefore, we impose $\overline{\lambda}^{(i)}$ as an upper bound in~\eqref{eq:Lambda+} to identify, if it exists, the most accessible point in the parameter space that meets the required criteria.
Note that the vertical bar preceding $t_\text{cut}$ in the expression for $F$ in~\eqref{eq:Lambda+} highlights that it is a non-hardware parameter, relevant only to teleportation in the intercity network. 
Furthermore, the specific form of $T_\text{cut}$ is provided in~\eqref{eq:t_cut_range}.
Furthermore, we define the baseline parameters, denoted as $\underline{\lambda}^{(i)}$, to represent the current state-of-the-art values such that $\underline{\lambda}^{(i)} \in [\underaccent{=}{\lambda}^{(i)},\overline{\lambda}^{(i)}]$, for all $i$.
In our evaluation, we specifically adopt the baseline and optimistic parameter sets corresponding to experiments with trapped-ion and ensemble-based memories; see Sec.~\ref{sec:baseline_and_optimistic_parameters} for more details.

It is evident that if the state-of-the-art (i.e., baseline) values of the hardware parameter lie in the desired range $\Lambda_{+}$, which achieves teleportation with target fidelity, no further hardware improvement is required for attaining our objective.
Otherwise, we aim to find the \emph{minimal} hardware improvements necessary, in the following space
\begingroup
    \setlength{\abovedisplayskip}{6pt}   
    \setlength{\belowdisplayskip}{6pt}   
    \setlength{\abovedisplayshortskip}{2pt}
    \setlength{\belowdisplayshortskip}{2pt}
    \begin{align}\label{eq:Lambda++}
        \Lambda_{++} &:= \{\vec\lambda \in {\Lambda}_\text{base} \!: \max_{t_\text{cut} \in T_\text{cut}} \mathbb{E}\big(F(\vec\lambda \mid t_\text{cut})\big) \geq f_\text{target}\}\,,~\text{where}~{\Lambda}_\text{base} := \times_{i=1}^{m} [\underline{\lambda}^{(i)}, \overline{\lambda}^{(i)}]~.
    \end{align}
\endgroup
Then, the \emph{minimality} of hardware improvement is defined in terms of the cost function~\cite{Avis2023DelftEindhoven, Ferreira_da_Silva_2021, da2024requirements} 
denoted  $h$, which leads to the set of resulting hardware parameter values as $\Lambda'_{*}$. 
That is,
\begingroup
    \setlength{\abovedisplayskip}{6pt}   
    \setlength{\belowdisplayskip}{6pt}   
    \setlength{\abovedisplayshortskip}{2pt}
    \setlength{\belowdisplayshortskip}{2pt}
    \begin{align}
        \label{eq:definition_optimal_point_set}
        \Lambda'_{*} := \{ \vec\lambda \in \Lambda_{++}\!: h(\vec\lambda,\underline{\vec{\lambda}}) = \min\limits_{\vec\lambda' \in \Lambda_{++}} h(\vec\lambda',\underline{\vec{\lambda}}) \}\,,~\text{where}~ h(\vec\lambda,\underline{\vec{\lambda}}) &:= \sum_{i=1}^m \text{IF}(\lambda^{(i)}\!, \underline{\lambda}^{(i)})~, \\
        \text{IF}(\lambda^{(i)}\!, \underline{\lambda}^{(i)}) &:= \frac{\ln(p^{(i)}_{NI}(\underline{\lambda}^{(i)}))}{\ln(p^{(i)}_{NI}(\lambda^{(i)}))}~, \label{eq:improvement_factor}
    \end{align}
\endgroup
and the functions $p^{(i)}_{NI},~ i \!\in\! [m]$ denote the \textit{probability of no-imperfection}, a metric~\cite{Coopmans2021Netsquid} that maps hardware parameters with varying ranges to $[0,1]$.
In this mapping, the functional value $1$ corresponds to ideal performance such as infinite coherence time, perfect link fidelity, and $100\%$ base efficiency.
We provide the specific forms of $p_{NI}$s in Tab.~\ref{tab:no_imperfection_probability_baseline_and_optimistic}.
Moreover, to quantify the degree of improvement required for each parameter relative to its baseline, we define the improvement factor (IF) as in~\eqref{eq:improvement_factor} and define the hardware cost $h$ as the sum of the improvement factors across all parameters.
It is important to note that the hardware cost computed through this methodology serves as a measure of the technical difficulty in enhancing hardware parameters to the specified levels, rather than representing any financial costs.

Observe that the definition of the set $\Lambda'_{*}$ in~\eqref{eq:definition_optimal_point_set} involves solving a constrained optimisation problem, where the constraint ${\vec\lambda' \in \Lambda_{++}}$ is imposed to ensure that the target performance in~\eqref{eq:Lambda+} is met.
Moreover, restricting the search space to $\Lambda_{++}$ rather than $\Lambda_{+}$ ensures that the optimisation process seeks only parameter improvements, preventing any parameter from being reduced below its baseline value.
We adopt the approach called scalarisation~\cite{pascoletti1984scalarizing, schaffer1985some} to convert this into an unconstrained optimisation problem, which is numerically convenient to solve. 
Specifically, we add the constraint as a large penalty term to the objective function $h$, yielding a new objective function $c$, which, due to the presence of the penalty term, \emph{may} also depend on the non-hardware parameter $t_\text{cut}$:
\begingroup
    \setlength{\abovedisplayskip}{6pt}   
    \setlength{\belowdisplayskip}{6pt}   
    \setlength{\abovedisplayshortskip}{2pt}
    \setlength{\belowdisplayshortskip}{2pt}
    \begin{align}
        c(\vec\lambda,\underline{\vec{\lambda}},F\mid t_\text{cut}) := \omega_1\big(1 \!+\! \big(f_\text{target} \!- \!\mathbb{E}\big(F(\vec\lambda \mid t_\text{cut})\big) \big)^2 \big)\, \mathds{1}_{\mathbb{E}(F(\vec\lambda \mid t_\text{cut})) < f_\text{target}} \!\!+\! \omega_2 h(\vec\lambda,\underline{\vec{\lambda}})~.
    \label{eq:defition_total_cost_function2}
    \end{align}
\endgroup
As mentioned, the indicator function $\mathds{1}_{\mathbb{E}(F(\vec\lambda) \mid t_\text{cut}) < f_\text{target}}$ introduces a penalty when the target fidelity is not achieved.
Further, the weights $\omega_1$ and $\omega_2$ control the relative importance of satisfying the fidelity constraint versus minimising hardware cost.
By imposing $\omega_1 \!\gg\! \omega_2$, we ensure that  the optimiser prioritises solutions satisfying $F(\vec\lambda \mid t_\text{cut}) \!\geq\! f_{\text{target}}$.
In our implementation, we set $w_1\!=\!10^{100}$ and $w_2\!=\!1$ following~\cite{da2024requirements} since the precise value of the total cost is irrelevant when the fidelity constraint is violated.

With the updated objective function $c$, we derive the set of desired hardware parameter values $\Lambda_{*}$ as below:
\begingroup
    \setlength{\abovedisplayskip}{2pt}   
    \setlength{\belowdisplayskip}{6pt}   
    \setlength{\abovedisplayshortskip}{2pt}
    \setlength{\belowdisplayshortskip}{2pt}
    \begin{align}
        \Lambda_{*} := \{ \vec\lambda \in \Lambda_\text{base} \!: c(\vec\lambda,\underline{\vec{\lambda}},F\mid t_\text{cut}) = \min\limits_{\vec\lambda' \in \Lambda_\text{base},\, t'_\text{cut} \in T_\text{cut}} c(\vec\lambda',\underline{\vec{\lambda}},F\mid t'_\text{cut}) \}~. \label{eq:lambdaStar}
    \end{align}
\endgroup
Finally, we use the optimisation heuristic from~\cite{Prielinger2024}, which outputs a point $\vec\lambda_{*} \in \Lambda_{*}$ that is \textit{practically} equivalent to any other point in $\Lambda_{*}$ considering our objective. 
We now address the individual questions in~\ref{Q:1}--\ref{Q:2c}.

\subsection{Reformulation of~\ref{Q:1}}
\label{sec:reformulation_of_Q1}
Recall that in~\ref{Q:1}, we consider teleportation in an MN, i.e., between nodes $P_1$ and $P_2$ (resp. $P_3$ and $P_4$) via the metropolitan hub $H_1$ (resp. $H_2$).
To identify the desired parameter space for the MN, we first introduce the following shorthand notations for the relevant hardware parameters and their ranges:
\begingroup
    \setlength{\abovedisplayskip}{6pt}   
    \setlength{\belowdisplayskip}{6pt}   
    \setlength{\abovedisplayshortskip}{2pt}
    \setlength{\belowdisplayshortskip}{2pt}
    \begin{align}
    \label{eq:paraemter_space_metropolitan}
        & \vec\lambda_{\text{m}'} := (p_\text{m}^0,t_\text{coh},f_{\text{m}'})\,, ~\vec{\underaccent{=}{\lambda}}_{\text{m}'} := (\underaccent{=}{p}_\text{m}^0, \underaccent{=}{t}_\text{coh}, \underaccent{=}{f}_{\text{m}'})\,, ~\vec{\underline{\lambda}}_{\text{m}'} := (\underline{p}_\text{m}^0, \underline{t}_\text{coh}, \underline{f}_{\text{m}'})\,,
        ~\vec{\overline{\lambda}}_{\text{m}'} := (\overline{p}_\text{m}^0, \overline{t}_\text{coh}, \overline{f}_{\text{m}'})\,, \\
        & \Lambda_{\text{m}'} := [\underaccent{=}{p}_\text{m}^0,\overline{p}_\text{m}^0] \times [\underaccent{=}{t}_\text{coh},\overline{t}_\text{coh}] \times [\underaccent{=}{f}_{\text{m}'},\overline{f}_{\text{m}'}]\,, ~\Lambda_{\text{m}',\text{base}} := [\underline{p}_\text{m}^0,\overline{p}_\text{m}^0] \times [\underline{t}_\text{coh},\overline{t}_\text{coh}] \times [\underline{f}_{\text{m}'},\overline{f}_{\text{m}'}]~.
    \end{align}
\endgroup
Recall that the parameters with lower bars denote baseline values, and those with upper bars denote optimistic values.

For the ER case, we denote the desired parameter range enabling teleportation in the MN with the threshold fidelity as
\begingroup
    \setlength{\abovedisplayskip}{6pt}   
    \setlength{\belowdisplayskip}{6pt}   
    \setlength{\abovedisplayshortskip}{2pt}
    \setlength{\belowdisplayshortskip}{2pt}
    \begin{align}
    \label{eq:desired_space_metro_ER}
        \Lambda_{1+}^\text{ER} := \{\vec\lambda_{\text{m}'} \in \Lambda_{\text{m}'} \!: \mathbb{E}\big(F_\text{m}^\text{ER} (\vec\lambda_{\text{m}'})\big)\geq f_\text{target} \}~.
    \end{align}
\endgroup
In case the baseline values lie outside this desired region, we proceed to find the set of points that minimises the hardware improvement cost as
\begingroup
    \setlength{\abovedisplayskip}{6pt}   
    \setlength{\belowdisplayskip}{6pt}   
    \setlength{\abovedisplayshortskip}{2pt}
    \setlength{\belowdisplayshortskip}{2pt}
    \begin{align}
    \label{eq:optimal_points_set_metro_ER}
        \Lambda_{1*}^\text{ER} \!:= \!\Big\{\vec\lambda_{\text{m}'} \in \Lambda_{\text{m}',\text{base}} \!: c(\vec{\lambda}_{\text{m}'},\underline{\vec{\lambda}}_{\text{m}'},F_\text{m}^\text{ER}) = \! \min\limits_{\vec{\lambda}_{\text{m}'}' \in \Lambda_{\text{m}',\text{base}}}  c(\vec{\lambda}_{\text{m}'}',\underline{\vec{\lambda}}_{\text{m}'},F_\text{m}^\text{ER}) \Big\}~.
    \end{align}
\endgroup
Depending on the cost function landscape, there could be multiple points in the parameter space satisfying the required criteria.
However, as discussed earlier, our numerical optimisation algorithm yields a specific solution $\vec\lambda_{1*}^\text{ER} \in \Lambda_{1*}^\text{ER}$, which is sufficient from a \textit{practical} point of view, since all such points are equivalent with respect to minimising the cost function~\eqref{eq:defition_total_cost_function2}.
Moreover, for better visualisation, we plot the following surface
\begingroup
    \setlength{\abovedisplayskip}{6pt}   
    \setlength{\belowdisplayskip}{6pt}   
    \setlength{\abovedisplayshortskip}{2pt}
    \setlength{\belowdisplayshortskip}{2pt}
    \begin{align}
    \label{eq:min_fm_definition_metro_ER}
        \tilde{\Lambda}_{1+}^\text{ER} := \{(p_\text{m}^0,t_\text{coh},f_{\text{m}'}) \in \Lambda_{1+}^\text{ER} \!: f_{\text{m}'} = \min_{(p_\text{m}^0,t_\text{coh},z)\in \Lambda_{1+}^\text{ER}} z \}~,
    \end{align}
\endgroup
instead of the the desired parameter range $\Lambda_{1+}^\text{ER}$.
Observe that $\tilde{\Lambda}_{1+}^\text{ER}$ denotes the set requiring the minimum link fidelity $f_\text{m}$ among the points in $\Lambda_{1+}^\text{ER}$.
Also, for each point on this surface, we calculate the corresponding teleportation rate as:
\begingroup
    \setlength{\abovedisplayskip}{6pt}   
    \setlength{\belowdisplayskip}{6pt}   
    \setlength{\abovedisplayshortskip}{2pt}
    \setlength{\belowdisplayshortskip}{2pt}
    \begin{align}
    \label{eq:rate_min_fm_definition_metro_ER}
        & R_1^{\text{ER}}(\vec{\lambda}_{\text{m}'}) := R_{\text{m}}(p_{\text{m}}^0)~,~~  \vec{\lambda}_{\text{m}'} \in \tilde{\Lambda}_\text{1+}^\text{ER}~.
    \end{align}
\endgroup

For the qubit-ready (QR) case, the desired parameter space, the surface, set of optimal points, and the corresponding rate for the points on this surface are, respectively, given by
\begingroup
    \setlength{\abovedisplayskip}{6pt}   
    \setlength{\belowdisplayskip}{6pt}   
    \setlength{\abovedisplayshortskip}{2pt}
    \setlength{\belowdisplayshortskip}{2pt}
    \begin{align}
        & \Lambda_{1+}^\text{QR} := \{\vec\lambda_{\text{m}'} \in \Lambda_{\text{m}'} \!: \mathbb{E}\big(F_\text{m}^\text{QR} (\vec\lambda_{\text{m}'})\big) \geq f_\text{target} \}~, 
        \label{eq:desired_space_metro_QR} \\
        & \tilde{\Lambda}_{1+}^\text{QR} := \{(p_\text{m}^0,t_\text{coh},f_{\text{m}'}) \in \Lambda_{1+}^\text{QR} \!: f_{\text{m}'} = \min_{(p_\text{m}^0,t_\text{coh},z)\in \Lambda_{1+}^\text{QR}} z \}~,
        \label{eq:min_fm_definition_metro_QR} \\
        & \Lambda_{1*}^\text{QR} \!:= \!\Big\{\vec\lambda_{\text{m}'} \in \Lambda_{\text{m}',\text{base}} \!: c(\vec{\lambda}_{\text{m}'},\underline{\vec{\lambda}}_{\text{m}'},F_\text{m}^\text{QR}) = \! \min\limits_{\vec{\lambda}_{\text{m}'}' \in \Lambda_{\text{m}',\text{base}}}  c(\vec{\lambda}_{\text{m}'}',\underline{\vec{\lambda}}_{\text{m}'},F_\text{m}^\text{QR}) \Big\}~, 
        \label{eq:optimal_points_set_metro_QR} \\
        & R_1^{\text{QR}}(\vec{\lambda}_{\text{m}'}) := R_{\text{m}}(p_{\text{m}}^0)
        ~, ~~\vec{\lambda}_{\text{m}'} \in \tilde{\Lambda}_\text{1+}^\text{QR}~.
        \label{eq:rate_min_fm_definition_metro_QR}
    \end{align}
\endgroup

The reformulations of~\ref{Q:2a}--\ref{Q:2c} are carried out in a similar way, and we define them formally in~\ref{sec:reformulation_of_Q2-Q4}. 
For example, in the reformulation of~\ref{Q:2a} for ER teleportation, the desired parameter space, surface, set of optimal points, and the corresponding rate are denoted respectively as $\Lambda_{2+}^\text{ER}$, $\tilde{\Lambda}_{2+}^\text{ER}$, $\Lambda_{2*}^\text{ER}$, and $R_{2}^\text{ER}$; see see~\eqref{eq:desired_space_metro_int_ER}--\eqref{eq:rate_min_fm_definition_metro_int_ER}.

\looseness = -1
For evaluations, we need to compute the sets~\eqref{eq:desired_space_metro_QR}--\eqref{eq:optimal_points_set_metro_QR} and their counterparts for~\ref{Q:2a}--\ref{Q:2c} provided in~\ref{sec:reformulation_of_Q2-Q4} .
To that end, we first express the performance metrics, i.e., the rate and fidelity of teleportation as mentioned in Tab.~\ref{tab:quantities_and_parameters}, in terms of the hardware parameters and the cut-off, where applicable.
The derivations for the metrics for intercity teleportation are provided in Sec.~\ref{sec:derivation_rate_an_fidelity_IN}, while the corresponding expressions for teleportation within an MN are provided in~\ref{sec:metro_teleportation_rate} and~\ref{sec:metro_teleportation_fidelity}.
In the next section, we outline the procedure used to arrive at the baseline and optimistic values of the hardware parameters, considering trapped-ion nodes for the MNs and ensemble-based memories for the repeater chain in the backbone.
\vspace{-7pt}

\subsection{Baseline and Optimistic Parameter Values}
\label{sec:baseline_and_optimistic_parameters}
We take the parameter values for the MN components corresponding to the quantum communication experiments with trapped ions.
Our choice is motivated by the popularity of these physical systems in the context of quantum networking primitives.
Empirical demonstrations of the efficacy of such systems include long-lived quantum memories~\cite{Krutyanskiy2023TelecomWavelength, drmota2022long}, remote entanglement generation~\cite{Krutyanskiy2023Entanglement230Meters, Krutyanskiy2019Light, Krutyanskiy2024Entanglement101km}, the ability to perform high-fidelity single- and two-qubit quantum gates~\cite{ghadimi2019towards, inlek2017multispecies, higgins2017single}, and entanglement swapping~\cite{Krutyanskiy2023Entanglement230Meters}.
Note that we do not assume all parameters to have been demonstrated in the same experimental setup.

The baseline parameter values listed in Tab.~\ref{tab:no_imperfection_probability_baseline_and_optimistic} reflect the state-of-the-art experimental capabilities, while their optimistic counterparts represent projected near-term improvements.
Recent experimental work~\cite{Krutyanskiy2023Entanglement230Meters} has demonstrated entanglement generation over a trapped-ion network with average fidelities ${f}_{\text{m}'}$ up to $0.88$ relative to a maximally entangled state.
Furthermore, coherence times (${t}_\text{coh}$) as long as $62$ ms have been achieved~\cite{Krutyanskiy2023TelecomWavelength}.
Our baseline for the base efficiency ${p}_\text{m}^0$ for ion-ion entanglement generation, i.e., without considering the fibre loss, is $5.95\times 10^{-4}$.
We provide a comprehensive derivation of $p_\text{m}^0$ from experimental parameters in~\ref{sec:derivation_metro_p_m} and obtain the baseline value.
Following the supplementary material of~\cite{Krutyanskiy2023TelecomWavelength}, we take $t_\text{prep}$ as $175\,\mu\text{s}$.

To obtain the optimistic values of these parameters, we incorporate anticipated near-term hardware advancements as indicated by experimental physicists at the Institute for Quantum Optics and Quantum Information, University of Innsbruck~\cite{privatecommunication2025}.
These improvements include narrowing the photon-detection coincidence window and enhancing photon detectors to increase ion–ion entanglement fidelity.
We expect such advances to yield a metropolitan link fidelity ${f}_{\text{m}'}$ of $0.95$.
Moreover, employing decoherence-free subspaces~\cite{zwerger2017quantum, haffner2005robust} in combination with sympathetic cooling techniques~\cite{rosenband2007observation, home2009memory}, or using alternative ion isotopes with inherently longer coherence times~\cite{drmota2023robust}, is expected to extend the memory coherence time ${t}_\text{coh}$ to $4$ seconds.
With respect to the base efficiency ${p}_\text{m}^0$, the ion-photon entanglement-generation probability can be enhanced by employing smaller cavities to achieve better coupling with the trapped ions, using improved photon detectors, and frequency conversion techniques for the emitted photons.
These advancements yield an optimistic estimate of ${p}_\text{m}^0 = 1.43\times 10^{-2}$.
We provide the corresponding derivation in~\ref{sec:derivation_metro_p_m}.
Note that we do not know a priori whether the optimistic parameter values are sufficient to satisfy our requirements.
While these values are based on projections of current experimental capabilities and plausible future improvements, further improvements beyond these estimates may be possible.
In this work, however, we do not consider such potential advancements and instead assess whether the optimistic parameters are adequate.

The backbone network, which spans a distance of 450 km, has not yet been experimentally realised.
However, multimode quantum memories based on atomic ensembles have been shown to enable heralded entanglement generation when integrated with spontaneous parametric down conversion sources~\cite{simon2007quantum, sinclair2014spectral}.
Motivated by these advances and the potential for long storage times~\cite{ruskuc2022nuclear, zhong2015optically}, we base our parameter estimates on theoretical models of hybrid architectures that combine trapped-ion nodes with atomic ensemble-based repeaters~\cite{tissot2025hybrid, tissot2025single}.
For the baseline backbone parameters, we adopt conservative estimates consistent with current experimental capabilities.
Using the formalism of~\cite{tissot2025hybrid}, this yields an entanglement generation rate of approximately $1/1610.15~\text{sec}^{-1}$ at a target fidelity ${f}_\text{b}\!=\!0.6$ 
\footnote{\label{fn:benedikt} Based on the \textit{two-single-click protocol} of~\cite{tissot2025hybrid, tissot2025single}, and values obtained using the corresponding code repository.}.
Following the supplementary material of~\cite{Krutyanskiy2023TelecomWavelength}, we take $t_\text{prep}$ to be $175\,\mu\text{s}$.
Thus, using~\eqref{eq:backbone_rate_definition} and~\eqref{eq:t_b_total}, we obtain the baseline of backbone entanglement-generation probability $p_\text{b}\!=\!R_\text{b}\times t_b\!=\!1.51\times 10^{-6}$, where we used $t_\text{b}^\text{class}\!=\!450/c$, and $c\!=\!200,\!000\,\text{km/s}$ is the speed of light in optical fibre.

\begin{table}[htbp]
  \centering
  \caption{Mapping of physical parameters to corresponding no-imperfection probabilities ($p_{NI}$), which re-scale physical parameter values to the interval $[0,1]$. Baseline values reflect the state-of-the-art experimental capabilities, and optimistic values represent projected near-term improvements. We vary parameters over this range to evaluate the hardware requirements to perform teleportation. The metropolitan and backbone parameter values are representative of trapped-ion-based networks and ensemble-based quantum repeater chains, respectively.
  }
  \label{tab:no_imperfection_probability_baseline_and_optimistic}
  \begin{tabular}{llll} 
    \toprule
    \textbf{Parameter $\lambda^{(i)}$} & \textbf{$p_{NI}(\lambda^{(i)})$} & \textbf{Baseline $\underline{\lambda}^{(i)}$} & \textbf{Optimistic $\overline{\lambda}^{(i)}$} \\
    \midrule
    Metropolitan base efficiency $p_\text{m}^0$ & $p_\text{m}^0$ & $5.95 \times 10^{-4}$ & $1.43\times 10^{-2}$ \\ 
    &  & (see \ref{sec:derivation_metro_p_m}) & (see \ref{sec:derivation_metro_p_m}) \\ 
    Coherence time $t_\text{coh}$ & $e^{-1/t_\text{coh}}$ & $62$ ms~\cite{Krutyanskiy2023TelecomWavelength} & $4$ s~\cite{privatecommunication2025}   \\ 
    Metropolitan ent. fidelity $f_\text{m}$ & $\frac{1}{3}(4f_\text{m}-1)$ & $0.88$~\cite{Krutyanskiy2023Entanglement230Meters} & $0.95$~\cite{privatecommunication2025}  \\
    Backbone ent. gen. probability $p_\text{b}$ & $p_\text{b}$ & $1.51\times 10^{-6}$   & $4.18\times 10^{-3}$~\cite{tissot2025hybrid} \\ 
    Backbone ent. fidelity $f_\text{b}$ & $\frac{1}{3}(4f_\text{b}-1)$ & $0.60$   & $0.90$~\cite{tissot2025hybrid} \\
    \bottomrule
  \end{tabular}
  \vspace{-6pt}
\end{table}

To derive the corresponding optimistic values, we follow~\cite{tissot2025hybrid} to arrive at an entanglement generation rate $R_\text{b}$ of approximately $1/0.58~\text{sec}^{-1}$ with a target entanglement fidelity ${f}_\text{b}\!=\!0.9$.
This gives the optimistic backbone entanglement-generation probability ${p_\text{b}\!=\!R_\text{b}\times t_b\!=\!4.18\times 10^{-3}}$.
We emphasise that although the parameter values for our evaluations are based on quantum-networking experiments involving trapped ions and atomic ensemble-based repeaters, our modelling framework remains hardware agnostic.
\vspace{-10pt}

\section{Derivation of the Teleportation Rate and Expected Fidelity in the Intercity Network}
\label{sec:derivation_rate_an_fidelity_IN}
\everypar{\looseness=-1}
To answer~\ref{Q:2a}--\ref{Q:2c}, we now derive the rate and fidelity of teleportation in the IN, which are then plugged into the formulations of the desired sets in~\ref{sec:reformulation_of_Q2}-\ref{sec:reformulation_of_Q4}.

\subsection{Derivation of Teleportation Rate in the Intercity Network}
\label{sec:derivation_rate_IN}

In our model of the IN, the rate at which a qubit can be teleported between end nodes belonging to different MNs is determined by the time required to establish end-to-end entanglement and subsequent teleportation time.
Due to~\ref{assumption:instantaneous_qubit_preparation_local_operation}, the teleportation time includes the transmission time of the Pauli correction message from the sender node to the receiver, which is given by
${t_{\text{int}}^\text{class}}$; see~\eqref{eq:t_class_int}.
Further, let \( X_\text{e2e} \) be the random variable representing the time to establish an end-to-end entanglement successfully.
Since we assume that data qubit preparation is instantaneous (see~\ref{assumption:instantaneous_qubit_preparation_local_operation}), the total time required in both ER and QR teleportation is  \( {t_{\text{int}}^\text{class}} + X_\text{e2e}\). 
Thus, the rate of intercity teleportation is given by
\begin{equation}
\label{eq:def_rate_teleportation_int}
    \text{(Teleportation rate)} \qquad \qquad \qquad R_{\text{int}} = \frac{1}{\mathbb{E}(t_\text{int}^{\text{class}} \!+\! X_\text{e2e})} = \frac{1}{t_\text{int}^{\text{class}} \!+\! \mathbb{E}(X_\text{e2e})}~. \qquad \qquad \qquad
\end{equation}

We now focus on the derivation of $\mathbb{E}(X_{\text{e2e}})$. 
To that end, we introduce certain notations and rephrase relevant implications from assumptions~\ref{assumption:werner_state}--\ref{assumption:instantaneous_qubit_preparation_local_operation}.

\begin{enumerate}[label={\textbf{R\arabic*}},nolistsep,leftmargin=*]
\item \label{assumptionRate:individualDurations} \textbf{Individual link generation times
}:
Recall from~\ref{assumption:cutoff_time} that a cut-off time $t_\text{cut}$ is imposed to ensure a certain quality of entanglement.  
We denote by $X_1$, $X_2$, and $X_\text{b}$ the duration of \emph{successfully} generating entanglement between $\text{P}_1$ (or $\text{P}_2$) and $\text{J}_1$, $\text{P}_3$ (or $\text{P}_4$) and $\text{J}_2$, and $\text{J}_1$ and $\text{J}_2$, respectively.
We further define
$${X_{\text{max}} := \max\{X_1, X_2, X_b\}}~,~ {X_{\text{min}} := \min\{X_1, X_2, X_b\}}~.$$
As explained in~\ref{assumption:cutoff_time}, we consider an end-to-end entanglement generation round to be successful if and only if all links are produced within a \( t_\text{cut} \) time block where $t_\text{cut} \in T_\text{cut}$; see~\eqref{eq:t_cut_range}.
That is, we require  
\begin{equation}\label{define-kutoff}
    X_{\text{max}} - X_{\text{min}} < t_\text{cut} \iff X_{\text{max}} - X_{\text{min}} \leq t_\text{cut}' := t_\text{cut} - 1.
\end{equation}

\item \label{assumptionRate:auxiliaryVar} \textbf{Auxiliary variables}:
To denote the success of the $i$-th end-to-end entanglement generation round, we introduce a Bernoulli random variable \( Y^{(i)} \). 
That is, \( {Y^{(i)} \!=\! 1} \) if end-to-end entanglement is successfully created during the \( i \)-th attempt and \( {Y^{(i)} \!=\! 0} \) otherwise.  
Also, we denote by $N$ the number of rounds until successful entanglement generation and by $Z^{(i)}$ the duration of the $i$-th round.  
We also assume
\begingroup
    \setlength{\abovedisplayskip}{6pt}   
    \setlength{\belowdisplayskip}{6pt}   
    \setlength{\abovedisplayshortskip}{2pt}
    \setlength{\belowdisplayshortskip}{2pt}
    \begin{align}
        &Y^{(i)} \overset{\text{IID}}{\sim} Y\,,~Z^{(i)} \overset{\text{IID}}{\sim} Z\,, \label{eq:def_Y_Z} \\
        &\mathbb{P}(Y=1) = \mathbb{P}(X_\text{max} \!-\! X_\text{min}<t_\text{cut}) =: p~. \label{eq:def_p}
    \end{align}
\endgroup

\looseness = -1 
Recall from~\ref{assumption:entanglement_swap_asap} that entanglement swap in memory always succeeds.
However, in a successful entanglement generation round, after completion of the final swap, a classical message must be transmitted to both end nodes to signal whether the actual teleportation process can commence.
As mentioned in~\eqref{eq:t_msg} and illustrated in Fig.~\ref{fig:intercity_network_time}, the duration required for this communication is denoted by $t_\text{msg}$.
On the other hand, recall from~\ref{assumption:cutoff_time} that a round is considered unsuccessful if the final entangled link (out of three) is not generated within the cut-off time $t_\text{cut}$ since the generation of the first elementary link.
At time $X_\text{min}$, a classical message is sent across nodes to communicate successful entanglement generation in the relevant link.
The message takes at most $t_{\text{msg}}$ time to reach the farthest node.
Since we have ${t_\text{cut} \geq t_{\text{msg}}}$ (see~\eqref{eq:t_cut_range}), this allows all nodes to abandon existing elementary link generation attempts by time $X_\text{min} + t_\text{cut}$ and restart the round.
Thus, the duration of a successful ($Y=1$) and failed ($Y=0$) end-to-end entanglement generation round is given by
\begingroup
    \setlength{\abovedisplayskip}{6pt}   
    \setlength{\belowdisplayskip}{8pt}   
    \setlength{\abovedisplayshortskip}{2pt}
    \setlength{\belowdisplayshortskip}{2pt}
    \begin{align}
        Z = 
        \begin{cases}
            X_{\text{max}} + t_{\text{msg}}, & \text{if } Y = 1~, \\[6pt]
            X_{\text{min}} + t_{\text{cut}}, & \text{if } Y = 0~.
        \end{cases}
        \label{eq:def_Z}
    \end{align}
\endgroup

\item \label{assumptionRate:events} \textbf{Events concerning end-to-end link generation rounds}:
We define the following events to facilitate our analysis:
\begingroup
    \setlength{\abovedisplayskip}{6pt}   
    \setlength{\belowdisplayskip}{6pt}   
    \setlength{\abovedisplayshortskip}{2pt}
    \setlength{\belowdisplayshortskip}{2pt}
    \begin{align}\label{eq:Aij}
        A_i \!&:=\! \left\{ \omega \!:\! X_\text{max}(\omega)\!=\! X_i(\omega)\right\}~,\thickspace 
        A_i A_j \!:=\! A_i \cap A_j~, \thickspace A_1 A_2 A_b \!:=\! \bigcap\limits_{k \in \{1,2,b\}} A_k~, \thickspace i,j \in \{1,2,b\}~.
    \end{align}
\endgroup
Further, the following events correspond to $Y^{(i)}=1$, i.e., success in the $i$-th round:
\begingroup
    \setlength{\abovedisplayskip}{6pt}   
    \setlength{\belowdisplayskip}{6pt}   
    \setlength{\abovedisplayshortskip}{2pt}
    \setlength{\belowdisplayshortskip}{2pt}
    \begin{align}\label{eq:Aij+}
    \begin{aligned}
        A_i^+ \!&:=\! \left\{ \omega \!:\! X_\text{max}(\omega)\!=\! X_i(\omega), X_{\text{max}} - X_{\text{min}} \leq t_\text{cut}'\right\}~; \\
        A_{ij}^+ \!&:=\! \left\{ \omega \!:\! X_\text{max}(\omega)\!=\! X_i(\omega), X_\text{min}(\omega)\!=\! X_j(\omega), X_{\text{max}} - X_{\text{min}} \leq t_\text{cut}'\right\}~; \\ 
        A_i^+ A_j^+ \!&:=\! A_i^+ \cap A_j^+~, \thickspace i,j \in \{1,2,b\}~; \thickspace A_1^+ A_2^+ A_b^+ \!:=\! \bigcap\limits_{k \in \{1,2,b\}} A_k^+~.
    \end{aligned}
    \end{align}
\endgroup
\looseness = -1 Characterisation of these events in terms of $M_1, M_2$ and $M_b$ is straightforward.
For example, on $A_1^+ A_b^+$, $M_1, M_2$ and $M_b$ will assume values of the following form 
\begingroup
    \setlength{\abovedisplayskip}{6pt}   
    \setlength{\belowdisplayskip}{6pt}   
    \setlength{\abovedisplayshortskip}{2pt}
    \setlength{\belowdisplayshortskip}{2pt}
    \begin{align} 
    \label{eq:m^*_and_k}
    \begin{aligned}
        & M_1 = k m^*, ~M_b = k m^* \frac{t_\text{m}}{t_\text{b}}, ~\text{for} ~k\in \mathbb{N}, ~m^* \!= \frac{t_\text{b}}{\text{gcd}(t_\text{m}, t_\text{b})} ~\text{and} \!\ceil{M_1 \!-\! \frac{t_\text{cut}^\prime}{t_\text{m}}} \leq M_2 \leq M_1.
    \end{aligned}
    \end{align}
\endgroup
An exhaustive list of characterisations for the events in~\eqref{eq:Aij+} is shown in Tab.~\ref{tab:sample_space_subset_A+}.

\begin{table}[h!]
    \centering
    \caption{Reformulation of the sets from~\eqref{eq:Aij+}. Here, $i,i^\prime \!\in\! \{1,2\}$, and ${i^\prime\!=\!(i\!+\!1)~\text{mod}~2}$.}
    \label{tab:sample_space_subset_A+}
    \renewcommand{\arraystretch}{1.6} 
    \begin{tabular}{|c|c|c|}
    \hline
    \textbf{Set} & \textbf{Definition} & \textbf{Interpretation in terms of $M_1, M_2, M_b$} \\ \hline
    \parbox{1.8cm}{\centering $A_i^+$,
    $i\!\in\! \{1,2\}$} & 
    \parbox{5cm}{\centering ${X_i \!-\! t_\text{cut}^\prime \leq X_{i^\prime} \leq X_i}$,\\ ${X_i \!-\! t_\text{cut}^\prime \leq X_b \leq X_i}$ } &
    \parbox{8cm}{\centering \vspace{0.2cm} ${\max(1, \left\lceil M_i \!-\! t_\text{cut}^\prime/t_\text{m}\right\rceil)} \!\leq\! M_{i^\prime} \!\leq\! M_i $, ${ \max(1, \left\lceil(M_i t_\text{m} \!-\! t_\text{cut}^\prime)/t_\text{b}\right\rceil) \!\leq\! M_b \!\leq\! \left\lfloor M_i t_\text{m}/t_\text{b}\right\rfloor}$, $M_i \ge \ceil{t_\text{b}/t_\text{m}}$ \vspace{0.2cm}} \\ \hline
    $A_b^+$ & 
    \parbox{5cm}{\centering ${X_b\!-\!t_\text{cut}^\prime \leq X_i \leq X_b}$} & 
    \parbox{8cm}{\centering \vspace{0.2cm} ${\max(1, \left\lceil (M_b t_\text{b}\!-\!t_\text{cut}^\prime)/t_\text{m} \right\rceil) \!\leq\! M_i \!\leq\! \left\lfloor M_b t_\text{b}/t_\text{m} \right\rfloor}$, $M_b \in \mathbb{N}$\vspace{0.2cm}} \\ \hline
    $A_1^+ A_2^+$ & 
    \parbox{5cm}{\centering ${X_1 \!=\! X_2}$, \\ ${X_1 \!-\! t_\text{cut}^\prime \leq X_b \leq X_1}$} & 
    \parbox{8cm}{\centering \vspace{0.1cm} ${\max(1, \left\lceil(M_1 t_\text{m} \!-\! t_\text{cut}^\prime)/t_\text{b}\right\rceil) \leq M_b \leq \left\lfloor M_1 t_\text{m}/t_\text{b}\right\rfloor}$, ${M_2 = M_1}$, $M_1 \ge \ceil{t_\text{b}/t_\text{m}}$ \vspace{0.2 cm}} \\ \hline
    \parbox{1.8cm}{\centering $A_i^+ A_b^+$,\\
    $i \!\in\! \{1,2\}$} & 
    \parbox{5cm}{\centering ${X_i \!=\! X_b}$, \\ ${X_b\!-\! t_\text{cut}^\prime \leq X_{i^\prime} \leq X_b}$ }& 
    \parbox{8cm}{\centering \vspace{0.2cm} ${\left\lceil M_i \!-\! t_\text{cut}^\prime/t_\text{m}\right\rceil \!\leq\! M_{i'} \!\leq\! M_i}$,\\ ${M_i \!=\! k m^*}$, ${M_b \!=\! k m^* t_\text{m}/t_\text{b}}$ for $k\in \mathbb{N}$; see~\eqref{eq:m^*_and_k} \vspace{0.2cm}} \\ \hline
    $A_1^+ A_2^+ A_b^+$ & 
    \parbox{5cm}{\centering ${X_1 \!=\! X_2 \!=\! X_b}$} & 
    \parbox{8cm}{\centering \vspace{0.2cm} ${M_1 t_\text{m} \!=\! M_2 t_\text{m} \!=\! M_b t_\text{b} \!=\! k m^* t_\text{m}}$ for $k\in \mathbb{N}$; see~\eqref{eq:m^*_and_k}\vspace{0.2cm}} \\ \hline
    \end{tabular}
    \vspace{-6pt}
\end{table}

\begin{table}[h!]
    \centering
    \caption{Reformulation of the sets from~\eqref{eq:Aij-}. Here, $i,i^\prime \in \{1,2\}$, and ${i^\prime\!=\!(i\!+\!1)~\text{mod}~2}$.}
    \label{tab:sample_space_subset_A-}
    \renewcommand{\arraystretch}{1.6} 
    \begin{tabular}{|c|c|c|} 
    \hline
    \textbf{Set} & \textbf{Definition} & \textbf{Interpretation in terms of $M_1, M_2, M_b$} \\ \hline
    $A_1^- A_2^-$ & 
    \parbox{3.8cm}{\centering ${X_b \leq X_1 \!=\! X_2}$, \\ ${X_b\leq X_1\!-\!t_\text{cut}}$} & 
    \parbox{9cm}{\centering \vspace{0.2cm}${M_1=M_2}$, \\ ${M_1 \ge \ceil{(M_b t_\text{b} \!+\! t_\text{cut})/t_\text{m}}}$, ${M_b \in \mathbb{N}}$ \vspace{0.2cm}} \\ \hline
    \parbox{1.7cm}{\centering $A_i^- A_b^-$ } & 
    \parbox{3.8cm}{\centering ${X_{i^\prime} \leq X_i\!=\!X_b}$, \\ ${X_{i^\prime} \leq X_i\!-\!t_\text{cut}}$} & 
    \parbox{9cm}{\centering \vspace{0.2cm} ${M_i = k m^*}$, ${M_b = k m^* t_\text{m}/t_\text{b}}$, for ${k\in \left[\ceil{(M_{i^\prime} t_\text{m} \!+\! t_\text{cut})/(m^* t_\text{m})}, \infty\right) \cap \mathbb{N}}$, ${M_{i^\prime}\in\mathbb{N}}$; see~\eqref{eq:m^*_and_k} \vspace{0.2cm}} \\ \hline
    \parbox{1.7cm}{\centering $A_{i i^\prime}^-$ } & 
    \parbox{3.8cm}{\centering ${X_{i^\prime} \leq X_b \leq X_i}$, \\${X_{i^\prime}\leq X_i\!-\!t_\text{cut}}$ } & 
    \parbox{9cm}{\centering \vspace{0.2cm}${ \left\lceil M_{i^\prime} t_\text{m}/t_\text{b}\right\rceil \leq M_b \leq \left\lfloor M_i t_\text{m}/t_\text{b}\right\rfloor}$, \\ ${M_i \geq M_{i^\prime} \!+\! \left\lceil t_\text{cut}/t_\text{m}\right\rceil}$, ${M_{i^\prime} \in \mathbb{N}}$ \vspace{0.2cm}}\\ \hline
    \parbox{1.7cm}{\centering $A_{ib}^-$ } & 
    \parbox{3.8cm}{\centering \vspace{0.2cm}${X_b \leq X_{i^\prime} \leq X_i}$, \\ ${X_b\leq X_i\!-\!t_\text{cut}}$ \vspace{0.2cm}} & 
    \parbox{9cm}{\centering \vspace{0.2cm}${\left\lceil M_b t_\text{b}/t_\text{m} \right\rceil \leq M_{i^\prime} \leq M_i}$, \\ ${M_i \ge \left\lceil ( M_b t_\text{b} \!+\! t_\text{cut})/t_\text{m} \right\rceil }$, ${M_b \in \mathbb{N}}$\vspace{0.2cm}} \\ \hline
    $A_{ii^\prime}^- A_{ib}^-$ & 
    \parbox{3.8cm}{\centering ${X_{i^\prime} \!=\! X_b \leq X_i}$, \\ ${X_{i^\prime} \leq X_i \!-\!t_\text{cut}}$} & 
    \parbox{9cm}{\centering \vspace{0.2cm}${M_i \ge M_{i^\prime}\!+\!\left\lceil t_\text{cut}/t_\text{m}\right\rceil}$, \\ ${M_{i^\prime} \!=\! k m^*}$, ${M_b \!=\! k m^* t_\text{m}/t_\text{b}}$, for $k\!\in\! \mathbb{N}$; see ~\eqref{eq:m^*_and_k} \vspace{0.2cm}} \\ \hline
    \parbox{1.7cm}{\centering $A_{bi}^-$ } & 
    \parbox{3.8cm}{\centering \vspace{0.2cm}${X_i \leq X_{i^\prime} \leq X_b}$, \\ ${X_i \leq X_b\!-\! t_\text{cut}}$ \vspace{0.2cm}} & 
    \parbox{9cm}{\centering \vspace{0.2cm}${M_i \leq M_{i^\prime} \leq \floor{M_b t_\text{b}/t_\text{m}}}$ \\ ${M_b \ge \ceil{(M_i t_\text{m} \!+\! t_\text{cut})/t_\text{b}} }$, ${M_i \in \mathbb{N}}$ \vspace{0.2cm}} \\ \hline
    $A_{b1}^- A_{b2}^-$ & 
    \parbox{3.8cm}{\centering ${X_1 \!=\! X_2 \leq X_b}$, \\ ${X_1\leq X_b\!-\!t_\text{cut}}$} & 
    \parbox{9cm}{\centering \vspace{0.2cm} ${M_1 = M_2}$, \\ ${M_b \ge \ceil{(M_1 t_\text{m} \!+\! t_\text{cut})/t_\text{b}} }$, ${M_1 \in \mathbb{N}}$ \vspace{0.2cm}} \\ \hline
    \end{tabular}
    \vspace{-6pt}
\end{table}

For failed rounds (i.e., $Y^{(i)}\!=\!0$), the relevant events are defined as
\begingroup
    \setlength{\abovedisplayskip}{6pt}   
    \setlength{\belowdisplayskip}{8pt}   
    \setlength{\abovedisplayshortskip}{2pt}
    \setlength{\belowdisplayshortskip}{2pt}
    \begin{align}\label{eq:Aij-}
    \begin{aligned}
        A_i^- &:= A_i - A_i^+~, \quad i \in \{1,2,\text{b}\}~; \\
        A_i^- A_j^- &:= A_i^- \cap A_j^-, \quad i,j \in \{1,2,\text{b}\} ~; \\
        A_{ij}^- &:= \left\{ \omega \!:\! X_\text{max}(\omega)\!=\! X_i(\omega), X_\text{min}(\omega)\!=\! X_j(\omega), X_{\text{max}} \!-\! X_{\text{min}} \geq t_\text{cut}\right\}~; \\
        A_{ij}^- A_{kl}^- \!&:=\! A_{ij}^- \cap A_{kl}^-~, \quad i,j,k,l \in \{1,2,\text{b}\}~.
    \end{aligned}
    \end{align}
\endgroup
Characterisation of these events in terms of $M_1, M_2$, and $M_b$ is given in Tab.~\ref{tab:sample_space_subset_A-}.
\end{enumerate}

\noindent
\ref{assumptionRate:individualDurations}--\ref{assumptionRate:events} help us derive $\mathbb{E}(X_\text{e2e})$, which we describe in the following theorem.

\begin{theorem}
    The expected time to successfully establish an end-to-end link is given by
    \vspace{-4pt}
    \begin{align}
        \mathbb{E}(X_\mathrm{e2e}) = \frac{\mathbb{E}\left(Z \mathds{1}_{Y\!=0}\right) + \mathbb{E}\left(Z \mathds{1}_{Y\!=1}\right)}{p}\,,~\text{where}
    \label{eq:ent_gen_time_int}
    \end{align}
    \vspace{-2.0em} 
    \begin{align}
    \mathbb{E}\big(Z \mathds{1}_{Y=1} \big) 
    =&~ 2\mathbb{E}\big(Z \mathds{1}_{A_1^+} \big) 
    \!+\! \mathbb{E}\big(Z \mathds{1}_{A_b^+} \big) 
    \!-\! \mathbb{E}\big(Z \mathds{1}_{A_1^+ A_2^+} \big) 
    \!-\! 2\mathbb{E}\big(Z \mathds{1}_{A_1^+ A_b^+} \big) 
    \!+\! \mathbb{E}\big(Z \mathds{1}_{A_1^+A_2^+A_b^+} \big)~, 
    \label{eq:zy1} \\
    \mathbb{E}\big(Z \mathds{1}_{Y\!=0} \big) 
    =&~ 2\bigg(\! 
        \mathbb{E}\big(Z \mathds{1}_{A_{12}^-} \big) 
        \!+\! \mathbb{E}\big(Z \mathds{1}_{A_{1b}^-} \big) 
        \!+\! \mathbb{E}\big(Z \mathds{1}_{A_{b1}^-} \big) 
    \!\bigg) 
    \!-\! 2\mathbb{E}\big(Z \mathds{1}_{A_{12}^-A_{1b}^-} \big) 
    \!-\! \mathbb{E}\big(Z \mathds{1}_{A_{b1}^- A_{b2}^-} \big) \nonumber \\
    & \!-\! \mathbb{E}\big(Z \mathds{1}_{A_1^- A_2^-} \big) 
    \!-\! 2\mathbb{E}\big(Z \mathds{1}_{A_1^-A_b^-} \big)~. 
    \label{eq:inclusion-exclusion_part_2}
\end{align}
\end{theorem}

\begin{proof}
\looseness = -1
Since $N$ denotes the total number of rounds required for successful end-to-end entanglement generation,
\begingroup
    \setlength{\abovedisplayskip}{-4pt}   
    \setlength{\belowdisplayskip}{8pt}   
    \setlength{\abovedisplayshortskip}{2pt}
    \setlength{\belowdisplayshortskip}{2pt}
    \begin{align}
        \mathbb{E}(X_\text{e2e}) &=\mathbb{E}\Big(\sum_{i=1}^{N}Z^{(i)}\Big) \\
        &\overset{\text{(i)}}{=} \mathbb{E}\Big(\mathbb{E}\Big(\sum_{i=1}^{N}Z^{(i)}\mid N\Big)\Big)    \\
        &\overset{\text{(ii)}}{=} \sum_{n=1}^{\infty} p (1-p)^{n-1} \mathbb{E}\Big(\sum_{i=1}^{N}Z^{(i)}\mid N\!=\!n \Big) \\
        &\overset{\text{(iii)}}{=} \sum_{n=1}^{\infty} p (1-p)^{n-1} \Big(\sum_{i=1}^{n-1} \mathbb{E}\left(Z^{(i)}\mid Y^{(i)}\!=\!0 \right) + \mathbb{E}\left(Z^{(n)}\mid Y^{(n)}\!=\!1 \right) \Big) \\
        &\overset{\text{(iv)}}{=} \sum_{n=1}^{\infty} p (1-p)^{n-1} \big((n-1) \mathbb{E}\left(Z \mid Y\!\!=\!0 \right) + \mathbb{E}\left(Z\mid Y\!\!=\!1 \right) \big) \\
        &= \Big(\frac{1}{p} - 1\Big) \mathbb{E}\left(Z \mid Y\!=\!0 \right) + \mathbb{E}\left(Z \mid Y\!=\!1 \right) \\
        &\overset{\text{(v)}}{=} \Big(\frac{1}{p} - 1\Big) \frac{\mathbb{E}\left(Z \mathds{1}_{Y\!=0} \right)}{\mathbb{P}(Y\!=\!0)} + \frac{\mathbb{E}\left(Z \mathds{1}_{Y\!=1} \right)}{\mathbb{P}(Y\!=\!1)} \\
        &= \frac{\mathbb{E}\left(Z \mathds{1}_{Y\!=0}\right) + \mathbb{E}\left(Z \mathds{1}_{Y\!=1}\right)}{p}~,
    \end{align}
\endgroup
which establishes~\eqref{eq:ent_gen_time_int}.
Here, in~(i) and~(v), we have used the definition of conditional expectation. 
Further, (ii) uses the fact that $N \sim \text{Geo}(p)$, (iii) uses linearity of expectation and~(iv) follows from the fact that $Y^{(i)}\!\overset{\text{iid}}{\sim}\! Y$ and $Z^{(i)}\!\overset{\text{iid}}{\sim}\! Z$. 

To show~\eqref{eq:zy1} and~\eqref{eq:inclusion-exclusion_part_2}, we use the principle of inclusion and exclusion. 
Recalling the definitions of $A_i^+$'s and $A_i^+ A_j^+$'s from~\eqref{eq:Aij+},
\begingroup
    \setlength{\abovedisplayskip}{6pt}   
    \setlength{\belowdisplayskip}{6pt}   
    \setlength{\abovedisplayshortskip}{2pt}
    \setlength{\belowdisplayshortskip}{2pt}
    \begin{align}
        &~\mathbb{E}\big(Z \mathds{1}_{Y=1} \big) \nonumber \\
        =& ~\mathbb{E}\big(Z \mathds{1}_{A_1^+} \big) \!+\! \mathbb{E}\big(Z \mathds{1}_{A_2^+} \big) \!+\!\mathbb{E}\big(Z \mathds{1}_{A_b^+} \big) \!-\! \mathbb{E}\big(Z \mathds{1}_{A_1^+ A_2^+} \big) \!-\! \mathbb{E}\big(Z \mathds{1}_{A_1^+ A_b^+} \big) \!-\! \mathbb{E}\big(Z \mathds{1}_{A_2^+ A_b^+} \big) \!+\! \mathbb{E}\big(Z \mathds{1}_{A_1^+A_2^+A_b^+} \big) \label{eq:zy1-1} \\
        \overset{(\text{i})}{=} & ~2\mathbb{E}\big(Z \mathds{1}_{A_1^+} \big) \!+\!\mathbb{E}\big(Z \mathds{1}_{A_b^+} \big) \!-\! \mathbb{E}\big(Z \mathds{1}_{A_1^+ A_2^+} \big) \!-\! 2\mathbb{E}\big(Z \mathds{1}_{A_1^+ A_b^+} \big) \!+\! \mathbb{E}\big(Z \mathds{1}_{A_1^+A_2^+A_b^+} \big) ~, \nonumber 
    \end{align}
\endgroup
where~(i) follows from the symmetry of the metropolitan links. 

For $Y \!=\!0$, we similarly have
\begin{align}
    \mathbb{E}\big(Z \mathds{1}_{Y\!=0} \big) = & ~2\mathbb{E}\big(Z \mathds{1}_{A_1^-} \big) \!+\!\mathbb{E}\big(Z \mathds{1}_{A_b^-} \big) \!-\! \mathbb{E}\big(Z \mathds{1}_{A_1^- A_2^-} \big) \!-\! 2\mathbb{E}\big(Z \mathds{1}_{A_1^- A_b^-} \big) \!+\! \mathbb{E}\big(Z \mathds{1}_{A_1^- A_2^- A_b^-} \big) \label{eq:inclusion-exclusion_part1} \\
    \overset{\text{(i)}}{=} & ~2\bigg(\! \mathbb{E}\big(Z \mathds{1}_{A_{12}^-}\big) \!+\! \mathbb{E}\big(Z \mathds{1}_{A_{1b}^-}\big) \!-\! \mathbb{E}\big(Z \mathds{1}_{A_{12}^-A_{1b}^-} \big) \!\bigg) \!+\! 2\mathbb{E}\big(Z \mathds{1}_{A_{b1}^-}\big) \!-\! \mathbb{E}\big(Z \mathds{1}_{A_{b1}^- A_{b2}^-}\big) \nonumber \\
    & \!-\! \mathbb{E}\big(Z \mathds{1}_{A_1^- A_2^-}\big) \!-\! 2\mathbb{E}\big(Z \mathds{1}_{A_1^-A_b^-} \big)~.  \nonumber 
\end{align}
Note that \( \mathbb{E}\big(Z \mathds{1}_{A_1^- A_2^- A_b^-}\big) \!=\! 0 \), as the subset \( A_1^- A_2^- A_b^- \) is empty.
In (i), we again use the principle of inclusion and exclusion as the events of the form $A_i^-$ are determined via $X_{\text{max}}$, while ${Z|Y=0 \sim X_\text{min}+t_\text{cut}}$.
Specifically, we have
\begingroup
    \setlength{\abovedisplayskip}{6pt}   
    \setlength{\belowdisplayskip}{6pt}   
    \setlength{\abovedisplayshortskip}{2pt}
    \setlength{\belowdisplayshortskip}{2pt}
    \begin{align}
        \mathds{1}_{A_1^-} = \mathds{1}_{A_{12}^-} + \mathds{1}_{A_{1b}^-} - \mathds{1}_{A_{12}^- {A_{1b}^-}}~, \quad
        \mathds{1}_{A_b^-} = \mathds{1}_{A_{b1}^-} + \mathds{1}_{A_{b2}^-} - \mathds{1}_{A_{b1}^- {A_{b2}^-}}~.
        \label{eq:term_expansion_A_1-andA_b-}
    \end{align}
\endgroup
Also by symmetry, $\mathbb{E}\big(Z\mathds{1}_{A_{b1}^-}\big) = \mathbb{E}\big(Z\mathds{1}_{A_{b2}^-}\big)$.
We have thus established~\eqref{eq:zy1} and~\eqref{eq:inclusion-exclusion_part_2}.
\end{proof}

We calculate the individual terms of~\eqref{eq:zy1} in~\ref{sec:individual_terms_for_rate}, which are given by~\eqref{eq:finite_sum_A_1^+}, \eqref{eq:finite_sum_A_b^+}, \eqref{eq:finite_sum_A_1^+A_2^+}, \eqref{eq:finite_sum_A_1^+A_b^+}, and \eqref{eq:finite_sum_A_1^+A_2^+A_b^+}.
For~\eqref{eq:inclusion-exclusion_part_2}, the calculable expressions of the individual terms  are given in~\eqref{eq:finite_sum_A_12^-}, \eqref{eq:finite_sum_A_1b^-}, \eqref{eq:finite_sum_A_b1^-}, \eqref{eq:finite_sum_A_12^-A_1b^-}, \eqref{eq:finite_sum_A_b1^-A_b2^-}, \eqref{eq:finite_sum_A_1^-A_2^-}, and \eqref{eq:finite_sum_A_1^-A_b^-}.
The calculation of the end-to-end entanglement generation success probability $p$ follows directly from the fidelity calculation.
Hence, we present the derivation in Sec.~\ref{sec:derivation_fidelity_IN}.
In particular, $p\!=\! U_1(0)$; see~\eqref{eq:p_from_U_1}.
Plugging these in~\eqref{eq:zy1},~\eqref{eq:inclusion-exclusion_part_2}, and~\eqref{eq:ent_gen_time_int} and subsequently in~\eqref{eq:def_rate_teleportation_int}, we obtain the teleportation rate in the IN.
Next, we derive the expressions for the expected teleportation fidelity.
\vspace{-5pt}

\subsection{Derivation of Expected Teleportation Fidelity in the Intercity Network}
\label{sec:derivation_fidelity_IN}

In our model of IN, the fidelity of the teleported qubit for ER teleportation depends on (i) the fidelity of the end-to-end entangled link and (ii) the time taken to transmit the classical Pauli correction message from the sender node to the receiver $t_\text{int}^\text{class}$; see~\eqref{eq:t_class_int}.
Furthermore, for the QR teleportation, since the data qubit decoheres until the end-to-end link generation, the teleportation fidelity also depends on the time required to generate this link.
Hence, we first focus on the derivation of the expected fidelity of the end-to-end link in the next lemma.
\begin{lemma}
\label{lemma_fidelity_e2e_entanglement}
The expected fidelity of an end-to-end entangled link in the IN is given by
\begingroup
    \setlength{\abovedisplayskip}{6pt}   
    \setlength{\belowdisplayskip}{6pt}   
    \setlength{\abovedisplayshortskip}{2pt}
    \setlength{\belowdisplayshortskip}{2pt}
    \begin{align}
        F_\mathrm{e2e} &= \frac{1 + 3\mathbb{E}({w}_\mathrm{e2e})}{4}\,, ~\textbf{}\text{where} \quad 
        \mathbb{E}({w}_\mathrm{e2e}) = \frac{w_\mathrm{m}^2\, w_\mathrm{b}\, e^{-k t_\mathrm{msg}}}{p} \, U_1(k)\,, ~\text{with} \label{eq:fidelity_e2e_entanglement} \\ 
        U_1(k) &:= \mathbb{E}\big(
            e^{-k X_\mathrm{diff}} \, \mathds{1}_{Y=1}
        \big)~, \quad k := \frac{2}{t_\mathrm{coh}}\,, ~\text{and}
        \label{eq:def_U_1_k} \\
        X_\mathrm{diff} &:= 
        \begin{cases}
            X_\mathrm{max} - X_\mathrm{min}, & \text{for } A_1^+ \text{ and } A_2^+, \\
            2X_b - X_1 - X_2, & \text{for } A_b^+~.
        \end{cases} \label{eq:def_X_diff}
    \end{align}
\endgroup
\looseness=-1
Here, $w_\text{m}$ and $w_\text{b}$ are the Werner parameters of freshly generated links between an end node and its neighbouring border node, and between two border nodes across the backbone, respectively.
Furthermore,
\begingroup
    \setlength{\abovedisplayskip}{6pt}   
    \setlength{\belowdisplayskip}{6pt}   
    \setlength{\abovedisplayshortskip}{2pt}
    \setlength{\belowdisplayshortskip}{2pt}
    \begin{align}
        U_1(v) &=  2\bigg( \!\mathbb{E}\big(e^{-v \thinspace (X_1 - X_{2})} \mathds{1}_{A_{12}^+}\big) \!+\! \mathbb{E}\big(e^{-v \thinspace (X_1 - X_{b})} \mathds{1}_{A_{1b}^+}\big) \!-\! \mathbb{E}\big(e^{-v \thinspace (X_1 - X_{2})} \mathds{1}_{A_{12}^+A_{1b}^+}\big) \!\bigg) \!+\! \mathbb{E}\big(e^{-v \thinspace (2X_b - X_1 - X_2)} \mathds{1}_{A_{b}^+}\big) \nonumber \\
        & \quad \!-\! \mathbb{E}\big(e^{-v \thinspace (X_1 - X_b)} \mathds{1}_{A_1^+A_2^+}\big) \!-\! 2\mathbb{E}\big(e^{-v \thinspace (X_1 - X_{2})} \mathds{1}_{A_1^+A_b^+}\big) \!+\! \mathbb{E}\big(\mathds{1}_{A_1^+ A_2^+ A_b^+} \big)~.
        \label{eq:term_expansion_U_1}
    \end{align}
\endgroup
\end{lemma}

\begin{proof}
    Recall from~\ref{assumptionRate:individualDurations} that $X_1$, $X_2$, and $X_b$ are the duration of successfully generating entanglement between $\text{P}_1$ (or $\text{P}_2$) and $\text{J}_1$, $\text{P}_3$ (or $\text{P}_4$) and $\text{J}_2$, and $\text{J}_1$ and $\text{J}_2$, respectively.
    Furthermore, recall from~\ref{assumption:werner_state} that an entanglement swap between two Werner states results in a Werner state with the corresponding Werner parameter being the product of those of the initial states.
    Since decoherence affects the states and the effect depends on the chronological order of the generation times, the end-to-end link fidelity differs across the events defined in~\eqref{eq:Aij+} and~\eqref{eq:Aij-}.
    We describe the evolution of the Werner parameter for each case in the Tab.~\ref{tab:evolution_e2e_werner_state}
    \begin{table}[htbp]
    \centering
    \caption{Derivation of the Werner parameter of the end-to-end entangled state. In the last column, we include the constant communication delay $t_\text{msg}$, which accounts for the time required to transmit the final swap message from border nodes to both end nodes. We set $k = 2/t_\text{coh}$.
    Note that due to symmetry of the metropolitan links, the contributions from the events $A_{21}^+$, $A_{2b}^+$, and $A_{b2}^+$ are identical to those from $A_{12}^+$, $A_{1b}^+$, and $A_{b1}^+$, respectively. The final column is precisely ${w_\text{m}^2 w_\text{b} e^{-k(X_\text{diff}+t_\text{msg})}}$; see~\eqref{eq:def_X_diff}.}
    \label{tab:evolution_e2e_werner_state}
    \renewcommand{\arraystretch}{1.2}
    \begin{tabular}{|c|l|l|l|l|}
        \hline
        \textbf{Set} & \textbf{$X_\text{min}$} & \textbf{$X_\text{mid}$} & \textbf{$X_\text{max}$} & \textbf{End-to-end link} \\
        \hline
        \multirow{3}{*}{$A_{12}^+$} 
        & \multirow{3}{*}{link 2: $w_\text{m}$} 
        & link 2: $w_\text{m} e^{-k(X_b-X_2)}$ 
        & link b2: $w_\text{m} w_\text{b} e^{-k(X_1-X_2)}$ 
        & \multirow{3}{*}{\(\begin{array}{c}
            w_\text{m}^2 w_\text{b} e^{-k(X_1-X_2)}\\
            \cdot \,e^{-k t_\text{msg}}
        \end{array}\)} \\
        & & link b: $w_\text{b}$ & link 1: $w_\text{m}$ & \\
        \cline{3-4}
        & & link b2: $w_\text{m} w_\text{b} e^{-k(X_b-X_2)}$ & link 1b2: $w_\text{m}^2w_\text{b}e^{-k(X_1-X_2)}$ & \\
        \hline
        \multirow{3}{*}{$A_{1b}^+$} 
        & \multirow{3}{*}{link b: $w_\text{b}$} 
        & link b: $w_\text{b} e^{-k(X_2-X_b)}$ 
        & link b2: $w_\text{m} w_\text{b} e^{-k(X_1-X_b)}$ 
        & \multirow{3}{*}{\(\begin{array}{c}
            w_\text{m}^2 w_\text{b} e^{-k(X_1-X_b)}\\
            \cdot \,e^{-k t_\text{msg}}
        \end{array}\)} \\
        & & link 2: $w_\text{m}$ & link 1: $w_\text{m}$ & \\
        \cline{3-4}
        & & link b2: $w_\text{m} w_\text{b} e^{-k(X_2-X_b)}$ & link 1b2: $w_\text{m}^2 w_\text{b} e^{-k(X_1-X_b)}$ & \\
        \hline
        \multirow{4}{*}{$A_{b1}^+$} 
        & \multirow{4}{*}{link 1: $w_\text{m}$} 
        & link 1: $w_\text{m} e^{-k(X_2-X_1)}$ 
        & link 1: $w_\text{m} e^{-k(X_b-X_1)}$ 
        & \multirow{4}{*}{\(\begin{array}{c}
            w_\text{m}^2 w_\text{b} e^{-k(2X_b-X_1-X_2)}\\
            \cdot \,e^{-k t_\text{msg}}
        \end{array}\)} \\
        & & link 2: $w_\text{m}$ & link 2: $w_\text{m} e^{-k(X_b-X_2)}$ & \\
        & &  & link b: $w_\text{b}$ & \\
        \cline{3-4}
        & & swap not possible & link 1b2: $w_\text{m}^2 w_\text{b} e^{-k(2X_b-X_1-X_2)}$ & \\
        \hline
    \end{tabular}
    \vspace{-5pt}
\end{table}

Since unsuccessful trials do not influence the fidelity of the final end-to-end entangled link, we restrict our analysis to the case ${Y \!=\! 1}$.
In this case, we use the shorthand $X_\text{diff}$ introduced in~\eqref{eq:def_X_diff} to succinctly express the contribution of different sub-events by 
\begingroup
    \setlength{\abovedisplayskip}{4pt}   
    \setlength{\belowdisplayskip}{4pt}   
    \setlength{\abovedisplayshortskip}{2pt}
    \setlength{\belowdisplayshortskip}{2pt}
    \begin{align}
    \label{eq:event_werner_parameter_e2e_entanglement}
        w_\text{e2e} = w_\text{m}^2 w_\text{b} e^{-k(X_\text{diff}+t_\text{msg})}~,
    \end{align}
\endgroup
\looseness = -1 see Tab.~\ref{tab:evolution_e2e_werner_state} for an explanation.
Thus, we obtain the expected Werner parameter of the end-to-end link as
\begingroup
    \setlength{\abovedisplayskip}{-8pt}   
    \setlength{\belowdisplayskip}{4pt}   
    \setlength{\abovedisplayshortskip}{2pt}
    \setlength{\belowdisplayshortskip}{2pt}
    \begin{align}
    \label{eq:expected_werner_parameter_e2e_entanglement}
      \mathbb{E}(w_\text{e2e}) = &~ \mathbb{E} (w_\text{m}^2 w_\text{b} e^{-k(X_\text{diff}+t_\text{msg})} \mid Y\!=\!1) \\
      = &~ w_\text{m}^2 w_\text{b} e^{-k t_\text{msg}} \, \mathbb{E}( 
        e^{-k X_\text{diff}} \mid Y \!=\! 1 
      ) \\
      = &~ w_\text{m}^2 w_\text{b} e^{-k t_\text{msg} } \,
        \frac{\mathbb{E}(
            e^{-{k X_\text{diff}}} \, \mathds{1}_{Y=1} 
          )}{\mathbb{P}(Y = 1)} \\
        \overset{\eqref{eq:def_p}, \eqref{eq:def_U_1_k}}{=} &~ \frac{w_\text{m}^2 w_\text{b} e^{-k t_\text{msg} }}{p}\, U_1(k)~, \nonumber
    \end{align}
\endgroup
and the corresponding expected fidelity is $(1+3\mathbb{E}({w}_\text{e2e}))/4$,
which establishes~\eqref{eq:fidelity_e2e_entanglement}.

To show~\eqref{eq:term_expansion_U_1}, we use the principle of inclusion and exclusion.
Recalling the definitions of $A_i^+$ and $A_{ij}^+$ from~\eqref{eq:Aij+}, we have
\begingroup
    \setlength{\abovedisplayskip}{4pt}   
    \setlength{\belowdisplayskip}{2pt}   
    \setlength{\abovedisplayshortskip}{2pt}
    \setlength{\belowdisplayshortskip}{2pt}
    \begin{align}
        U_1(v) =&~ \mathbb{E}\big( e^{-v X_\mathrm{diff}} \, \mathds{1}_{Y=1} \big)  \\
        =&~\mathbb{E}\big(e^{-v X_{\text{diff}}} \mathds{1}_{A_{1}^+}\big) \!+\! \mathbb{E}\big(e^{-v X_{\text{diff}}} \mathds{1}_{A_{2}^+}\big) \!+\! \mathbb{E}\big(e^{-v X_{\text{diff}}} \mathds{1}_{A_{b}^+}\big) \!-\! \mathbb{E}\big(e^{-v X_{\text{diff}}} \mathds{1}_{A_1^+A_2^+}\big) \!-\! \mathbb{E}\big(e^{-v X_{\text{diff}}} \mathds{1}_{A_1^+A_b^+}\big) \nonumber \\
        &~\!-\! \mathbb{E}\big(e^{-v X_{\text{diff}}} \mathds{1}_{A_2^+A_b^+}\big) 
         \!+\! \mathbb{E}\big(e^{-v X_{\text{diff}}}\mathds{1}_{A_1^+ A_2^+ A_b^+} \big) \\
        \overset{(\text{i})}{=}&~ 2\mathbb{E}\big(e^{-v X_{\text{diff}}} \mathds{1}_{A_{1}^+}\big) + \mathbb{E}\big(e^{-v X_{\text{diff}}} \mathds{1}_{A_{b}^+}\big) - \mathbb{E}\big(e^{-v X_{\text{diff}}} \mathds{1}_{A_1^+A_2^+}\big) - 2\mathbb{E}\big(e^{-v X_{\text{diff}}} \mathds{1}_{A_1^+A_b^+}\big) \nonumber \\
        &~ + \mathbb{E}\big(e^{-v X_{\text{diff}}}\mathds{1}_{A_1^+ A_2^+ A_b^+} \big) \\
        \overset{(\text{ii})}{=}&~ 2\bigg(\! \mathbb{E}\big(e^{-v \thinspace (X_1 - X_{2})} \mathds{1}_{A_{12}^+}\big) \!+\! \mathbb{E}\big(e^{-v \thinspace (X_1 - X_{b})} \mathds{1}_{A_{1b}^+}\big) \!-\! \mathbb{E}\big(e^{-v \thinspace (X_1 - X_{2})} \mathds{1}_{A_{12}^+A_{1b}^+}\big) \!\bigg) \!+\! \mathbb{E}\big(e^{-v \thinspace (2X_b - X_1 - X_2)} \mathds{1}_{A_{b}^+}\big) \nonumber \\
        &~ \!-\! \mathbb{E}\big(e^{-v \thinspace (X_1 - X_b)} \mathds{1}_{A_1^+A_2^+}\big) \!-\! 2\mathbb{E}\big(e^{-v \thinspace (X_1 - X_{2})} \mathds{1}_{A_1^+A_b^+}\big) \!+\! \mathbb{E}\big(\mathds{1}_{A_1^+ A_2^+ A_b^+} \big)~, \nonumber
    \end{align}
\endgroup
where (i) follows from the symmetry of the metropolitan links.
In (ii), we have used the principle of inclusion and exclusion for the event $A_1^+$ since $A_1^+$ is defined via $X_\text{max}$, while ${X_\text{diff}\!=\!X_\text{max} \!-\! X_\text{min}}$ for $A_1^+$ (see~\eqref{eq:def_X_diff}).
Specifically, we have
\begingroup
    \setlength{\abovedisplayskip}{6pt}   
    \setlength{\belowdisplayskip}{6pt}   
    \setlength{\abovedisplayshortskip}{2pt}
    \setlength{\belowdisplayshortskip}{2pt}
    \begin{align}
        \mathds{1}_{A_1^+} = \mathds{1}_{A_{12}^+} + \mathds{1}_{A_{1b}^+} - \mathds{1}_{A_{12}^+ {A_{1b}^+}}~,
        \label{eq:term_expansion_A_1+}
    \end{align}
\endgroup
Furthermore, we have substituted the value of $X_\text{diff}$ from~\eqref{eq:def_X_diff}.
Note that, for $A_1^+A_2^+A_b^+$, we have ${X_\text{diff} \!=\! 0}$ which follows from the definition~\eqref{eq:def_X_diff}.
Thus, we have established~\eqref{eq:term_expansion_U_1}.
\end{proof}

Note that the calculation of $U_1(v)$ facilitates the calculation of $p$.
Specifically, we have
\begingroup
    \setlength{\abovedisplayskip}{6pt}   
    \setlength{\belowdisplayskip}{6pt}   
    \setlength{\abovedisplayshortskip}{2pt}
    \setlength{\belowdisplayshortskip}{2pt}
    \begin{align}
    \label{eq:p_from_U_1}
        p = \mathbb{E}(\mathds{1}_{Y=1}) = U_1(0)~.
    \end{align}
\endgroup
We describe the computation procedure for $U_1(v)$ in~\ref{subsec:terms_U_alpha} and subsequently calculate $p$ as $U_1(0)$, which helps us calculate the rate and fidelity via~\eqref{eq:ent_gen_time_int}, \eqref{eq:fidelity_e2e_entanglement}, \eqref{eq:fidelity_ent_ready}, and~\eqref{eq:fidelity_qubit_ready}.

The fidelity of teleportation now follows straightforwardly from Lemma~\ref{lemma_fidelity_e2e_entanglement} as described in the next theorem.

\begin{theorem}
    The expected fidelity of a teleported qubit in the ER case is given by
    \begin{align}
    \label{eq:fidelity_ent_ready}
      \textnormal{(Expected fidelity ER)} \qquad F_\mathrm{int}^\mathrm{ER} \!=\! \frac{1 \!+\! \mathbb{E}(w_\mathrm{e2e})e^{-kt_{\mathrm{int}}^{\mathrm{class}}/2}}{2} \!=\! \frac{1}{2} \!+\! \frac{w_\mathrm{m}^2 w_\mathrm{b} e^{-k( t_\mathrm{msg} \!+\! t_{\mathrm{int}}^{\mathrm{class}}/2)}}{2p} \, U_1(k)~, \quad
    \end{align}
    where $U_1(\cdot)$ and $k$ is defined in~\eqref{eq:def_U_1_k}.
    $w_\text{m}$ and $w_\text{b}$ are the Werner parameters of freshly generated links between an end node and its neighbouring border node, and between two border nodes across the backbone, respectively. 
\end{theorem}

\begin{proof}
    In the ER scenario, only the qubit stored in the receiver node undergoes decoherence during the propagation time of the Pauli correction message between the end nodes $t_\text{int}^{\text{class}}$.
    Since the end-to-end link has Werner parameter $w_{e2e}$~\eqref{eq:event_werner_parameter_e2e_entanglement}, using~\eqref{eq:teleportation_fidelity_general_expression} we obtain the expected fidelity of the teleported qubit as
    \begin{align*}
    \mathbb{E}\Big(\frac{1+w_\text{e2e}e^{-kt_\text{int}^\text{class}/2}}{2} \Big) &= \frac{1+\mathbb{E}(w_\text{e2e})e^{-kt_\text{int}^\text{class}/2}}{2} \overset{\eqref{eq:fidelity_e2e_entanglement}}{=} \frac{1}{2} + \frac{w_\text{m}^2 w_\text{b} e^{-k( t_\text{msg} + t_\text{int}^{\text{class}}/2)}}{2p} 
         \, U_1(k)~,
    \end{align*}
    which establishes~\eqref{eq:fidelity_ent_ready}.
\end{proof}

\begin{theorem}
    The expected fidelity of a teleported qubit in the QR case is given by
    \begin{align}
      \textnormal{(Expected fidelity QR)} \qquad F_\mathrm{int}^\mathrm{QR} =& ~ \frac{1}{2} + \frac{1}{2}w_\mathrm{m}^2 w_\mathrm{b} e^{-k(3t_\mathrm{msg} + t_\mathrm{int}^{\mathrm{class}}/2)} \frac{e^{-k t_\mathrm{msg}/2} \, U_2(k)}{1 - e^{-k t_\mathrm{cut}/2} \, U_3(k)}\,, ~\text{where} \label{eq:fidelity_qubit_ready} \\
    U_2(v) :=&~ \mathbb{E}\big(
        e^{-v (X_\mathrm{diff} + X_\mathrm{max})/2} \, \mathds{1}_{Y=1}
    \big)~, \quad k = \frac{2}{t_\mathrm{coh}}~, \label{eq:def_U_2} \\
    U_3(v) :=&~ \mathbb{E}\big(
        e^{-v X_\mathrm{min}/2} \, \mathds{1}_{Y=0} \big)~.  \label{eq:def_U_3}
\end{align}
Here,  $t_\mathrm{coh}$ is the memory coherence time of an end or border node. $w_\mathrm{m}$ and $w_\mathrm{b}$ are the Werner parameters of freshly generated links between an end node and its neighbouring border node, and between two border nodes across the backbone, respectively.
Furthermore,
\begingroup
    \setlength{\abovedisplayskip}{6pt}   
    \setlength{\belowdisplayskip}{6pt}   
    \setlength{\abovedisplayshortskip}{2pt}
    \setlength{\belowdisplayshortskip}{2pt}
    \begin{align}
        U_2(v) =&~ 2 \bigg(\! \mathbb{E}\big(e^{-k\thinspace(3 X_1/2-X_{2})}\mathds{1}_{A_{12}^+}) + \mathbb{E}\big(e^{-k\thinspace(3 X_1/2-X_{b})}\mathds{1}_{A_{1b}^+}) - \mathbb{E}\big(e^{-k\thinspace(3 X_1/2-X_{2})}\mathds{1}_{A_{12}^+ A_{1b}^+}) \!\bigg) \nonumber \\
        &~ + \mathbb{E}(e^{-k\thinspace (5 X_b/2-X_1-X_2)}\mathds{1}_{A_b^+}) - \mathbb{E}(e^{-k\thinspace(3 X_1/2-X_b)}\mathds{1}_{A_1^+A_2^+}) - 2 \mathbb{E}(e^{-k\thinspace (3 X_1/2-X_2)}\mathds{1}_{A_1^+A_b^+}) \nonumber \\
        &~ + \mathbb{E}(e^{-k\thinspace X_{1}/2}\mathds{1}_{A_1^+ A_2^+ A_b^+}) ~,    \label{eq:term_expansion_U_2} \\
        U_3(v) =&~ 2\bigg(\! \mathbb{E}\big(e^{-kX_{2}/2} \thinspace\mathds{1}_{A_{12}^-}\big) + \mathbb{E}\big(e^{-kX_{b}/2} \thinspace\mathds{1}_{A_{1b}^-}\big) - \mathbb{E}\big(e^{-kX_{2}/2} \thinspace\mathds{1}_{A_{12}^- A_{1b}^-}\big)  \!\bigg) + 2\mathbb{E}\big(e^{-kX_{1}/2} \thinspace\mathds{1}_{A_{b1}^-}\big) \nonumber \\
        &  - \mathbb{E}\big(e^{-kX_{1}/2} \thinspace\mathds{1}_{A_{b1}^- A_{b2}^-}\big) - \mathbb{E}\big(e^{-kX_{b}/2} \thinspace\mathds{1}_{A_1^- A_2^-}\big) - 2 \mathbb{E}\big(e^{-kX_{2}/2} \thinspace\mathds{1}_{A_1^- A_b^-}\big)~. \label{eq:term_expansion_U_3}
    \end{align}
\endgroup
\end{theorem}

\begin{proof}
    In the QR scenario, the data qubit is first prepared at the sender node in a pure state, and the end-to-end link generation process begins immediately.
    During the failed rounds in end-to-end link generation, the data qubit experiences decoherence for a duration of
    \begingroup
    \setlength{\abovedisplayskip}{4pt}   
    \setlength{\belowdisplayskip}{4pt}   
    \setlength{\abovedisplayshortskip}{2pt}
    \setlength{\belowdisplayshortskip}{2pt}
        \begin{align}
            \sum_{i=1}^{N-1} Z^{(i)}_0~,\quad 
        \text{where}~~ {Z^{(i)}_0 \overset{\text{iid}}{\sim} Z \mid Y\!=\!0}~, \nonumber
        \end{align}
    \endgroup
    and $N$ denotes the number of required rounds until success; see~\eqref{eq:def_Z}.
    During the successful round, the data qubit decoheres further for a duration of 
    \begingroup
    \setlength{\abovedisplayskip}{4pt}   
    \setlength{\belowdisplayskip}{4pt}   
    \setlength{\abovedisplayshortskip}{2pt}
    \setlength{\belowdisplayshortskip}{2pt}
        \begin{align}
            Z_1 \sim Z\mid Y\!=\!1~; \nonumber
        \end{align}
    \endgroup
    see~\eqref{eq:def_Z}.
    Therefore, at the beginning of teleportation, the data qubit has decohered to the state 
\begingroup
    \setlength{\abovedisplayskip}{4pt}   
    \setlength{\belowdisplayskip}{4pt}   
    \setlength{\abovedisplayshortskip}{2pt}
    \setlength{\belowdisplayshortskip}{2pt}
    \begin{align}
    \label{eq:p_d_in_QR}
        p_d \ket{\phi}\bra{\phi} + (1 - p_d) \frac{\mathbb{I}}{2}~,~ \text{where}~ p_d := e^{-(\sum_{i=1}^{N-1} Z^{(i)}_0 + Z_1) / t_{\text{coh}}}
        = e^{-k(\sum_{i=1}^{N-1} Z^{(i)}_0 + Z_1) / 2}~,
        ~k = 2 / t_\text{coh}~.
    \end{align}
\endgroup

After the BSM is performed on the data qubit and one-half of the end-to-end link, the qubit stored at the receiver node memory undergoes decoherence during the propagation time of the Pauli correction message between end nodes $t_\text{int}^{\text{class}}$.
Therefore, using~\eqref{eq:teleportation_fidelity_general_expression}, we obtain the expected fidelity of the teleported qubit as  
\begingroup
    \setlength{\abovedisplayskip}{4pt}   
    \setlength{\belowdisplayskip}{4pt}   
    \setlength{\abovedisplayshortskip}{2pt}
    \setlength{\belowdisplayshortskip}{2pt}
    \begin{align}
        &\mathbb{E}\Big(\frac{1+w_\text{e2e} e^{-k(\sum_{i=1}^{N-1} Z^{(i)}_0 + Z_1) / 2} e^{-kt_{\mathrm{int}}^{\mathrm{class}}/2}}{2}\Big) \nonumber \\
        \overset{\eqref{eq:fidelity_e2e_entanglement}}{=} &\frac{1}{2} + \frac{1}{2}w_\text{m}^2 w_\text{b} e^{-k(3t_{\text{msg}}+t_\text{int}^{\text{class}})/2} \underbrace{\mathbb{E}\big(e^{-k\sum_{i=1}^{N-1} Z_0^{(i)} / 2) }\big) \thinspace \mathbb{E} \big(e^{-k(X_{\text{diff}} + Z_1/2)} \big)}_{=:\,U'} \label{eq:U_prime_in_werner_paramater} ~.
    \end{align}
\endgroup
Recall from~\eqref{eq:def_p} that $p = \mathbb{P}(Y\!=\!1)$.
Now the term $U^\prime$ in~\eqref{eq:U_prime_in_werner_paramater} can be simplified further as follows: 
\begingroup
    \setlength{\abovedisplayskip}{-8pt}   
    \setlength{\belowdisplayskip}{4pt}   
    \setlength{\abovedisplayshortskip}{2pt}
    \setlength{\belowdisplayshortskip}{2pt}
    \begin{align}
        U^\prime = & ~\mathbb{E}\big(e^{-k\sum_{i=1}^{N-1} Z_0^{(i)} / 2) }\big) \thinspace \mathbb{E} \big(e^{-k(X_{\text{diff}} + Z_1/2)} \big) \\
        = & ~\sum_{n=1}^\infty p(1-p)^{n-1} \Big(\mathbb{E}\big(e^{-kZ_0/2}\big)\Big)^{n-1} \thinspace \mathbb{E}\big(e^{-k(X_{\text{diff}} + Z_1/2)}\big) \\
        = & ~\frac{p \thickspace \mathbb{E}\big(e^{-k(X_{\text{diff}} + Z_1/2)} \big)}{1 - (1-p) \thickspace \mathbb{E}\big(e^{-kZ_0/2} \big)} \\
        = & ~\frac{\mathbb{E}\big(e^{-k(X_{\text{diff}} + Z/2)} \thinspace\mathds{1}_{Y=1} \big)}{1 - \mathbb{E}\big(e^{-kZ/2} \thinspace\mathds{1}_{Y=0} \big)} \\
        \overset{\eqref{eq:def_Z}}{=} 
        &~\frac{e^{-k t_{\text{msg}} / 2} \thinspace \mathbb{E}\!\left(e^{-k\left( X_{\text{diff}} + X_{\text{max}}/2 \right)} \thinspace\mathds{1}_{Y=1} \right)}
        {1 - e^{-k t_\text{cut} / 2} \thinspace \mathbb{E}\!\left(e^{-k X_{\text{min}}/2} \thinspace\mathds{1}_{Y=0} \right)} \\
        \overset{\eqref{eq:def_U_2},\eqref{eq:def_U_3}}{=} 
        &~\frac{e^{-k t_{\text{msg}} / 2} \, U_2(k) }
        {1 - e^{-k t_\text{cut} / 2} \, U_3(k) } ~.
        \label{eq:expected_term_qubit_ready}
    \end{align}
\endgroup
Substituting~\eqref{eq:expected_term_qubit_ready} in~\eqref{eq:U_prime_in_werner_paramater} establishes~\eqref{eq:fidelity_qubit_ready}.
To show~\eqref{eq:term_expansion_U_2} and~\eqref{eq:term_expansion_U_3}, we use the principle of inclusion and exclusion.
Recalling the definitions of $A_i^+$ and $A_{ij}^+$ from~\eqref{eq:Aij+}, we have
\begingroup
    \setlength{\abovedisplayskip}{4pt}   
    \setlength{\belowdisplayskip}{4pt}   
    \setlength{\abovedisplayshortskip}{2pt}
    \setlength{\belowdisplayshortskip}{2pt}
    \begin{align}
        U_2(v) =&~ \mathbb{E}\big(e^{-v( X_{\text{diff}}+X_{\text{max}} / 2)} \thinspace\mathds{1}_{Y=1} \big)  \\
        =&~ \mathbb{E}\big(e^{-v (X_{\text{diff}}+X_{\text{max}} / 2)}\mathds{1}_{A_{1}^+}) \!+\! \mathbb{E}\big(e^{-v(X_{\text{diff}}+X_{\text{max}} / 2)}\mathds{1}_{A_{2}^+}) \!+\! \mathbb{E}\big(e^{-v(X_{\text{diff}}+X_{\text{max}} / 2)}\mathds{1}_{A_{b}^+}) \nonumber \\
        &~ \!-\! \mathbb{E}(e^{-v (X_{\text{diff}}+X_{\text{max}} / 2)}\mathds{1}_{A_1^+A_2^+}) \!-\! \mathbb{E}(e^{-v (X_{\text{diff}}+X_{\text{max}} / 2)}\mathds{1}_{A_1^+A_b^+}) \nonumber \\
        &~  \!-\! \mathbb{E}\big(e^{-v (X_{\text{diff}}+X_{\text{max}} / 2)}\mathds{1}_{A_{2}^+ A_{b}^+}) \!+\! \mathbb{E}(e^{-v(X_{\text{diff}}+X_{\text{max}} / 2)}\mathds{1}_{A_1^+ A_2^+ A_b^+})  \\
        \overset{(\text{i})}{=}&~ 2\mathbb{E}\big(e^{-v (X_{\text{diff}}+X_{\text{max}} / 2)}\mathds{1}_{A_{1}^+}) \!+\! \mathbb{E}\big(e^{-v(X_{\text{diff}}+X_{\text{max}} / 2)}\mathds{1}_{A_{b}^+}) \!-\! \mathbb{E}(e^{-v (X_{\text{diff}}+X_{\text{max}} / 2)}\mathds{1}_{A_1^+A_2^+}) \nonumber \\
        &~  \!-\! 2\mathbb{E}(e^{-v (X_{\text{diff}}+X_{\text{max}} / 2)}\mathds{1}_{A_1^+A_b^+}) \!+\! \mathbb{E}(e^{-v(X_{\text{diff}}+X_{\text{max}} / 2)}\mathds{1}_{A_1^+ A_2^+ A_b^+}) \\
        \overset{(\text{ii})}{=}&~ 2 \bigg(\! \mathbb{E}\big(e^{-v\thinspace(3 X_1/2-X_{2})}\mathds{1}_{A_{12}^+}) \!+\! \mathbb{E}\big(e^{-v\thinspace(3 X_1/2-X_{b})}\mathds{1}_{A_{1b}^+}) \!-\! \mathbb{E}\big(e^{-v\thinspace(3 X_1/2-X_{2})}\mathds{1}_{A_{12}^+ A_{1b}^+}) \!\bigg) \nonumber \\
        &~ \!+\! \mathbb{E}(e^{-v\thinspace (5 X_b/2-X_1-X_2)}\mathds{1}_{A_b^+}) \!-\! \mathbb{E}(e^{-v\thinspace(3 X_1/2-X_b)}\mathds{1}_{A_1^+A_2^+}) - 2 \mathbb{E}(e^{-v\thinspace (3 X_1/2-X_2)}\mathds{1}_{A_1^+A_b^+}) \nonumber \\
        &~ \!+\! \mathbb{E}(e^{-v\thinspace X_{1}/2}\mathds{1}_{A_1^+ A_2^+ A_b^+}) ~, \nonumber
    \end{align}
\endgroup
where (i) follows from the symmetry of the metropolitan links.
In (ii), we have used the principle of inclusion and exclusion on the event $A_1^+$ using~\eqref{eq:term_expansion_A_1+} and expanded the shorthand $X_\text{diff}$ from~\eqref{eq:def_X_diff}.
Thus, we established~\eqref{eq:term_expansion_U_2}.
Similarly, we can expand $U_3(v)$ as
\begingroup
    \setlength{\abovedisplayskip}{4pt}   
    \setlength{\belowdisplayskip}{4pt}   
    \setlength{\abovedisplayshortskip}{2pt}
    \setlength{\belowdisplayshortskip}{2pt}
    \begin{align}
        U_3(v) =&~ \mathbb{E}\big(
            e^{-v X_\mathrm{min}/2} \, \mathds{1}_{Y=0} \big) \\
        =&~ \mathbb{E}\big(
            e^{-v X_\mathrm{min}/2} \, \mathds{1}_{A_1^-} \big) \!+\! \mathbb{E}\big(
            e^{-v X_\mathrm{min}/2} \, \mathds{1}_{A_2^-} \big) \!+\! \mathbb{E}\big(
            e^{-v X_\mathrm{min}/2} \, \mathds{1}_{A_b^-} \big) \!-\! \mathbb{E}\big(
            e^{-v X_\mathrm{min}/2} \, \mathds{1}_{A_1^-A_2^-} \big) \nonumber \\
            &~ \!-\! \mathbb{E}\big(
            e^{-v X_\mathrm{min}/2} \, \mathds{1}_{A_1^-A_b^-} \big) \!-\! \mathbb{E}\big(
            e^{-v X_\mathrm{min}/2} \, \mathds{1}_{A_2^-A_b^-} \big) \!+\! \mathbb{E}\big(
            e^{-v X_\mathrm{min}/2} \, \mathds{1}_{A_1^-A_2^-A_b^-} \big) \\
        \overset{(\text{i})}{=}&~ 2\mathbb{E}\big(
            e^{-v X_\mathrm{min}/2} \, \mathds{1}_{A_1^-} \big) \!+\! \mathbb{E}\big(
            e^{-v X_\mathrm{min}/2} \, \mathds{1}_{A_b^-} \big) \!-\! \mathbb{E}\big(
            e^{-v X_\mathrm{min}/2} \, \mathds{1}_{A_1^-A_2^-} \big) \!-\! 2\mathbb{E}\big(
            e^{-v X_\mathrm{min}/2} \, \mathds{1}_{A_1^-A_b^-} \big)~ \\
        \overset{(\text{ii})}{=}&~ 2\bigg(\! \mathbb{E}\big(
            e^{-v X_\mathrm{min}/2} \, \mathds{1}_{A_{12}^-} \big) \!+\! \mathbb{E}\big(
            e^{-v X_\mathrm{min}/2} \, \mathds{1}_{A_{1b}^-} \big) \!-\! \mathbb{E}\big(
            e^{-v X_\mathrm{min}/2} \, \mathds{1}_{A_{12}^-A_{1b}^-} \big) \!\bigg) \nonumber \\
            &~ \!+\! \bigg(\! \mathbb{E}\big(
            e^{-v X_\mathrm{min}/2} \, \mathds{1}_{A_{b1}^-} \big) \!+\! \mathbb{E}\big(
            e^{-v X_\mathrm{min}/2} \, \mathds{1}_{A_{b2}^-} \big) \!-\! \mathbb{E}\big(
            e^{-v X_\mathrm{min}/2} \, \mathds{1}_{A_{b1}^-A_{b2}^-} \big) \!\bigg)  \nonumber \\
            &~ \!-\! \mathbb{E}\big(
            e^{-v X_\mathrm{min}/2} \, \mathds{1}_{A_1^-A_2^-} \big) \!-\! 2\mathbb{E}\big(
            e^{-v X_\mathrm{min}/2} \, \mathds{1}_{A_1^-A_b^-} \big) \\
            =&~ 2\bigg(\!\mathbb{E}\big(e^{-kX_{2}/2} \thinspace\mathds{1}_{A_{12}^-}\big) \!+\! \mathbb{E}\big(e^{-kX_{b}/2} \thinspace\mathds{1}_{A_{1b}^-}\big) \!-\! \mathbb{E}\big(e^{-kX_{2}/2} \thinspace\mathds{1}_{A_{12}^- A_{1b}^-}\big) \!\bigg) \nonumber \\
            &~ \!+\! 2\mathbb{E}\big(e^{-kX_{1}/2} \thinspace\mathds{1}_{A_{b1}^-}\big) \!-\! \mathbb{E}\big(e^{-kX_{1}/2} \thinspace\mathds{1}_{A_{b1}^- A_{b2}^-}\big) \!-\! \mathbb{E}\big(e^{-kX_{b}/2} \thinspace\mathds{1}_{A_1^- A_2^-}\big) \!-\! 2 \mathbb{E}\big(e^{-kX_{2}/2} \thinspace\mathds{1}_{A_1^- A_b^-}\big)~, \nonumber
    \end{align}
\endgroup
where (i) follows from the symmetry of the metropolitan links.
Note that ${\mathbb{E}\big( e^{-v X_\mathrm{min}/2} \, \mathds{1}_{A_1^-A_2^-A_b^-} \big) \!=\! 0}$, as the subset $A_1^- A_2^- A_b^-$ is empty.
In (ii), we have again used the principle of inclusion and exclusion using~\eqref{eq:term_expansion_A_1-andA_b-} as the events of the form $A_i^-$ are determined via $X_\text{max}$, while here we need $X_\text{min}$.
Furthermore, by symmetry, $\mathbb{E}\big(Z\mathds{1}_{A_{b1}^-}\big) = \mathbb{E}\big(Z\mathds{1}_{A_{b2}^-}\big)$.
Thus, we established~\eqref{eq:term_expansion_U_3} by plugging in the value of $X_\text{min}$.

\end{proof}

\vspace{-4pt}
Observe that instead of evaluating $U_1(v)$ and $U_2(v)$ separately, we can compute
\begingroup
    \setlength{\abovedisplayskip}{2pt}   
    \setlength{\belowdisplayskip}{4pt}   
    \setlength{\abovedisplayshortskip}{2pt}
    \setlength{\belowdisplayshortskip}{2pt}
    \begin{align}
        U(v,\alpha) =&~ 2\bigg(\!\mathbb{E}(e^{-k((\alpha-1)X_1 - X_2)} \mathds{1}_{A_{12}^+}) + \mathbb{E}(e^{-k((\alpha-1)X_1 - X_b)} \mathds{1}_{A_{1b}^+}) - \mathbb{E}(e^{-k((\alpha-1)X_1 - X_{2})} \mathds{1}_{A_{12}^+ A_{1b}^+}) \!\bigg) \nonumber \\
        & + \mathbb{E}(e^{-k(\alpha X_b - X_1 - X_2)} \mathds{1}_{A_{b}^+}) - \mathbb{E}(e^{-k((\alpha-1)X_1-X_b)}\mathds{1}_{A_1^+ A_2^+}) - 2 \mathbb{E}(e^{-k((\alpha-1)X_1-X_2)}\mathds{1}_{A_1^+ A_b^+}) \nonumber \\
        & + \mathbb{E}(e^{-k(\alpha-2) X_{1}} \mathds{1}_{A_1^+ A_2^+ A_b^+})~, \label{eq:def_U_alpha}
    \end{align}
\endgroup

\noindent which gives
\begingroup
    \setlength{\abovedisplayskip}{3pt}   
    \setlength{\belowdisplayskip}{6pt}   
    \setlength{\abovedisplayshortskip}{2pt}
    \setlength{\belowdisplayshortskip}{2pt}
    \begin{align}
        U_1(v) = U(v,2)~, \quad U_2(v) = U(v,5/2)~.
        \label{eq:U_1_andU_2_to_U_alpha}
    \end{align}
\endgroup
We evaluate the individual terms of~\eqref{eq:def_U_alpha} in~\ref{subsec:terms_U_alpha}, which are given by~\eqref{eq:indiv_term_U_alpha_A_12+}, \eqref{eq:indiv_term_U_alpha_A_1b+}, \eqref{eq:indiv_term_U_alpha_A_12+A_1b+}, \eqref{eq:indiv_term_U_alpha_A_b+}, \eqref{eq:indiv_term_U_alpha_A_1+A_2+}, \eqref{eq:indiv_term_U_alpha_A_1+A_b+}, and~\eqref{eq:indiv_term_U_alpha_A_1+A_2+A_b+}.
Plugging in these computations in~\eqref{eq:def_U_alpha}, we can compute $U_1$ using ~\eqref{eq:U_1_andU_2_to_U_alpha}, which gives the expected teleportation fidelity in the ER case via~\eqref{eq:fidelity_ent_ready}.
Note that we can obtain $U_2$ in the same way.
For $U_3$ described in~\eqref{eq:term_expansion_U_3}, the calculable expressions of the individual terms are given in~\eqref{eq:indiv_term_U_3_A_12-}, \eqref{eq:indiv_term_U_3_A_1b-}, \eqref{eq:indiv_term_U_3_A_b1-}, \eqref{eq:indiv_term_U_3_A_12-A_1b-}, \eqref{eq:indiv_term_U_3_A_b1-A_b2-}, \eqref{eq:indiv_term_U_3_A_1-A_2-}, and~\eqref{eq:indiv_term_U_3_A_1-A_b-} in~\ref{sec:calculation_indiv_terms_U_3}. 
Having obtained $U_2$ and $U_3$, we can compute the expected fidelity in the QR case via~\eqref{eq:fidelity_qubit_ready}.
The analytical expressions derived in this section help us analyse the behaviour of the performance metrics as functions of the hardware parameters without running extensive simulations, which we take up in the next section.
\vspace{-7pt}

\section{Evaluations}
\label{sec:results}
\everypar{\looseness=-1}

In this section, we specify exact empirical requirements corresponding to~\ref{Q:1}--\ref{Q:2c} for the network from Fig.~\ref{fig:intercity_network_diagram}, when the baseline and optimistic parameters are set according to  Tab.~\ref{tab:no_imperfection_probability_baseline_and_optimistic}.
As already mentioned, the questions~\ref{Q:1}--\ref{Q:2c} involve the performance metrics teleportation rate and fidelity, which we derive analytically.
This lets us answer the questions without having to run extensive simulations.
In the following, we present our findings for teleportation in the MN and the IN separately.
We also provide empirical validation of our analytical results for the performance metrics before we delve into~\ref{Q:1}--\ref{Q:2c}.

\subsection{Teleportation in a Metropolitan Network}
\label{sec:results_teleportation_in_metropolitan_network}

Recall that in~\ref{Q:1}, we investigate the hardware requirements for teleportation in an MN of Fig.~\ref{fig:intercity_network_diagram}.
Here, the end nodes are separated by $50\,\text{km}$.
The key hardware parameters that influence the performance metrics, i.e., the expected teleportation fidelity and rate, are the base efficiency $p_\text{m}^0$, memory coherence time $t_\text{coh}$, and metropolitan link fidelity $f_{\text{m}'}$.
We explain in~\ref{sec:metro_teleportation_fidelity} that the teleportation fidelity in the ER case is deterministic while the rate equals that of the QR case.
Thus, we compare the analytically derived expected teleportation fidelity and rate with their empirical values only in the QR case.
The empirical values are obtained by simulating~\cite{git_repo_teleportation} in NetSquid, and we present the comparison in Fig.~\ref{fig:fidelity_metro_analytical_vs_simulation_both_cases}.

To simulate average fidelities, we fix the memory coherence time $t_\text{coh}$ and the metropolitan link fidelity $f_{\text{m}'}$ at their baseline and optimistic values, i.e. $(62\,\text{ms}, 4 \,\text{s})$, and $(0.88, 0.95)$, respectively, while we let the base efficiency $p_\text{m}^0$ vary between $10^{-4}$ and $1$.
We keep $t_\text{prep}\!=\!175\,\text{ms}$, following the supplementary material of~\cite{Krutyanskiy2023TelecomWavelength}.
For each observation point, we run the experiment for $100$ batches with each batch comprising $100$ independent runs.
Each batch produces a single observation for the average teleportation fidelity. 
We also compute the $5$th and the $95$th percentiles of the (batch) average fidelities.
Further, the analytical values of the expected teleportation fidelity are derived using~\eqref{eq:metro_fidelity_QR}.
We observe close agreement between the analytical and simulation results in Fig.~\ref{fig:fidelity_metro_analytical_vs_simulationQubitReady}.
Since the teleportation rate only depends on the base efficiency $p_\text{m}^0$, we calculate the simulated rate for each batch as a function of $p_\text{m}^0$.
The number of batches and runs per batch remains the same as in Fig.~\ref{fig:fidelity_metro_analytical_vs_simulationQubitReady}.
Further, the analytical value of the rate is derived using~\eqref{eq:metro_tel_rate}.
We observe in Fig.~\ref{fig:rate_metro_analytical_vs_simulation} that the analytical and simulation results are consistent in this case as well.

We now address~\ref{Q:1}, i.e., identify the hardware requirements necessary to attain teleportation in the MN with the target teleportation fidelity $f_\text{target}$.
As mentioned, the baseline hardware parameter values are set as per Tab.~\ref{tab:no_imperfection_probability_baseline_and_optimistic}.
Since the hardware parameter space is three-dimensional, we plot the surface corresponding to the minimum required link fidelity $f_{\text{m}'}$ among the parameter space that achieves the target fidelity of teleportation.
Recall that the surface $\tilde{\Lambda}_{1+}^\text{ER}$ (resp. $\tilde{\Lambda}_{1+}^\text{QR}$) in the ER (resp. QR) case is given by \eqref{eq:min_fm_definition_metro_ER} (resp. \eqref{eq:min_fm_definition_metro_QR}), while the region above this surface corresponds to the desired parameter space $\Lambda_{1+}^\text{ER}$ (resp. $\Lambda_{1+}^\text{QR}$) as defined in~\eqref{eq:desired_space_metro_ER} (resp. \eqref{eq:desired_space_metro_QR}).

\begin{figure}[t]
    \centering
    \begin{subfigure}[b]{0.48\textwidth}
        \includegraphics[width=\textwidth, trim=5pt 8pt 5pt 5pt, clip]{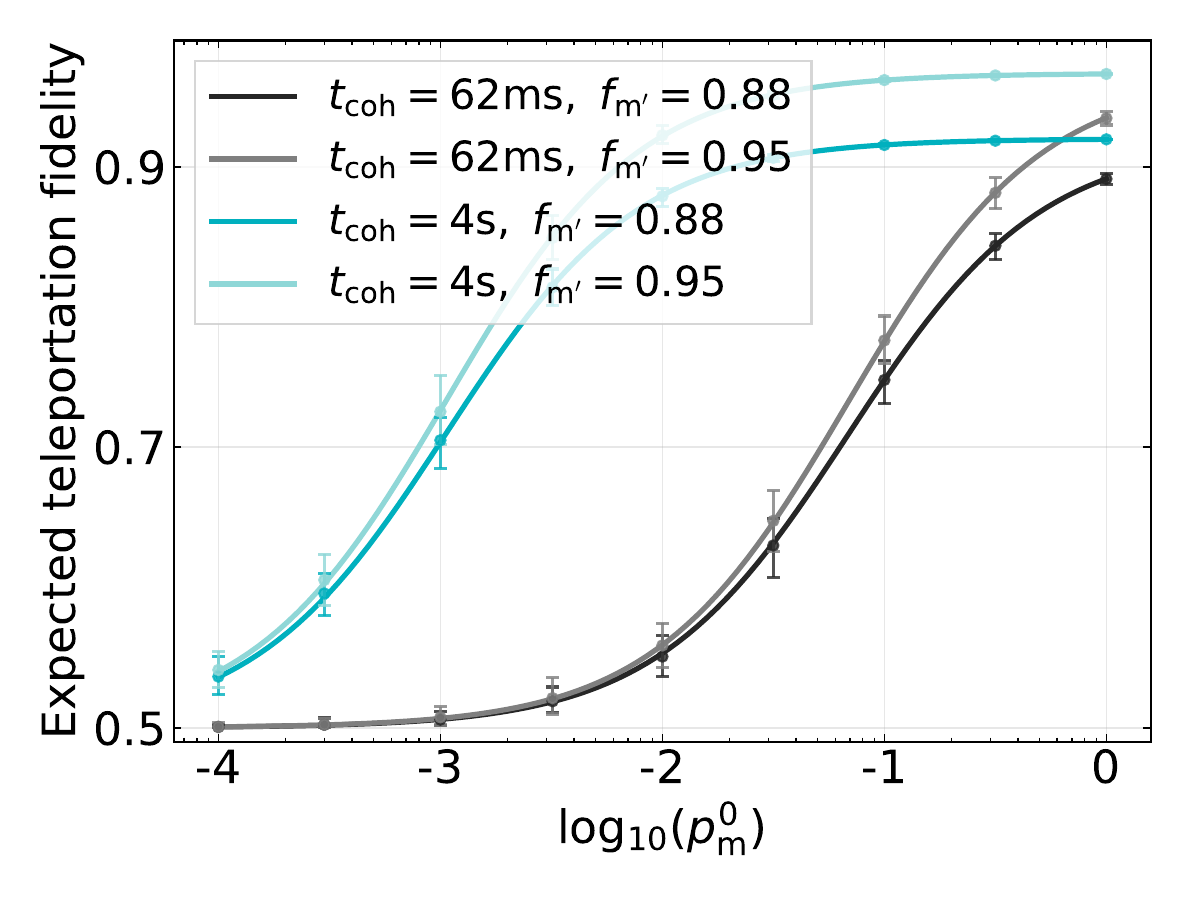}
        \caption{Expected teleportation fidelity}
        \label{fig:fidelity_metro_analytical_vs_simulationQubitReady}
    \end{subfigure}
    \hfill 
    \begin{subfigure}[b]{0.48\textwidth}
        \includegraphics[width=\textwidth, trim=5pt 8pt 5pt 5pt, clip]{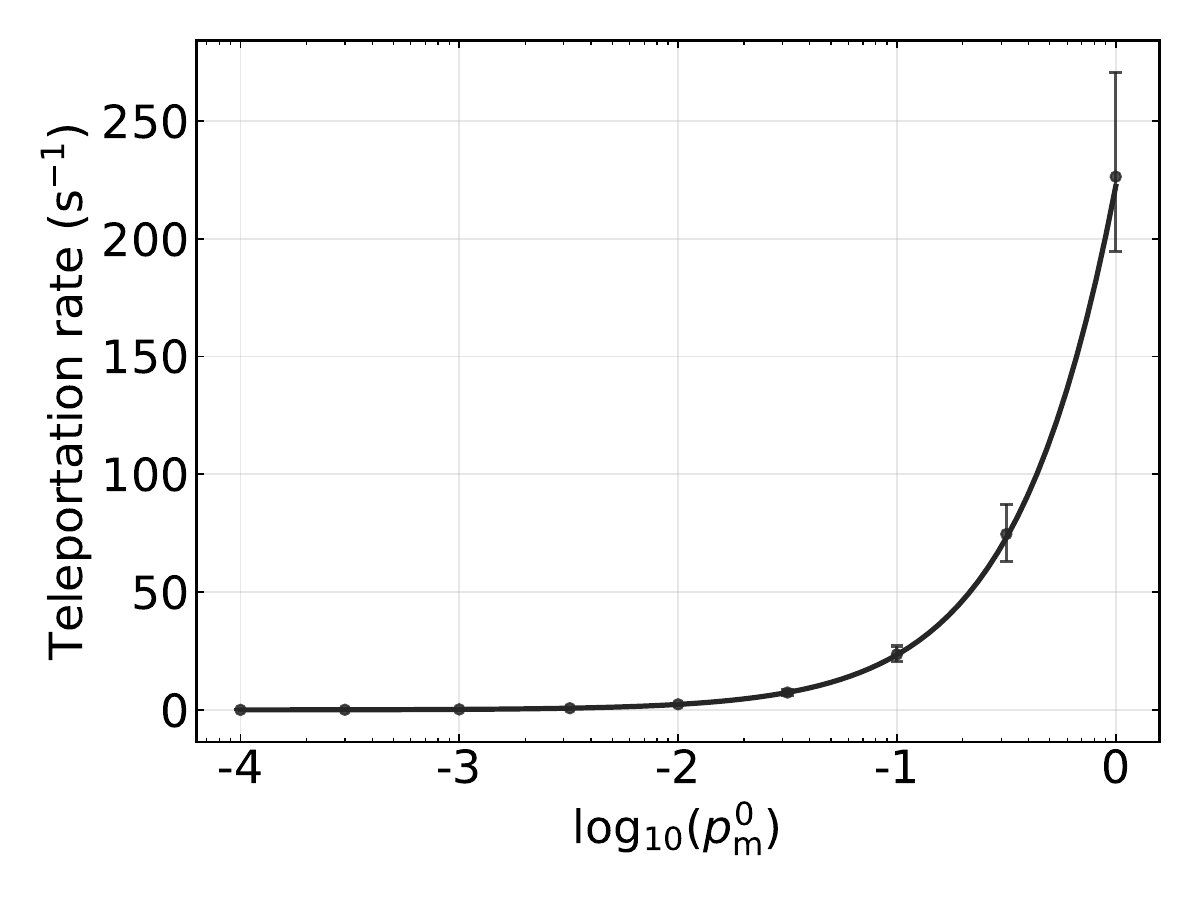}
        \caption{Teleportation rate}
        \label{fig:rate_metro_analytical_vs_simulation}
    \end{subfigure}
    \caption{Comparison of analytical (lines) and simulation results (dots) for qubit-ready teleportation in an MN. The mean, along with the 5th and 95th percentiles of the performance metrics, is shown as a function of the base efficiency $p_\text{m}^0$. (a) The expected teleportation fidelity is evaluated for both baseline and optimistic values of memory coherence time $t_\text{coh}$ and metropolitan link fidelity $f_{\text{m}'}$, while (b) the teleportation rate depends solely on $p_\text{m}^0$. Simulation results closely match the analytical values.}
    \vspace{-15pt}
    \label{fig:fidelity_metro_analytical_vs_simulation_both_cases}
\end{figure}

We plot the surface $\tilde{\Lambda}_{1+}^\text{ER}$ in Fig.~\ref{fig:fidelity_surface_teleportEntanglementReady} and its colour represents the rate of teleportation $R_1^\text{ER}$ as defined in~\eqref{eq:rate_min_fm_definition_metro_ER}.
As expected, the rate only depends on the base efficiency $p_\text{m}^0$.
We also observe that the baseline values of the parameters $(p_\text{m}^0,t_\text{coh},f_{\text{m}'})$ denoted by $\text{B}:\!({5.95\times 10^{-4}, 62\,\text{ms}, 0.88})$ lie above the surface, which implies that we can already achieve teleportation in the MN with target fidelity based on state-of-the-art hardware.
Specifically, the baseline yields an expected teleportation fidelity of $0.92$ and a corresponding teleportation rate of $0.14\,\text{s}^{-1}$.

We show the surface $\tilde{\Lambda}_{1+}^\text{QR}$ in Fig.~\ref{fig:fidelity_surface_teleportQubitReady} where the colour of the surface corresponds to the teleportation rate $R_1^\text{QR}$.
In contrast with Fig.~\ref{fig:fidelity_surface_teleportEntanglementReady}, the baseline (B) in this case lies below the surface, indicating that the current hardware capabilities are insufficient to meet the target for expected teleportation fidelity.
Thus, we aim to numerically calculate an optimal point in the set $\Lambda_{1*}^\text{QR}$, defined in~\eqref{eq:optimal_points_set_metro_QR}.
Recall that the optimal points refer to hardware parameter configurations which achieve the target fidelity and are easiest to achieve from the baseline values in terms of the cost function $c$ from~\eqref{eq:defition_total_cost_function2}.
For the numerical optimisation, we employ a global optimisation heuristic from~\cite{Prielinger2024}.
To improve its performance, we run the optimiser $50$ times and select the solution with the lowest hardware cost among all runs.
The resulting optimal point O: \!(${1.43\times 10^{-2}, 196\,\text{ms}, 0.88}$) is shown in Fig.~\ref{fig:fidelity_surface_teleportQubitReady}, which yields the target expected teleportation fidelity $2/3$ at a rate of $3.36\,\text{s}^{-1}$.

\subsection{Teleportation in Intercity Network}
\label{sec:results_teleportation_in_intercity_network}

As discussed in~\ref{Q:2a}--\ref{Q:2c}, we investigate the requirements for teleportation in an IN of Fig.~\ref{fig:intercity_network_diagram}.
Specifically, the scenario involves teleporting a data qubit across $500\,\text{km}$ between two nodes located in separate MNs.
Here, each end node is located $25\,\text{km}$ from its respective metropolitan hubs, and the hubs are connected by a $450\,\text{km}$ backbone.
The key hardware parameters influencing the teleportation fidelity and rate include the metropolitan base efficiency $p_\text{m}^0$, memory coherence time $t_\text{coh}$, metropolitan link fidelity $f_\text{m}$, entanglement-generation probability in backbone $p_\text{b}$, and backbone fidelity $f_\text{b}$.
The non-hardware parameter cut-off time ($t_\text{cut}$) plays a role in shaping the performance metrics as well.
As in Sec.~\ref{sec:results_teleportation_in_metropolitan_network}, we rely on analytical expressions for the teleportation rate and expected fidelity derived in~\ref{sec:cutOff_rate} and~\ref{sec:cutOff_fidelity} to answer~\ref{Q:2a}--\ref{Q:2c} as it obviates the need to run extensive simulations.
We also provide a comparison of the analytical values of the performance metrics with corresponding empirical estimates from NetSquid-based simulations~\cite{git_repo_teleportation} to show their accuracy.

\begin{figure}[t]
    \centering
    \begin{subfigure}[b]{0.49\textwidth}
        \includegraphics[width=\textwidth, trim=90pt 10pt 70pt 30pt, clip]{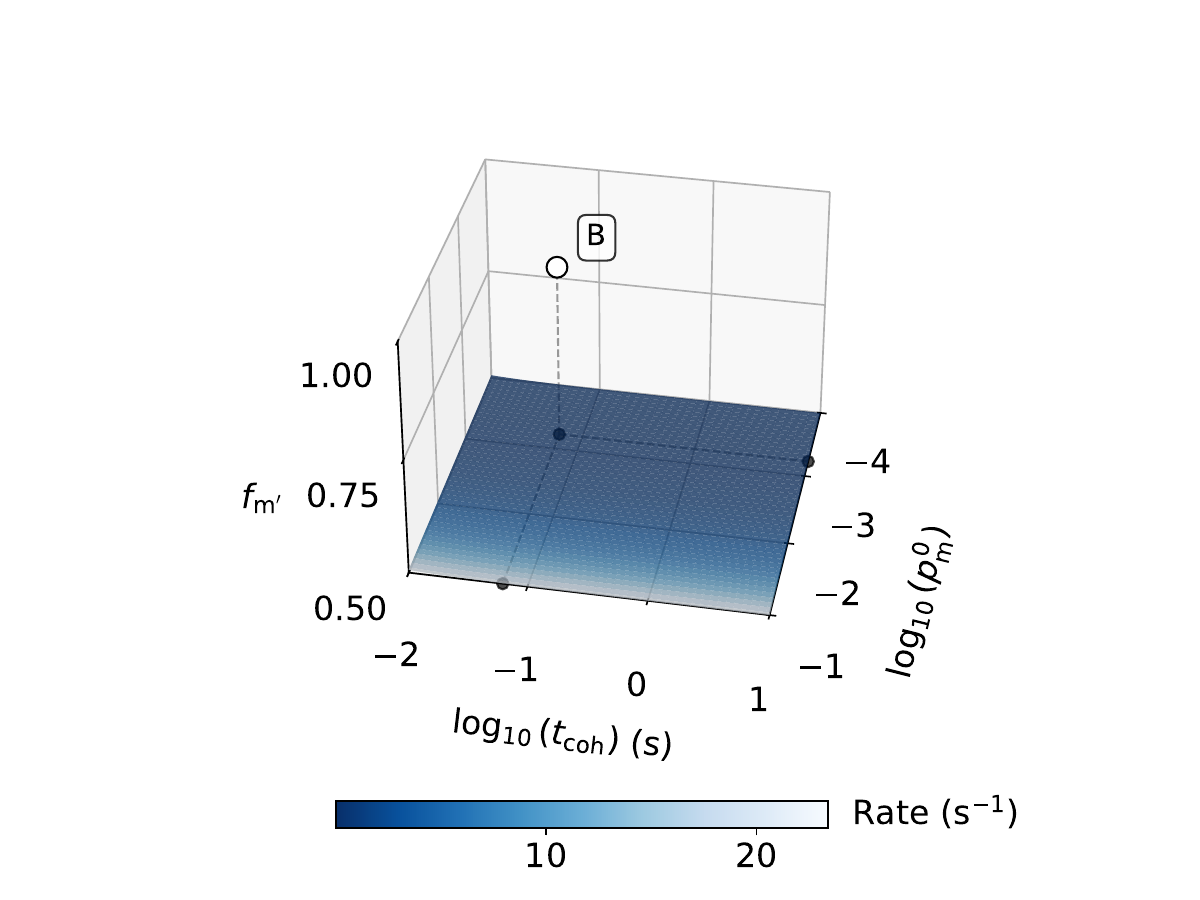}
        \caption{Entanglement-ready}
        \label{fig:fidelity_surface_teleportEntanglementReady}
    \end{subfigure}
    \hfill
    \begin{subfigure}[b]{0.49\textwidth}
        \includegraphics[width=\textwidth, trim=90pt 10pt 70pt 30pt, clip]{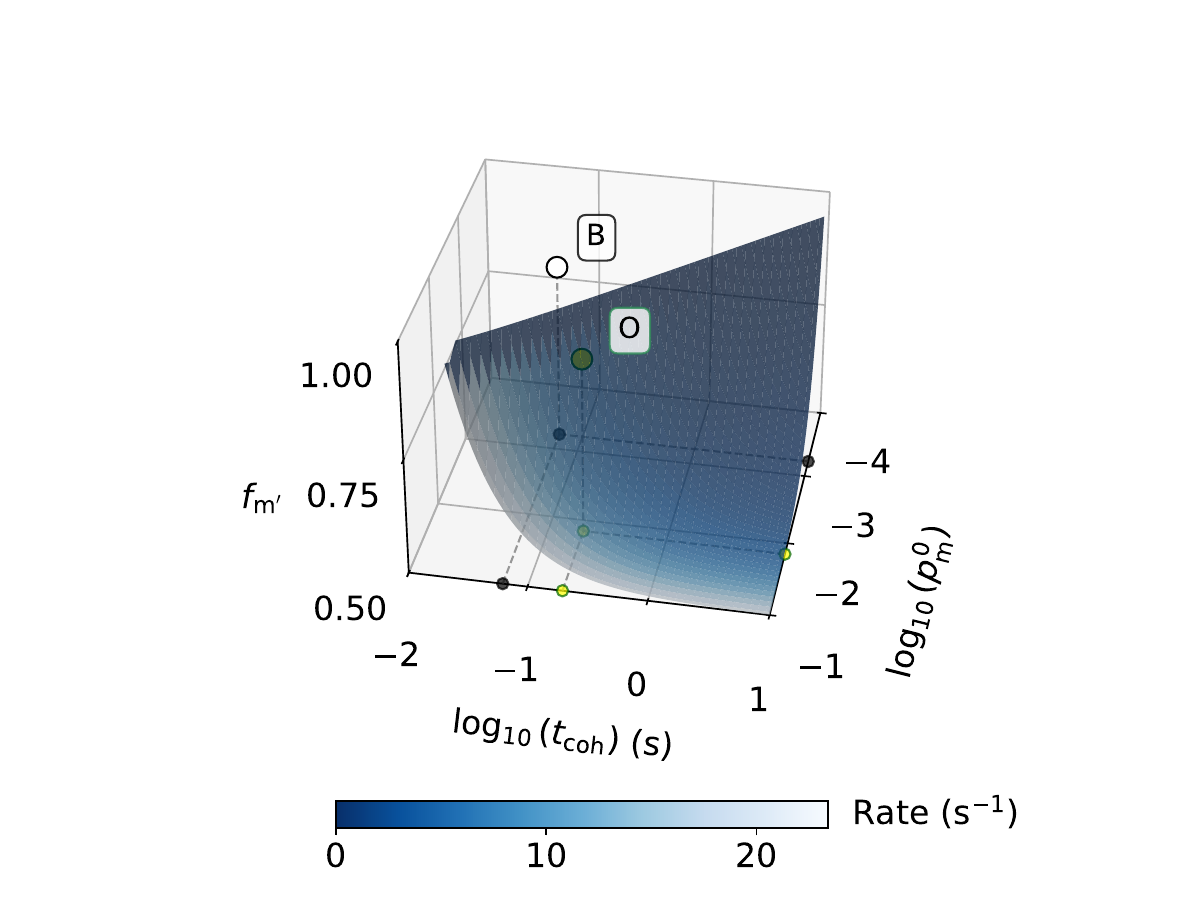}
        \caption{Qubit-ready}
        \label{fig:fidelity_surface_teleportQubitReady}
    \end{subfigure}
    \caption{Requirements to achieve the target teleportation fidelity of $2/3$ in an MN.
    The surfaces represent the minimum required metro link fidelity $f_{\text{m}'}$ as a function of $p_\text{m}^0$ and $t_\text{coh}$ for (a) entanglement-ready and (b) qubit-ready teleportation.
    The colour of the surface shows the corresponding teleportation rate.
    The baseline (denoted B) in (a) already achieves the target fidelity of teleportation, while an optimal (denoted O) parameter configuration subject to hardware cost $c$ is shown in (b).
    The $(p_\text{m}^0, t_\text{coh}, f_{\text{m}'})$ coordinates of the points in (b) are B: ${(5.95\times 10^{-4}, 62\,\text{ms}, 0.88)}$ and O: ${(1.43\times 10^{-2}, 196\,\text{ms}, 0.88)}$. 
    }
    \vspace{-15pt}
    \label{fig:fidelity_surface_metro_teleportation}
\end{figure}

\begin{figure}[t]
    \centering
    \begin{subfigure}[b]{0.48\textwidth}
        \includegraphics[width=\textwidth, trim=5pt 8pt 5pt 5pt, clip]{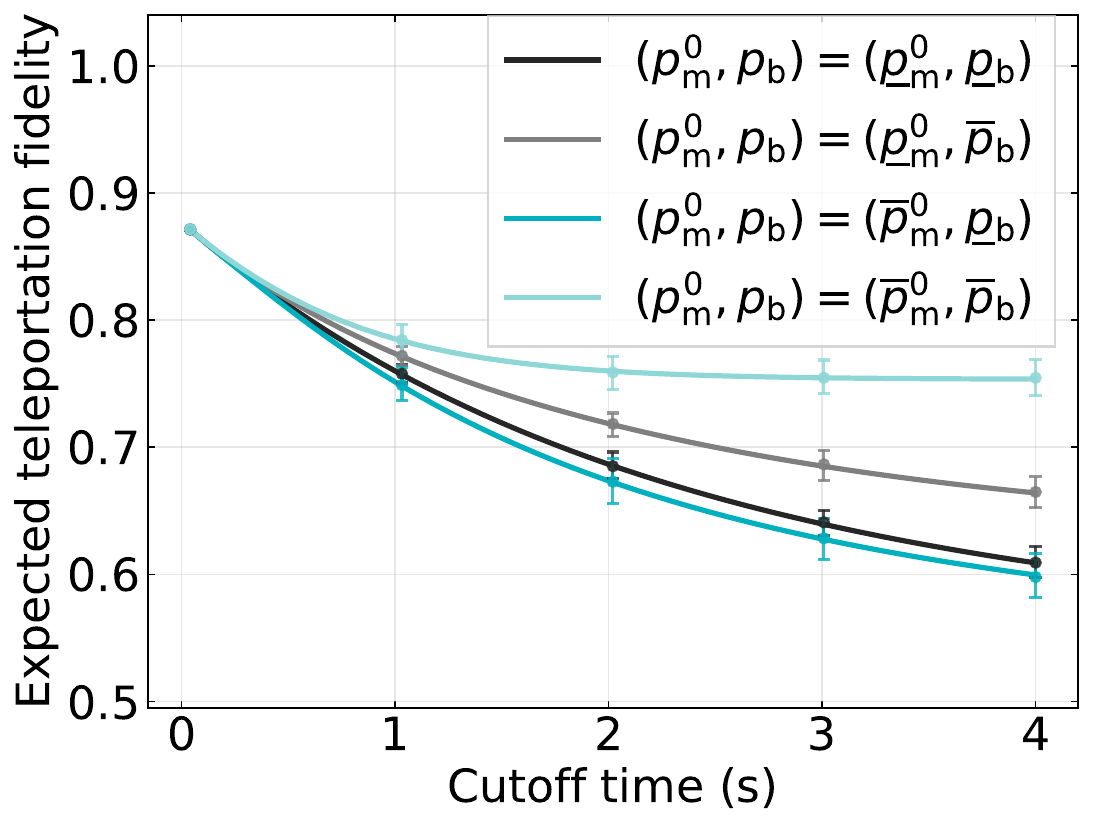}
        \caption{Entanglement-ready}
        \label{fig:fidelity_analytical_vs_simulationEntanglementReady}
    \end{subfigure}
    \hfill 
    \begin{subfigure}[b]{0.48\textwidth}
        \includegraphics[width=\textwidth, trim=5pt 8pt 5pt 5pt, clip]{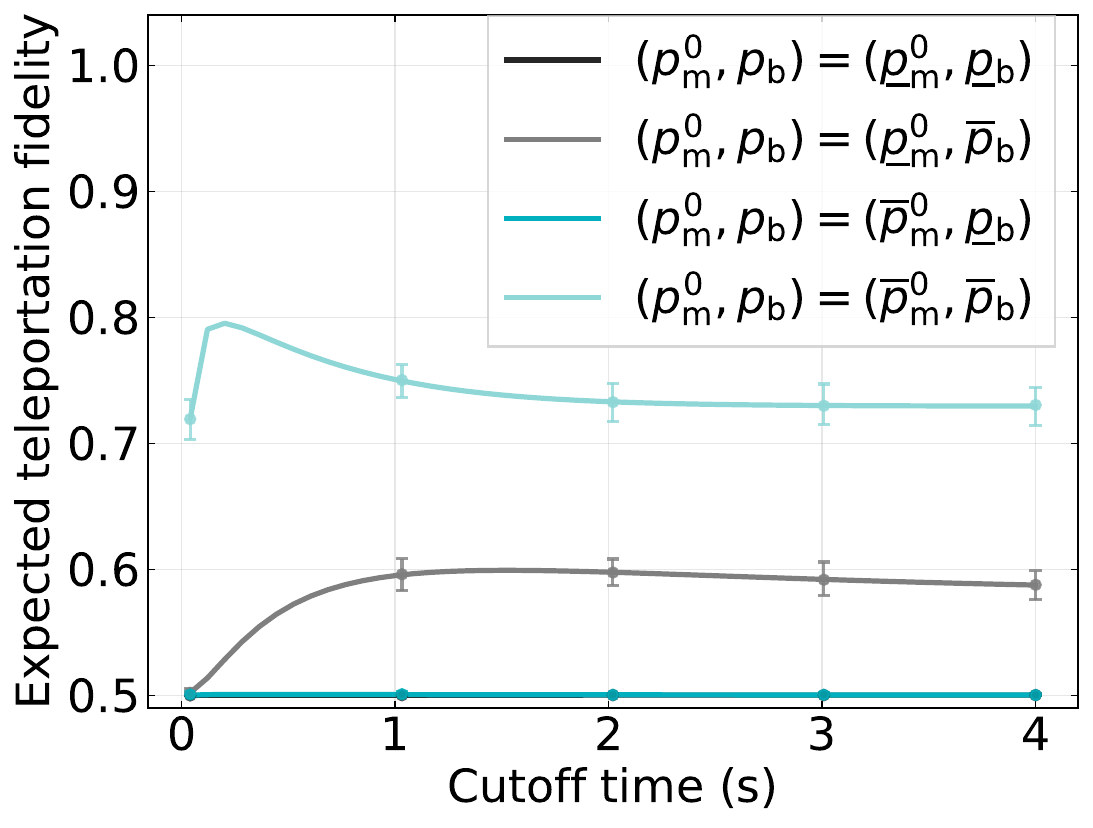}
        \caption{Qubit-ready}
        \label{fig:fidelity_analytical_vs_simulationQubitReady}
    \end{subfigure}
    \caption{\looseness = -1 Comparison of analytical (lines) and simulation results (dots) for the expected teleportation fidelity as a function of cut-off time $t_\text{cut}$ for (a) entanglement-ready and (b) qubit-ready teleportation.
    Each simulation point corresponds to the mean fidelity of batch average values, with error bars denoting the $5$th and $95$th percentiles across batch averages.
    The entanglement generation probabilities $p_\text{m}^0$, and $p_\text{b}$ are set at their baseline ($\underline{p}_\text{m}^0\!=\! 5.95\times 10^{-4}$, $\underline{p}_\text{b}\!=\!1.51\times 10^{-6}$) and optimistic values ($\overline{p}_\text{m}^0\!=\!1.43\times 10^{-2}$, $\overline{p}_\text{b}\!=\!4.18\times 10^{-3}$), as indicated in the legend.
    All remaining parameters are fixed at their respective optimistic values, i.e., $\overline{f}_\text{m}\!=\!0.95$, $\overline{t}_\text{coh}\!=\!4\,\text{s}$, and $\overline{f}_\text{b}\!=\!0.90$.
    The analytical expected teleportation fidelity matches closely with the empirical mean.
    }
    \vspace{-15pt}
    \label{fig:fidelity_analytical_vs_simulation_both_cases}
\end{figure}

\looseness = -1 To simulate average fidelities and rate, we fix the base efficiency $p_\text{m}^0$ and the backbone entanglement-generation probability $p_\text{b}$ at their baseline and optimistic values, i.e., $(5.95\times 10^{-4}, 1.43\times 10^{-2})$ and $(1.51\times 10^{-6}, 4.18\times 10^{-3})$, respectively, while we vary the cut-off time $t_\text{coh}$ between $0.04\,\text{s}$ and $4\,\text{s}$.
Note that the other parameters, namely, the memory coherence time $t_\text{coh}$, metropolitan link fidelity $f_\text{m}$, and backbone link fidelity $f_\text{b}$, are kept at their optimistic values, i.e., $4\,\text{s}$, $0.95$, and $0.9$, respectively.
Similar to Sec.~\ref{sec:results_teleportation_in_metropolitan_network}, we derive empirical estimates of average teleportation fidelity and rate by running simulations for $100$ batches, each batch consisting of $100$ independent runs.
Each batch yields an estimate of the average fidelity and rate, and we subsequently derive the mean and the $5$th and $95$th percentiles of the batch averages.
We compare the analytical values of the expected teleportation fidelity with their empirical estimates in the entanglement-ready (ER) and qubit-ready (QR) cases in Fig.~\ref{fig:fidelity_analytical_vs_simulationEntanglementReady} and Fig.~\ref{fig:fidelity_analytical_vs_simulationQubitReady}, respectively.
Since the rate of teleportation is the same for both cases, we do not make any distinction.
The corresponding comparison is shown in Fig.~\ref{fig:rate_analytical_vs_simulationQubitReady}.
The close agreement between analytical and simulation results in Fig.~\ref{fig:fidelity_analytical_vs_simulation_both_cases} and ~\ref{fig:rate_analytical_vs_simulationQubitReady} demonstrates the accuracy of our analytical results.

\looseness = -1 In the ER case (Fig.~\ref{fig:fidelity_analytical_vs_simulationEntanglementReady}), the expected fidelity decreases monotonically with increasing cut-off time.
This occurs because longer cut-off times allow noisier elementary links to contribute, degrading the overall end-to-end fidelity.
In contrast, the QR case (Fig.~\ref{fig:fidelity_analytical_vs_simulationQubitReady}) exhibits a non-monotonic behaviour with respect to the cut-off time $t_\text{cut}$.
Here, short cut-off times result in frequent link discards, leading to prolonged end-to-end link generation times and causing increased decoherence for the data qubit.
Although this yields high fidelity for the end-to-end link, the overall teleportation fidelity suffers due to data qubit degradation.
For higher cut-off times, the waiting time decreases on average but we allow elementary links with lower fidelity to form the end-to-end link.
This naturally leads to the formation of low fidelity end-to-end link and, subsequently, low teleportation fidelity. 
In contrast, observe from Fig.~\ref{fig:rate_analytical_vs_simulationQubitReady} that the teleportation rate, which is the inverse of the expected end-to-end link generation time, always increases with $t_\text{cut}$, as expected.

\begin{figure}[t]
    \centering
    \includegraphics[width=0.5\linewidth]{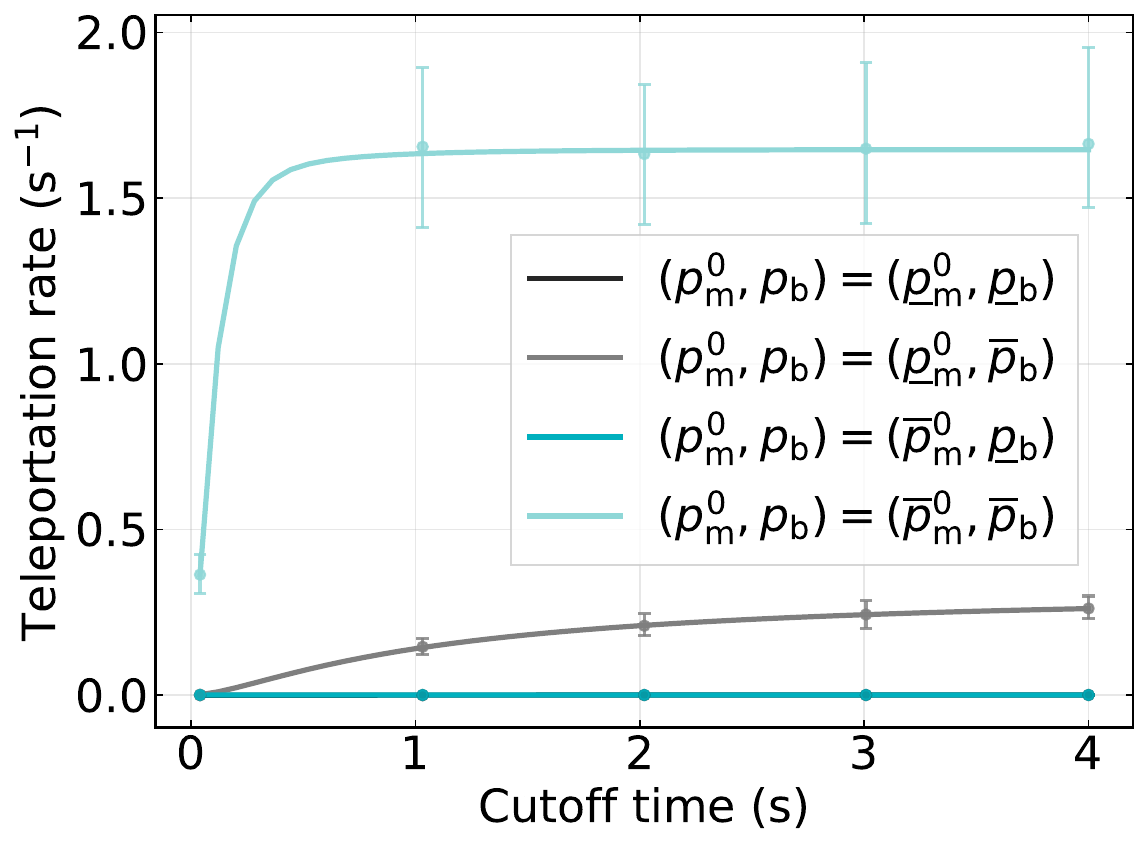}
    \caption{Comparison of analytical (lines) and simulation results (dots) for the teleportation rate as a function of cut-off time $t_\text{cut}$.
    Each simulation point represents the mean of rates computed across batches, and error bars indicate the $5$th and $95$th percentiles of batch averages.
    The entanglement generation probabilities $p_\text{m}^0$, and $p_\text{b}$ are set at their baseline ($\underline{p}_\text{m}^0\!=\! 5.95\times 10^{-4}$, $\underline{p}_\text{b}\!=\!1.51\times 10^{-6}$) and optimistic values ($\overline{p}_\text{m}^0\!=\!1.43\times 10^{-2}$, $\overline{p}_\text{b}\!=\!4.18\times 10^{-3}$), as indicated in the legend.
    The coherence time $t_\text{coh}$ is fixed at its optimistic value $4\,\text{s}$.
    The strong agreement between simulation results and analytical predictions confirms the validity of the analytical model.
    }
    \vspace{-15pt}
    \label{fig:rate_analytical_vs_simulationQubitReady}
\end{figure}

Recall that in~\ref{Q:2a}, we aim to find the minimal metropolitan hardware requirements needed to reach $f_\text{target}$ while the backbone parameters are fixed at their optimistic values: $\overline{p}_\text{b}\!=\!4.18\times 10^{-3}, \overline{f}_\text{b} 
\!=\!0.90$; see Tab.~\ref{tab:no_imperfection_probability_baseline_and_optimistic}.
In Fig.~\ref{fig:fixed_backbone_3D_EntanglementReady},
we identify the surface $\tilde{\Lambda}_{2+}^\text{ER}$  given by ~\eqref{eq:min_fm_definition_metro_int_ER} representing the minimum required metropolitan link fidelity $f_\text{m}$ among the metropolitan hardware parameter space that achieves the target teleportation fidelity in the ER case.
The region above this surface corresponds to the desired parameter space $\Lambda_{2+}^\text{ER}$ defined in~\eqref{eq:desired_space_metro_int_ER}. 
The colour of the surface $\tilde{\Lambda}_{2+}^\text{ER}$ in Fig.~\ref{fig:fixed_backbone_3D_EntanglementReady} represents the maximum achievable teleportation rate $R_2^\text{ER}$ defined in~\eqref{eq:rate_min_fm_definition_metro_int_ER}.
We observe that for a fixed coherence time $t_\text{coh}$, the required link fidelity $f_\text{m}$ does not change much with varying base efficiency $p_\text{m}^0$.
Moreover, for $t_\text{coh}\!\leq\!0.02\,\text{s}$, the required $f_\text{m}$ increases sharply and eventually renders the target fidelity unachievable.
Moreover, since we vary the cut-off time $t_\text{cut}$ in the range given by~\eqref{eq:t_cut_range}, the rate is implicitly influenced by $t_\text{coh}$ via $t_\text{cut}$.
As a result, for a fixed $p_\text{m}^0$, increasing $t_\text{coh}$ allows for longer cut-off times, leading to higher rates.
Also, as expected, the rate increases monotonically with increasing $p_\text{m}^0$.
This is more prominent for higher values of $t_\text{coh}$ and $p_\text{m}^0$ in the plot.
A crucial observation from Fig.~\ref{fig:fixed_backbone_3D_EntanglementReady} is that the baseline value of the parameters B: \!(${5.95\times 10^{-4}, 62\,\text{ms}, 0.88}$) (shown in white) lies above the surface $\tilde{\Lambda}_{2+}^\text{ER}$, indicating that we can already achieve ER teleportation in the IN under state-of-the-art metropolitan hardware and optimistic backbone parameter estimates.
The corresponding maximum rate achievable under this configuration is $4.00\times 10^{-4}\,\text{s}^{-1}$, associated with an expected fidelity $2/3$.

We plot the analogous surface $\tilde{\Lambda}_{2+}^\text{QR}$ in the QR scenario in Fig.~\ref{fig:fixed_backbone_3D_QubitReady}. Here, the baseline point B (shown in white) falls below the surface $\tilde{\Lambda}_{2+}^\text{QR}$, indicating that the current hardware cannot achieve the fidelity threshold $f_\text{target}$.
Thus, we perform numerical optimisation over the parameter space to identify an optimal point in the set $\Lambda_{2*}^\text{QR}$ defined in~\eqref{eq:optimal_points_set_metro_int_QR}.
Similar to Sec.~\ref{sec:results_teleportation_in_metropolitan_network}, we run the optimiser $50$ times and select the solution with the lowest hardware cost among all runs.
The resulting optimal point O: \!(${1.40\times 10^{-2}, 1095\,\text{ms}, 0.94}$) is shown in yellow in Fig.~\ref{fig:fixed_backbone_3D_QubitReady}, which yields the target fidelity at a rate of $1.32\,\text{s}^{-1}$.

\begin{figure}[t]
    \centering
    \begin{subfigure}[b]{0.49\textwidth}
        \includegraphics[width=\textwidth, trim=110pt 15pt 70pt 40pt, clip]{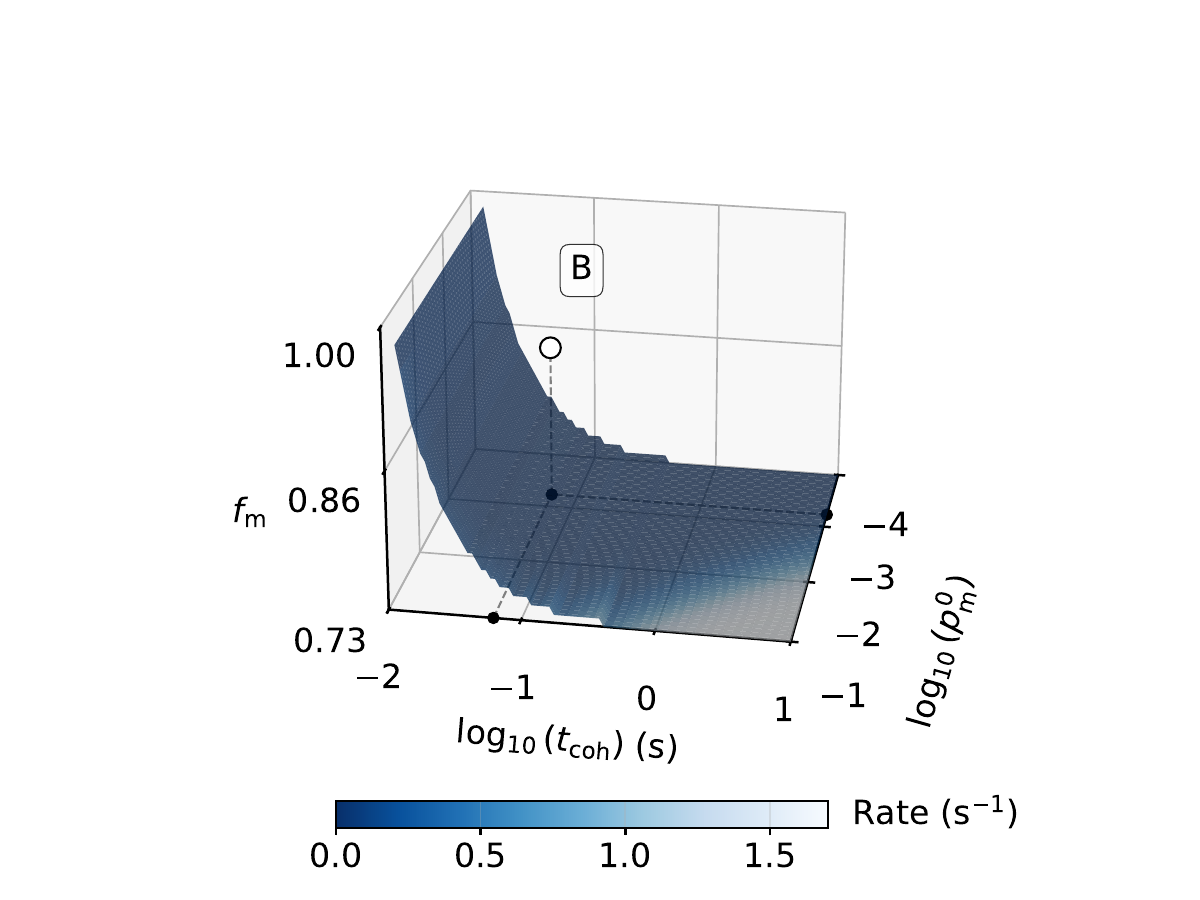}
        \caption{Entanglement-ready}
        \label{fig:fixed_backbone_3D_EntanglementReady}
    \end{subfigure}
    \hfill
    \begin{subfigure}[b]{0.49\textwidth}
        \includegraphics[width=\textwidth, trim=110pt 15pt 70pt 40pt, clip]{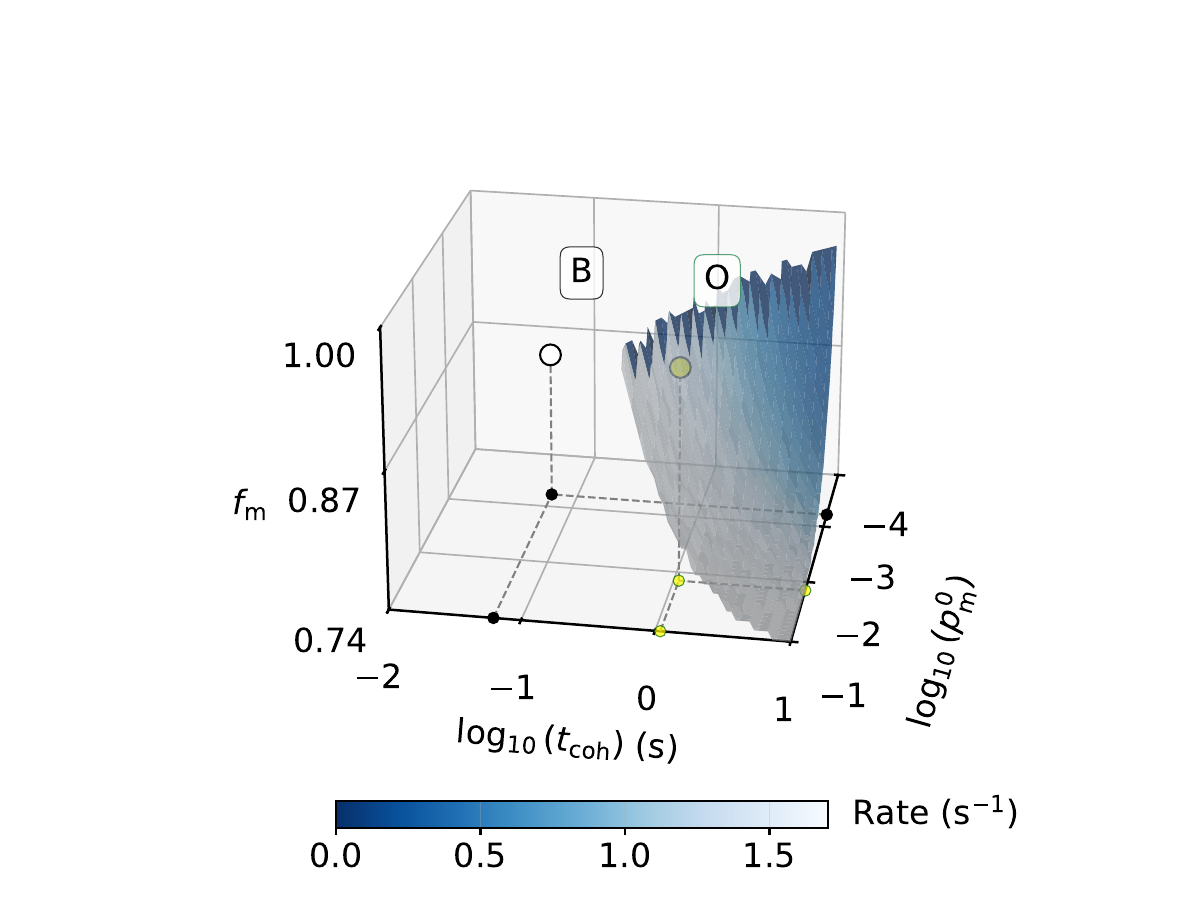}
        \caption{Qubit-ready}
        \label{fig:fixed_backbone_3D_QubitReady}
    \end{subfigure}
    \caption{Requirements to achieve the target teleportation fidelity of $2/3$ in an IN with the backbone parameters fixed at optimistic values: $\overline{p}_\text{b} \!=\! 4.18\times 10^{-3}, \overline{f}_\text{b} \!=\! 0.9$.
    The surface represents the minimum required link fidelity $f_\text{m}$ as a function of the base efficiency $p_\text{m}^0$ and the memory coherence time $t_\text{coh}$ for (a) entanglement-ready and (b) qubit-ready teleportation.
    The colour of the surface encodes the maximum achievable teleportation rate, subject to meeting the target teleportation fidelity.
    The baseline (denoted B) in (a) already achieves the target teleportation fidelity, while an optimal (denoted O) parameter configuration subject to minimising the hardware cost $c$ is shown in (b).
    The $(p_\text{m}^0, t_\text{coh}, f_\text{m})$ coordinates of the points are B: \!${(5.95\times 10^{-4}, 62\,\text{ms}, 0.88)}$, and O: \!${(1.40\times 10^{-2}, 1095\,\text{ms}, 0.94)}$. 
    }
    \vspace{-15pt}
    \label{fig:both_images}
\end{figure}

We now address~\ref{Q:2b}, where we fix the metropolitan parameters at their respective optimistic values $\overline{p}_\text{m}^0\!=\!1.43\times 10^{-2},\overline{t}_\text{coh}\!=\!4\,\text{s},\overline{f}_\text{m}\!=\!0.95$ (see Tab.~\ref{tab:no_imperfection_probability_baseline_and_optimistic}) and identify the backbone requirements necessary to achieve target teleportation fidelity.
Since the relevant hardware parameter space for the backbone is two-dimensional, comprising the entanglement-generation probability $p_\text{b}$ and link fidelity $f_\text{b}$, we plot the desired parameter space $\Lambda_{3+}^\text{ER}$ (resp. $\Lambda_{3+}^\text{QR}$) described in~\eqref{eq:desired_space_backbone_int_ER} (resp.~\eqref{eq:desired_space_backbone_int_QR}) that achieves the target expected teleportation fidelity in the ER (resp. QR) case.
The colour of the space $\Lambda_{3+}^\text{ER}$ in Fig.~\ref{fig:fixed_metro_EntanglementReady_mono_blue} represents the corresponding maximum achievable teleportation rate $R_{3}^\text{ER}$, defined in~\eqref{eq:rate_backbone_fixed_metro_ER}.
The contour lines indicate the parameter combinations that yield equal rates in units of $\text{s}^{-1}$.
For a fixed $p_\text{b}$, lower values of $f_\text{b}$ require shorter cut-off times to meet the target expected fidelity, which in turn reduces the teleportation rate.
As $f_\text{b}$ increases, the configuration achieves the target teleportation fidelity with a much higher cut-off, leading to a higher rate.
However, beyond a certain threshold, further increasing $f_\text{b}$ offers no additional benefit, as the rate becomes limited by $p_\text{b}$ within the allowed range of cut-off time given in~\eqref{eq:t_cut_range}.
This effect is reflected in the shape of the contour lines: they are curved at low $f_\text{b}$ and tend to become vertical as $f_\text{b}$ increases.
Furthermore, we observe that the baseline value of the backbone parameters ($p_\text{b},\,f_\text{b}$), denoted by B: \!($1.51\times 10^{-6},0.60$) lies within the $\Lambda_{3+}^\text{ER}$, indicating that this configuration achieves the target fidelity of teleportation.
The baseline yields a maximum teleportation rate $6.16\times10^{-4}~\text{s}^{-1}$ with an expected teleportation fidelity of $f_\text{target}$.

\begin{figure}[t]
    \centering
    \begin{subfigure}[b]{0.48\textwidth}
        \includegraphics[width=\textwidth]{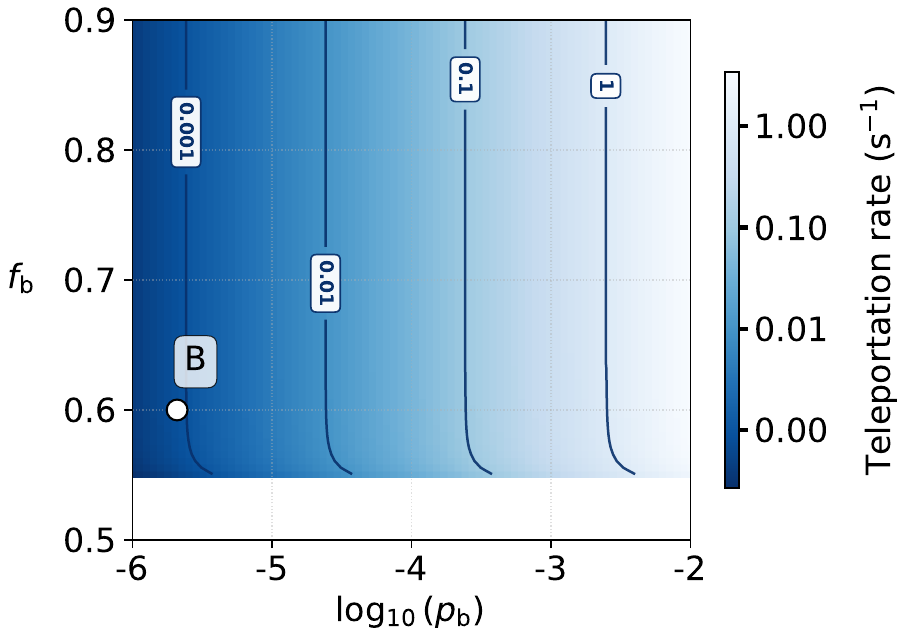}
        \caption{Entanglement-ready}
        \label{fig:fixed_metro_EntanglementReady_mono_blue}
    \end{subfigure}
    \hfill
    \begin{subfigure}[b]{0.48\textwidth}
        \includegraphics[width=\textwidth]{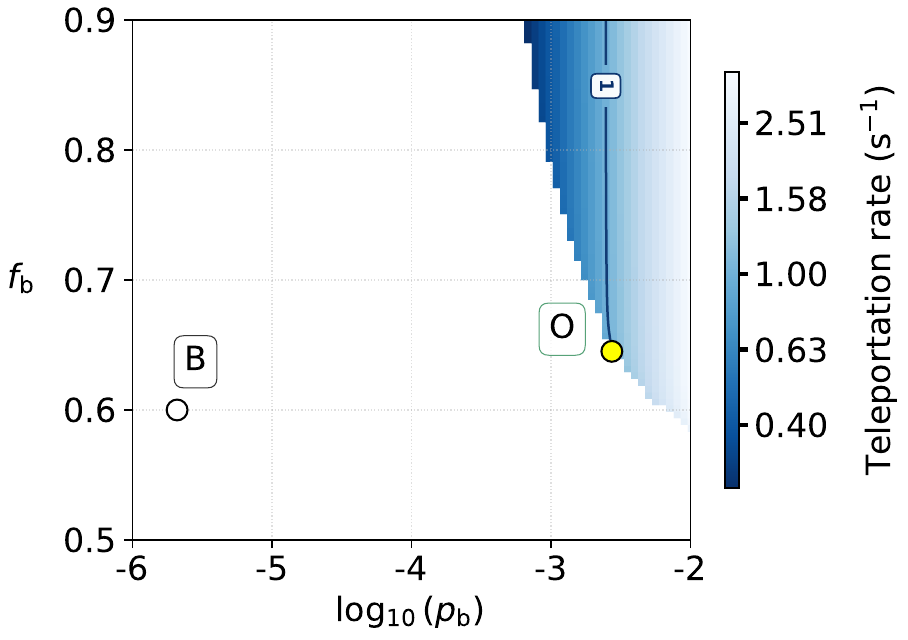}
        \caption{Qubit-ready}
        \label{fig:fixed_metro_QubitReady_mono_blue}
    \end{subfigure}
    \caption{Requirements to achieve the target teleportation fidelity of $2/3$ in an IN where the metropolitan parameters are fixed at their respective optimistic values ${\overline{p}_\text{m}^0\!=\!1.43\times 10^{-2}}$, ${\overline{t}_\text{coh}\!=\!4\,\text{s}}$, ${\overline{f}_\text{m}\!=\!0.95}$.
    The coloured region represents the desired parameter space $\Lambda_{3+}^\text{ER}$ (a) and $\Lambda_{3+}^\text{QR}$ (b) in the entanglement-ready and qubit-ready case, respectively. 
    The colour encodes the maximum achievable teleportation rate $R_{3}^\text{ER}$ (a) and $R_{3}^\text{QR}$ (b) subject to meeting the fidelity threshold.
    The baseline (denoted B) in (a) already achieves the target teleportation fidelity, while the optimal (denoted O) parameter configuration subject to hardware cost $c$ is shown in (b).
    The $(p_\text{b}, f_\text{b})$ coordinates of the points are B: \!${(1.51\times 10^{-6},0.60)}$ and O: \!${(2.73\times 10^{-3}, 0.64)}$.
    The contour lines indicate parameter combinations that yield equal rates.
    }
    \label{fig:fixed_metro_both_images}
    \vspace{-15pt}
\end{figure}

\looseness = -1 For the QR case, we plot the desired space $\Lambda_{3+}^\text{QR}$ defined in~\eqref{eq:desired_space_backbone_int_QR} in Fig.~\ref{fig:fixed_metro_QubitReady_mono_blue}.
The colour of the space represents the maximum achievable rate $R_{3}^\text{QR}$ defined in~\eqref{eq:rate_backbone_fixed_metro_QR}.
In contrast to the ER case, the baseline parameter value falls outside the desired region, indicating the need for improved hardware.
Thus, we numerically optimise over the backbone parameter space to identify an optimal point in the set of minimal hardware requirements $\Lambda_{3*}^\text{QR}$ enabling teleportation with the target fidelity as defined in~\eqref{eq:optimal_points_set_backbone_int_QR}.
Using a similar optimisation procedure as before, our optimiser produces the optimal point O: \!($2.73\times 10^{-3}, 0.64$) shown in yellow in Fig.~\ref{fig:fixed_metro_QubitReady_mono_blue}.
The optimal point achieves an expected teleportation fidelity of $f_\text{target}$ with a maximum achievable rate of $0.92\,\text{s}^{-1}$.

Finally, in~\ref{Q:2c}, we aim to determine the minimal hardware improvements over the baseline values of both metropolitan and backbone parameters necessary to reach the target teleportation fidelity.
In this setting, all five hardware parameters, namely the metropolitan base efficiency $p_\text{m}^0$, memory coherence time $t_\text{coh}$, metropolitan link fidelity $f_\text{m}$, backbone entanglement-generation probability $p_\text{b}$, and backbone link fidelity $f_\text{b}$ are allowed to vary, in addition to the non-hardware parameter, cut-off time $t_\text{cut}$.
Using the baseline values for these parameters, the maximum achievable expected teleportation fidelities are $0.61$ in the ER case and $0.50$ in the QR case, both below the target threshold.
This indicates that hardware improvements are necessary in both scenarios.
Due to the high dimensionality of the parameter space, we do not visualise the full desired space and optimisation landscape.
Instead, we numerically identify an optimal parameter configuration in the optimal set $\Lambda_{4*}^\text{ER}$ for the ER case and in $\Lambda_{4*}^\text{QR}$ in the QR case; see~\eqref{eq:optimal_points_set_int_ER} and~\eqref{eq:optimal_points_set_int_QR}, respectively, for the definition of the optimal set.
Similar to Sec.~\ref{sec:results_teleportation_in_metropolitan_network}, we conduct $50$ independent optimisation runs and select the output corresponding to the minimum hardware cost.
For the ER case, the ${(p_\text{m}^0, t_\text{coh}, f_\text{m}, p_\text{b}, f_\text{b})}$ coordinates of the resulting optimal point is ${(6.45\times10^{-4}, 64\,\text{ms}, 0.89, 1.57\times 10^{-6}, 0.67)}$, which achieves the expected teleportation fidelity $2/3$ with a corresponding rate $2.58\times 10^{-9}\,\text{s}^{-1}$.
Note that the backbone link fidelity $f_\text{b}$ stands out as the dominant bottleneck, while improvements in the other parameters remain modest.
In the QR case, the coordinates of the corresponding optimal point are given by ${(1.41\times10^{-2}, 1128\,\text{ms}, 0.95, 4.16\times 10^{-3}, 0.87)}$, which achieve the target expected fidelity $2/3$ with corresponding rate $1.24\,\text{s}^{-1}$.
As expected, the improvements are significantly more demanding than in the ER case, with the coherence time $t_\text{coh}$ requiring the most substantial enhancement.

\begin{figure}
    \centering
    \includegraphics[width=0.7\linewidth]{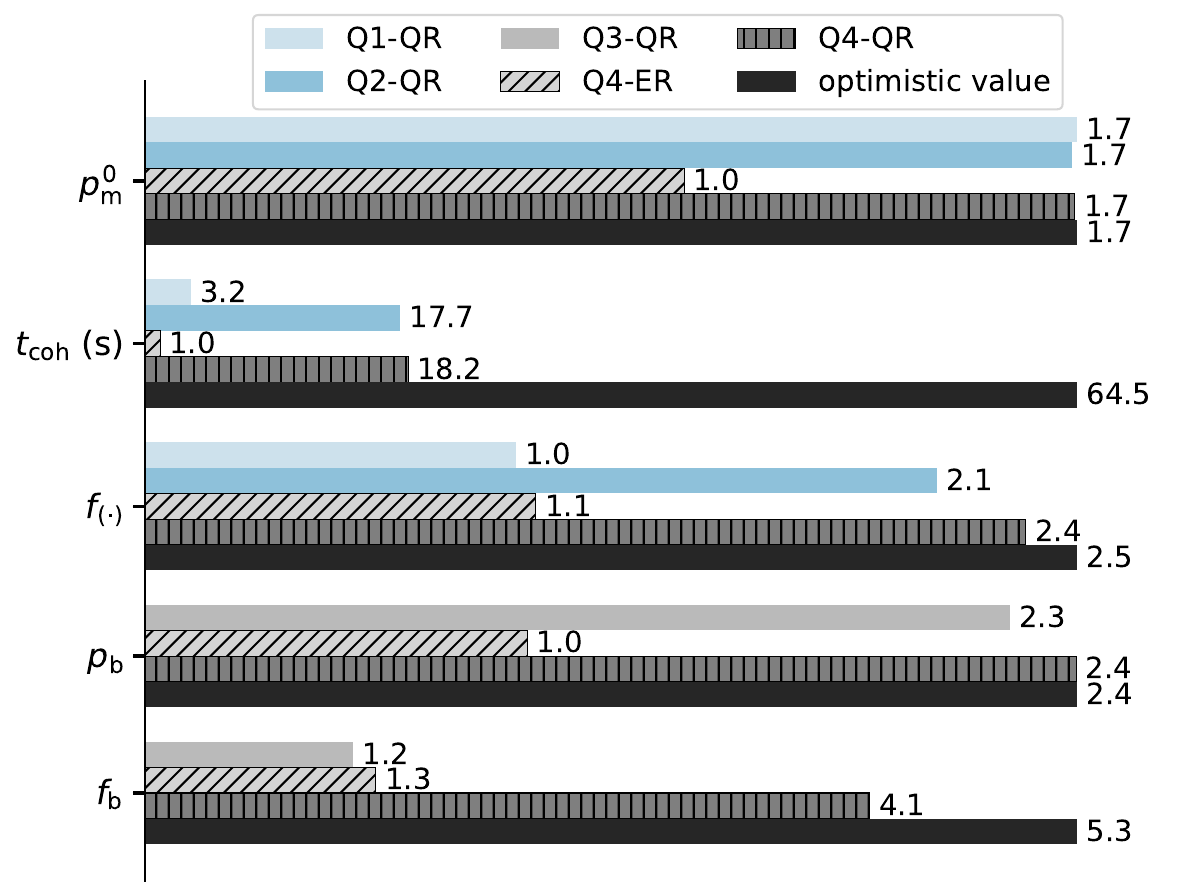}
    \caption{Required improvements for individual hardware parameters for enabling teleportation with target fidelity for~\ref{Q:1}--\ref{Q:2c}.
    Improvement factors quantify the extent to which a given parameter must be enhanced relative to its baseline value to meet the fidelity threshold.
    A factor of $1$ indicates no required improvement, while larger values reflect increasingly stringent requirements.
    The length of each bar represents the improvement factor corresponding to enhancing a parameter, and is normalised relative to the corresponding optimistic values shown in black.
    Annotations beside the bars indicate the actual improvement factor for the relevant question.
    Here $f_{(\cdot)}$ refers to $f_{\text{m}'}$ for Q1-QR while for all other cases, $f_{(\cdot)}$ corresponds to $f_\text{m}$.
    We represent them together as they share the same baseline and the same optimistic values. 
    Note that not all hardware parameters are relevant for every question considered.
    }
    \vspace{-15pt}
    \label{fig:bar_plot_all_requirements}
\end{figure}

In Fig.~\ref{fig:bar_plot_all_requirements}, we show the improvement factors corresponding to the optimal hardware parameter values, with corresponding legends indicating ER or QR cases of relevant questions.
The length of each bar is normalised relative to the improvement factor associated with the optimistic value of the corresponding parameter, which is shown in solid black.
Annotations next to the bars indicate the actual improvement factors determined by the optimisation procedure.
Recall from~\eqref{eq:improvement_factor} that the baseline value of a parameter corresponds to the improvement factor of $1$.
Along the vertical axis, we represent different hardware parameters.
Here $f_{(\cdot)}$ refers to $f_{\text{m}'}$ for Q1-QR, and $f_\text{m}$ elsewhere.
We represented $f_{\text{m}'}$ and $f_\text{m}$ together, as they share the same baseline and the same optimistic values.
Note that across all instances of QR teleportation, the optimal value of the metropolitan base efficiency $p_\text{m}^0$ reaches its optimistic value due to the fact that the room for improvement for $p_\text{m}^0$ in terms of IF is low compared to other parameters.
In Q1-QR, for instance, metropolitan link fidelity $f_{\text{m}'}$ remains near its baseline value, i.e. corresponding to an improvement factor of $1$, 
indicating that enhancing $p_\text{m}^0$ and $t_\text{coh}$ is relatively more beneficial from a cost-performance perspective.
In contrast, for more demanding scenarios such as Q2-QR and Q4-QR, $f_\text{m}$ is increased close to its optimistic value while $t_\text{coh}$ is moderately improved, suggesting its effectiveness over $t_\text{coh}$ from the perspective of cost-performance trade-off under these scenarios.
Regarding the backbone parameters, optimisation in Q3-QR returns a substantial improvement in its entanglement-generation probability $p_\text{b}$, compared to the corresponding link fidelity $f_\text{b}$, reflecting the advantage of the former in terms of cost-performance trade-off.
Further, in Q4-ER, all parameters require only modest improvements from their baselines, except for $f_\text{b}$, which stands out as the key bottleneck.
In contrast, Q4-QR demands near-optimistic improvements across all parameters, except for $t_\text{coh}$ and $f_\text{b}$, with moderate improvement in $f_\text{b}$, indicating that enhancing the remaining parameters is more efficient under the considered cost function.
\vspace{-8pt}

\section{Conclusion and Future Work}
\label{sec:conclusion}
\vspace{-7pt}
\everypar{\looseness=-1}
In this work, we present an analysis to identify the hardware requirements for enabling quantum teleportation over metropolitan and intercity-scale distances.
To this end, we derive closed-form analytical expressions for the expected fidelities and rates of end-to-end entanglement and teleported qubit as functions of relevant hardware parameters.
Throughout, we adopt a simple yet realistic noise model that accounts for memory decoherence and a finite cut-off time for entanglement generation rounds.
Determining the minimal hardware requirements is formulated as an optimisation problem over the parameters, with the analytical expressions enabling efficient exploration of this parameter space without resorting to computationally intensive simulations.
Using this approach, we obtain the minimal parameter values required to achieve the target teleportation fidelity of $2/3$.
We perform the analysis in a hardware-agnostic manner, in which network performance is fully characterised by parameters describing elementary link generation probabilities, entanglement fidelities, and memory coherence time. 

While the analysis remains hardware-agnostic, we also provide a concrete case study based on realistic parameter values from trapped-ion and ensemble-based memories. 
Within this framework, we address four central questions: (i) the feasibility of teleportation over metropolitan-scale distances with present-day hardware; (ii) the parameter improvements required within the metropolitan network with fixed backbone performance; 
(iii) the parameter improvements required within the backbone with fixed metropolitan network performance; and (iv) the minimal joint improvements necessary when both the metropolitan network and backbone operate near state-of-the-art experimental performance to enable teleportation over intercity-scale distances.

Our results indicate that near-term advances in trapped-ion and ensemble-based platforms, under optimistic parameter regimes, could enable quantum teleportation across metropolitan and intercity distances.
Achieving these optimistic parameters would require improvements in hardware performance, which is feasible according to the current projections of hardware capabilities in the near term.
Collectively, these findings establish clear design principles for making experimental progress: the elementary link generation probabilities in both metropolitan and backbone networks emerge as the most practical parameters to optimise, as the heuristic consistently favoured their improvement and drove them to optimistic values when necessary.
This behaviour arises because the range of feasible improvements, from baseline to optimistic values, is smaller for the link generation probability parameters compared to others.
Furthermore, in scenarios (ii) and (iii), entanglement-ready teleportation was feasible provided that at least one of the networks, either metropolitan or backbone, is operating at the optimistic parameter regime.
However, it is important to emphasise that these optimistic estimates are obtained for individual parameters and do not always account for potential interdependencies among them.
For instance, repeated photon generation attempts in trapped-ion systems can induce additional decoherence, thereby reducing the effective coherence time.
Consequently, achieving all optimistic parameters simultaneously within a single device or experiment poses a significant challenge, and such analysis is left as future work.
Also, building on this work, one can further investigate the performance of longer repeater chains, identify hardware requirements for more involved applications such as blind quantum computation, and determine the conditions for simultaneously meeting target thresholds for teleportation rate and fidelity by appropriately modifying the cost function.
\vspace{5pt}
\section*{Data and Code Availability}
\label{sec:data_availabililty}
\vspace{-5pt}
All source code used for data generation, processing, and plots is available at~\cite{git_repo_teleportation}. 
The datasets supporting the findings of this study are openly available at \url{https://doi.org/10.4121/6b7d42d3-ff78-4d64-bc9a-449480cc8550}.

\vspace{-7pt}
\section*{Acknowledgement}
\label{sec:acknowledgement}
\vspace{-5pt}
SM, GA, and SW acknowledge funding from the Quantum Internet Alliance (QIA).
QIA received funding from the European Union’s Horizon Europe research and innovation programme under grant agreement No.~101102140 and project name `QIA-Phase 1’.
SW also acknowledge funding from NWO VICI.
GA acknowledges funding by NSF CQN Grant No. 1941583, NSF Grant No. 2346089, and NSF Grant No. 2402861.
We thank Anders S S{\o}rensen, Benedikt Tissot, Hugues de Riedmatten, Tracy E Northup, and Victor Krutyanskiy for fruitful discussions.
HR, TN, and VK provided parameter values in the state-of-the-art and optimistic regime.
We also thank Francisco Ferreira da Silva for critical feedback on the content of this manuscript and Michal van Hooft for advice and assistance with NetSquid.
\vspace{-7pt}
\section*{Author Contribution}
\label{sec:contribution}
\vspace{-5pt}
SM led the development and execution of the project.
SM and SK worked on the analytics and preparation of the manuscript.
SM and GA worked on the simulations using NetSquid.
SW conceived and supervised the project.

\bibliographystyle{IEEEtran}
\section*{References}
\bibliography{references}

@article{ekert1991quantum,
  title = {Quantum cryptography based on Bell's theorem},
  author = {Ekert, Artur K.},
  journal = {Phys. Rev. Lett.},
  volume = {67},
  issue = {6},
  pages = {661--663},
  numpages = {0},
  year = {1991},
  month = {Aug},
  publisher = {American Physical Society},
  doi = {10.1103/PhysRevLett.67.661},
}

@article{arrighi2006blind,
  title={Blind quantum computation},
  author={Arrighi, Pablo and Salvail, Louis},
  journal={International Journal of Quantum Information},
  volume={4},
  number={05},
  pages={883--898},
  year={2006},
  publisher={World Scientific}
}

@inproceedings{broadbent2009universal,
  title={Universal blind quantum computation},
  author={Broadbent, Anne and Fitzsimons, Joseph and Kashefi, Elham},
  booktitle={2009 50th annual IEEE symposium on foundations of computer science},
  pages={517--526},
  year={2009},
  organization={IEEE}
}

@article{giovannetti2004quantum,
  title={Quantum-enhanced measurements: beating the standard quantum limit},
  author={Giovannetti, Vittorio and Lloyd, Seth and Maccone, Lorenzo},
  journal={Science},
  volume={306},
  number={5700},
  pages={1330--1336},
  year={2004},
  publisher={American Association for the Advancement of Science}
}

@article{jozsa2000quantum,
  title={Quantum clock synchronization based on shared prior entanglement},
  author={Jozsa, Richard and Abrams, Daniel S and Dowling, Jonathan P and Williams, Colin P},
  journal={Physical Review Letters},
  volume={85},
  number={9},
  pages={2010},
  year={2000},
  publisher={APS}
}

@article{mylavarapu2024teleportation,
  title={Teleportation fidelity of quantum repeater networks},
  author={Mylavarapu, Ganesh and Ghosh, Subrata and Hens, Chittaranjan and Chakrabarty, Indranil and Mitra, Subhadip},
  journal={Physical Review A},
  volume={112},
  number={3},
  pages={032618},
  year={2025},
  publisher={APS}
}

@inproceedings{chia2012quantum,
  title={Quantum blind computation with teleportation-based computation},
  author={Chia, Nai-Hui and Chien, Chia-Hung and Chung, Wei-Ho and Kuo, Sy-Yen},
  booktitle={2012 Ninth International Conference on Information Technology-New Generations},
  pages={769--774},
  year={2012},
  organization={IEEE}
}

@article{da2024requirements,
  title={Requirements for upgrading trusted nodes to a repeater chain over 900 km of optical fiber},
  author={da Silva, Francisco F. and Avis, Guus and Slater, Joshua A and Wehner, Stephanie},
  journal={Quantum Science and Technology},
  volume={9},
  number={4},
  pages={045041},
  year={2024},
  publisher={IOP Publishing}
}

@article{rosenband2007observation,
  title={Observation of the S 0 1→ P 0 3 clock transition in Al+ 27},
  author={Rosenband, Till and Schmidt, Piet O and Hume, David B and Itano, Wayne M and Fortier, Tara M and Stalnaker, Jason E and Kim, K and Diddams, Scott A and Koelemeij, JCJ and Bergquist, JC and others},
  journal={Physical review letters},
  volume={98},
  number={22},
  pages={220801},
  year={2007},
  publisher={APS}
}

@article{home2009memory,
  title={Memory coherence of a sympathetically cooled trapped-ion qubit},
  author={Home, Jonathan P and McDonnell, MJ and Szwer, DJ and Keitch, BC and Lucas, DM and Stacey, DN and Steane, AM},
  journal={Physical Review A—Atomic, Molecular, and Optical Physics},
  volume={79},
  number={5},
  pages={050305},
  year={2009},
  publisher={APS}
}

@article{lago2021telecom,
  title={Telecom-heralded entanglement between multimode solid-state quantum memories},
  author={Lago-Rivera, Dario and Grandi, Samuele and Rakonjac, Jelena V and Seri, Alessandro and de Riedmatten, Hugues},
  journal={Nature},
  volume={594},
  number={7861},
  pages={37--40},
  year={2021},
  publisher={Nature Publishing Group UK London}
}

@article{liu2021heralded,
  title={Heralded entanglement distribution between two absorptive quantum memories},
  author={Liu, Xiao and Hu, Jun and Li, Zong-Feng and Li, Xue and Li, Pei-Yun and Liang, Peng-Jun and Zhou, Zong-Quan and Li, Chuan-Feng and Guo, Guang-Can},
  journal={Nature},
  volume={594},
  number={7861},
  pages={41--45},
  year={2021},
  publisher={Nature Publishing Group UK London}
}

@article{ruskuc2025multiplexed,
  title={Multiplexed entanglement of multi-emitter quantum network nodes},
  author={Ruskuc, A and Wu, C-J and Green, E and Hermans, SLN and Pajak, W and Choi, J and Faraon, A},
  journal={Nature},
  pages={1--6},
  year={2025},
  publisher={Nature Publishing Group UK London}
}

@article{wehner2018quantum,
  title={Quantum internet: A vision for the road ahead},
  author={Wehner, Stephanie and Elkouss, David and Hanson, Ronald},
  journal={Science},
  volume={362},
  number={6412},
  pages={eaam9288},
  year={2018},
  publisher={American Association for the Advancement of Science}
}

@article{wu2020near,
  title={Near-term performance of quantum repeaters with imperfect ensemble-based quantum memories},
  author={Wu, Yufeng and Liu, Jianlong and Simon, Christoph},
  journal={Physical Review A},
  volume={101},
  number={4},
  pages={042301},
  year={2020},
  publisher={APS}
}

@inproceedings{drmota2022long,
  title={Long-lived entanglement memory in a trapped-ion quantum network node},
  author={Drmota, Peter and Nadlinger, David and Nichol, Bethan and Main, Dougal and Ainley, Ellis and Araneda, Gabriel and Srinivas, Raghavendra and Ballance, Chris and Lucas, David},
  booktitle={APS Division of Atomic, Molecular and Optical Physics Meeting Abstracts},
  volume={2022},
  pages={Q03--004},
  year={2022}
}

@article{drmota2023robust,
  title={Robust quantum memory in a trapped-ion quantum network node},
  author={Drmota, Peter and Main, Dougal and Nadlinger, DP and Nichol, BC and Weber, Marius Alfons and Ainley, EM and Agrawal, Ayush and Srinivas, Raghavendra and Araneda, Gabriel and Ballance, CJ and others},
  journal={Physical Review Letters},
  volume={130},
  number={9},
  pages={090803},
  year={2023},
  publisher={APS}
}

@article{inlek2017multispecies,
  title={Multispecies trapped-ion node for quantum networking},
  author={Inlek, Ismail Volkan and Crocker, Clayton and Lichtman, Martin and Sosnova, Ksenia and Monroe, Christopher},
  journal={Physical review letters},
  volume={118},
  number={25},
  pages={250502},
  year={2017},
  publisher={APS}
}

@article{bennett1993teleporting,
  title={Teleporting an unknown quantum state via dual classical and Einstein-Podolsky-Rosen channels},
  author={Bennett, Charles H and Brassard, Gilles and Cr{\'e}peau, Claude and Jozsa, Richard and Peres, Asher and Wootters, William K},
  journal={Physical review letters},
  volume={70},
  number={13},
  pages={1895},
  year={1993},
  publisher={APS}
}

@article{bennett1996mixed,
  title={Mixed-state entanglement and quantum error correction},
  author={Bennett, Charles H and DiVincenzo, David P and Smolin, John A and Wootters, William K},
  journal={Physical Review A},
  volume={54},
  number={5},
  pages={3824},
  year={1996},
  publisher={APS}
}

@article{brand2020efficient,
  title={Efficient computation of the waiting time and fidelity in quantum repeater chains},
  author={Brand, Sebastiaan and Coopmans, Tim and Elkouss, David},
  journal={IEEE Journal on Selected Areas in Communications},
  volume={38},
  number={3},
  pages={619--639},
  year={2020},
  publisher={IEEE}
}

@article{jiang2007fast,
  title={Fast and robust approach to long-distance quantum communication with atomic ensembles},
  author={Jiang, L and Taylor, JM and Lukin, MD},
  journal={Physical Review A—Atomic, Molecular, and Optical Physics},
  volume={76},
  number={1},
  pages={012301},
  year={2007},
  publisher={APS}
}

@article{kamin2023exact,
  title={Exact rate analysis for quantum repeaters with imperfect memories and entanglement swapping as soon as possible},
  author={Kamin, Lars and Shchukin, Evgeny and Schmidt, Frank and van Loock, Peter},
  journal={Physical Review Research},
  volume={5},
  number={2},
  pages={023086},
  year={2023},
  publisher={APS}
}

@article{goodenough2025noise,
  title={On noise in swap ASAP repeater chains: exact analytics, distributions and tight approximations},
  author={Goodenough, Kenneth and Coopmans, Tim and Towsley, Don},
  journal={Quantum},
  volume={9},
  pages={1744},
  year={2025},
  publisher={Verein zur F{\"o}rderung des Open Access Publizierens in den Quantenwissenschaften}
}

@article{rozpkedek2019near,
  title={Near-term quantum-repeater experiments with nitrogen-vacancy centers: Overcoming the limitations of direct transmission},
  author={Rozp{\k{e}}dek, Filip and Yehia, Raja and Goodenough, Kenneth and Ruf, Maximilian and Humphreys, Peter C and Hanson, Ronald and Wehner, Stephanie and Elkouss, David},
  journal={Physical Review A},
  volume={99},
  number={5},
  pages={052330},
  year={2019},
  publisher={APS}
}

@article{rozpkedek2018parameter,
  title={Parameter regimes for a single sequential quantum repeater},
  author={Rozp{\k{e}}dek, Filip and Goodenough, Kenneth and Ribeiro, Jeremy and Kalb, Norbert and Vivoli, V Caprara and Reiserer, Andreas and Hanson, Ronald and Wehner, Stephanie and Elkouss, David},
  journal={Quantum Science and Technology},
  volume={3},
  number={3},
  pages={034002},
  year={2018},
  publisher={IOP Publishing}
}

@article{collins2007multiplexed,
  title={Multiplexed memory-insensitive quantum repeaters},
  author={Collins, OA and Jenkins, SD and Kuzmich, A and Kennedy, TAB},
  journal={Physical review letters},
  volume={98},
  number={6},
  pages={060502},
  year={2007},
  publisher={APS}
}

@article{khatri2019practical,
  title={Practical figures of merit and thresholds for entanglement distribution in quantum networks},
  author={Khatri, Sumeet and Matyas, Corey T and Siddiqui, Aliza U and Dowling, Jonathan P},
  journal={Physical Review Research},
  volume={1},
  number={2},
  pages={023032},
  year={2019},
  publisher={APS}
}

@article{li2021efficient,
  title={Efficient optimization of cutoffs in quantum repeater chains},
  author={Li, Boxi and Coopmans, Tim and Elkouss, David},
  journal={IEEE Transactions on Quantum Engineering},
  volume={2},
  pages={1--15},
  year={2021},
  publisher={IEEE}
}

@article{shor2000simple,
  title={Simple proof of security of the BB84 quantum key distribution protocol},
  author={Shor, Peter W and Preskill, John},
  journal={Physical review letters},
  volume={85},
  number={2},
  pages={441},
  year={2000},
  publisher={APS}
}

@INPROCEEDINGS{chandra2022scheduling,
  author={Chandra, Aparimit and Dai, Wenhan and Towsley, Don},
  booktitle={2022 IEEE International Conference on Quantum Computing and Engineering (QCE)}, 
  title={Scheduling Quantum Teleportation with Noisy Memories}, 
  year={2022},
  volume={},
  number={},
  pages={437-446},
  doi={10.1109/QCE53715.2022.00065}}

@article{pascoletti1984scalarizing,
  title={Scalarizing vector optimization problems},
  author={Pascoletti, Adriano and Serafini, Paolo},
  journal={Journal of Optimization Theory and Applications},
  volume={42},
  number={4},
  pages={499--524},
  year={1984},
  publisher={Springer}
}

@techreport{schaffer1985some,
  title={Some experiments in machine learning using vector evaluated genetic algorithms},
  author={Schaffer, James David},
  year={1985},
  institution={Vanderbilt Univ., Nashville, TN (USA)}
}

@article{zwerger2017quantum,
  title={Quantum repeaters based on trapped ions with decoherence-free subspace encoding},
  author={Zwerger, Michael and Lanyon, Ben P and Northup, Tracy E and Muschik, Christine A and D{\"u}r, Wolfgang and Sangouard, Nicolas},
  journal={Quantum Science and Technology},
  volume={2},
  number={4},
  pages={044001},
  year={2017},
  publisher={IOP Publishing}
}

@article{haffner2005robust,
  title={Robust entanglement},
  author={H{\"a}ffner, H and Schmidt-Kaler, F and H{\"a}nsel, W and Roos, CF and K{\"o}rber, T and Chwalla, M and Riebe, M and Benhelm, J and Rapol, UD and Becher, C and others},
  journal={Applied Physics B},
  volume={81},
  number={2},
  pages={151--153},
  year={2005},
  publisher={Springer}
}

@article{higgins2017single,
  title={Single strontium Rydberg ion confined in a Paul trap},
  author={Higgins, Gerard and Li, Weibin and Pokorny, Fabian and Zhang, Chi and Kress, Florian and Maier, Christine and Haag, Johannes and Bodart, Quentin and Lesanovsky, Igor and Hennrich, Markus},
  journal={Physical Review X},
  volume={7},
  number={2},
  pages={021038},
  year={2017},
  publisher={APS}
}

@inproceedings{ghadimi2019towards,
  title={Towards long-distance quantum communication using trapped ions and frequency qubits},
  author={Ghadimi, M and Connell, S and Scarabel, J and Shimizu, K and Streed, E and Lobino, M},
  booktitle={AOS Australian Conference on Optical Fibre Technology (ACOFT) and Australian Conference on Optics, Lasers, and Spectroscopy (ACOLS) 2019},
  volume={11200},
  pages={69--70},
  year={2019},
  organization={SPIE}
}

@article{zhong2015optically,
  title={Optically addressable nuclear spins in a solid with a six-hour coherence time},
  author={Zhong, Manjin and Hedges, Morgan P and Ahlefeldt, Rose L and Bartholomew, John G and Beavan, Sarah E and Wittig, Sven M and Longdell, Jevon J and Sellars, Matthew J},
  journal={Nature},
  volume={517},
  number={7533},
  pages={177--180},
  year={2015},
  publisher={Nature Publishing Group UK London}
}

@article{herbauts2013demonstration,
  title={Demonstration of active routing of entanglement in a multi-user network},
  author={Herbauts, Isabelle and Blauensteiner, Bibiane and Poppe, Andreas and Jennewein, Thomas and Huebel, Hannes},
  journal={Optics express},
  volume={21},
  number={23},
  pages={29013--29024},
  year={2013},
  publisher={Optical Society of America}
}

@article{ursin2007entanglement,
  title={Entanglement-based quantum communication over 144 km},
  author={Ursin, Rupert and Tiefenbacher, Felix and Schmitt-Manderbach, Tobias and Weier, Henning and Scheidl, Thomas and Lindenthal, Michael and Blauensteiner, Bibiane and Jennewein, Thomas and Perdigues, Josep and Trojek, Pavel and others},
  journal={Nature physics},
  volume={3},
  number={7},
  pages={481--486},
  year={2007},
  publisher={Nature Publishing Group UK London}
}

@article{yin2017satellite,
  title={Satellite-based entanglement distribution over 1200 kilometers},
  author={Yin, Juan and Cao, Yuan and Li, Yu-Huai and Liao, Sheng-Kai and Zhang, Liang and Ren, Ji-Gang and Cai, Wen-Qi and Liu, Wei-Yue and Li, Bo and Dai, Hui and others},
  journal={Science},
  volume={356},
  number={6343},
  pages={1140--1144},
  year={2017},
  publisher={American Association for the Advancement of Science}
}

@article{Bruzewicz2019,
    author = {Bruzewicz, Colin D. and Chiaverini, John and McConnell, Robert and Sage, Jeremy M.},
    title = {Trapped-ion quantum computing: Progress and challenges},
    journal = {Applied Physics Reviews},
    volume = {6},
    number = {2},
    pages = {021314},
    year = {2019},
    month = {05},
    issn = {1931-9401},
    doi = {10.1063/1.5088164},
}

@article{Avis2023DelftEindhoven,
  title={Requirements for a processing-node quantum repeater on a real-world fiber grid},
  author={Avis, Guus and Ferreira da Silva, Francisco and Coopmans, Tim and Dahlberg, Axel and Jirovsk{\'a}, Hana and Maier, David and Rabbie, Julian and Torres-Knoop, Ariana and Wehner, Stephanie},
  journal={npj Quantum Information},
  volume={9},
  number={1},
  pages={100},
  year={2023},
  publisher={Nature Publishing Group},
  doi={10.1038/s41534-023-00765-x},
}

@article{Schupp2021InterfaceBetweenTrapped-Ion,
  title = {Interface between Trapped-Ion Qubits and Traveling Photons with Close-to-Optimal Efficiency},
  author = {Schupp, J. and Krcmarsky, V. and Krutyanskiy, V. and Meraner, M. and Northup, T.E. and Lanyon, B.P.},
  journal = {PRX Quantum},
  volume = {2},
  issue = {2},
  pages = {020331},
  numpages = {16},
  year = {2021},
  month = {Jun},
  publisher = {American Physical Society},
  doi = {10.1103/PRXQuantum.2.020331},
}

@article{Krutyanskiy2019Light,
  title={Light-matter entanglement over 50 km of optical fibre},
  author={Krutyanskiy, Viktor and Meraner, Martin and Schupp, Josef and Krcmarsky, Vojtech and Hainzer, Helene and Lanyon, Ben P},
  journal={npj Quantum Information},
  volume={5},
  number={1},
  pages={72},
  year={2019},
  publisher={Nature Publishing Group UK London}
}

@ARTICLE{Munro2015,
  author={Munro, William J. and Azuma, Koji and Tamaki, Kiyoshi and Nemoto, Kae},
  journal={IEEE Journal of Selected Topics in Quantum Electronics}, 
  title={Inside Quantum Repeaters}, 
  year={2015},
  volume={21},
  number={3},
  pages={78-90},
  doi={10.1109/JSTQE.2015.2392076}}

@article{Sangouard2011,
  title = {Quantum repeaters based on atomic ensembles and linear optics},
  author = {Sangouard, Nicolas and Simon, Christoph and de Riedmatten, Hugues and Gisin, Nicolas},
  journal = {Rev. Mod. Phys.},
  volume = {83},
  issue = {1},
  pages = {33--80},
  numpages = {0},
  year = {2011},
  month = {Mar},
  publisher = {American Physical Society},
  doi = {10.1103/RevModPhys.83.33},
}

@misc{privatecommunication2025,
  author = {Northup, Tracy E. and Lanyon, Benjamin P. and Krutyanskiy, Victor},
  note = {private communications, 2025},
}

@article{borregaard2025testing,
  title={Testing quantum theory on curved spacetime with quantum networks},
  author={Borregaard, Johannes and Pikovski, Igor},
  journal={Physical Review Research},
  volume={7},
  number={2},
  pages={023192},
  year={2025},
  publisher={APS}
}

@article{ruskuc2022nuclear,
  title={Nuclear spin-wave quantum register for a solid-state qubit},
  author={Ruskuc, Andrei and Wu, Chun-Ju and Rochman, Jake and Choi, Joonhee and Faraon, Andrei},
  journal={Nature},
  volume={602},
  number={7897},
  pages={408--413},
  year={2022},
  publisher={Nature Publishing Group UK London}
}

@article{Prielinger2024,
  title={Surrogate-guided optimization in quantum networks},
  author={Prielinger, Luise and I{\~n}esta, {\'A}lvaro G and Vardoyan, Gayane},
  journal={npj Quantum Information},
  volume={11},
  number={1},
  pages={89},
  year={2025},
  publisher={Nature Publishing Group UK London}
}

@article{childress2006fault,
  title={Fault-tolerant quantum communication based on solid-state photon emitters},
  author={Childress, Lilian and Taylor, JM and S{\o}rensen, Anders S{\o}ndberg and Lukin, MD},
  journal={Physical review letters},
  volume={96},
  number={7},
  pages={070504},
  year={2006},
  publisher={APS}
}

@article{yin2012quantum,
  title={Quantum teleportation and entanglement distribution over 100-kilometre free-space channels},
  author={Yin, Juan and Ren, Ji-Gang and Lu, He and Cao, Yuan and Yong, Hai-Lin and Wu, Yu-Ping and Liu, Chang and Liao, Sheng-Kai and Zhou, Fei and Jiang, Yan and others},
  journal={Nature},
  volume={488},
  number={7410},
  pages={185--188},
  year={2012},
  publisher={Nature Publishing Group UK London}
}

@misc{beauchamp2025modularquantumnetworkarchitecture,
      title={A Modular Quantum Network Architecture for Integrating Network Scheduling with Local Program Execution}, 
      author={Thomas R. Beauchamp and Hana Jirovská and Scarlett Gauthier and Stephanie Wehner},
      year={2025},
      eprint={2503.12582},
      archivePrefix={arXiv},
      primaryClass={quant-ph},
      url={https://arxiv.org/abs/2503.12582}, 
}

@article{Ruf2021quantum,
  title={Quantum networks based on color centers in diamond},
  author={Ruf, Maximilian and Wan, Noel H and Choi, Hyeongrak and Englund, Dirk and Hanson, Ronald},
  journal={Journal of Applied Physics},
  volume={130},
  number={7},
  year={2021},
  publisher={AIP Publishing}
}

@article{pompili2021realization,
  title={Realization of a multinode quantum network of remote solid-state qubits},
  author={Pompili, Matteo and Hermans, Sophie LN and Baier, Simon and Beukers, Hans KC and Humphreys, Peter C and Schouten, Raymond N and Vermeulen, Raymond FL and Tiggelman, Marijn J and dos Santos Martins, Laura and Dirkse, Bas and others},
  journal={Science},
  volume={372},
  number={6539},
  pages={259--264},
  year={2021},
  publisher={American Association for the Advancement of Science}
}

@article{maunz2007quantum,
  title={Quantum interference of photon pairs from two remote trapped atomic ions},
  author={Maunz, Peter and Moehring, DL and Olmschenk, Steven and Younge, Kelly Cooper and Matsukevich, DN and Monroe, Christopher},
  journal={Nature Physics},
  volume={3},
  number={8},
  pages={538--541},
  year={2007},
  publisher={Nature Publishing Group UK London}
}

@article{covey2023quantum,
  title={Quantum networks with neutral atom processing nodes},
  author={Covey, Jacob P and Weinfurter, Harald and Bernien, Hannes},
  journal={npj Quantum Information},
  volume={9},
  number={1},
  pages={90},
  year={2023},
  publisher={Nature Publishing Group UK London}
}

@article{uphoff2016integrated,
  title={An integrated quantum repeater at telecom wavelength with single atoms in optical fiber cavities},
  author={Uphoff, Manuel and Brekenfeld, Manuel and Rempe, Gerhard and Ritter, Stephan},
  journal={Applied Physics B},
  volume={122},
  pages={1--15},
  year={2016},
  publisher={Springer}
}

@article{stephenson2020high,
  title={High-rate, high-fidelity entanglement of qubits across an elementary quantum network},
  author={Stephenson, LJ and Nadlinger, DP and Nichol, BC and An, Shuoming and Drmota, P and Ballance, Timothy G and Thirumalai, K and Goodwin, Joseph Francis and Lucas, David M and Ballance, CJ},
  journal={Physical review letters},
  volume={124},
  number={11},
  pages={110501},
  year={2020},
  publisher={APS}
}

@article{bernien2013heralded,
  title={Heralded entanglement between solid-state qubits separated by three metres},
  author={Bernien, Hannes and Hensen, Bas and Pfaff, Wolfgang and Koolstra, Gerwin and Blok, Machiel S and Robledo, Lucio and Taminiau, Tim H and Markham, Matthew and Twitchen, Daniel J and Childress, Lilian and others},
  journal={Nature},
  volume={497},
  number={7447},
  pages={86--90},
  year={2013},
  publisher={Nature Publishing Group UK London}
}

@article{gottesman1999demonstrating,
  title={Demonstrating the viability of universal quantum computation using teleportation and single-qubit operations},
  author={Gottesman, Daniel and Chuang, Isaac L},
  journal={Nature},
  volume={402},
  number={6760},
  pages={390--393},
  year={1999},
  publisher={Nature Publishing Group},
  doi={10.1038/46503}
}

@article{hensen2015loophole,
  title={Loophole-free Bell inequality violation using electron spins separated by 1.3 kilometres},
  author={Hensen, Bas and Bernien, Hannes and Dr{\'e}au, Ana{\"\i}s E and Reiserer, Andreas and Kalb, Norbert and Blok, Machiel S and Ruitenberg, Just and Vermeulen, Raymond FL and Schouten, Raymond N and Abell{\'a}n, Carlos and others},
  journal={Nature},
  volume={526},
  number={7575},
  pages={682--686},
  year={2015},
  publisher={Nature Publishing Group UK London}
}

@article{bell1964einstein,
  title={On the einstein podolsky rosen paradox},
  author={Bell, John S},
  journal={Physics Physique Fizika},
  volume={1},
  number={3},
  pages={195},
  year={1964},
  publisher={APS}
}

@article{chou2007functional,
  title={Functional quantum nodes for entanglement distribution over scalable quantum networks},
  author={Chou, Chin-Wen and Laurat, Julien and Deng, Hui and Choi, Kyung Soo and De Riedmatten, Hugues and Felinto, Daniel and Kimble, H Jeff},
  journal={Science},
  volume={316},
  number={5829},
  pages={1316--1320},
  year={2007},
  publisher={American Association for the Advancement of Science}
}

@article{sinclair2014spectral,
  title={Spectral multiplexing for scalable quantum photonics using an atomic frequency comb quantum memory and feed-forward control},
  author={Sinclair, Neil and Saglamyurek, Erhan and Mallahzadeh, Hassan and Slater, Joshua A and George, Mathew and Ricken, Raimund and Hedges, Morgan P and Oblak, Daniel and Simon, Christoph and Sohler, Wolfgang and others},
  journal={Physical review letters},
  volume={113},
  number={5},
  pages={053603},
  year={2014},
  publisher={APS}
}

@article{gauthier2023architecture,
  title={An architecture for control of entanglement generation switches in quantum networks},
  author={Gauthier, Scarlett and Vardoyan, Gayane and Wehner, Stephanie},
  journal={IEEE Transactions on Quantum Engineering},
  volume={4},
  pages={1--17},
  year={2023},
  publisher={IEEE}
}

@article{duan2001long,
  title={Long-distance quantum communication with atomic ensembles and linear optics},
  author={Duan, L-M and Lukin, Mikhail D and Cirac, J Ignacio and Zoller, Peter},
  journal={Nature},
  volume={414},
  number={6862},
  pages={413--418},
  year={2001},
  publisher={Nature Publishing Group UK London}
}

@article{simon2007quantum,
  title={Quantum repeaters with photon pair sources and multimode memories},
  author={Simon, Christoph and De Riedmatten, Hugues and Afzelius, Mikael and Sangouard, Nicolas and Zbinden, Hugo and Gisin, Nicolas},
  journal={Physical review letters},
  volume={98},
  number={19},
  pages={190503},
  year={2007},
  publisher={APS}
}

@article{werner1989quantum,
  title={Quantum states with Einstein-Podolsky-Rosen correlations admitting a hidden-variable model},
  author={Werner, Reinhard F},
  journal={Physical Review A},
  volume={40},
  number={8},
  pages={4277},
  year={1989},
  publisher={APS}
}

@article{Barrett2005DoubleClick,
  title = {Efficient high-fidelity quantum computation using matter qubits and linear optics},
  author = {Barrett, Sean D. and Kok, Pieter},
  journal = {Phys. Rev. A},
  volume = {71},
  issue = {6},
  pages = {060310},
  numpages = {4},
  year = {2005},
  month = {Jun},
  publisher = {American Physical Society},
  doi = {10.1103/PhysRevA.71.060310},
}

@article{Krutyanskiy2024Entanglement101km,
  title = {Multimode Ion-Photon Entanglement over 101 Kilometers},
  author = {Krutyanskiy, V. and Canteri, M. and Meraner, M. and Krcmarsky, V. and Lanyon, B.P.},
  journal = {PRX Quantum},
  volume = {5},
  issue = {2},
  pages = {020308},
  numpages = {11},
  year = {2024},
  month = {Apr},
  publisher = {American Physical Society},
  doi = {10.1103/PRXQuantum.5.020308},
}

@article{azuma2015all,
  title={All-photonic quantum repeaters},
  author={Azuma, Koji and Tamaki, Kiyoshi and Lo, Hoi-Kwong},
  journal={Nature communications},
  volume={6},
  number={1},
  pages={6787},
  year={2015},
  publisher={Nature Publishing Group UK London}
}

@article{liorni2021quantum,
  title={Quantum repeaters in space},
  author={Liorni, Carlo and Kampermann, Hermann and Bru{\ss}, Dagmar},
  journal={New Journal of Physics},
  volume={23},
  number={5},
  pages={053021},
  year={2021},
  publisher={IOP Publishing}
}

@article{de2024analysis,
  title={On the analysis of quantum repeater chains with sequential swaps},
  author={de Andrade, Matheus Guedes and Van Milligen, Emily A and Bacciottini, Leonardo and Chandra, Aparimit and Pouryousef, Shahrooz and Panigrahy, Nitish K and Vardoyan, Gayane and Towsley, Don},
  journal={arXiv preprint arXiv:2405.18252},
  year={2024},
  url={https://arxiv.org/abs/2405.18252}
}

@article{Ferreira_da_Silva_2021,
   title={Optimizing entanglement generation and distribution using genetic algorithms},
   volume={6},
   ISSN={2058-9565},
   DOI={10.1088/2058-9565/abfc93},
   number={3},
   journal={Quantum Science and Technology},
   publisher={IOP Publishing},
   author={F. da Silva, Francisco and Torres-Knoop, Ariana and Coopmans, Tim and Maier, David and Wehner, Stephanie},
   year={2021},
   month=jun, pages={035007} 
}

@article{wootters1982single,
  title={A single quantum cannot be cloned},
  author={Wootters, William K and Zurek, Wojciech H},
  journal={Nature},
  volume={299},
  number={5886},
  pages={802--803},
  year={1982},
  publisher={Nature Publishing Group UK London}
}

@article{asadi2018quantum,
  title={Quantum repeaters with individual rare-earth ions at telecommunication wavelengths},
  author={Asadi, F Kimiaee and Lauk, Nikolai and Wein, S and Sinclair, Neil and O'Brien, Chris and Simon, Christoph},
  journal={Quantum},
  volume={2},
  pages={93},
  year={2018},
  publisher={Verein zur F{\"o}rderung des Open Access Publizierens in den Quantenwissenschaften}
}

@article{bernardes2011rate,
  title={Rate analysis for a hybrid quantum repeater},
  author={Bernardes, Nadja K and Praxmeyer, Ludmi{\l}a and van Loock, Peter},
  journal={Physical Review A—Atomic, Molecular, and Optical Physics},
  volume={83},
  number={1},
  pages={012323},
  year={2011},
  publisher={APS}
}

@article{Coopmans2021Netsquid,
  title = {NetSquid, a NETwork Simulator for QUantum Information using Discrete events},
  author = {Coopmans, Tim and Knegjens, Robert and Dahlberg, Axel and Maier, David and Nijsten, Loek and de Oliveira Filho, Julio and Papendrecht, Martijn and Rabbie, Julian and Rozp{\k{e}}dek, Filip and Skrzypczyk, Matthew and Wubben, Leon and de Jong, Walter and Podareanu, Damian and Torres-Knoop, Ariana and Elkouss, David and Wehner, Stephanie},
  journal = {Communications Physics},
  volume = {4},
  number = {1},
  pages = {164},
  year = {2021},
  publisher = {Nature Publishing Group},
  doi = {10.1038/s42005-021-00647-8},
}

@article{vanDam2024hardware,
  title={Hardware requirements for trapped-ion-based verifiable blind quantum computing with a measurement-only client},
  author={van Dam, Janice and Avis, Guus and Propp, Tz B and da Silva, F Ferreira and Slater, Joshua A and Northup, Tracy E and Wehner, Stephanie},
  journal={Quantum Science and Technology},
  volume={9},
  number={4},
  pages={045031},
  year={2024},
  publisher={IOP Publishing}
}

@article{Krutyanskiy2023TelecomWavelength,
  title = {Telecom-Wavelength Quantum Repeater Node Based on a Trapped-Ion Processor},
  author = {Krutyanskiy, V. and Canteri, M. and Meraner, M. and Bate, J. and Krcmarsky, V. and Schupp, J. and Sangouard, N. and Lanyon, B. P.},
  journal = {Phys. Rev. Lett.},
  volume = {130},
  issue = {21},
  pages = {213601},
  numpages = {7},
  year = {2023},
  month = {May},
  publisher = {American Physical Society},
  doi = {10.1103/PhysRevLett.130.213601},
}

@article{Krutyanskiy2023Entanglement230Meters,
  title = {Entanglement of Trapped-Ion Qubits Separated by 230 Meters},
  author = {Krutyanskiy, V. and Galli, M. and Krcmarsky, V. and Baier, S. and Fioretto, D. A. and Pu, Y. and Mazloom, A. and Sekatski, P. and Canteri, M. and Teller, M. and Schupp, J. and Bate, J. and Meraner, M. and Sangouard, N. and Lanyon, B. P. and Northup, T. E.},
  journal = {Phys. Rev. Lett.},
  volume = {130},
  issue = {5},
  pages = {050803},
  numpages = {7},
  year = {2023},
  month = {Feb},
  publisher = {American Physical Society},
  doi = {10.1103/PhysRevLett.130.050803},
}

@article{pfaff2014unconditional,
  title={Unconditional quantum teleportation between distant solid-state quantum bits},
  author={Pfaff, Wolfgang and Hensen, Bas J and Bernien, Hannes and van Dam, Suzanne B and Blok, Machiel S and Taminiau, Tim H and Tiggelman, Marijn J and Schouten, Raymond N and Markham, Matthew and Twitchen, Daniel J and others},
  journal={Science},
  volume={345},
  number={6196},
  pages={532--535},
  year={2014},
  publisher={American Association for the Advancement of Science}
}

@article{briegel1998quantum,
  title={Quantum repeaters: the role of imperfect local operations in quantum communication},
  author={Briegel, H-J and D{\"u}r, Wolfgang and Cirac, Juan I and Zoller, Peter},
  journal={Physical Review Letters},
  volume={81},
  number={26},
  pages={5932},
  year={1998},
  publisher={APS}
}

@article{hartmann2007role,
  title={Role of memory errors in quantum repeaters},
  author={Hartmann, Lorenz and Kraus, Barbara and Briegel, H-J and D{\"u}r, W},
  journal={Physical Review A—Atomic, Molecular, and Optical Physics},
  volume={75},
  number={3},
  pages={032310},
  year={2007},
  publisher={APS}
}

@article{humphreys2018deterministic,
  title={Deterministic delivery of remote entanglement on a quantum network},
  author={Humphreys, Peter C and Kalb, Norbert and Morits, Jaco PJ and Schouten, Raymond N and Vermeulen, Raymond FL and Twitchen, Daniel J and Markham, Matthew and Hanson, Ronald},
  journal={Nature},
  volume={558},
  number={7709},
  pages={268--273},
  year={2018},
  publisher={Nature Publishing Group UK London}
}

@article{yuan2008experimental,
  title={Experimental demonstration of a BDCZ quantum repeater node},
  author={Yuan, Zhen-Sheng and Chen, Yu-Ao and Zhao, Bo and Chen, Shuai and Schmiedmayer, J{\"o}rg and Pan, Jian-Wei},
  journal={Nature},
  volume={454},
  number={7208},
  pages={1098--1101},
  year={2008},
  publisher={Nature Publishing Group UK London}
}

@article{yu2020entanglement,
  title={Entanglement of two quantum memories via fibres over dozens of kilometres},
  author={Yu, Yong and Ma, Fei and Luo, Xi-Yu and Jing, Bo and Sun, Peng-Fei and Fang, Ren-Zhou and Yang, Chao-Wei and Liu, Hui and Zheng, Ming-Yang and Xie, Xiu-Ping and others},
  journal={Nature},
  volume={578},
  number={7794},
  pages={240--245},
  year={2020},
  publisher={Nature Publishing Group UK London}
}

@article{taminiau2014universal,
  title={Universal control and error correction in multi-qubit spin registers in diamond},
  author={Taminiau, Tim Hugo and Cramer, Julia and van der Sar, Toeno and Dobrovitski, Viatcheslav V and Hanson, Ronald},
  journal={Nature nanotechnology},
  volume={9},
  number={3},
  pages={171--176},
  year={2014},
  publisher={Nature Publishing Group UK London}
}

@article{reiserer2016robust,
  title={Robust quantum-network memory using decoherence-protected subspaces of nuclear spins},
  author={Reiserer, Andreas and Kalb, Norbert and Blok, Machiel S and van Bemmelen, Koen JM and Taminiau, Tim H and Hanson, Ronald and Twitchen, Daniel J and Markham, Matthew},
  journal={Physical Review X},
  volume={6},
  number={2},
  pages={021040},
  year={2016},
  publisher={APS}
}

@article{kalb2017entanglement,
  title={Entanglement distillation between solid-state quantum network nodes},
  author={Kalb, Norbert and Reiserer, Andreas A and Humphreys, Peter C and Bakermans, Jacob JW and Kamerling, Sten J and Nickerson, Naomi H and Benjamin, Simon C and Twitchen, Daniel J and Markham, Matthew and Hanson, Ronald},
  journal={Science},
  volume={356},
  number={6341},
  pages={928--932},
  year={2017},
  publisher={American Association for the Advancement of Science}
}

@article{sangouard2009quantum,
  title={Quantum repeaters based on single trapped ions},
  author={Sangouard, Nicolas and Dubessy, Romain and Simon, Christoph},
  journal={Physical Review A—Atomic, Molecular, and Optical Physics},
  volume={79},
  number={4},
  pages={042340},
  year={2009},
  publisher={APS}
}

@article{lloyd2001long,
  title={Long distance, unconditional teleportation of atomic states via complete Bell state measurements},
  author={Lloyd, S and Shahriar, MS and Shapiro, JH and Hemmer, PR},
  journal={Physical Review Letters},
  volume={87},
  number={16},
  pages={167903},
  year={2001},
  publisher={APS}
}

@article{razavi2006long,
  title={Long-distance quantum communication with neutral atoms},
  author={Razavi, Mohsen and Shapiro, Jeffrey H},
  journal={Physical Review A—Atomic, Molecular, and Optical Physics},
  volume={73},
  number={4},
  pages={042303},
  year={2006},
  publisher={APS}
}

@article{horodecki1999general,
  title={General teleportation channel, singlet fraction, and quasidistillation},
  author={Horodecki, Micha{\l} and Horodecki, Pawe{\l} and Horodecki, Ryszard},
  journal={Physical Review A},
  volume={60},
  number={3},
  pages={1888},
  year={1999},
  publisher={APS}
}

@article{bennett1992quantum,
  title={Quantum cryptography without Bell’s theorem},
  author={Bennett, Charles H and Brassard, Gilles and Mermin, N David},
  journal={Physical review letters},
  volume={68},
  number={5},
  pages={557},
  year={1992},
  publisher={APS}
}

@article{tissot2025hybrid,
  title={Hybrid Single-Ion Atomic-Ensemble Node for High-Rate Remote Entanglement Generation},
  author={Tissot, Benedikt and Maiti, Soubhadra and Hellebek, Emil R and S{\o}rensen, Anders S{\o}ndberg},
  journal={arXiv preprint arXiv:2511.04488},
  year={2025},
  url={https://arxiv.org/abs/2511.04488}
}

@article{tissot2025single,
  title={Single and Double-click High-Rate Entanglement Generation Between Distant Ions Using Multiplexed Atomic Ensembles},
  author={Tissot, Benedikt and Maiti, Soubhadra and Hellebek, Emil and S{\o}rensen, Anders S{\o}ndberg},
  journal={arXiv preprint arXiv:2511.04987},
  year={2025},
  url={https://arxiv.org/abs/2511.04987}
}

@article{sun2025hybrid,
  title={Hybrid Quantum Repeater Chains with Atom-based Quantum Processing Units and Quantum Memory Multiplexers},
  author={Sun, Shin and Bhatti, Daniel and Gao, Shaobo and Elkouss, David and Takahashi, Hiroki},
  journal={arXiv preprint arXiv:2512.21655},
  year={2025},
  url={https://arxiv.org/abs/2512.21655}
}

@misc{git_repo_teleportation,
  author = {Maiti, Soubhadra},
  title = {Code underlying the paper: Requirements for Teleportation in an Intercity Quantum Network},
  url = {https://gitlab.tudelft.nl/wehner-research/teleportation_requirements_paper.git},
  version = {paper_version},
  year = {2026},
  publisher = {GitLab (TU Delft)}
}
\appendix
\section{Derivation of Teleportation Rate and Expected Fidelity in the Metropolitan Network}
\label{sec:metro_rate_and_fidelity}
\everypar{\looseness=-1}

In our model, teleportation in an MN proceeds by first generating entanglement between two end nodes in the same MN via HEG, followed by the teleportation step.
We characterise the network by the expected entanglement fidelity, corresponding rate, and the memory coherence time of the end nodes.
The teleportation rate is a function of the entanglement-generation probability between end nodes at zero separation, i.e., the base efficiency, and for the base and optimistic values, we use experimental and projected values, respectively, from trapped-ion platforms.
In Sec.~\ref{sec:derivation_metro_p_m}, we derive the dependence of the base efficiency on experimentally relevant parameters.
Next, we derive the expressions for the teleportation rate and fidelity in~\ref{sec:metro_teleportation_rate} and~\ref{sec:metro_teleportation_fidelity}, respectively.
Note that we do not model the entanglement fidelity as a function of the experimental parameters but rather use the corresponding baseline value as the one obtained in state-of-the-art experiments~\cite{Krutyanskiy2023Entanglement230Meters}.
We apply the same for the coherence time which we model as a decoherence channel. 

\subsection{Relating Metropolitan Entanglement-Generation Probability to Trapped-Ion Experiments}
\label{sec:derivation_metro_p_m}

In this section, we provide a detailed explanation of the hardware parameters and entanglement generation between remote nodes using trapped-ion devices~\cite{Krutyanskiy2023Entanglement230Meters, Krutyanskiy2023TelecomWavelength, Krutyanskiy2019Light, Krutyanskiy2024Entanglement101km}.
Note that, in this work, we model the entanglement generation and teleportation scheme under the assumptions~\ref{assumption:werner_state}--\ref{assumption:instantaneous_qubit_preparation_local_operation}.
The purpose of this is to provide a simple and general model for a platform-agnostic setup.
The parameters required to obtain the form of the entanglement-generation probability are
\begin{itemize}
    \item \textbf{Efficiency for collecting a photon emitted from the ion} ($\eta_{\text{ion}}$):
        The probability of successfully collecting a photon emitted from the ion upon excitation, into the fibre, excluding fibre transmission losses.
        If you push a button to get the photon out, this is the probability that it comes out. This includes every loss except detector and fibre transmission losses.
    \item \textbf{Detection efficiency for photons at ion frequency} ($\eta_{\text{det}}^\text{ion-freq}$):
        The efficiency of detectors optimised for the natural emission frequency of Ca$^+$ ions.
    \item \textbf{Frequency conversion efficiency, ion to telecom} ($\eta_{\text{FC}}$):
    The fraction of photons successfully converted from the Ca$^+$ emission frequency to telecom frequency, including additional losses in fibre.

    \item \textbf{Detector efficiency at telecom frequency} ($\eta_\text{det}^\text{telecom}$):
    The probability of detecting a photon incident on the detector, without any other inefficiencies being included here.
    
    \item \textbf{Efficiency of using a `truncated detection window'} ($\eta_{\text{penalty}}$):
    In the two-click scheme experiment, a wider detection window allows for more photons but may include more noise, while a narrower window improves fidelity but reduces the detection rate.
    Hence, selecting a narrow coincidence window helps improve ion-ion entanglement fidelity, reducing the total success probability of the protocol.

    \item \textbf{Average preparation time between shots} ($t_\text{m}^\text{prep}$): 
        The average time required for sequential photon generation via laser excitation, including experimental cycle delays such as initialisation, cooling, optical pumping, etc.
\end{itemize}

\looseness=-1 Recall from~\eqref{eq:t_m_metro} that since two end nodes in a metropolitan network are separated by a distance $2d_\text{m}'$ with the hub being at the midpoint (see Fig.~\ref{fig:intercity_network_diagram}), the average time per entanglement generation attempt is given by ${t_\text{m} \!=\! t_\text{m}^\text{prep} \!+\! 2t_\text{m}^{\text{class}}}$, where $t_\text{m}^\text{class} \!=\! d_{\text{m}'}/c$ and $c$ is the speed of light in fibre.

\looseness=-1 Recall from~\ref{assumption:entangling_metro} that in our model, we adopt the sequential double-click HEG protocol~\cite{Barrett2005DoubleClick}
Note that we consider a simple model with parameters explained above and do not consider other factors such as imperfect indistinguishability, non-photon-number-resolving detectors, detector dark counts~\cite{Avis2023DelftEindhoven}, etc.
First, both end nodes, here trapped-ion devices, generate matter-photon entanglement and send the photon towards the central hub.
The photons emitted by the ions with efficiency $\eta_\text{ion}$ are frequency converted to telecom frequency with an efficiency $\eta_\text{FC}$.
Node-to-node entanglement is heralded by the BSM at the hub.
Specifically, this corresponds to the detection of two photons after interference, with a detection efficiency $\eta_\text{det}^\text{telecom}$.
Thus, using detectors optimised for telecom frequencies and using a truncated detection window, the base efficiency $p_\text{m}^0$ is given by
\begingroup
    \setlength{\abovedisplayskip}{6pt}   
    \setlength{\belowdisplayskip}{6pt}   
    \setlength{\abovedisplayshortskip}{2pt}
    \setlength{\belowdisplayshortskip}{2pt}
    \begin{equation}
    \label{eq:BaseEfficiency_IonIon_TruncatedWindow}
        p_\text{m}^0 = \frac{1}{2} \thinspace \eta_\text{penalty} \big(\eta_{\text{ion}} \thinspace\eta_{\text{FC}} \thinspace \eta_\text{det}^{\text{telecom}}\big)^2~.
    \end{equation}
\endgroup
Substituting the parameter values from Tab.~\ref{tab:baseline_and_optimistic_paramvals_trapped_ion} into~\eqref{eq:BaseEfficiency_IonIon_TruncatedWindow} yields the baseline and optimistic values of $p_\text{m}^0$ as $5.95\times 10^{-4}$ and $1.43\times 10^{-2}$, respectively.
Consequently, recall from~\eqref{eq:p_m_metro} that the entanglement-generation probability when the two nodes are separated by a distance $2d_{\text{m}'}$ (km) is given by
\begingroup
    \setlength{\abovedisplayskip}{0pt}   
    \setlength{\belowdisplayskip}{8pt}   
    \setlength{\abovedisplayshortskip}{2pt}
    \setlength{\belowdisplayshortskip}{2pt}
    \begin{equation}
        p_{\text{m}'} = p_\text{m}^0 10^{-2\alpha d_{\text{m}'}/10}~. \nonumber
    \end{equation}
\endgroup
Next, we derive the teleportation rate and expected fidelity for entanglement-ready and qubit-ready modes of teleportation in a metropolitan network.

\begin{table}[htbp]
  \centering
  \caption{Baseline and optimistic parameter values from trapped-ion experiments. The optimistic parameter values represent projected technological advancements anticipated within the next 5-10 years, based on current development trajectories and experimental progress rates, and obtained in consultation with experimental physicists at the University of Innsbruck. 
  The value for $\eta_\text{penalty} =0.12$ is obtained from \cite{Krutyanskiy2023Entanglement230Meters} by accepting 3.5 clicks per minute with a truncated detection window, instead of 0.49 clicks per second with an almost-whole detection window.
  We set both the baseline and optimistic value of $t^\text{prep}$ to be $175\,\mu\text{s}$, following the supplementary material of~\cite{Krutyanskiy2023TelecomWavelength}.
  }
  \label{tab:baseline_and_optimistic_paramvals_trapped_ion}
  \begin{tabular}{lll}
    \toprule
    \textbf{Parameter} & \textbf{Baseline value} & \textbf{Optimistic value} \\
    \midrule
    $\eta_\text{ion}$ & $0.462/0.87$~\cite{Schupp2021InterfaceBetweenTrapped-Ion} & $0.5/0.87$~\cite{privatecommunication2025}   \\
    $\eta_\text{det}^\text{ion-freq}$ & $0.87$~\cite{Schupp2021InterfaceBetweenTrapped-Ion} & $0.87$~\cite{privatecommunication2025}   \\
    $\eta_{\text{FC}}$ & $0.25$~\cite{Krutyanskiy2019Light} & $0.70$~\cite{privatecommunication2025}   \\
    $\eta_\text{penalty}$ & $0.12$~\cite{Krutyanskiy2023Entanglement230Meters} & $0.20$~\cite{privatecommunication2025}   \\
    $\eta_\text{det}^\text{telecom}$ & $0.75$~\cite{Krutyanskiy2024Entanglement101km} & $0.94$~\cite{privatecommunication2025}   \\
    $t_\text{m}^\text{prep}$ & $175\,\mu$s~\cite{privatecommunication2025} & $175\,\mu$s~\cite{privatecommunication2025}   \\
    \bottomrule
  \end{tabular}
\end{table}

\subsection{Teleportation Rate in the Metropolitan Network}
\label{sec:metro_teleportation_rate}

In our model, the rate at which a qubit can be teleported between end nodes in an MN is determined by the time required to establish end-to-end entanglement and subsequent teleportation time.
Due to~\ref{assumption:instantaneous_qubit_preparation_local_operation}, the teleportation time includes the transmission time of the Pauli correction message from the sender node to the receiver, given by ${t_{\text{m}}^\text{class}}$.
Further, let the random variable \( X_{\text{m}'} \) represent the time to establish a node-to-node entanglement successfully.
Since we assume that data qubit preparation is instantaneous (see~\ref{assumption:instantaneous_qubit_preparation_local_operation}), the total time required in both ER and QR teleportation is given by \( t_\text{m}^{\text{class}} + X_{\text{m}'}\). 
Thus, the teleportation rate in both cases is given by
\begingroup
    \setlength{\abovedisplayskip}{2pt}   
    \setlength{\belowdisplayskip}{8pt}   
    \setlength{\abovedisplayshortskip}{2pt}
    \setlength{\belowdisplayshortskip}{2pt}
    \begin{equation}
    \label{eq:teleportation_rate_metro}
        R_{\text{m}} = \frac{1}{\mathbb{E}(t_{\text{class}} \!+\! X_{\text{m}'})} = \frac{1}{t_{\text{class}} \!+\! \mathbb{E}(X_{\text{m}'})}~.    
    \end{equation}
\endgroup
Let $M'$ be the number of attempts to create an entangled pair in the MN and $M' \sim \text{Geo}(p_{\text{m}'})$~\ref{eq:attempts_unitl_success_metro}, such that $X_{\text{m}'} \!=\! t_{\text{m}'} M'$.
Hence, we have
\begingroup
    \setlength{\abovedisplayskip}{6pt}   
    \setlength{\belowdisplayskip}{8pt}   
    \setlength{\abovedisplayshortskip}{2pt}
    \setlength{\belowdisplayshortskip}{2pt}
    \begin{align}
        R_\text{m} = \frac{1}{2t_\text{m}^\text{class} \!+\! \mathbb{E}(t_{\text{m}'} M')} \overset{(\text{i})}{=} \frac{1}{2t_\text{m}^\text{class} \!+\! t_{\text{m}'}/p_{\text{m}'}} = \frac{p_{\text{m}'}}{2p_{\text{m}'} t_\text{m}^\text{class} \!+\! t_{\text{m}'}}~, \label{eq:metro_tel_rate}
    \end{align}
\endgroup
where in $(\text{i})$, we used $\mathbb{E}(M') \!=\! 1 / p_{\text{m}'}$, for geometric distribution.
Since the heralded entanglement generation process and hardware parameter values are related to trapped-ion experiments~\cite{Krutyanskiy2023Entanglement230Meters, Krutyanskiy2023TelecomWavelength, Krutyanskiy2019Light, Krutyanskiy2024Entanglement101km}, the derivation of the entanglement-generation probability $p_{\text{m}'}$ is given in~\ref{sec:derivation_metro_p_m}, and the state-of-the-art and optimistic values are presented in Tab.~\ref{tab:baseline_and_optimistic_paramvals_trapped_ion}.

\subsection{Expected Teleportation Fidelity in the Metropolitan Network}
\label{sec:metro_teleportation_fidelity}

\looseness=-1
In the ER teleportation, we teleport a pure state once we have successfully prepared an entangled link with fidelity $f_{\text{m}'}$, or corresponding Werner parameter ${w_{\text{m}'} \!=\!(4f_{\text{m}'} \!-\! 1)/3}$.
Observe that in the ER case, the link fidelity is deterministic.
Thus, using~\eqref{eq:teleportation_fidelity_general_expression}, the expected teleportation fidelity is given by
\begingroup
    \setlength{\abovedisplayskip}{2pt}   
    \setlength{\belowdisplayskip}{6pt}   
    \setlength{\abovedisplayshortskip}{2pt}
    \setlength{\belowdisplayshortskip}{2pt}
    \begin{align}
    \label{eq:metro_fidelity_ER}
        \mathbb{E}(F_\text{m}^\text{ER}) = F_\text{m}^\text{ER} = \frac{1+w_\text{m}e^{-t_\text{m}^\text{class} / t_\text{coh}}}{2}~.
    \end{align}
\endgroup

For the QR case, the data qubit is initially prepared at the sender node and stored in memory, where it undergoes decoherence until an entangled link is established in the network.
This waiting time for one teleportation round is given by $X_{\text{m}'} \!=\! t_{\text{m}'} M'$, where $M'$ is the number of entanglement generation attempts until success.
Thus, from~\eqref{eq:teleportation_fidelity_general_expression}, we obtain the expected fidelity in QR teleportation as
\begingroup
    \setlength{\abovedisplayskip}{4pt}   
    \setlength{\belowdisplayskip}{4pt}   
    \setlength{\abovedisplayshortskip}{2pt}
    \setlength{\belowdisplayshortskip}{2pt}
    \begin{align}
        \mathbb{E}(F_\text{m}^\text{QR}) &= \mathbb{E}\Big(\frac{1+w_\text{m} e^{-X_{\text{m}'}/t_\text{coh}} e^{-t_\text{m}^\text{class} / t_\text{coh}}}{2} \Big) \\
        &= \frac{1}{2} + \frac{1}{2} w_{\text{m}'} e^{-t_\text{m}^\text{class} / t_\text{coh}} \mathbb{E}(e^{-t_{\text{m}'} M'/t_\text{coh}})  \\
        &= \frac{1}{2} + \frac{1}{2} w_{\text{m}'} e^{-t_\text{m}^\text{class} / t_\text{coh}} \frac{p_{\text{m}'}}{e^{t_{\text{m}'}/t_\text{coh}} + p_{\text{m}'\textbf{}} - 1}~. \label{eq:metro_fidelity_QR}
    \end{align}
\endgroup

\section{Reformulation of~\ref{Q:2a}--\ref{Q:2c}}
\label{sec:reformulation_of_Q2-Q4}
\everypar{\looseness=-1}

Recall that in~\ref{Q:2a}--\ref{Q:2c}, we consider teleportation across nodes in an IN, e.g., between nodes $P_1$ and $P_3$ in Fig.~\ref{fig:intercity_network_diagram}.
To identify the desired parameter space for the backbone, we introduce the following shorthands for the relevant hardware parameters and the parameter space:
\begingroup
    \setlength{\abovedisplayskip}{6pt}   
    \setlength{\belowdisplayskip}{8pt}   
    \setlength{\abovedisplayshortskip}{2pt}
    \setlength{\belowdisplayshortskip}{2pt}
    \begin{align}
    \label{eq:paraemter_space_intercity_start}
        & \vec\lambda_\text{m} := (p_\text{m}^0,t_\text{coh},f_\text{m})\,, ~~\vec{\underaccent{=}{\lambda}}_\text{m} := (\underaccent{=}{p}_\text{m}^0, \underaccent{=}{t}_\text{coh}, \underaccent{=}{f}_\text{m})\,, ~~\vec{\underline{\lambda}}_\text{m} := (\underline{p}_\text{m}^0, \underline{t}_\text{coh}, \underline{f}_\text{m})\,, ~~\vec{\overline{\lambda}}_\text{m} := (\overline{p}_\text{m}^0, \overline{t}_\text{coh}, \overline{f}_\text{m})\,, \\
        & \Lambda_\text{m} := [\underaccent{=}{p}_\text{m}^0,\overline{p}_\text{m}^0] \times [\underaccent{=}{t}_\text{coh},\overline{t}_\text{coh}] \times [\underaccent{=}{f}_\text{m},\overline{f}_\text{m}]\,, ~~\Lambda_{\text{m},\text{base}} := [\underline{p}_\text{m}^0,\overline{p}_\text{m}^0] \times [\underline{t}_\text{coh},\overline{t}_\text{coh}] \times [\underline{f}_\text{m},\overline{f}_\text{m}]\,, \\
        & \vec\lambda_\text{b} := (p_\text{b},f_\text{b})\,, ~~\vec{\underaccent{=}{\lambda}}_\text{b} := (\underaccent{=}{p}_\text{b},\underaccent{=}{f}_\text{b})\,, ~~\underline{\vec{\lambda}}_\text{b} := (\underline{p}_\text{b},\underline{f}_\text{b})\,, ~~\vec{\overline{\lambda}}_\text{b} := (\overline{p}_\text{b},\overline{f}_\text{b})\,, \\
        & {\Lambda}_\text{b} := [\underaccent{=}{p}_\text{b},\overline{p}_\text{b}] \times [\underaccent{=}{f}_\text{b},\overline{f}_\text{b}]\,, \quad \quad \!\Lambda_{\text{b},\text{base}} := [\underline{p}_\text{b},\overline{p}_\text{b}] \times [\underline{f}_\text{b},\overline{f}_\text{b}]~. \label{eq:paraemter_space_intercity_end}
    \end{align}
\endgroup
Given that we have two sets of parameters for the MN ($\vec{\lambda}_\text{m}$) and the backbone ($\vec{\lambda}_\text{b}$),~\ref{Q:2a}-\ref{Q:2b} investigate the minimal hardware improvement necessary in the following sense.
We fix parameters of one segment (e.g., $\vec{\lambda}_\text{b}$) to its optimistic values and vary the other (e.g., $\vec{\lambda}_\text{m}$) along with the cut-off time to determine whether we can achieve the target performance.
If this condition is satisfied in both cases, i.e., by fixing either $\vec{\lambda}_\text{m} = \vec{\overline{\lambda}}_\text{m}$ or $\vec{\lambda}_\text{b} = \vec{\overline{\lambda}}_\text{b}$ and without exceeding the respective optimistic bounds subsequently, then this indicates that near-term experimental realisation is feasible using the considered hardware platforms.
Furthermore, this approach provides a lower bound on the minimal improvements required for each segment.
Conversely, if the target cannot be achieved under these constraints, alternative hardware platforms would be required to support intercity teleportation with the desired fidelity. 
We first reformulate~\ref{Q:2a}--\ref{Q:2b} below.
\vspace{10pt}

\subsection{Reformulation of~\ref{Q:2a}}
\label{sec:reformulation_of_Q2}

\looseness=-1 
In~\ref{Q:2a}, we investigate the achievable performance of intercity teleportation while keeping the backbone hardware parameters fixed at their optimistic values $\vec{\overline{\lambda}}_\text{b}$.
Thus, for the ER teleportation in the IN, we obtain the desired hardware parameter space, the surface for visualisation, and the set of optimal points (if the baseline does not belong to the desired parameter space), respectively, as follows
\vspace{-5pt}
\begin{align}
    & \Lambda_{2+}^\text{ER} := \{ \vec\lambda_\text{m} \in \Lambda_\text{m} \!: \max\limits_{t_\text{cut} \in T_\text{cut}} \mathbb{E}\big(F_\text{int}^\text{ER}
    (\vec\lambda_\text{m}, \vec{\overline{\lambda}}_\text{b} \mid t_\text{cut})\big) \geq f_\text{target} \}~, 
    \label{eq:desired_space_metro_int_ER} \\
    & \tilde{\Lambda}_{2+}^\text{ER} := \{(p_\text{m}^0,t_\text{coh},f_\text{m}) \in \Lambda_{2+}^\text{ER} \!: f_\text{m} = \min_{(p_\text{m}^0,t_\text{coh}, z)\in\Lambda_{2+}^\text{ER}\!, \,t_\text{cut}\in T_\text{cut}} z \}~, 
    \label{eq:min_fm_definition_metro_int_ER} \\
    & \Lambda_{2*}^\text{ER} \!:= \!\Big\{\vec\lambda_\text{m} \in \Lambda_{\text{m},\text{base}} \!:   c\big((\vec{\lambda}_\text{m},\vec{\overline{\lambda}}_\text{b}),(\vec{\underline{\lambda}}_\text{m},\vec{\overline{\lambda}}_\text{b}),F_\text{int}^\text{ER} \mid t_\text{cut}) \nonumber \\
    & \hspace{5cm} = \! \min\limits_{\vec{\lambda}_\text{m}'\in\Lambda_{\text{m},\text{base}}, \,t'_\text{cut}\in T_\text{cut})}  c\big((\vec{\lambda}_\text{m}',\vec{\overline{\lambda}}_\text{b}),(\underline{\vec{\lambda}}_\text{m}, \vec{\overline\lambda}_\text{b}),F_\text{int}^\text{ER} \mid t'_\text{cut}\big) \Big\}~.
    \label{eq:optimal_points_set_metro_int_ER}
\end{align}
For each set of points on the surface $\tilde{\Lambda}_{2+}^\text{ER}$, we define the set of cut-off times that achieves the fidelity threshold as
\begingroup
    \setlength{\abovedisplayskip}{6pt}   
    \setlength{\belowdisplayskip}{6pt}   
    \setlength{\abovedisplayshortskip}{2pt}
    \setlength{\belowdisplayshortskip}{2pt}
    \begin{align}
    \label{eq:cutoff_time_2_ER}
         T_2^\text{ER}(\vec{\lambda}_\text{m}) = \{ t_\text{cut}\in T_\text{cut}\!: \mathbb{E}\big(F_\text{int}^\text{ER}(\vec{\lambda}_\text{m},\vec{\overline{\lambda}}_\text{b}\mid t_\text{cut})\big) \geq f_\text{target} \}~.
    \end{align}
\endgroup
Since the teleportation rate is a non-decreasing function of the cut-off time, we define the rate for $\vec{\lambda}_\text{m} \in \tilde{\Lambda}_{2+}^\text{ER}$ as the best achievable rate when $t_\text{cut}\in T_2^\text{ER}(\vec{\lambda}_\text{m})$:
\begingroup
    \setlength{\abovedisplayskip}{6pt}   
    \setlength{\belowdisplayskip}{6pt}   
    \setlength{\abovedisplayshortskip}{2pt}
    \setlength{\belowdisplayshortskip}{2pt}
    \begin{align}
    \label{eq:rate_min_fm_definition_metro_int_ER}
        R_{2}^\text{ER}(\vec{\lambda}_\text{m}) = R_\text{int}\big(p_\text{m}^0,t_\text{coh},\overline{p}_\text{b}\mid \max (T_2^\text{ER}(\vec{\lambda}_\text{m}) ) \big)~,
    \end{align}
\endgroup
where $R_\text{int}$ denotes the teleportation rate in the IN as mentioned in Tab.~\ref{tab:quantities_and_parameters}.

Similarly, for QR teleportation, the required sets are given as 
\begingroup
    \setlength{\abovedisplayskip}{6pt}   
    \setlength{\belowdisplayskip}{6pt}   
    \setlength{\abovedisplayshortskip}{2pt}
    \setlength{\belowdisplayshortskip}{2pt}
    \begin{align}
        & \Lambda_{2+}^\text{QR} := \{ \vec\lambda_\text{m} \in \Lambda_\text{m} \!: \max\limits_{t'_\text{cut}\in T_\text{cut}} \mathbb{E}\big(F_\text{int}^\text{QR}
        (\vec\lambda_\text{m}, \vec{\overline{\lambda}}_\text{b} \mid t'_\text{cut})\big) \geq f_\text{target} \}~, 
        \label{eq:desired_space_metro_int_QR} \\
        & \tilde{\Lambda}_{2+}^\text{QR} := \{(p_\text{m}^0,t_\text{coh},f_\text{m}) \in \Lambda_{2+}^\text{QR} \!: f_\text{m} = \min_{(p_\text{m}^0,t_\text{coh}, z)\in\Lambda_{2+}^\text{QR}\!, \,t_\text{cut}\in T_\text{cut}} z \}~, 
        \label{eq:min_fm_definition_metro_int_QR} \\
        & \Lambda_{2*}^\text{QR} \!:= \!\Big\{\vec\lambda_\text{m} \in \Lambda_{\text{m},\text{base}} \!:   c\big((\vec{\lambda}_\text{m},\vec{\overline{\lambda}}_\text{b}),(\vec{\underline{\lambda}}_\text{m},\vec{\overline{\lambda}}_\text{b}),F_\text{int}^\text{QR} \mid t_\text{cut}) \nonumber \\
        & \hspace{5cm} = \! \min\limits_{\vec{\lambda}_\text{m}'\in\Lambda_{\text{m},\text{base}}, \,t'_\text{cut}\in T_\text{cut})}  c\big((\vec{\lambda}_\text{m}',\vec{\overline{\lambda}}_\text{b}),(\underline{\vec{\lambda}}_\text{m}, \vec{\overline\lambda}_\text{b}),F_\text{int}^\text{QR} \mid t'_\text{cut}\big) \Big\}~.
        \label{eq:optimal_points_set_metro_int_QR}
    \end{align}
\endgroup
Further, the best achievable teleportation rate for $\vec{\lambda}_\text{m} \in \tilde{\Lambda}_{2+}^\text{QR}$ is given by
\begingroup
    \setlength{\abovedisplayskip}{6pt}   
    \setlength{\belowdisplayskip}{6pt}   
    \setlength{\abovedisplayshortskip}{2pt}
    \setlength{\belowdisplayshortskip}{2pt}
    \begin{align}
        & R_{2}^\text{QR}(\vec{\lambda}_\text{m}) = R_\text{int}\big(p_\text{m}^0,t_\text{coh},\overline{p}_\text{b}\mid \max  (T_2^\text{QR}(\vec{\lambda}_\text{m}) )\big)\,, ~\text{where}
        \label{eq:rate_min_fm_definition_metro_int_QR}  \\
        &  T_2^\text{QR}(\vec{\lambda}_\text{m}) = \{ t_\text{cut}\in T_\text{cut}\!: \mathbb{E}\big(F_\text{int}^\text{QR}(\vec{\lambda}_\text{m},\vec{\overline{\lambda}}_\text{b}\mid t_\text{cut})\big) \geq f_\text{target}\}~.
        \label{eq:cutoff_time_2_QR}
    \end{align}
\endgroup

\subsection{Reformulation of~\ref{Q:2b}}
\label{sec:reformulation_of_Q3}

In contrast with~\ref{Q:2a}, we fix the metropolitan parameters at their optimistic values $\vec{\overline{\lambda}}_\text{m}$ in~\ref{Q:2b}.
In this setting, we aim to identify the desired parameter space and, if necessary, the set of optimal hardware parameters for the backbone to achieve the target fidelity of intercity teleportation.
We characterise these sets for ER teleportation, respectively, as follows
\begingroup
    \setlength{\abovedisplayskip}{6pt}   
    \setlength{\belowdisplayskip}{6pt}   
    \setlength{\abovedisplayshortskip}{2pt}
    \setlength{\belowdisplayshortskip}{2pt}
    \begin{align}
        & \Lambda_{3+}^\text{ER} := \{ \vec\lambda_\text{b} \in \Lambda_\text{b} \!: \max\limits_{t'_\text{cut}\in T_\text{cut}} \mathbb{E}\big(F_\text{int}^\text{ER}
        (\vec{\overline{\lambda}}_\text{m}, \vec{\lambda}_\text{b} \mid t'_\text{cut})\big) \geq f_\text{target} \}~, 
        \label{eq:desired_space_backbone_int_ER} \\
        & \Lambda_{3*}^\text{ER} \!:= \!\Big\{\vec\lambda_\text{b} \in \Lambda_{\text{b},\text{base}} \!:   c\big((\vec{\overline{\lambda}}_\text{m},\vec{{\lambda}}_\text{b}),(\vec{\overline{\lambda}}_\text{m},\underline{\vec{\lambda}}_\text{b}),F_\text{int}^\text{ER} \mid t_\text{cut}\big) \nonumber \\
        & \hspace{5cm} = \! \min\limits_{\vec{\lambda}_\text{b}'\in\Lambda_{\text{b},\text{base}}, \,t'_\text{cut}\in T_\text{cut}}  c\big((\vec{\overline{\lambda}}_\text{m},\vec{{\lambda'}}_\text{b}),(\vec{\overline{\lambda}}_\text{m},\underline{\vec{\lambda}}_\text{b}),F_\text{int}^\text{ER} \mid t'_\text{cut} \big) \Big\}~.
        \label{eq:optimal_points_set_backbone_int_ER}
    \end{align}
\endgroup
Since the parameter space of interest $\Lambda_\text{b}$ is already two-dimensional, we do not plot the surface showing the minimum required backbone fidelity.
Also, for $\vec\lambda_\text{b} \in \Lambda_{3*}^\text{ER}$, the maximum teleportation rate is given by
\begingroup
    \setlength{\abovedisplayskip}{6pt}   
    \setlength{\belowdisplayskip}{6pt}   
    \setlength{\abovedisplayshortskip}{2pt}
    \setlength{\belowdisplayshortskip}{2pt}
    \begin{align}
        & R_{3}^\text{ER}(\vec{\lambda}_\text{b}) := R_\text{int}\big(\overline{p}_\text{m}^0,\overline{t}_\text{coh},{p}_\text{b}\mid \max  (T_3^\text{ER}(\vec{\lambda}_\text{b}) )\big)\,, ~\text{where}
        \label{eq:rate_backbone_fixed_metro_ER} \\
        &  T_3^\text{ER}(\vec{\lambda}_\text{b}) := \{t_\text{cut} \in T_\text{cut} \!: \mathbb{E}\big(F_\text{int}^\text{ER}(\vec{\overline{\lambda}}_\text{m},\vec{\lambda}_\text{b} \mid t_\text{cut}) \big) \geq f_\text{target} \}~.
        \label{eq:cutoff_time_3_ER}
        \end{align}
\endgroup
Further, for the QR case, the corresponding sets are respectively given by
\begingroup
    \setlength{\abovedisplayskip}{6pt}   
    \setlength{\belowdisplayskip}{4pt}   
    \setlength{\abovedisplayshortskip}{2pt}
    \setlength{\belowdisplayshortskip}{2pt}
    \begin{align}
        & \Lambda_{3+}^\text{QR} := \{ \vec\lambda_\text{b} \in \Lambda_\text{b} \!: \max\limits_{t'_\text{cut}\in T_\text{cut}} \mathbb{E}\big(F_\text{int}^\text{QR}
        (\vec{\overline{\lambda}}_\text{m}, \vec{\lambda}_\text{b}' \mid t'_\text{cut})\big) \geq f_\text{target} \}~, 
        \label{eq:desired_space_backbone_int_QR} \\
        & \Lambda_{3*}^\text{QR} \!:= \!\Big\{\vec\lambda_\text{b} \in \Lambda_{\text{b},\text{base}} \!:   c\big((\vec{\overline{\lambda}}_\text{m},\vec{{\lambda}}_\text{b}),(\vec{\overline{\lambda}}_\text{m},\underline{\vec{\lambda}}_\text{b}),F_\text{int}^\text{QR} \mid t_\text{cut}\big) \nonumber \\
        & \hspace{5cm} = \! \min\limits_{\vec{\lambda}_\text{b}'\in\Lambda_{\text{b},\text{base}}, \,t'_\text{cut}\in T_\text{cut}}  c\big((\vec{\overline{\lambda}}_\text{m},\vec{{\lambda'}}_\text{b}),(\vec{\overline{\lambda}}_\text{m},\underline{\vec{\lambda}}_\text{b}),F_\text{int}^\text{QR} \mid t'_\text{cut} \big) \Big\}~. 
        \label{eq:optimal_points_set_backbone_int_QR}
    \end{align}
\endgroup
For $\vec{\lambda}_\text{b}\in \Lambda_{3+}^\text{QR}$, the maximum teleportation rate is given by
\begingroup
    \setlength{\abovedisplayskip}{6pt}   
    \setlength{\belowdisplayskip}{2pt}   
    \setlength{\abovedisplayshortskip}{2pt}
    \setlength{\belowdisplayshortskip}{2pt}
    \begin{align}
        & R_{3}^\text{QR}(\vec{\lambda}_\text{b}) := R_\text{int}\big(\overline{p}_\text{m}^0,\overline{t}_\text{coh},{p}_\text{b}\mid \max  (T_3^\text{QR}(\vec{\lambda}_\text{b}) ) \big)\,, ~\text{where} 
        \label{eq:rate_backbone_fixed_metro_QR} \\
        &  T_3^\text{QR}(\vec{\lambda}_\text{b}) := \{t_\text{cut} \in T_\text{cut} \!: \mathbb{E}\big(F_\text{int}^\text{ER}(\vec{\overline{\lambda}}_\text{m},\vec{\lambda}_\text{b} \mid t_\text{cut}) \big) \geq f_\text{target}  \}~.
        \label{eq:cutoff_time_3_QR}
    \end{align}
\endgroup

\subsection{Reformulation of~\ref{Q:2c}}
\label{sec:reformulation_of_Q4}
In~\ref{Q:2a} and~\ref{Q:2b}, we assess whether teleportation is feasible in the IN by fixing the hardware parameter values of each segment (metropolitan or backbone) at their respective optimistic values.
If feasibility is established, we proceed to~\ref{Q:2c}, which focuses on identifying the minimal hardware improvements necessary over the baseline values of both segments to achieve the target teleportation fidelity.
The desired space, optimal configurations enabling ER teleportation, and the corresponding rate in this case are, respectively, given by
\begingroup
    \setlength{\abovedisplayskip}{6pt}   
    \setlength{\belowdisplayskip}{6pt}   
    \setlength{\abovedisplayshortskip}{2pt}
    \setlength{\belowdisplayshortskip}{2pt}
    \begin{align}
    \label{eq:int_opt_point_int_ER}
        & \Lambda_{4+}^\text{ER} := \{ (\vec{\lambda}_\text{m},\vec{\lambda}_\text{b}) \in \Lambda_\text{m} \times \Lambda_\text{b} \!: \max\limits_{t'_\text{cut}\in T_\text{cut}} \mathbb{E}\big(F_\text{int}^\text{ER}
        (\vec{\lambda}_\text{m},\vec{\lambda}_\text{b}\mid t'_\text{cut})\big) \geq f_\text{target} \}~,
        \\ 
        & \Lambda_{4*}^\text{ER}
        \!:= \!\Big\{(\vec{\lambda}_\text{m},\vec{\lambda}_\text{b}) \in \Lambda_{\text{m},\text{base}}\times\Lambda_{\text{b},\text{base}} \!: c\big((\vec{\lambda}_\text{m},\vec{\lambda}_\text{b}), (\underline{\vec{\lambda}}_\text{m},\underline{\vec{\lambda}}_\text{b}),F_\text{int}^\text{ER} \mid t_\text{cut}\big) \nonumber \\
        & \hspace{3cm} = \! \min\limits_{\vec{\lambda}_\text{m}'\in \Lambda_{\text{m},\text{base}},\,\vec{\lambda}_\text{b}'\in \Lambda_{\text{b},\text{base}}, \,t'_\text{cut}\in T_\text{cut}}  c\big((\vec{\lambda}_\text{m}',\vec{\lambda}_\text{b}'), (\underline{\vec{\lambda}}_\text{m},\underline{\vec{\lambda}}_\text{b}),F_\text{int}^\text{ER} \mid t'_\text{cut}\big) \Big\}~.
        \label{eq:optimal_points_set_int_ER}
    \end{align}
\endgroup
For $(\vec{\lambda}_\text{m},\vec{\lambda}_\text{b})\in \Lambda_\text{4*}^\text{ER}$, we further define the maximum achievable teleportation rate as
\begingroup
    \setlength{\abovedisplayskip}{6pt}   
    \setlength{\belowdisplayskip}{6pt}   
    \setlength{\abovedisplayshortskip}{2pt}
    \setlength{\belowdisplayshortskip}{2pt}
    \begin{align}
        & R_{4}^\text{ER}(\vec{\lambda}_\text{m},\vec{\lambda}_\text{b}) := R_\text{int}\big({p}_\text{m}^0,{t}_\text{coh},{p}_\text{b}\mid \max  \big(T_4^\text{ER}(\vec{\lambda}_\text{m},\vec{\lambda}_\text{b}) \big)\big)\,, ~\text{where} \label{eq:rate_4_ER} \\
        &  T_4^\text{ER}(\vec{\lambda}_\text{m},\vec{\lambda}_\text{b}) := \{t_\text{cut} \in T_\text{cut} \!: \mathbb{E}\big(F_\text{int}^\text{ER}(\vec{{\lambda}}_\text{m},\vec{\lambda}_\text{b} \mid t_\text{cut}) \big) \geq f_\text{target} \}~.
        \label{eq:cutoff_time_4_ER}
    \end{align}
\endgroup
For QR teleportation, these required sets are given by
\begingroup
    \setlength{\abovedisplayskip}{6pt}   
    \setlength{\belowdisplayskip}{6pt}   
    \setlength{\abovedisplayshortskip}{2pt}
    \setlength{\belowdisplayshortskip}{2pt}
    \begin{align}
    \label{eq:int_opt_point_int_QR}
        & \Lambda_{4+}^\text{QR} := \{ (\vec{\lambda}_\text{m},\vec{\lambda}_\text{b}) \in \Lambda_\text{m} \times \Lambda_\text{b} \!: \max\limits_{t'_\text{cut}\in T_\text{cut}} \mathbb{E}\big(F_\text{int}^\text{QR}
        (\vec{\lambda}_\text{m},\vec{\lambda}_\text{b}\mid t'_\text{cut})\big) \geq f_\text{target} \}~,
        \\ 
        & \Lambda_{4*}^\text{QR}
        \!:= \!\Big\{(\vec{\lambda}_\text{m},\vec{\lambda}_\text{b}) \in \Lambda_{\text{m},\text{base}}\times\Lambda_{\text{b},\text{base}} \!: c\big((\vec{\lambda}_\text{m},\vec{\lambda}_\text{b}), ( \underline{\vec{\lambda}}_\text{m},\underline{\vec{\lambda}}_\text{b}),F_\text{int}^\text{QR} \mid t_\text{cut}\big) \nonumber \\
        & \hspace{3cm} = \! \min\limits_{\vec{\lambda}_\text{m}'\in \Lambda_{\text{m},\text{base}},\,\vec{\lambda}_\text{b}'\in \Lambda_{\text{b},\text{base}}, \,t'_\text{cut}\in T_\text{cut}}  c\big((\vec{\lambda}_\text{m}',\vec{\lambda}_\text{b}'), (\underline{\vec{\lambda}}_\text{m},\underline{\vec{\lambda}}_\text{b}),F_\text{int}^\text{QR} \mid t'_\text{cut}\big) \Big\}~,
        \label{eq:optimal_points_set_int_QR}
    \end{align}
\endgroup
and the maximum allowed rate for a representative point from the optimisation algorithm, $(\vec{\lambda}_\text{m},\vec{\lambda}_\text{b}) \in \Lambda_{4*}^\text{QR}$, is given by
\begingroup
    \setlength{\abovedisplayskip}{6pt}   
    \setlength{\belowdisplayskip}{6pt}   
    \setlength{\abovedisplayshortskip}{2pt}
    \setlength{\belowdisplayshortskip}{2pt}
    \begin{align}
        & R_{4}^\text{QR}(\vec{\lambda}_\text{m},\vec{\lambda}_\text{b}) := R_\text{int}\big({p}_\text{m}^0,{t}_\text{coh},{p}_\text{b}\mid \max  \big(T_4^\text{QR}(\vec{\lambda}_\text{m},\vec{\lambda}_\text{b}) \big)\big)\,, ~\text{where} \\
        &  T_4^\text{QR}(\vec{\lambda}_\text{m},\vec{\lambda}_\text{b}) := \{t_\text{cut} \in T_\text{cut} \!: \mathbb{E}\big(F_\text{int}^\text{ER}(\vec{{\lambda}}_\text{m},\vec{\lambda}_\text{b} \mid t_\text{cut}) \big) \geq f_\text{target} \}~.
        \label{eq:cutoff_time_4_QR}
    \end{align}
\endgroup

\newpage
\section{Calculation of Individual Terms in the Expression of Teleportation Rate in the Intercity Network}
\label{sec:cutOff_rate}
\everypar{\looseness=-1}

In this appendix, we calculate the individual terms of~\eqref{eq:ent_gen_time_int} in the expression of teleportation rate given by~\eqref{eq:def_rate_teleportation_int}.
Specifically, we calculate the terms defined in~\eqref{eq:zy1} and~\eqref{eq:inclusion-exclusion_part_2}.
\

\subsection{Calculation of Individual Terms}
\label{sec:individual_terms_for_rate}

To calculate the individual expectations in~\eqref{eq:zy1}, we use the ranges of $M_1$, $M_2$, and $m_\text{b}$ from Tab.~\ref{tab:sample_space_subset_A+}.
We also introduce the following shorthand for $x, y, s \in [0,1)$, ${z \in \mathbb{N}\cup \{\infty\}}$, ${q, r \in \mathbb{Q}}$, and ${\alpha,\sigma,\kappa \in \mathbb{R}}$:
\begingroup
    \setlength{\abovedisplayskip}{6pt}   
    \setlength{\belowdisplayskip}{6pt}   
    \setlength{\abovedisplayshortskip}{2pt}
    \setlength{\belowdisplayshortskip}{2pt}
    \begin{align}
        \Pi_c(x,y,q,\alpha,l,u) &:= \sum_{i=l}^{u} x^i y^{\ceil{iq-\alpha}}~,
        \quad\quad\quad\quad\thickspace \Pi_f(x,y,q,\alpha,l,u) := \sum_{i=l}^{u} x^i y^{\floor{iq-\alpha}}~, \label{eq:def_Pi} \\
        \Pi_{cc}(x,y,q,\alpha,\sigma,l,u) &:=\! \sum_{i=l}^{u} x^i y^{\ceil{iq-\alpha} + \ceil{iq-\sigma}}~, \thickspace \Pi_{cf}(x,y,q,\alpha,\sigma,l,u) :=\! \sum_{i=l}^{u} x^i y^{\ceil{iq-\alpha} + \floor{iq-\sigma}}  ~,\label{eq:def_Pi_cc, Pi_cf} \\
        \Pi_{ff}(x,y,q,\alpha,\sigma,l,u) &:= \sum_{i=l}^{u} x^i y^{\floor{iq-\alpha} + \floor{iq-\sigma}} ~,\label{eq:def_Pi_ff} \\
        \Theta_c(x,y,q,\alpha,l,u) &:= \sum_{i=l}^{u} i x^i y^{\ceil{iq-\alpha}}~, \quad\quad\quad\quad \Theta_f(x,y,q,\alpha,l,u) := \sum_{i=l}^{u} i x^i y^{\floor{iq-\alpha}}~. \label{eq:def_Theta} \\
        \Theta_{cc}(x,y,q,\alpha,\sigma,l,u) &:=\! \sum_{i=l}^{u} i x^i y^{\ceil{iq-\alpha} + \ceil{iq-\sigma}}, \thickspace\Theta_{cf}(x,y,q,\alpha,\sigma,l,u) :=\! \sum_{i=l}^{u} i x^i y^{\ceil{iq-\alpha} + \floor{iq-\sigma}}, \label{eq:def_Theta_cc, Theta_cf} \\
        \Theta_{ff}(x,y,q,\alpha,\sigma,l,u) &:= \sum_{i=l}^{u} \floor{iq\!-\!\alpha} x^i y^{\floor{iq-\alpha} + \floor{iq-\sigma}}~, \label{eq:def_Theta_ff} \\
        \Gamma(x,y,s,q,\alpha,l,u) &:= \sum_{i=l}^{u} x^i y^{\floor{iq}} s^{\floor{iq-\alpha}}~, \label{eq:def_Gamma} \\
        \Delta(x,y,s,r,\kappa,q,l,u) &:= \sum_{i=l}^{u} x^i \!\sum_{j=\ceil{ir-\kappa}}^{\floor{ir}} \! y^j s^{\ceil{jq}}~. \label{eq:def_Delta}
    \end{align}
\endgroup
For $u=\infty$, the infinite sums admit the following closed-form expressions:
\begin{align}
    \Pi_c(x,y,q,\alpha,l,\infty) &= \frac{\Pi_c(x,y,q,\alpha,l,l\!+\!z^*\!-\!1)}{1 - (xy^q)^{z^*}}~, \label{eq:Pi_c_finite}\\
    \Pi_f(x,y,q,\alpha,l,\infty) &= \frac{\Pi_f(x,y,q,\alpha,l,l\!+\!z^*\!-\!1)}{1 - (xy^q)^{z^*}}~, \label{eq:Pi_f_finite}\\
    \Pi_{cc}(x,y,q,\alpha,\sigma,l,\infty) &= \frac{\Pi_{cc}(x,y,q,\alpha,\sigma,l,l\!+\!z^*\!-\!1)}{1 - (xy^{2q})^{z^*}}~, \\
    \Pi_{cf}(x,y,q,\alpha,\sigma,l,\infty) &= \frac{\Pi_{cf}(x,y,q,\alpha,\sigma,l,l\!+\!z^*\!-\!1)}{1 - (xy^{2q})^{z^*}}~, \\
    \Pi_{ff}(x,y,q,\alpha,\sigma,l,\infty) &= \frac{\Pi_{ff}(x,y,q,\alpha,\sigma,l,l\!+\!z^*\!-\!1)}{1 - (xy^{2q})^{z^*}}~, \\
    \Theta_c(x,y,q,\alpha,l,\infty) &= \frac{\Theta_c(x,y,q,\alpha,l,l\!+\!z^*\!-\!1)}{1 - (xy^q)^{z^*}} + \frac{z^* (xy^q)^{z^*} \Pi_c(x,y,q,\alpha,l,l\!+\!z^*\!-\!1)}{(1 - (xy^q)^{z^*})^2}~, \label{eq:Theta_c_finite}\\
    \Theta_f(x,y,q,\alpha,l,\infty) &= \frac{\Theta_f(x,y,q,\alpha,l,l\!+\!z^*\!-\!1)}{1 - (xy^q)^{z^*}} + \frac{z^* (xy^q)^{z^*} \Pi_f(x,y,q,\alpha,l,l\!+\!z^*\!-\!1)}{(1 - (xy^q)^{z^*})^2}~, \label{eq:Theta_f_finite}
\end{align}
\begingroup
    \setlength{\abovedisplayskip}{0pt}   
    \setlength{\belowdisplayskip}{2pt}   
    \setlength{\abovedisplayshortskip}{2pt}
    \setlength{\belowdisplayshortskip}{2pt}
\begin{align}
    \Theta_{cc}(x,y,q,\alpha,\sigma,l,\infty) \!&=\! \frac{\Theta_{cc}(x,y,q,\alpha,\sigma,l,l\!+\!z^*\!-\!1)}{1 - (xy^{2q})^{z^*}} \!+\! \frac{z^* (xy^{2q})^{z^*} \Pi_{cc}(x,y,q,\alpha,\sigma,l,l\!+\!z^*\!-\!1)}{(1 - (xy^{2q})^{z^*})^2}~, \\
    \Theta_{cf}(x,y,q,\alpha,\sigma,l,\infty) \!&=\! \frac{\Theta_{cf}(x,y,q,\alpha,\sigma,l,l\!+\!z^*\!-\!1)}{1 - (xy^{2q})^{z^*}} \!+\! \frac{z^* (xy^{2q})^{z^*} \Pi_{cf}(x,z,q,\alpha,\sigma,l,l\!+\!z^*\!-\!\!1)}{(1 - (xy^{2q})^{z^*})^2}~, \\
    \Theta_{ff}(x,y,q,\alpha,\sigma,l,\infty) &= \frac{\Theta_{ff}(x,y,q,\alpha,\sigma,l,l\!+\!z^*\!-\!1)}{1 - (xy^{2q})^{z^*}} + \frac{z^* q (xy^{2q})^{z^*} \Pi_{ff}(x,y,q,\alpha,\sigma,l,l\!+\!z^*\!-\!1)}{(1 - (xy^{2q})^{z^*})^2}~, \\
    \Gamma(x,y,s,q,\alpha,l,\infty) &= \frac{\Gamma(x,y,s,q,\alpha,l,l\!+\!z^*\!-\!1)}{1\!-\!(xy^qs^q)^{z^*}}  \\
    \Delta(x,y,s,r,\kappa,q,l,\infty) &= \frac{\Delta(x,y,s,r,\kappa,q,l,l\!+\!\Bar{z}\!-\!1)}{1\!-\!(x y^r s^{rq})^{\Bar{z}}} ~,
\end{align}
\endgroup
with
\begingroup
    \setlength{\abovedisplayskip}{-2pt}   
    \setlength{\belowdisplayskip}{4pt}   
    \setlength{\abovedisplayshortskip}{2pt}
    \setlength{\belowdisplayshortskip}{2pt}
    \begin{align}
        &z^* = z^*(q) := \min\{z \in \mathbb{N} : zq \in \mathbb{N}\}~. \label{eq:Define_z*} \\
        &\Bar{z} = \Bar{z}(r,q) := \min \big\{z\in\mathbb{N} : zr, zrq \in \mathbb{N} \big\}~. \label{eq:Define_z_bar}
    \end{align}
\endgroup
For derivations, see~\ref{subsec:InfiniteToFiniteSum}.
We also use the following identity in our calculations:
\begingroup
    \setlength{\abovedisplayskip}{6pt}   
    \setlength{\belowdisplayskip}{6pt}   
    \setlength{\abovedisplayshortskip}{2pt}
    \setlength{\belowdisplayshortskip}{2pt}
    \begin{align}
    \label{eq:sum_i_x^i}
        \sum_{i=1}^{n} i x^i &= \frac{x \!-\! (n+1)x^{n+1} \!+\! n x^{n+2}}{(1-x)^2}, \quad \text{for } x \neq 1~.
    \end{align}
\endgroup
Now,
\begingroup
    \setlength{\abovedisplayskip}{6pt}   
    \setlength{\belowdisplayskip}{6pt}   
    \setlength{\abovedisplayshortskip}{2pt}
    \setlength{\belowdisplayshortskip}{2pt}
    \begin{align}
        & ~\mathbb{E}\big(Z \mathds{1}_{A_1^+}\big) \nonumber \\
        = & ~\mathbb{E}\big((X_{\text{max}} \!+\!t_{\text{msg}}) \mathds{1}_{A_1^+}\big) \\ 
        = & \!\!\!\!\sum_{m_1=\left\lceil t_\text{b} / t_\text{m}\right\rceil}^{\infty} \thinspace \sum_{\substack{m_2=\text{max}(1,\\ \left\lceil m_1-t_\text{cut}^{\prime}/t_\text{m}\right\rceil)}}^{m_1} \thinspace \sum_{\substack{m_\text{b}= \text{max}(1,\\ \left\lceil (m_1 t_\text{m} - t_\text{cut}^\prime)/t_\text{b}\right\rceil)}}^{\left\lfloor m_1 t_\text{m}/t_\text{b}\right\rfloor} \!\!\!\!\!\!\!(m_1 t_\text{m}\!+\!t_{\text{msg}}) \thinspace\mathbb{P}(X_1\!=\!m_1 t_\text{m}) \thinspace\mathbb{P}(X_2\!=\!m_2 t_\text{m}) \thinspace\mathbb{P}(X_b\!=\!m_\text{b} t_\text{b}) \\
        =& \!\!\!\sum_{m_1=\left\lceil t_\text{b} / t_\text{m}\right\rceil}^{\infty} \thinspace \sum_{\substack{m_2=\text{max}(1,\\ \left\lceil m_1-t_\text{cut}^{\prime}/t_\text{m}\right\rceil)}}^{m_1} \thinspace \sum_{\substack{m_\text{b}= \text{max}(1,\\ \left\lceil (m_1 t_\text{m} - t_\text{cut}^\prime)/t_\text{b}\right\rceil)}}^{\left\lfloor m_1 t_\text{m}/t_\text{b}\right\rfloor} \!\!\!\!\!\!\!(m_1 t_\text{m}\!+\!t_{\text{msg}}) p_\text{m}(\underbrace{1\!-\!p_\text{m}}_{=:q_\text{m}})^{m_1-1} ~p_\text{m}(\underbrace{1\!-\!p_\text{m}}_{q_\text{m}})^{m_2-1} ~p_\text{b}(\underbrace{1\!-\!p_\text{b}}_{=:q_\text{b}})^{m_\text{b}-1} \\
        = & ~ \underbrace{\Big(\frac{p_\text{m}^2 p_\text{b}}{q_\text{m}^2q_\text{b}}\Big)}_{=:\thinspace \delta}
        \thinspace\sum_{m_1=\left\lceil t_\text{b} / t_\text{m}\right\rceil}^{\infty} \!\!\!\left(m_1 t_\text{m} \!+\!t_{\text{msg}}\right) {q_\text{m}}^{m_1}
        \!\!\!\! \sum_{\substack{m_2= \text{max}(1,\\ \left\lceil m_1-t_\text{cut}^{\prime}/t_\text{m}\right\rceil)}}^{m_1} \!\!\! q_\text{m}^{m_2} 
        \!\!\!\!\!\!\sum_{\substack{m_\text{b}= \text{max}(1,\\ \left\lceil (m_1 t_\text{m} - t_\text{cut}^\prime)/t_\text{b}\right\rceil)}}^{\left\lfloor m_1 t_\text{m}/t_\text{b}\right\rfloor} \!\!\!q_\text{b}^{m_\text{b}} \\
        \overset{\text{(i)}}{=} & ~\delta\bigg( \sum_{m_1=\left\lceil t_\text{b} / t_\text{m}\right\rceil}^{1+\floor{t_\text{cut}^\prime/t_\text{m}}} \!\!\!\!(m_1 t_\text{m} \!+\!t_{\text{msg}}) q_\text{m}^{m_1} \sum_{m_2=1}^{m_1} q_\text{m}^{m_2} \sum_{m_\text{b}=1}^{\left\lfloor m_1 t_\text{m}/t_\text{b}\right\rfloor} \!\!\!q_\text{b}^{m_\text{b}} \nonumber \\
        &~ + \sum_{m_1=2+\floor{t_\text{cut}^\prime/t_\text{m}}}^{\floor{(t_\text{b}+t_\text{cut}^\prime)/t_\text{m}}} \!\!\!\!\!\!(m_1 t_\text{m} \!+\!t_{\text{msg}}) q_\text{m}^{m_1} \!\!\!\!\!\sum_{m_2=\ceil{m_1-t_\text{cut}^\prime/t_\text{m}}}^{m_1} \!\!\!\!\!\!\!\!q_\text{m}^{m_2} \sum_{m_\text{b}=1}^{\left\lfloor m_1 t_\text{m}/t_\text{b}\right\rfloor} \!\!\!\!q_\text{b}^{m_\text{b}} \nonumber \\
        &~ + \sum_{m_1=1+\floor{(t_\text{b}+t_\text{cut}^\prime)/t_\text{m}}}^\infty \!\!\!\!(m_1 t_\text{m} \!+\!t_{\text{msg}}) q_\text{m}^{m_1} \!\!\!\!\sum_{m_2=\ceil{m_1-t_\text{cut}^\prime/t_\text{m}}}^{m_1} \!\!\!\!\!\!\!\!q_\text{m}^{m_2} \sum_{m_\text{b}=\ceil{(m_1 t_\text{m}-t_\text{cut}^\prime)/t_\text{b})}}^{\left\lfloor m_1 t_\text{m}/t_\text{b}\right\rfloor} \!\!\!\!\!\!\!\!q_\text{b}^{m_\text{b}}
        \bigg) \label{eq:sum_with_pieces_without_Gamma_A_1^+}
    \end{align}
\endgroup
\begin{align}
    = &~ \frac{\delta}{p_\text{m} p_\text{b}} \bigg( \underbrace{\sum_{m_1=\left\lceil t_\text{b} / t_\text{m}\right\rceil}^{1+\floor{t_\text{cut}^\prime/t_\text{m}}} \!\!\!\!(m_1 t_\text{m} \!+\!t_{\text{msg}}) q_\text{m}^{m_1} \big(q_\text{m} \!-\! q_\text{m}^{m_1+1}\big) \big(q_\text{b} \!-\! q_\text{b}^{\floor{m_1t_\text{m}/t_\text{b}}+1} \big)}_{=:T_{1^+}^{(1)}} \nonumber \\
    &~ + \underbrace{\sum_{m_1=2+\floor{t_\text{cut}^\prime/t_\text{m}}}^{\floor{(t_\text{b}+t_\text{cut}^\prime)/t_\text{m}}} \!\!\!\!\!\!(m_1 t_\text{m} \!+\!t_{\text{msg}}) q_\text{m}^{m_1} \big(q_\text{m}^{\ceil{m_1-t_\text{cut}^\prime/t_\text{m}}} \!-\!q_\text{m}^{m_1+1} \big) \big(q_\text{b} \!-\! q_\text{b}^{\floor{m_1t_\text{m}/t_\text{b}}+1}\big)}_{=:T_{1^+}^{(2)}} \nonumber \\
    &~ + \!\!\!\!\!\! \sum_{m_1=1+\floor{(t_\text{b}+t_\text{cut}^\prime)/t_\text{m})}}^{\infty} \!(m_1 t_\text{m} \!+\!t_{\text{msg}}) q_\text{m}^{m_1} \!\Big(q_\text{m}^{\left\lceil m_1-t_\text{cut}^{\prime} / t_\text{m}\right\rceil}\!-\!q_\text{m}^{m_1+1}\Big)\!\Big(q_\text{b}^{\left\lceil\left(m_1 t_\text{m}-t_\text{cut}^{\prime}\right)/ t_\text{b}\right\rceil}\!-\!q_\text{b}^{\left\lfloor m_1 t_\text{m} / t_\text{b}\right\rfloor+1} \Big) \label{eq:infinite_sum_A_1^+} \\
    = ~& \frac{\delta}{p_\text{m} p_\text{b}}\big(T_{1^+}^{(1)} \!+\! \textbf{}T_{1^+}^{(2)} \big) + \frac{\delta \big(q_\text{m}^{-\floor{t_\text{cut}^\prime/t_\text{m}}} - q_\text{m} \big)}{p_\text{m} p_\text{b}}\bigg(t_\text{m} \!\!\!\!\!\! \sum_{m_1=1+\floor{(t_\text{b}+t_\text{cut}^\prime)/t_\text{m})}}^\infty \!\!\!\!\!\! m_1 q_\text{m}^{2m_1} q_\text{b}^{\ceil{m_1t_\text{m}/t_\text{b}-t_\text{cut}^\prime/t_\text{b}}} \nonumber \\
    & \quad - t_\text{m} q_\text{b} \!\!\!\!\!\! \sum_{m_1=1+\floor{(t_\text{b}+t_\text{cut}^\prime)/t_\text{m})}}^\infty \!\!\!\!\!\! m_1 q_\text{m}^{2m_1} q_\text{b}^{\floor{m_1t_\text{m}/t_\text{b}}} + t_{\text{msg}} \!\!\!\!\!\! \sum_{m_1=1+\floor{(t_\text{b}+t_\text{cut}^\prime)/t_\text{m})}}^\infty \!\!\!\!\!\!\!\! q_\text{m}^{2m_1} q_\text{b}^{\ceil{m_1t_\text{m}/t_\text{b}-t_\text{cut}^\prime/t_\text{b}}} \nonumber \\
    & \quad - t_{\text{msg}} q_\text{b} \sum_{m_1=1}^\infty q_\text{m}^{2m_1} q_\text{b}^{\floor{m_1t_\text{m}/t_\text{b}}}\!\bigg) \\
    \overset{\text{(ii)}}{=} & \frac{\delta}{p_\text{m} p_\text{b}}\big(T_{1^+}^{(1)} \!+\! T_{1^+}^{(2)} \big) + \frac{\delta\big(q_\text{m}^{-\floor{t_\text{cut}^\prime/t_\text{m}}} - q_\text{m} \big)}{p_\text{m} p_\text{b}}\Big(\!t_\text{m}\Theta_c\big(q_\text{m}^2,q_\text{b},\frac{t_\text{m}}{t_\text{b}},\frac{t_\text{cut}^\prime}{t_\text{b}},1\!+\!\floor{\frac{t_\text{b}\!+\!t_\text{cut}^\prime}{t_\text{m}}},\infty\big)  \nonumber \\
    & \quad \! - t_\text{m} q_\text{b} \thinspace\Theta_f\big(q_\text{m}^2,q_\text{b},\frac{t_\text{m}}{t_\text{b}},0,1\!+\!\floor{\frac{t_\text{b}\!+\!t_\text{cut}^\prime}{t_\text{m}}},\infty\big) \!+\! t_{\text{msg}}\Pi_c\big(q_\text{m}^2,q_\text{b},\frac{t_\text{m}}{t_\text{b}},\frac{t_\text{cut}^\prime}{t_\text{b}},1\!+\!\floor{\frac{t_\text{b}\!+\!t_\text{cut}^\prime}{t_\text{m}}},\infty \big) \nonumber \\
    & \quad - t_{\text{msg}}q_\text{b} \thinspace\Pi_f\big(q_\text{m}^2,q_\text{b},\frac{t_\text{m}}{t_\text{b}},0,1\!+\!\floor{\frac{t_\text{b}\!+\!t_\text{cut}^\prime}{t_\text{m}}},\infty\big)\!\Big) \label{eq:PiThetaWithZ=infty_A_1^+} \\
    \overset{\text{(iii)}}{=} & \frac{\delta}{p_\text{m} p_\text{b}}\big(T_{1^+}^{(1)} \!+\! T_{1^+}^{(2)} \big) \nonumber \\
    & \quad + \frac{\delta \big(q_\text{m}^{-\floor{t_\text{cut}^\prime/t_\text{m}}} \!-\! q_\text{m} \big)}{p_\text{m} p_\text{b}}\bigg(t_\text{m} \Big(\frac{1}{1\!-\!(q_\text{m}^2q_\text{b}^{t_\text{m}/t_\text{b}})^{m^*}}{\Theta_c\big(q_\text{m}^2,q_\text{b},\frac{t_\text{m}}{t_\text{b}},\frac{t_\text{cut}^\prime}{t_\text{b}}, 1\!+\!\floor{\frac{t_\text{b}+t_\text{cut}^\prime}{t_\text{m}}}\!,\floor{\frac{t_\text{b}+t_\text{cut}^\prime}{t_\text{m}}}\!+\!m^*\big)} \nonumber \\
    & \quad + \frac{m^* (q_\text{m}^2q_\text{b}^{t_\text{m}/t_\text{b}})^{m^*}}{\big(1\!-\!(q_\text{m}^2q_\text{b}^{t_\text{m}/t_\text{b}})^{m^*}\big)^2} { \Pi_c(q_\text{m}^2,q_\text{b},\frac{t_\text{m}}{t_\text{b}},\frac{t_\text{cut}^\prime}{t_\text{b}},1\!+\!\floor{\frac{t_\text{b}+t_\text{cut}^\prime}{t_\text{m}}}\!,\floor{\frac{t_\text{b}+t_\text{cut}^\prime}{t_\text{m}}}\!+\!m^*)} \Big) \nonumber \\
    & \quad - t_\text{m} q_\text{b} \Big(\frac{1}{1\!-\!(q_\text{m}^2q_\text{b}^{t_\text{m}/t_\text{b}})^{m^*}} {\Theta_f(q_\text{m}^2,q_\text{b},\frac{t_\text{m}}{t_\text{b}},0,1\!+\!\floor{\frac{t_\text{b}+t_\text{cut}^\prime}{t_\text{m}}}\!,\floor{\frac{t_\text{b}+t_\text{cut}^\prime}{t_\text{m}}}\!+\!m^*)} \nonumber
     \\
     & \quad + \frac{m^* (q_\text{m}^2q_\text{b}^{t_\text{m}/t_\text{b}})^{m^*}}{\big(1\!-\!(q_\text{m}^2q_\text{b}^{t_\text{m}/t_\text{b}})^{m^*}\big)^2} { \Pi_f(q_\text{m}^2,q_\text{b},\frac{t_\text{m}}{t_\text{b}},0,1\!+\!\floor{\frac{t_\text{b}+t_\text{cut}^\prime}{t_\text{m}}}\!,\floor{\frac{t_\text{b}+t_\text{cut}^\prime}{t_\text{m}}}\!+\!m^*)} \Big) \nonumber \\
    & \quad + \frac{t_{\text{msg}}}{1\!-\!(q_\text{m}^2q_\text{b}^{t_\text{m}/t_\text{b}})^{m^*}} \Pi_c\big(q_\text{m}^2, q_\text{b},\frac{t_\text{m}}{t_\text{b}},\frac{t_\text{cut}^\prime}{t_\text{b}},1\!+\!\floor{\frac{t_\text{b}+t_\text{cut}^\prime}{t_\text{m}}}\!,\floor{\frac{t_\text{b}+t_\text{cut}^\prime}{t_\text{m}}}\!+\!m^*\big) \nonumber \\
    & \quad - \frac{t_{\text{msg}} q_\text{b}}{1\!-\!(q_\text{m}^2q_\text{b}^{t_\text{m}/t_\text{b}})^{m^*}} {\Pi_f\big(q_\text{m}^2, q_\text{b},\frac{t_\text{m}}{t_\text{b}},0,1\!+\!\floor{\frac{t_\text{b}+t_\text{cut}^\prime}{t_\text{m}}}\!,\floor{\frac{t_\text{b}+t_\text{cut}^\prime}{t_\text{m}}}\!+\!m^*\big)} \bigg) ~. \label{eq:finite_sum_A_1^+}
\end{align}

Note that in (i), we have split the range of $m_1$ into three sets: 
$\ceil{t_\text{b}/t_\text{m}} \!\leq\! m_1 \!\leq\! 1 \!+\! \floor{t_\text{cut}^\prime/t_\text{m}}$, 
$2 \!+\! \floor{t_\text{cut}^\prime/t_\text{m}} \!\leq\! m_1 \!\leq\! \floor{(t_\text{b} \!+\! t_\text{cut}^\prime)/t_\text{m}}$, 
and $m_1 \!\geq\! 1 \!+\! \floor{(t_\text{b} \!+\! t_\text{cut}^\prime)/t_\text{m}}$, which simplifies the corresponding ranges of $m_2$ and $m_\text{b}$.
In (ii), we have used the definitions of $\Pi_c$, $\Pi_f$, $\Theta_c$, and $\Theta_f$ from~\eqref{eq:def_Pi} and~\eqref{eq:def_Theta}, and $z^*({t_\text{m}}/{t_\text{b}}) \!=\! t_\text{b}/\text{gcd}(t_\text{m},t_\text{b}) \!=\! m^*$.
Lastly, in (iii), we use the closed-form expressions for the infinite sums over $m_1$, which exploits the geometric nature of these sums.
The derivation is provided in~\ref{subsec:InfiniteToFiniteSum}.

Here onwards, by using the definitions of $\Pi_x$, and $\Theta_x$, where $x \in \{ c,f,cc,cf,ff\}$ from~\eqref{eq:def_Pi}, ~\eqref{eq:def_Pi_cc, Pi_cf}, ~\eqref{eq:def_Pi_ff}, ~\eqref{eq:def_Theta}, ~\eqref{eq:def_Theta_cc, Theta_cf} and ~\eqref{eq:def_Theta_ff}, we evaluate the following terms.
On $A_b^+$, we have 
\begingroup
    \setlength{\abovedisplayskip}{8pt}   
    \setlength{\belowdisplayskip}{6pt}   
    \setlength{\abovedisplayshortskip}{2pt}
    \setlength{\belowdisplayshortskip}{2pt}
    \begin{align}
        & ~\mathbb{E}\big(Z \mathds{1}_{A_b^+}\big) \nonumber \\
        = & ~\mathbb{E}\big((X_{\text{max}} \!+\!t_{\text{msg}}) \mathds{1}_{A_b^+}\big)  \\ 
        = & ~ \underbrace{\bigg(\frac{p_\text{m}^2 p_\text{b}}{(1\!-\!p_\text{m})^2(1\!-\!p_\text{b})}\bigg)}_{\delta}
        \thinspace\sum_{m_\text{b}=1}^{\infty} (m_\text{b}t_\text{b} \!+\!t_{\text{msg}}) {(\underbrace{1\!-\!p_\text{b}}_{q_\text{b}})}^{m_\text{b}}
        \! \Big( \!\!\sum_{\substack{m_1= \text{max}(1,\\ \left\lceil (m_\text{b}t_\text{b}-t_\text{cut}^{\prime})/t_\text{m}\right\rceil)}}^{\floor{m_\text{b}t_\text{b}/t_\text{m}}} \!\!\!{(\underbrace{1\!-\!p_\text{m}}_{q_\text{m}})}^{m_1} \Big)^2 \\
        = & ~\delta \sum_{m_\text{b}=1}^{\infty} \left(m_\text{b} t_\text{b} \!+\!t_{\text{msg}}\right) q_\text{b}^{m_\text{b}} \bigg(\sum_{\substack{m_1=\text{max}(1, \\ \left\lceil(m_\text{b} t_\text{b}-t_\text{cut}^{\prime}) / t_\text{m}\right\rceil)}}^{\left\lfloor m_\text{b} t_\text{b} / t_\text{m}\right\rfloor}  \!\!\!\!\!\! q_\text{m}^{m_1}\bigg)^2 \\
        = & ~ \underbrace{\frac{\delta}{p_\text{m}^2}\sum_{m_\text{b}=1}^{\floor{(t_\text{m}+t_\text{cut}^\prime)/t_\text{b}}} \!\!\!\!\!\!(m_\text{b} t_\text{b} + t_{\text{msg}}) q_\text{b}^{m_\text{b}} \left( q_\text{m} \!-\! q_\text{m}^{\floor{m_\text{b} t_\text{b}/t_\text{m}}+1}\right)^{\!2}}_{=:\thinspace T_{b^+} } \nonumber \\
        & \quad + \frac{\delta}{p_\text{m}^2} \sum_{\substack{m_\text{b}=1+ \floor{(t_\text{m}+t_\text{cut}^\prime)/t_\text{b}}}}^\infty \!\!\!\!\!\!\!\! (m_\text{b}t_\text{b}\!+\!t_{\text{msg}}) q_\text{b}^{m_\text{b}} \Big(q_\text{m}^{\ceil{(m_\text{b}t_\text{b}-t_\text{cut}^\prime)/t_\text{m}}} \!-\! q_\text{m}^{\floor{m_\text{b}t_\text{b}/t_\text{m}}+1} \Big)^{\!2} \\
        = & ~\! T_{b^+} + \frac{\delta t_\text{b}}{p_\text{m}^2} \Big(\sum_{\substack{m_\text{b}=1+ \floor{(t_\text{m}+t_\text{cut}^\prime)/t_\text{b}}}}^\infty \!\!\!\!\!\!\!\! m_\text{b} q_\text{b}^{m_\text{b}} q_\text{m}^{2\ceil{(m_\text{b}t_\text{b}-t_\text{cut}^\prime)/t_\text{m}}} + q_\text{m}^2 \!\!\!\!\!\!\!\! \sum_{\substack{m_\text{b}=1+ \floor{(t_\text{m}+t_\text{cut}^\prime)/t_\text{b}}}}^\infty \!\!\!\!\!\!\!\! m_\text{b} q_\text{b}^{m_\text{b}} q_\text{m}^{2\floor{m_\text{b}t_\text{b}/t_\text{m}}}  \nonumber \\
        & \quad - 2q_\text{m} \!\!\!\!\!\!\!\! \sum_{\substack{m_\text{b}=1+ \floor{(t_\text{m}+t_\text{cut}^\prime)/t_\text{b}}}}^\infty m_\text{b} q_\text{b}^{m_\text{b}} q_\text{m}^{\ceil{(m_\text{b}t_\text{b}-t_\text{cut}^\prime)/t_\text{m}} + \floor{m_\text{b}t_\text{b}/t_\text{b}}} \Big)  \nonumber  \\
        & \quad + \frac{\delta t_{\text{msg}}}{p_\text{m}^2} \Big( \sum_{\substack{m_\text{b}=1+ \floor{(t_\text{m}+t_\text{cut}^\prime)/t_\text{b}}}}^\infty \!\!\!\!\!\!\!\! q_\text{b}^{m_\text{b}} q_\text{m}^{2\ceil{(m_\text{b}t_\text{b}-t_\text{cut}^\prime)/t_\text{m}}} + q_\text{m}^2 \!\!\!\!\!\!\!\! \sum_{\substack{m_\text{b}=1+ \floor{(t_\text{m}+t_\text{cut}^\prime)/t_\text{b}}}}^\infty \!\!\!\!\!\!\!\! q_\text{b}^{m_\text{b}} q_\text{m}^{2\floor{m_\text{b}t_\text{b}/t_\text{m}}}  \nonumber \\
        & \quad - 2q_\text{m} \!\!\!\!\!\!\!\! \sum_{\substack{m_\text{b}=1+ \floor{(t_\text{m}+t_\text{cut}^\prime)/t_\text{b}}}}^\infty q_\text{b}^{m_\text{b}} q_\text{m}^{\ceil{(m_\text{b}t_\text{b}-t_\text{cut}^\prime)/t_\text{m}} + \floor{m_\text{b}t_\text{b}/t_\text{b}}}    \Big) \\
        = & ~\! T_{b^+} + \frac{\delta t_\text{b}}{p_\text{m}^2} \bigg(\Theta_c\big(q_\text{b},q_\text{m}^2,\frac{t_\text{b}}{t_\text{m}},\frac{t_\text{cut}^\prime}{t_\text{m}}, 1\!+\!\floor{\frac{t_\text{m}\!+\!t_\text{cut}^\prime}{t_\text{b}}}\!,\infty \big) + q_\text{m}^2\thinspace\Theta_f\big(q_\text{b},q_\text{m}^2,\frac{t_\text{b}}{t_\text{m}},0, 1\!+\!\floor{\frac{t_\text{m}\!+\!t_\text{cut}^\prime}{t_\text{b}}}\!,\infty\big) \nonumber \\
        & \quad - 2 q_\text{m}\thinspace \Theta_{cf}\big(q_\text{b},q_\text{m},\frac{t_\text{b}}{t_\text{m}},\frac{t_\text{cut}^\prime}{t_\text{m}},0, 1\!+\!\floor{\frac{t_\text{m}\!+\!t_\text{cut}^\prime}{t_\text{b}}}\!,\infty \big)\! \bigg) \nonumber \\
        & \quad + \frac{\delta t_{\text{msg}}}{p_\text{m}^2} \bigg(\! \Pi_c\big(q_\text{b},q_\text{m}^2,\frac{t_\text{b}}{t_\text{m}},\frac{t_\text{cut}^\prime}{t_\text{m}}, 1\!+\!\floor{\frac{t_\text{m}\!+\!t_\text{cut}^\prime}{t_\text{b}}}\!,\infty\big) + q_\text{m}^2\thinspace\Pi_f\big(q_\text{b},q_\text{m}^2,\frac{t_\text{b}}{t_\text{m}},0, 1\!+\!\floor{\frac{t_\text{m}\!+\!t_\text{cut}^\prime}{t_\text{b}}}\!,\infty\big) \nonumber \\
        & \quad  - 2 q_\text{m}\thinspace \Pi_{cf}\big(q_\text{b},q_\text{m},\frac{t_\text{b}}{t_\text{m}},\frac{t_\text{cut}^\prime}{t_\text{m}},0, 1\!+\!\floor{\frac{t_\text{m}\!+\!t_\text{cut}^\prime}{t_\text{b}}}\!,\infty\big)\! \bigg)
    \end{align}
\endgroup
\begingroup
    \setlength{\abovedisplayskip}{-4pt}   
    \setlength{\belowdisplayskip}{6pt}   
    \setlength{\abovedisplayshortskip}{2pt}
    \setlength{\belowdisplayshortskip}{2pt}
    \begin{align}
        \overset{(\text{i})}{=} & ~\! T_{b^+} + \frac{\delta t_\text{b}}{p_\text{m}^2} \bigg(\frac{1}{1\!-\!\big(q_\text{b}q_\text{m}^{2t_\text{b}/t_\text{m}}\big)^{m^*t_\text{m}/t_\text{b}}}\Theta_c\big(q_\text{b},q_\text{m}^2,\frac{t_\text{b}}{t_\text{m}},\frac{t_\text{cut}^\prime}{t_\text{m}}, 1\!+\!\floor{\frac{t_\text{m}\!+\!t_\text{cut}^\prime}{t_\text{b}}}\!, \floor{\frac{t_\text{m}\!+\!t_\text{cut}^\prime}{t_\text{b}}} \!+\! \frac{m^*t_\text{m}}{t_\text{b}}\big) \nonumber \\
        & \quad + \frac{m^*t_\text{m}}{t_\text{b}} \frac{\big(q_\text{b}q_\text{m}^{2t_\text{b}/t_\text{m}}\big)^{m^*t_\text{m}/t_\text{b}}}{\big(1\!-\!\big(q_\text{b}q_\text{m}^{2t_\text{b}/t_\text{m}}\big)^{m^*t_\text{m}/t_\text{b}}\big)^2} \Pi_c\big(q_\text{b},q_\text{m}^2,\frac{t_\text{b}}{t_\text{m}},\frac{t_\text{cut}^\prime}{t_\text{m}}, 1\!+\!\floor{\frac{t_\text{m}\!+\!t_\text{cut}^\prime}{t_\text{b}}}\!, \floor{\frac{t_\text{m}\!+\!t_\text{cut}^\prime}{t_\text{b}}} \!+\! \frac{m^*t_\text{m}}{t_\text{b}}\big)   \nonumber \\
        & \quad + q_\text{m}^2\Big(\frac{1}{1\!-\!\big(q_\text{b}q_\text{m}^{2t_\text{b}/t_\text{m}}\big)^{m^*t_\text{m}/t_\text{b}}} \Theta_f\big(q_\text{b},q_\text{m}^2,\frac{t_\text{b}}{t_\text{m}},0, 1\!+\!\floor{\frac{t_\text{m}\!+\!t_\text{cut}^\prime}{t_\text{b}}}\!,\floor{\frac{t_\text{m}\!+\!t_\text{cut}^\prime}{t_\text{b}}} \!+\!\frac{m^*t_\text{m}}{t_\text{b}}\big) \nonumber \\
        & \quad + \frac{m^*t_\text{m}}{t_\text{b}}\frac{q_\text{b}q_\text{m}^{2t_\text{b}/t_\text{m}}\big)^{m^*t_\text{m}/t_\text{b}}}{\big(1\!-\!\big(q_\text{b}q_\text{m}^{2t_\text{b}/t_\text{m}}\big)^{m^*t_\text{m}/t_\text{b}}\big)^2} \Pi_f\big(q_\text{b},q_\text{m}^2,\frac{t_\text{b}}{t_\text{m}},0, 1\!+\!\floor{\frac{t_\text{m}\!+\!t_\text{cut}^\prime}{t_\text{b}}}\!,\floor{\frac{t_\text{m}\!+\!t_\text{cut}^\prime}{t_\text{b}}} \!+\!\frac{m^*t_\text{m}}{t_\text{b}}\big)    \Big) \nonumber \\
        & \quad - 2 q_\text{m}\Big( \frac{1}{1\!-\!\big(q_\text{b}q_\text{m}^{2t_\text{b}/t_\text{m}}\big)^{m^*t_\text{m}/t_\text{b}}} \Theta_{cf}\big(q_\text{b},q_\text{m},\frac{t_\text{b}}{t_\text{m}},\frac{t_\text{cut}^\prime}{t_\text{m}},0, 1\!+\!\floor{\frac{t_\text{m}\!+\!t_\text{cut}^\prime}{t_\text{b}}}\!, \floor{\frac{t_\text{m}\!+\!t_\text{cut}^\prime}{t_\text{b}}} \!+\! \frac{m^*t_\text{m}}{t_\text{b}} \big) \nonumber \\
        & \quad + \frac{m^*t_\text{m}}{t_\text{b}} \frac{\big(q_\text{b}q_\text{m}^{2t_\text{b}/t_\text{m}}\big)^{m^*t_\text{m}/t_\text{b}}}{\big(1\!-\!\big(q_\text{b}q_\text{m}^{2t_\text{b}/t_\text{m}}\big)^{m^*t_\text{m}/t_\text{b}}\big)^2} \Pi_{cf}\big(q_\text{b},q_\text{m},\frac{t_\text{b}}{t_\text{m}},\frac{t_\text{cut}^\prime}{t_\text{m}},0, 1\!+\!\floor{\frac{t_\text{m}\!+\!t_\text{cut}^\prime}{t_\text{b}}}\!,\floor{\frac{t_\text{m}\!+\!t_\text{cut}^\prime}{t_\text{b}}} \!+\! \frac{m^*t_\text{m}}{t_\text{b}} \big)  \Big)  \bigg)  \nonumber \\
        & \quad + \frac{\delta t_{\text{msg}}}{p_\text{m}^2} \bigg(\frac{1}{1\!-\!\big(q_\text{b}q_\text{m}^{2t_\text{b}/t_\text{m}}\big)^{m^*t_\text{m}/t_\text{b}}} \Pi_c\big(q_\text{b},q_\text{m}^2,\frac{t_\text{b}}{t_\text{m}},\frac{t_\text{cut}^\prime}{t_\text{m}}, 1\!+\!\floor{\frac{t_\text{m}\!+\!t_\text{cut}^\prime}{t_\text{b}}}\!,\floor{\frac{t_\text{m}\!+\!t_\text{cut}^\prime}{t_\text{b}}} \!+\! \frac{m^*t_\text{m}}{t_\text{b}} \big) \nonumber \\
        & \quad + q_\text{m}^2 \frac{1}{1\!-\!\big(q_\text{b}q_\text{m}^{2t_\text{b}/t_\text{m}}\big)^{m^*t_\text{m}/t_\text{b}}} \Pi_f\big(q_\text{b},q_\text{m}^2,\frac{t_\text{b}}{t_\text{m}},0,
        \! 1\!+\!\floor{\frac{t_\text{m}\!+\!t_\text{cut}^\prime}{t_\text{b}}}\!,\floor{\frac{t_\text{m}\!+\!t_\text{cut}^\prime}{t_\text{b}}} \!+\! \frac{m^*t_\text{m}}{t_\text{b}} \big) \nonumber \\
        & \quad \!\! - 2 q_\text{m} \frac{1}{1\!-\!\big(q_\text{b}q_\text{m}^{2t_\text{b}/t_\text{m}}\big)^{m^*t_\text{m}/t_\text{b}}} \Pi_{cf}\big(q_\text{b},q_\text{m},\frac{t_\text{b}}{t_\text{m}},\frac{t_\text{cut}^\prime}{t_\text{m}},0, 1\!+\!\floor{\frac{t_\text{m}\!+\!t_\text{cut}^\prime}{t_\text{b}}}\!\!, \!\floor{\frac{t_\text{m}\!+\!t_\text{cut}^\prime}{t_\text{b}}} \!+\! \frac{m^*t_\text{m}}{t_\text{b}}\big)\! \bigg).
        \label{eq:finite_sum_A_b^+}
    \end{align}
\endgroup
Note that in (i), we used $z^*(t_\text{b}/t_\text{m}) \!=\! t_\text{m}/\text{gcd}(t_\text{m},t_\text{b}) \!=\! m^*t_\text{m}/t_\text{b}$.

On $A_1^+ A_2^+$, we observe that
\begingroup
    \setlength{\abovedisplayskip}{6pt}   
    \setlength{\belowdisplayskip}{6pt}   
    \setlength{\abovedisplayshortskip}{2pt}
    \setlength{\belowdisplayshortskip}{2pt}
    \begin{align}
        & ~\mathbb{E}\big(Z \mathds{1}_{A_1^+ A_2^+}\big) \nonumber \\
        = & ~ \underbrace{\bigg(\frac{p_\text{m}^2 p_\text{b}}{(1\!-\!p_\text{m})^2(1\!-\!p_\text{b})}\bigg)}_{\delta}
        \thinspace\sum_{m_1=\ceil{t_\text{b}/t_\text{m}}}^{\infty} \!\!\!(m_1t_\text{m} \!+\!t_{\text{msg}}) {(\underbrace{1\!-\!p_\text{m}}_{q_\text{m}})}^{m_1} {(\underbrace{1\!-\!p_\text{m}}_{q_\text{m}})}^{m_1}
        \!\!\!\!\!\!\sum_{\substack{m_\text{b}=\text{max}(1,\\ \left\lceil(m_1 t_\text{m}-t_\text{cut}^{\prime})/t_\text{b}\right\rceil)}}^{\floor{m_1t_\text{m}/t_\text{b}}} \!\!\!{(\underbrace{1\!-\!p_\text{b}}_{q_\text{b}})}^{m_\text{b}} \\
        = & ~\delta \!\!\!\sum_{m_1=\left\lceil t_\text{b}/t_\text{m}\right\rceil}^{\infty}\!\!\! (m_1 t_\text{m} \!+\!t_{\text{msg}}) q_\text{m}^{2 m_1} \!\!\!\sum_{\substack{m_\text{b}=\text{max}(1,\\ \left\lceil(m_1 t_\text{m}-t_\text{cut}^{\prime})/t_\text{b}\right\rceil)}}^{\left\lfloor m_1 t_\text{m} / t_\text{b}\right\rfloor} \!\!\! q_\text{b}^{m_\text{b}} \\
         = &~ \underbrace{\frac{\delta}{p_\text{b}}\sum_{m_1=\ceil{t_\text{b}/t_\text{m}}}^{\floor{(t_\text{b}+t_\text{cut}^\prime)/t_\text{m}}} \!\!\!\!\!\!(m_1 t_\text{m} \!+\! t_{\text{msg}}) q_\text{m}^{2 m_1}\big(q_\text{b} \!-\! q_\text{b}^{\left\lfloor m_1 t_\text{m} / t_\text{b}\right\rfloor+1}\big)}_{=:\thinspace T_{1^+ 2^+}} \nonumber \\
         & \quad + \frac{\delta}{p_\text{b}}\sum_{\substack{m_1=1+ \floor{(t_\text{b}+t_\text{cut}^\prime)/t_\text{m}}}}^\infty \!\!\!\!\!\!(m_1 t_\text{m} \!+\!t_{\text{msg}}) q_\text{m}^{2m_1} \big( q_\text{b}^{\ceil{(m_1 t_\text{m}-t_\text{cut}^\prime)/t_\text{b}}} \!-\! q_\text{b}^{\floor{m_1t_\text{m}/t_\text{b}}+1}  \big) \\
         = &~ T_{1^+ 2^+} + \frac{\delta t_\text{m}}{p_\text{b}}\bigg(\! \Theta_c\big(q_\text{m}^2,q_\text{b},\frac{t_\text{m}}{t_\text{b}},\frac{t_\text{cut}^\prime}{t_\text{b}},1\!+\!\floor{\frac{t_\text{b}\!+\!t_\text{cut}^\prime}{t_\text{m}}},\infty\big) - q_\text{b}\thinspace\Theta_f\big(q_\text{m}^2,q_\text{b},\frac{t_\text{m}}{t_\text{b}},0,1\!+\!\floor{\frac{t_\text{b}\!+\!t_\text{cut}^\prime}{t_\text{m}}},\infty \big)\!\bigg) \nonumber \\
        & \quad + \frac{\delta t_{\text{msg}}}{p_\text{b}} \bigg(\! \Pi_c\big(q_\text{m}^2,q_\text{b},\frac{t_\text{m}}{t_\text{b}},\frac{t_\text{cut}^\prime}{t_\text{b}},1\!+\!\floor{\frac{t_\text{b}\!+\!t_\text{cut}^\prime}{t_\text{m}}},\infty \big) - q_\text{b}\thinspace\Pi_f\big(q_\text{m}^2,q_\text{b},\frac{t_\text{m}}{t_\text{b}},0,1\!+\!\floor{\frac{t_\text{b}\!+\!t_\text{cut}^\prime}{t_\text{m}}},\infty \big)\!\bigg)
    \end{align}
\endgroup

\begin{align}
    \overset{(\text{i})}{=} &~ T_{1^+ 2^+} + \frac{\delta t_\text{m}}{p_\text{b}}\bigg(\frac{1}{1\!-\!\big(q_\text{m}^{2}q_\text{b}^{t_\text{m}/t_\text{b}}\big)^{m^*}} \Theta_c\big(q_\text{m}^2,q_\text{b},\frac{t_\text{m}}{t_\text{b}},\frac{t_\text{cut}^\prime}{t_\text{b}},1\!+\!\floor{\frac{t_\text{b}\!+\!t_\text{cut}^\prime}{t_\text{m}}}\!,\floor{\frac{t_\text{b}\!+\!t_\text{cut}^\prime}{t_\text{m}}} \!+\! m^*\big) \nonumber \\
    & \quad + \frac{m^*}{\big(1\!-\!\big(q_\text{m}^{2}q_\text{b}^{t_\text{m}/t_\text{b}}\big)^{m^*} \big)^2} \Pi_c\big(q_\text{m}^2,q_\text{b},\frac{t_\text{m}}{t_\text{b}},\frac{t_\text{cut}^\prime}{t_\text{b}},1\!+\!\floor{\frac{t_\text{b}\!+\!t_\text{cut}^\prime}{t_\text{m}}}\!,\floor{\frac{t_\text{b}\!+\!t_\text{cut}^\prime}{t_\text{m}}} \!+\! m^*\big) \nonumber \\
    & \quad - q_\text{b}\Big( \frac{1}{1\!-\!\big(q_\text{m}^2q_\text{b}^{t_\text{m}/t_\text{b}}\big)^{m^*}} \Theta_f\big(q_\text{m}^2,q_\text{b},\frac{t_\text{m}}{t_\text{b}},0,1\!+\!\floor{\frac{t_\text{b}\!+\!t_\text{cut}^\prime}{t_\text{m}}}\!,\floor{\frac{t_\text{b}\!+\!t_\text{cut}^\prime}{t_\text{m}}} \!+\! m^* \big) \nonumber \\
    & \quad + \frac{m^*}{\big(1\!-\!\big(q_\text{m}^{2}q_\text{b}^{t_\text{m}/t_\text{b}}\big)^{m^*} \big)^2} \Pi_f\big(q_\text{m}^2,q_\text{b},\frac{t_\text{m}}{t_\text{b}},0,1\!+\!\floor{\frac{t_\text{b}\!+\!t_\text{cut}^\prime}{t_\text{m}}}\!,\floor{\frac{t_\text{b}\!+\!t_\text{cut}^\prime}{t_\text{m}}} \!+\! m^* \big) \Big) \bigg) \nonumber \\
    & \quad + \frac{\delta t_{\text{msg}}}{p_\text{b}} \bigg(\frac{1}{1\!-\!\big(q_\text{m}^2q_\text{b}^{t_\text{m}/t_\text{b}}\big)^{m^*}}\Pi_c\big(q_\text{m}^2,q_\text{b},\frac{t_\text{m}}{t_\text{b}},\frac{t_\text{cut}^\prime}{t_\text{b}},1\!+\!\floor{\frac{t_\text{b}\!+\!t_\text{cut}^\prime}{t_\text{m}}}\!, \floor{\frac{t_\text{b}\!+\!t_\text{cut}^\prime}{t_\text{m}}}\!+\! m^* \big) \nonumber \\
    &\quad - q_\text{b} \frac{1}{1\!-\!\big(q_\text{m}^2q_\text{b}^{t_\text{m}/t_\text{b}}\big)^{m^*}} \Pi_f\big(q_\text{m}^2,q_\text{b},\frac{t_\text{m}}{t_\text{b}},0,1\!+\!\floor{\frac{t_\text{b}\!+\!t_\text{cut}^\prime}{t_\text{m}}}\!,\floor{\frac{t_\text{b}\!+\!t_\text{cut}^\prime}{t_\text{m}}} \!+\! m^* \big)\!\bigg) ~,\qquad \qquad \quad
    \label{eq:finite_sum_A_1^+A_2^+}
\end{align}
where, in (i), we have used $z^*(t_\text{m}/t_\text{b}) \!=\! m^*$.

On $A_1^+ A_b^+$, it follows that
\begingroup
    \setlength{\abovedisplayskip}{6pt}   
    \setlength{\belowdisplayskip}{6pt}   
    \setlength{\abovedisplayshortskip}{2pt}
    \setlength{\belowdisplayshortskip}{2pt}
    \begin{align}
        & ~\mathbb{E}\big(Z \mathds{1}_{A_1^+ A_b^+}\big) \nonumber \\
        = & ~ \underbrace{\bigg(\frac{p_\text{m}^2 p_\text{b}}{(1\!-\!p_\text{m})^2(1\!-\!p_\text{b})}\bigg)}_{\delta}
        \thinspace\sum_{k=1}^{\infty} (km^*t_\text{m} \!+\!t_{\text{msg}}) {(\underbrace{1\!-\!p_\text{m}}_{q_\text{m}})}^{km^*} {(\underbrace{1\!-\!p_\text{b}}_{q_\text{b}})}^{k m^*t_\text{m}/t_\text{b}}
        \!\!\!\!\!\!\sum_{\substack{m_2=\text{max}(1,\\ \ceil{km^*-t_\text{cut}^{\prime}/t_\text{m}})}}^{km^*} \!\!\!\!{(\underbrace{1\!-\!p_\text{m}}_{q_\text{m}})}^{m_2} \\
        = & ~\delta ~\sum_{k=1}^{\infty} (k m^* t_\text{m} \!+\!t_{\text{msg}}) ~(\underbrace{q_\text{m}^{m^*} q_\text{b}^{m^* t_\text{m}/t_\text{b}}}_{=:\gamma})^k \sum_{\substack{m_2=\text{max}(1,\\ \ceil{km^*-t_\text{cut}^{\prime}/t_\text{m}})}}^{k m^*} \!\!\!\! q_\text{m}^{m_2} \\
        = & ~\frac{\delta}{p_\text{m}} \Big(\sum_{k=1}^{\floor{(t_\text{m}+t_\text{cut}^\prime)/(m^* t_\text{m})}} \!\!\!\!\!\!\!\!(k m^* t_\text{m} \!+\! t_{\text{msg}}) \gamma^{k} (q_\text{m} \!-\! q_\text{m}^{k m^*+1}) \nonumber \\
        & \quad +  \!\!\!\!\! \sum_{\substack{k=1 + \floor{(t_\text{m}+t_\text{cut}^\prime)/(m^* t_\text{m})}}}^\infty \!\!\!\!\!\!\!\!\!\!\! (km^*t_\text{m} + t_{\text{msg}})\gamma^k (q_\text{m}^{\ceil{km^*-t_\text{cut}^\prime/t_\text{m}}} \!-\!q_\text{m}^{km^*+1})  \Big) \\
        \overset{(\text{i})}{=} & ~\frac{\delta q_\text{m}}{p_\text{m}}\Big( m^*t_\text{m}\!\sum_{k=1}^{\xi} k\gamma^{k} \!+\! t_{\text{msg}}\!\sum_{k=1}^\xi \gamma^k \!-\! m^*t_\text{m}\!\sum_{k=1}^{\xi} k(\underbrace{\gamma q_\text{m}^{m^*}}_{=:\thinspace \beta})^{k} \!-\! t_{\text{msg}}\!\sum_{k=1}^\xi (\underbrace{\gamma q_\text{m}^{m^*}}_{\beta})^{k} \Big) \nonumber \\
        & \quad + \frac{\delta\big(q_\text{m}^{-\floor{t_\text{cut}^\prime/t_\text{m}}} \!-\!q_\text{m} \big)}{p_\text{m}} \Big( m^*t_\text{m} \!\sum_{k=\xi+1}^\infty \!\!k(\underbrace{\gamma q_\text{m}^{m^*}}_{\beta})^k + t_{\text{msg}} \!\sum_{k=\xi+1}^\infty \!(\underbrace{\gamma q_\text{m}^{m^*}}_{\beta})^k \Big) \\
        =&~ \frac{\delta q_\text{m} m^* t_\text{m}}{p_\text{m}} \Big(\frac{\gamma(1 \!-\!(\xi\!+\!1)\gamma^\xi \!+\! \xi \gamma^{\xi+1})}{(1\!-\!\gamma)^2} -\frac{\beta(1 \!-\!(\xi\!+\!1)\beta^\xi \!+\! \xi \beta^{\xi+1})}{(1\!-\!\beta)^2} \Big) \nonumber \\
        & \quad +\frac{\delta q_\text{m} t_{\text{msg}}}{p_\text{m}} \Big(\frac{\gamma(1\!-\!\gamma^\xi)}{1\!-\!\gamma} - \frac{\beta(1\!-\!\beta^\xi)}{1\!-\!\beta}\Big) + \frac{\delta\big(q_\text{m}^{-\floor{t_\text{cut}^\prime/t_\text{m}}} \!-\! q_\text{m} \big) \beta^{\xi+1}}{p_\text{m} (1\!-\!\beta)} \Big( m^* t_\text{m} \frac{\xi\!+\!1\!-\!\xi\beta}{1\!-\!\beta} \!+\! t_{\text{msg}}\Big) ~,
        \label{eq:finite_sum_A_1^+A_b^+}
    \end{align}
\endgroup
where in (i), we have used the notation ${\xi \!:=\! \floor{(t_\text{m}+t_\text{cut}^\prime)/(m^* t_\text{m})}}$.

Lastly, on $A_1^+ A_2^+ A_b^+$, we have
\begingroup
    \setlength{\abovedisplayskip}{6pt}   
    \setlength{\belowdisplayskip}{6pt}   
    \setlength{\abovedisplayshortskip}{2pt}
    \setlength{\belowdisplayshortskip}{2pt}
    \begin{align}
        &~\mathbb{E}\big(Z \mathds{1}_{A_1^+ A_2^+ A_b^+}\big) \nonumber \\
        = & ~ \underbrace{\bigg(\frac{p_\text{m}^2 p_\text{b}}{(1\!-\!p_\text{m})^2(1\!-\!p_\text{b})}\bigg)}_{\delta}
        \thinspace\sum_{k=1}^{\infty} (km^*t_\text{m} \!+\!t_{\text{msg}}) {(\underbrace{1\!-\!p_\text{m}}_{q_\text{m}})}^{m^*k} {(\underbrace{1\!-\!p_\text{m}}_{q_\text{m}})}^{m^*k} {(\underbrace{1\!-\!p_\text{b}}_{q_\text{b}})}^{m^*kt_\text{m}/t_\text{b}} \\
        = & ~\delta \sum_{k=1}^{\infty} (k m^* t_\text{m} \!+\!t_{\text{msg}}) {(\underbrace{q_\text{m}^{2m^*} q_\text{b}^{m^*t_\text{m}/t_\text{b}}}_{\beta})}^{k} \\
        =&~ \delta \Big( m^* t_\text{m} \sum_{k=1}^\infty k\beta^k + t_{\text{msg}} \sum_{k=1}^\infty \beta^k \Big) \\
        =& ~ \frac{\delta\beta}{1\!-\!\beta}\Big( \frac{m^* t_\text{m}}{1\!-\!\beta} + t_{\text{msg}} \Big) ~.
        \label{eq:finite_sum_A_1^+A_2^+A_b^+}
    \end{align}
\endgroup

By plugging ~\eqref{eq:finite_sum_A_1^+}, \eqref{eq:finite_sum_A_b^+}, \eqref{eq:finite_sum_A_1^+A_2^+}, \eqref{eq:finite_sum_A_1^+A_b^+}, and \eqref{eq:finite_sum_A_1^+A_2^+A_b^+} in~\eqref{eq:zy1}, we obtain the value of $\mathbb{E}(Z\mathds{1}_{Y=1})$.
Next, to calculate the individual expectations of~\eqref{eq:inclusion-exclusion_part_2}, we observe the ranges of \( M_1 \), \( M_2 \), and \( m_\text{b} \) on \( A_i^- \) as summarised in Table~\ref{tab:sample_space_subset_A-}.
On \( A_{12}^- \), we have
\begingroup
    \setlength{\abovedisplayskip}{3pt}   
    \setlength{\belowdisplayskip}{2pt}   
    \setlength{\abovedisplayshortskip}{2pt}
    \setlength{\belowdisplayshortskip}{2pt}
    \begin{align}
        & ~\mathbb{E}\big(Z \mathds{1}_{A_{12}^-}\big) \nonumber \\
        = & ~\mathbb{E}\big((X_{\text{min}}\!+\!t_\text{cut}) \mathds{1}_{A_{12}^-}\big) \\
        = & \!\underbrace{\bigg(\frac{p_\text{m}^2 p_\text{b}}{(1\!-\!p_\text{m})^2(1\!-\!p_\text{b})}\bigg)}_{\delta}\sum_{m_2=1}^{\infty} \! \sum_{\substack{m_1=(m_2 \\ + \ceil{t_\text{cut} / t_\text{m}})}}^{\infty} \sum_{m_\text{b}=\left\lceil m_2 t_\text{m} / t_\text{b}\right\rceil}^{\left\lfloor m_1 t_\text{m}/ t_\text{b}\right\rfloor} \!\!\!\!(m_2 t_\text{m}\!+\!t_\text{cut}) (\underbrace{1\!-\!p_\text{m}}_{q_\text{m}})^{m_1} (\underbrace{1\!-\!p_\text{m}}_{q_\text{m}})^{m_2} (\underbrace{1\!-\!p_\text{b}}_{q_\text{b}})^{m_\text{b}}   \\
        = & ~\delta \sum_{m_2=1}^{\infty} \left(m_2 t_\text{m}\!+\!t_\text{cut}\right)q_\text{m}^{m_2} \!\!\!\!\!\! \sum_{m_1=m_2+\left\lceil t_\text{cut} / t_\text{m}\right\rceil}^{\infty} \!\!\!\!\!\!\!\! q_\text{m}^{m_1} \sum_{m_\text{b}=\left\lceil m_2 t_\text{m} / t_\text{b}\right\rceil}^{\left\lfloor m_1 t_\text{m}/ t_\text{b}\right\rfloor} \!\!\!\!\!\! q_\text{b}^{m_\text{b}} \\
        = & ~\frac{\delta}{p_\text{b}} \sum_{m_2=1}^{\infty} \left(m_2 t_\text{m}\!+\!t_\text{cut}\right)q_\text{m}^{m_2} \!\!\!\!\!\! \sum_{m_1=m_2+\left\lceil t_\text{cut} / t_\text{m}\right\rceil}^{\infty} \!\!\!\!\!\!\!\! q_\text{m}^{m_1} \big( q_\text{b}^{\ceil{m_2t_\text{m}/t_\text{b}}} \!-\! q_\text{b}^{\floor{m_1t_\text{m}/t_\text{b}}+1}\big)        \\
        = &~\frac{\delta q_\text{m}^{\ceil{t_\text{cut}/t_\text{m}}}}{p_\text{m} p_\text{b}}\sum_{m_2=1}^{\infty}(m_2t_\text{m}\!+\!t_\text{cut})q_\text{m}^{2m_2} q_\text{b}^{\ceil{m_2t_\text{m}/t_\text{b}}} - \frac{\delta q_\text{b}}{p_\text{b}}\sum_{m_2=1}^\infty (m_2 t_\text{m}\!+\!t_\text{cut})  q_\text{m}^{m_2}\!\!\!\!\!\!\!\!\!\sum_{m_1=m_2+\ceil{t_\text{cut}/t_\text{m}}}^\infty\!\!\!\!\!\! q_\text{m}^{m_1} q_\text{b}^{\floor{m_1t_\text{m}/t_\text{b}}} \\
        \overset{\text{(i)}}{=} & ~\frac{\delta q_\text{m}^{\ceil{t_\text{cut}/t_\text{m}}}}{p_\text{m} p_\text{b}}\Big(t_\text{m}\Theta_c\big(q_\text{m}^2,q_\text{b},\frac{t_\text{m}}{t_\text{b}},0,1,\infty \big) + t_\text{cut}\Pi_c\big(q_\text{m}^2,q_\text{b},\frac{t_\text{m}}{t_\text{b}},0,1,\infty \big) \Big) \nonumber \\
        & \quad - \frac{\delta q_\text{b}}{p_\text{b}} \sum_{m_1=1+\ceil{t_\text{cut}/t_\text{m}}}^\infty \!\!\!\!\!\!\!\!q_\text{m}^{m_1} q_\text{b}^{\floor{m_1t_\text{m}/t_\text{b}}} \sum_{m_2=1}^{m_1-\ceil{t_\text{cut}/t_\text{m}}} \!\!\!\!(t_\text{cut}\!+\!m_2t_\text{m}) q_\text{m}^{m_2} \\
        \overset{\text{(ii)}}{=} &~ \frac{\delta q_\text{m}^{\ceil{t_\text{cut}/t_\text{m}}}}{p_\text{m} p_\text{b}} \bigg(t_\text{m} \Big( \frac{\Theta_c(q_\text{m}^2,q_\text{b},{t_\text{m}}/{t_\text{b}},0,1,m^*)}{1-\big(q_\text{m}^2 q_\text{b}^{t_\text{m}/t_\text{b}}\big)^{m^*}} + m^* \big(q_\text{m}^2q_\text{b}^{t_\text{m}/t_\text{b}}\big)^{m^*} \frac{\Pi_c(q_\text{m}^2,q_\text{b},{t_\text{m}}/{t_\text{b}},0,1,m^*)}{\big(1-\big(q_\text{m}^2 q_\text{b}^{t_\text{m}/t_\text{b}}\big)^{m^*}\big)^2} \Big) \nonumber \\
        & \quad + t_\text{cut} \frac{\Pi_c(q_\text{m}^2,q_\text{b},{t_\text{m}}/{t_\text{b}},0,1,m^*)}{1-\big(q_\text{m}^2 q_\text{b}^{t_\text{m}/t_\text{b}}\big)^{m^*}} \bigg) - \frac{\delta q_\text{b}}{p_\text{b}} \sum_{m_1=1+\ceil{t_\text{cut}/t_\text{m}}}^\infty \!\!\!\!\!\!\!\! q_\text{m}^{m_1}q_\text{b}^{\floor{m_1t_\text{m}/t_\text{b}}} \bigg( t_\text{cut}\frac{q_\text{m} \!-\!q_\text{m}^{m_1-\ceil{t_\text{cut}/t_\text{m}}+1}}{1\!-\!q_\text{m}} \nonumber \\
        & \quad + t_\text{m}\frac{q_\text{m} \!-\! (m_1\!-\!\ceil{t_\text{cut}/t_\text{m}}\!+\!1) q_\text{m}^{m_1-\ceil{t_\text{cut}/t_\text{m}}+1} + (m_1\!-\!\ceil{t_\text{cut}/t_\text{m}})q_\text{m}^{m_1-\ceil{t_\text{cut}/t_\text{m}}+2}}{(1\!-\!q_\text{m})^2} \bigg)
    \end{align}
\endgroup
\begingroup
    \setlength{\abovedisplayskip}{3pt}   
    \setlength{\belowdisplayskip}{6pt}   
    \setlength{\abovedisplayshortskip}{2pt}
    \setlength{\belowdisplayshortskip}{2pt}
    \begin{align}
        = & \underbrace{\frac{\delta q_\text{m}^{\ceil{t_\text{cut}/t_\text{m}}}}{p_\text{m} p_\text{b}(1\!-\!\beta)}\Big(t_\text{m}\Theta_c\big(q_\text{m}^2,q_\text{b},m^*,\frac{t_\text{m}}{t_\text{b}},0\big) + \Big(\frac{\beta m^*t_\text{m}}{1\!-\!\beta} + t_\text{cut}\Big)\Pi_c\big(q_\text{m}^2,q_\text{b},m^*,\frac{t_\text{m}}{t_\text{b}},0\big) \Big)}_{=:\thinspace T_{12^-} } \nonumber \\
        & \quad - \frac{\delta q_\text{b}}{p_\text{b}}\Bigg( \frac{q_\text{m}}{p_\text{m}}\Big(\frac{t_\text{m}}{p_\text{m}}\!+\!t_\text{cut} \Big)\!\!\!\! \sum_{m_1=1+\ceil{t_\text{cut}/t_\text{m}}}^\infty\!\!\!\!\!\! q_\text{m}^{m_1} q_\text{b}^{\floor{m_1t_\text{m}/t_\text{b}}} \nonumber \\
        & \quad -\! \frac{q_\text{m}^{1-\ceil{t_\text{cut}/t_\text{m}}}}{p_\text{m}}\Big(\frac{t_\text{m}}{p_\text{m}}\Big(1\!-\!\ceil{\frac{t_\text{cut}}{t_\text{m}}} \!+\!q_\text{m}\ceil{\frac{t_\text{cut}}{t_\text{m}}}\Big) \!+\!t_\text{cut} \Big) \!\!\!\! \sum_{m_1=1+\ceil{t_\text{cut}/t_\text{m}}}^\infty\!\!\!\!\!\! q_\text{m}^{2m_1} q_\text{b}^{\floor{m_1t_\text{m}/t_\text{b}}} \nonumber \\
        & \quad -\! \frac{q_\text{m}^{1-\ceil{t_\text{cut}/t_\text{m}}}}{p_\text{m}^2} ~t_\text{m}(1\!-\!q_\text{m})\!\!\!\! \sum_{m_1=1+\ceil{t_\text{cut}/t_\text{m}}}^\infty\!\!\!\!\!\! m_1 q_\text{m}^{2m_1} q_\text{b}^{\floor{m_1t_\text{m}/t_\text{b}}} \Bigg) \\
        = & ~T_{12^-} - \frac{\delta q_\text{m}q_\text{b}}{p_\text{m}p_\text{b}}\Bigg(\Big(\frac{t_\text{m}}{p_\text{m}}\!+\!t_\text{cut}\Big)\Pi_f\big(q_\text{m},q_\text{b},\frac{t_\text{m}}{t_\text{b}},0, 1\!+\!\ceil{\frac{t_\text{cut}}{t_\text{m}}},\infty\big) \nonumber \\
        & \quad - { q_\text{m}^{-\ceil{t_\text{cut}/t_\text{m}}}}\Big(\frac{t_\text{m}}{p_\text{m}}\!-\!t_\text{m}\ceil{\frac{t_\text{cut}}{t_\text{m}}}\!+\!t_\text{cut}\Big) \Pi_f\big(q_\text{m}^2,q_\text{b},\frac{t_\text{m}}{t_\text{b}},0, 1\!+\!\ceil{\frac{t_\text{cut}}{t_\text{m}}},\infty\big) \nonumber \\
        & \quad - { q_\text{m}^{-\ceil{t_\text{cut}/t_\text{m}}} t_\text{m}}\Theta_f\big(q_\text{m}^2,q_\text{b},\frac{t_\text{m}}{t_\text{b}},0, 1\!+\!\ceil{\frac{t_\text{cut}}{t_\text{m}}},\infty\big) \Bigg) \\
        \overset{(\text{iii})}{=} & ~T_{12^-} - \frac{\delta q_\text{m}q_\text{b}}{p_\text{m}p_\text{b}}\Bigg(\Big(\frac{t_\text{m}}{p_\text{m}}\!+\!t_\text{cut}\Big) \frac{1}{1\!-\!\big(q_\text{m} q_\text{b}^{t_\text{m}/t_\text{b}} \big)^{m^*}} \Pi_f\big(q_\text{m},q_\text{b},\frac{t_\text{m}}{t_\text{b}},0, 1\!+\!\ceil{\frac{t_\text{cut}}{t_\text{m}}}\!,\!\ceil{\frac{t_\text{cut}}{t_\text{m}}} \!+\! m^*\big) \nonumber \\
        & \quad - { q_\text{m}^{-\ceil{t_\text{cut}/t_\text{m}}}}\Big(\frac{t_\text{m}}{p_\text{m}}\!-\!t_\text{m}\ceil{\frac{t_\text{cut}}{t_\text{m}}}\!+\!t_\text{cut}\Big) \frac{1}{1\!-\!\big(q_\text{m}^2 q_\text{b}^{t_\text{m}/t_\text{b}} \big)^{m^*}} \Pi_f\big(q_\text{m}^2,q_\text{b},\frac{t_\text{m}}{t_\text{b}},0, 1\!+\!\ceil{\frac{t_\text{cut}}{t_\text{m}}}\!,\!\ceil{\frac{t_\text{cut}}{t_\text{m}}} \!+\! m^*\big) \nonumber \\
        & \quad - { q_\text{m}^{-\ceil{t_\text{cut}/t_\text{m}}} t_\text{m}} \Big( \frac{1}{1\!-\!\big(q_\text{m} q_\text{b}^{t_\text{m}/t_\text{b}} \big)^{m^*}}\Theta_f\big(q_\text{m}^2,q_\text{b},\frac{t_\text{m}}{t_\text{b}},0, 1\!+\!\ceil{\frac{t_\text{cut}}{t_\text{m}}}\!,\ceil{\frac{t_\text{cut}}{t_\text{m}}} \!+\!m^*\big) \nonumber \\
        & \quad + \frac{m^*}{\big(1\!-\!\big(q_\text{m} q_\text{b}^{t_\text{m}/t_\text{b}} \big)^{m^*}\big)^2}\Pi_f\big(q_\text{m}^2,q_\text{b},\frac{t_\text{m}}{t_\text{b}},0, 1\!+\!\ceil{\frac{t_\text{cut}}{t_\text{m}}}\!,\ceil{\frac{t_\text{cut}}{t_\text{m}}} \!+\!m^*\big)\Big) \Bigg)~,
        \label{eq:finite_sum_A_12^-}
    \end{align}
\endgroup
where we have used the definition of $\Pi_c$, $\Pi_f$, $\Theta_c$, and $\Theta_f$ from~\eqref{eq:def_Pi} and~\eqref{eq:def_Theta}.
Also, note that we have changed the order of summation in (i).
In (ii) and (iii), we recalled that ${z^*(t_\text{m}/t_\text{b}) \!=\! m^*}$, and applied the identity~\eqref{eq:sum_i_x^i}.

On $A_{1b}^-$, we observe that
\begingroup
    \setlength{\abovedisplayskip}{2pt}   
    \setlength{\belowdisplayskip}{2pt}   
    \setlength{\abovedisplayshortskip}{2pt}
    \setlength{\belowdisplayshortskip}{2pt}
    \begin{align}
        & ~\mathbb{E}\big(Z \mathds{1}_{A_{1b}^-}\big) \nonumber \\
        = & ~\mathbb{E}\big((X_{\text{min}}\!+\!t_\text{cut}) \mathds{1}_{A_{1b}^-}\big) \\
        = & \underbrace{\bigg(\frac{p_\text{m}^2 p_\text{b}}{(1\!-\!p_\text{m})^2(1\!-\!p_\text{b})}\bigg)}_{\delta} \sum_{m_\text{b}=1}^{\infty}(m_\text{b} t_\text{b}\!+\!t_\text{cut}) (\underbrace{1\!-\!p_\text{b}}_{q_\text{b}})^{m_\text{b}} \!\!\!\!\!\! \sum_{\substack{m_1=\\\left\lceil(m_\text{b} t_\text{b}+t_\text{cut}) / t_\text{m} \right\rceil } }^{\infty} \!\!\!\! (\underbrace{1\!-\!p_\text{m}}_{q_\text{m}})^{m_1} \!\! \sum_{m_2=\left\lceil m_\text{b} t_\text{b} / t_\text{m}\right\rceil}^{m_1} \!\!\!\! (\underbrace{1\!-\!p_\text{m}}_{q_\text{m}})^{m_2} \\
        = & ~\frac{\delta}{p_\text{m}} \sum_{m_\text{b}=1}^{\infty}(m_\text{b} t_\text{b}\!+\!t_\text{cut}) q_\text{b}^{m_\text{b}} \!\!\!\!\!\!\!\! \sum_{m_1=\left\lceil(m_\text{b} t_\text{b}+t_\text{cut}) / t_\text{m} \right\rceil}^{\infty} \!\!\!\!\!\!\!\! q_\text{m}^{m_1} \big( {q_\text{m}^{\left\lceil m_\text{b} t_\text{b} / t_\text{m}\right\rceil} - q_\text{m}^{m_1+1}} \big) \\
        = & ~\frac{\delta}{p_\text{m}^2}\sum_{m_\text{b}=1}^\infty(m_\text{b} t_\text{b}+t_\text{cut}) q_\text{b}^{m_\text{b}} q_\text{m}^{\ceil{{(m_\text{b}t_\text{b}+t_\text{cut})}/{t_\text{m}}}\!+\!\ceil{{m_\text{b}t_\text{b}}/{t_\text{m}}}} - \frac{\delta q_\text{m}}{p_\text{m}(1\!-\!q_\text{m}^2)}\sum_{m_\text{b}=1}^\infty(m_\text{b} t_\text{b}+t_\text{cut}) q_\text{b}^{m_\text{b}} q_\text{m}^{2\ceil{{(m_\text{b}t_\text{b}+t_\text{cut})}/{t_\text{m}}}}
    \end{align}
\endgroup
\begingroup
    \setlength{\abovedisplayskip}{6pt}   
    \setlength{\belowdisplayskip}{6pt}   
    \setlength{\abovedisplayshortskip}{2pt}
    \setlength{\belowdisplayshortskip}{2pt}
    \begin{align}
        = & ~\frac{\delta}{p_\text{m}^2}\bigg(t_\text{b}\Theta_{cc} \big(q_\text{b},q_\text{m},\frac{t_\text{b}}{t_\text{m}},-\frac{t_\text{cut}}{t_\text{m}},0,1,\infty\big) + t_\text{cut} \Pi_{cc} \big(q_\text{b},q_\text{m},\frac{t_\text{b}}{t_\text{m}},-\frac{t_\text{cut}}{t_\text{m}},0,1,\infty\big) \!\bigg) \nonumber \\
        & \quad - \frac{\delta q_\text{m}}{p_\text{m}(1\!-\!q_\text{m}^2)}\bigg(t_\text{b}\Theta_c\big(q_\text{b},q_\text{m}^2,\frac{t_\text{b}}{t_\text{m}},-\frac{t_\text{cut}}{t_\text{m}},1,\infty\big) + t_\text{cut}\Pi_c\big(q_\text{b},q_\text{m}^2,\frac{t_\text{b}}{t_\text{m}},-\frac{t_\text{cut}}{t_\text{m}},1,\infty\big) \!\bigg) \qquad \quad \\
        \overset{(\text{i})}{=} & ~\frac{\delta}{p_\text{m}^2}\bigg(t_\text{b}\Big(\frac{\Theta_{cc} (q_\text{b},q_\text{m},{t_\text{b}}/{t_\text{m}},-{t_\text{cut}}/{t_\text{m}},0,1,{m^*t_\text{m}}/{t_\text{b}})}{1\!-\!\big(q_\text{b}q_\text{m}^{2t_\text{b}/t_\text{m}}\big)^{m^*t_\text{m}/t_\text{b}}} \nonumber \\
        & \quad + \frac{m^*t_\text{m}}{t_\text{b}}\big(q_\text{b} q_\text{m}^{2t_\text{b}/t_\text{m}}\big)^{m^*t_\text{m}/t_\text{b}} \frac{\Pi_{cc} (q_\text{b},q_\text{m},{t_\text{b}}/{t_\text{m}},-{t_\text{cut}}/{t_\text{m}},0,1,{m^*t_\text{m}}/{t_\text{b}})}{\big(1\!-\!\big(q_\text{b}q_\text{m}^{2t_\text{b}/t_\text{m}}\big)^{m^*t_\text{m}/t_\text{b}}\big)^2} \Big) \nonumber \\
        & \quad + t_\text{cut} \frac{\Pi_{cc} (q_\text{b},q_\text{m},{t_\text{b}}/{t_\text{m}},-{t_\text{cut}}/{t_\text{m}},0,1,{m^*t_\text{m}}/{t_\text{b}} )}{1\!-\!\big(q_\text{b}q_\text{m}^{2t_\text{b}/t_\text{m}}\big)^{m^*t_\text{m}/t_\text{b}}} \!\bigg) \nonumber \\
        & \quad - \frac{\delta q_\text{m}}{p_\text{m}(1\!-\!q_\text{m}^2)} \bigg(t_\text{b}\Big(\frac{\Theta_c(q_\text{b},q_\text{m}^2,{t_\text{b}}/{t_\text{m}},-{t_\text{cut}}/{t_\text{m}},1,{m^*t_\text{m}}/{t_\text{b}})}{1\!-\!\big(q_\text{b}q_\text{m}^{2t_\text{b}/t_\text{m}}\big)^{m^*t_\text{m}/t_\text{b}}} \nonumber \\
        & \quad + \frac{m^* t_\text{m}}{t_\text{b}}\big(q_\text{b} q_\text{m}^{2t_\text{b}/t_\text{m}}\big)^{m^*t_\text{m}/t_\text{b}} \frac{\Pi_c(q_\text{b},q_\text{m}^2,{t_\text{b}}/{t_\text{m}},-{t_\text{cut}}/{t_\text{m}},1,{m^*t_\text{m}}/{t_\text{b}})}{\big(1\!-\!\big(q_\text{b}q_\text{m}^{2t_\text{b}/t_\text{m}}\big)^{m^*t_\text{m}/t_\text{b}}\big)^2} \Big) \nonumber \\
        & \quad + t_\text{cut} \frac{\Pi_c(q_\text{b},q_\text{m}^2,{t_\text{b}}/{t_\text{m}},-{t_\text{cut}}/{t_\text{m}},1,{m^*t_\text{m}}/{t_\text{b}})}{1\!-\!\big(q_\text{b}q_\text{m}^{2t_\text{b}/t_\text{m}} \big)^{m^*t_\text{m}/t_\text{b}}} \!\bigg) ~,
        \label{eq:finite_sum_A_1b^-}
    \end{align}
\endgroup
where we recalled, in (i), that ${z^*(t_\text{b}/t_\text{m}) \!=\! m^*t_\text{m}/t_\text{b}}$.

On $A_{b1}^-$, it follows that
\begingroup
    \setlength{\abovedisplayskip}{6pt}   
    \setlength{\belowdisplayskip}{6pt}   
    \setlength{\abovedisplayshortskip}{2pt}
    \setlength{\belowdisplayshortskip}{2pt}
    \begin{align}
        & ~\mathbb{E}\left(Z \mathds{1}_{A_{b1}^-}\right) \nonumber \\
        = & ~\mathbb{E}\big((X_{\text{min}}\!+\!t_\text{cut}) \mathds{1}_{A_{b1}^-}\big) \\
        = & \underbrace{\!\bigg(\frac{p_\text{m}^2 p_\text{b}}{(1\!-\!p_\text{m})^2(1\!-\!p_\text{b})}\bigg)\!}_{\delta} \sum_{m_1=1}^{\infty}\!(m_1 t_\text{m}\!+\!t_\text{cut}) (\underbrace{1\!-\!p_\text{m}}_{q_\text{m}})^{m_1} \!\!\!\!\!\!\!\sum_{m_\text{b}=\left\lceil (m_1 t_\text{m}+t_\text{cut})/{t_\text{b}}\right\rceil}^{\infty} \!\! (\underbrace{1\!-\!p_\text{b}}_{q_\text{b}})^{m_\text{b}} \!\! \sum_{m_2=m_1}^{\left\lfloor m_\text{b} t_\text{b} / t_\text{m}\right\rfloor} \!\!(\underbrace{1\!-\!p_\text{m}}_{q_\text{m}})^{m_2} \\
        =&~ \frac{\delta}{p_\text{m}} \sum_{m_1=1}^{\infty}(m_1 t_\text{m}\!+\!t_\text{cut}) q_\text{m}^{m_1} \!\!\! \sum_{m_\text{b}=\left\lceil (m_1 t_\text{m}+t_\text{cut})/{t_\text{b}}\right\rceil}^{\infty} \!\!\!\!\!\!\!\!\! q_\text{b}^{m_\text{b}} \big(q_\text{m}^{m_1} - q_\text{m}^{\floor{m_\text{b}t_\text{b}/t_\text{m}}+1} \big) \\
        = & ~\frac{\delta}{p_\text{m} p_\text{b}} \sum_{m_1=1}^{\infty}(m_1 t_\text{m}\!+\!t_\text{cut}) q_\text{m}^{2m_1}q_\text{b}^{\ceil{(m_1t_\text{m}+t_\text{cut})/t_\text{b}}} - \frac{\delta q_\text{m}}{p_\text{m}} \sum_{m_1=1}^{\infty}(m_1 t_\text{m}\!+\!t_\text{cut}) q_\text{m}^{m_1} \!\!\!\!\!\!\!\!\! \sum_{m_\text{b}=\ceil{(m_1t_\text{m}+t_\text{cut})/t_\text{b}}}^{\infty} \!\!\!\!\!\!\!\!\!\!\! q_\text{b}^{m_\text{b}} q_\text{m}^{\floor{m_\text{b}t_\text{b}/t_\text{m}}} \\
        \overset{\text{(i)}}{=} &~ \frac{\delta}{p_\text{m} p_\text{b}}\Big(t_\text{m}\Theta_c\big(q_\text{m}^2,q_\text{b},\frac{t_\text{m}}{t_\text{b}},-\frac{t_\text{cut}}{t_\text{b}},1,\infty\big) + t_\text{cut}\Pi_c\big(q_\text{m}^2,q_\text{b},\frac{t_\text{m}}{t_\text{b}},-\frac{t_\text{cut}}{t_\text{b}},1,\infty\big) \Big) \nonumber \\
        & \quad - \frac{\delta q_\text{m}}{p_\text{m}} \!\! \sum_{m_\text{b}=\ceil{(t_\text{m}+t_\text{cut})/t_\text{b}}}^{\infty} \!\!\!\!\!\!\!\!\!\! q_\text{b}^{m_\text{b}} q_\text{m}^{\floor{m_\text{b}t_\text{b}/t_\text{m}}} \!\!\!\!\!\!\sum_{m_1=1}^{\floor{(m_\text{b}t_\text{b}-t_\text{cut})/t_\text{m}}} \!\!\!\!\!\!\!\!(m_1t_\text{m}\!+\!t_\text{cut})q_\text{m}^{m_1} 
    \end{align}
\endgroup
\begin{align}
    \overset{\text{(ii)}}{=} &~ \frac{\delta}{p_\text{m} p_\text{b}}\bigg( t_\text{m} \Big( \frac{\Theta_c(q_\text{m}^2,q_\text{b},m^*,{t_\text{m}}/{t_\text{b}},-{t_\text{cut}}/{t_\text{b}},1,m^*)}{1\!-\!\big(q_\text{m}^2q_\text{b}^{t_\text{m}/t_\text{b}}\big)^{m^*}} + m^* \underbrace{\big(q_\text{m}^2 q_\text{b}^{t_\text{m}/t_\text{b}} \big)^{m^*}}_{\beta} \frac{\Pi_c(q_\text{m}^2,q_\text{b},{t_\text{m}}/{t_\text{b}},-{t_\text{cut}}/{t_\text{b}},1,m^*)}{\big(1\!-\!\big(q_\text{m}^2q_\text{b}^{t_\text{m}/t_\text{b}}\big)^{m^*}\big)^2}\Big) \nonumber \\
    & \quad + t_\text{cut}\frac{\Pi_c(q_\text{m}^2,q_\text{b},{t_\text{m}}/{t_\text{b}},-{t_\text{cut}}/{t_\text{b}},1,m^*)}{1\!-\!\big(q_\text{m}^2q_\text{b}^{t_\text{m}/t_\text{b}}\big)^{m^*}}  \bigg) - \frac{\delta q_\text{m}}{p_\text{m}} \sum_{m_\text{b}=\ceil{(t_\text{m}+t_\text{cut})/t_\text{b}}}^\infty \!\!\!\!\!\! q_\text{b}^{m_\text{b}} q_\text{m}^{\floor{m_\text{b}t_\text{m}/t_\text{m}}} \Bigg( \frac{t_\text{m}}{(\underbrace{1\!-\!q_\text{m}}_{p_\text{m}})^2} \bigg(q_\text{m} \nonumber \\
    & \quad -\!\Big(\floor{\frac{m_\text{b}t_\text{b}\!-\!t_\text{cut}}{t_\text{m}}}+1\Big) q_\text{m}^{\floor{(m_\text{b}t_\text{b}-t_\text{cut})/t_\text{m}}+1} \!+\! {\floor{\frac{m_\text{b}t_\text{b}\!-\!t_\text{cut}}{t_\text{m}}}}q_\text{m}^{\floor{(m_\text{b}t_\text{b}-t_\text{cut})/t_\text{m}}+2} \bigg) \nonumber \\
    &\quad + \frac{t_\text{cut}}{\underbrace{1\!-\!q_\text{m}}_{p_\text{m}}} \Big(q_\text{m} \!-\! q_\text{m}^{\floor{(m_\text{b}t_\text{b}-t_\text{cut})/t_\text{m}}+1}\Big) \Bigg)  \\
    = & \underbrace{\frac{\delta}{p_\text{m} p_\text{b}(1\!-\!\beta)}\bigg(t_\text{m}\Theta_c\big(q_\text{m}^2,q_\text{b},\frac{t_\text{m}}{t_\text{b}},-\frac{t_\text{cut}}{t_\text{b}},1,m^*\big) + \Big(\frac{m^* t_\text{m}\beta}{1\!-\! \beta}+t_\text{cut}\Big)\Pi_c\big(q_\text{m}^2,q_\text{b},\frac{t_\text{m}}{t_\text{b}},-\frac{t_\text{cut}}{t_\text{b}},1,m^*\big) \!\bigg)}_{=:\thinspace T_{b1^-}} \nonumber \\
    & \quad - \frac{\delta q_\text{m}}{p_\text{m}}\bigg( \Big(\frac{t_\text{m} q_\text{m}}{p_\text{m}^2}\!+\!\frac{t_\text{cut}q_\text{m}}{p_\text{m}}\Big) \!\!\! \sum_{m_\text{b}=\ceil{(t_\text{m}+t_\text{cut})/t_\text{b}}}^{\infty} \!\!\!\!\!\!\!\!\!\! q_\text{b}^{m_\text{b}} q_\text{m}^{\floor{m_\text{b}t_\text{b}/t_\text{m}}} \nonumber \\
    & \quad - \Big(\frac{t_\text{m} q_\text{m}}{p_\text{m}^2}\!+\!\frac{t_\text{cut}q_\text{m}}{p_\text{m}}\Big)\!\!\sum_{m_\text{b}=\ceil{(t_\text{m}+t_\text{cut})/t_\text{b}}}^{\infty} \!\!\!\!\!\!\!\!\!\! q_\text{b}^{m_\text{b}} q_\text{m}^{\floor{m_\text{b}t_\text{b}/t_\text{m}}+\floor{(m_\text{b}t_\text{b}-t_\text{cut})/t_\text{m}}} \nonumber \\
    & \quad - \frac{t_\text{m} q_\text{m}}{p_\text{m}^2}(\underbrace{1\!-\!q_\text{m}}_{p_\text{m}})\!\!\!\! \sum_{m_\text{b}=\ceil{(t_\text{m}+t_\text{cut})/t_\text{b}}}^{\infty} \!\!\floor{\frac{m_\text{b}t_\text{b}\!-\!t_\text{cut}}{t_\text{m}}} q_\text{b}^{m_\text{b}} q_\text{m}^{\floor{m_\text{b}t_\text{b}/t_\text{m}}+\floor{(m_\text{b}t_\text{b}-t_\text{cut})/t_\text{m}}} \bigg) \\
    = & ~T_{b1^-} - \frac{\delta q_\text{m}^2}{p_\text{m}^2}\bigg(\Big(\frac{t_\text{m}}{p_\text{m}}\!+\!t_\text{cut}\Big) \Big( \!\!\! \sum_{m_\text{b}=\ceil{(t_\text{m}+t_\text{cut})/t_\text{b}}}^{\infty} \!\!\!\!\!\!\!\!\!\! q_\text{b}^{m_\text{b}} q_\text{m}^{\floor{m_\text{b}t_\text{b}/t_\text{m}}}~ - \!\!\sum_{m_\text{b}=\ceil{(t_\text{m}+t_\text{cut})/t_\text{b}}}^{\infty} \!\!\!\!\!\!\!\!\!\! q_\text{b}^{m_\text{b}} q_\text{m}^{\floor{m_\text{b}t_\text{b}/t_\text{m}}+\floor{(m_\text{b}t_\text{b}-t_\text{cut})/t_\text{m}}} \Big) \nonumber \\
    & \quad - t_\text{m}\!\!\!\! \sum_{m_\text{b}=\ceil{(t_\text{m}+t_\text{cut})/t_\text{b}}}^{\infty} \!\!\floor{\frac{m_\text{b}t_\text{b}\!-\!t_\text{cut}}{t_\text{m}}} q_\text{b}^{m_\text{b}} q_\text{m}^{\floor{m_\text{b}t_\text{b}/t_\text{m}}+\floor{(m_\text{b}t_\text{b}-t_\text{cut})/t_\text{m}}} \bigg) \\
    = & ~T_{b1^-} \!-\! \frac{\delta q_\text{m}^2}{p_\text{m}^2} \bigg( \!\Big(\frac{t_\text{m}}{p_\text{m}}\!+\!t_\text{cut}\Big) \Big(\Pi_f\big(q_\text{b},q_\text{m},\frac{t_\text{b}}{t_\text{m}},0,\ceil{\frac{t_\text{m}\!+\!t_\text{cut}}{t_\text{b}}},\infty\big) \nonumber \\
    & \quad \! - \!\Pi_{ff}\big(q_\text{b},q_\text{m},\frac{t_\text{b}}{t_\text{m}},0,\frac{t_\text{cut}}{t_\text{m}},\!\ceil{\frac{t_\text{m}\!+\!t_\text{cut}}{t_\text{b}}}\!,\infty\big)\! \Big)  -\! {t_\text{m}} \Theta_{ff}\big(q_\text{b},q_\text{m},\frac{t_\text{b}}{t_\text{m}},0,\frac{t_\text{cut}}{t_\text{m}}, \!\ceil{\frac{t_\text{m}\!+\!t_\text{cut}}{t_\text{b}}}\!,\infty\big) \!\bigg)
\end{align}
\begin{align}
    \overset{(\text{iii})}{=} & ~T_{b1^-}^{(1)} \!-\! \frac{\delta q_\text{m}^2}{p_\text{m}^2} \bigg( \Big(\frac{t_\text{m}}{p_\text{m}}\!+\!t_\text{cut}\Big) \Big(\frac{1}{1 \!-\! \big(q_\text{b} q_\text{m}^{t_\text{b}/t_\text{m}}\big)^{m^*t_\text{m}/t_\text{b}} } \Pi_f\big(q_\text{b},q_\text{m},\frac{t_\text{b}}{t_\text{m}},0,\ceil{\frac{t_\text{m}\!+\!t_\text{cut}}{t_\text{b}}},\ceil{\frac{t_\text{m}\!+\!t_\text{cut}}{t_\text{b}}} \!+\!\frac{m^*t_\text{m}}{t_\text{b}}\!-\!1\big) \nonumber \\
    & \quad - \frac{1}{1 \!-\! \big(q_\text{b} q_\text{m}^{2t_\text{b}/t_\text{m}}\big)^{m^*t_\text{m}/t_\text{b}}} \Pi_{ff}\big(q_\text{b},q_\text{m},\frac{t_\text{b}}{t_\text{m}},0,\frac{t_\text{cut}}{t_\text{m}},\ceil{\frac{t_\text{m}\!+\!t_\text{cut}}{t_\text{b}}},\ceil{\frac{t_\text{m}\!+\!t_\text{cut}}{t_\text{b}}}\!+\!\frac{m^*t_\text{m}}{t_\text{b}}\!-\!1\big) \Big) \nonumber \\
    & \quad - {t_\text{m}} \Big(\frac{1}{1 \!-\! \big(q_\text{b} q_\text{m}^{2t_\text{b}/t_\text{m}}\big)^{m^*t_\text{m}/t_\text{b}}} \Theta_{ff}\big(q_\text{b},q_\text{m},\frac{t_\text{b}}{t_\text{m}},0,\frac{t_\text{cut}}{t_\text{m}},\ceil{\frac{t_\text{m}\!+\!t_\text{cut}}{t_\text{b}}},\ceil{\frac{t_\text{m}\!+\!t_\text{cut}}{t_\text{b}}}\!+\!\frac{m^*t_\text{m}}{t_\text{b}}\!-\!1\big) \nonumber \\
    & \quad + \frac{m^* t_\text{m}}{t_\text{b}}\frac{t_\text{b}}{t_\text{m}}\frac{\big(q_\text{b} q_\text{m}^{2t_\text{b}/t_\text{m}}\big)^{m^*t_\text{m}/t_\text{b}}}{\big(1 \!-\! \big(q_\text{b} q_\text{m}^{2t_\text{b}/t_\text{m}}\big)^{m^*t_\text{m}/t_\text{b}}\big)^2} \Pi_{ff}\big(q_\text{b},q_\text{m},\frac{t_\text{b}}{t_\text{m}},0,\frac{t_\text{cut}}{t_\text{m}},\ceil{\frac{t_\text{m}\!+\!t_\text{cut}}{t_\text{b}}},\ceil{\frac{t_\text{m}\!+\!t_\text{cut}}{t_\text{b}}}\!+\!\frac{m^*t_\text{m}}{t_\text{b}}\!-\!1 \big) \Big)\bigg) ~.
    \label{eq:finite_sum_A_b1^-}
\end{align}
\looseness=-1 
Note that we have changed the order of summation in (i).
We used in (ii) and (iii) that ${z^*(t_\text{m}/t_\text{b})\!=\!m^*}$ and ${z^*(t_\text{b}/t_\text{m})\!=\!m^*t_\text{m}/t_\text{b}}$, respectively, and applied the identity~\eqref{eq:sum_i_x^i}.

On $A_{12}^- A_{1b}^-$, we have
\begingroup
    \setlength{\abovedisplayskip}{6pt}   
    \setlength{\belowdisplayskip}{10pt}   
    \setlength{\abovedisplayshortskip}{2pt}
    \setlength{\belowdisplayshortskip}{2pt}
    \begin{align}
        & ~\mathbb{E}\big(Z \mathds{1}_{A_{12}^- A_{1b}^-}\big) \nonumber \\
        = & ~\mathbb{E}\big((X_{\text{min}}\!+\!t_\text{cut}) \mathds{1}_{A_{12}^-A_{1b}^-}\big) \\
        = & ~\underbrace{\bigg(\frac{p_\text{m}^2 p_\text{b}}{(1\!-\!p_\text{m})^2(1\!-\!p_\text{b})}\bigg)}_{\delta} \sum_{k=1}^{\infty} (km^* t_\text{m}\!+\!t_\text{cut}) (\underbrace{1\!-\!p_\text{m}}_{q_\text{m}})^{km^*} (\underbrace{1\!-\!p_\text{b}}_{q_\text{b}})^{k m^*t_\text{m}/t_\text{b}} \!\!\!\!\!\! \sum_{m_1=k m^*+\left\lceil t_\text{cut} / t_\text{m}\right\rceil}^{\infty}\!\!\! (\underbrace{1\!-\!p_\text{m}}_{q_\text{m}})^{m_1} \\
        = & ~\delta \sum_{k=1}^{\infty} (k m^* t_\text{m}\!+\!t_\text{cut}) \big(q_\text{m}^{m^*} q_\text{b}^{m^* t_\text{m} / t_\text{b}}\big)^k ~\frac{q_\text{m}^{km^*+\left\lceil t_\text{cut} / t_\text{m}\right\rceil}}{1\!-\!q_\text{m}} \\
        = & ~\frac{\delta q_\text{m}^{\ceil{t_\text{cut}/t_\text{m}}}}{p_\text{m}} \sum_{k=1}^\infty (k m^* t_\text{m} + t_\text{cut}) ~\big( \underbrace{q_\text{m}^{2m^*}q_\text{b}^{m^*t_\text{m}/t_\text{b}}}_{\beta}\big)^k        \\
        =&~ \frac{\delta q_\text{m}^{\ceil{t_\text{cut}/t_\text{m}}}}{p_\text{m}} \Big( m^*t_\text{m}\sum_{k=1}^{\infty} k\beta^k + t_\text{cut} \sum_{k=1}^{\infty} \beta^k \Big)  \\
        = & ~\frac{\delta q_\text{m}^{\ceil{t_\text{cut}/t_\text{m}}}\beta}{p_\text{m}(1\!-\!\beta)}\Big(\frac{m^* t_\text{m}}{1\!-\!\beta} + t_\text{cut}\Big)~
        \label{eq:finite_sum_A_12^-A_1b^-}
    \end{align}
\endgroup
Further, on $A_{b1}^- A_{b2}^-$,
\begingroup
    \setlength{\abovedisplayskip}{6pt}   
    \setlength{\belowdisplayskip}{6pt}   
    \setlength{\abovedisplayshortskip}{2pt}
    \setlength{\belowdisplayshortskip}{2pt}
    \begin{align}
        & ~\mathbb{E}\big(Z \mathds{1}_{A_{b1}^- A_{b2}^-}\big) \nonumber \\
        = & ~\mathbb{E}\big((X_{\text{min}}\!+\!t_\text{cut}) \mathds{1}_{A_{b1}^-A_{b2}^-}\big) \\
        = & ~\underbrace{\bigg(\frac{p_\text{m}^2 p_\text{b}}{(1\!-\!p_\text{m})^2(1\!-\!p_\text{b})}\bigg)}_{\delta} \sum_{m_1=1}^{\infty}(m_1 t_\text{m}\!+\!t_\text{cut}) \big( \underbrace{1\!-\!p_\text{m}}_{q_\text{m}}\big)^{2 m_1} \!\!\!\!\!\!\! \sum_{m_\text{b}=\left\lceil(m_1 t_\text{m}+t_\text{cut})/t_\text{b} \right\rceil}^{\infty} \!\!\!\! \big( \underbrace{1\!-\!p_\text{b}}_{q_\text{b}}\big)^{m_\text{b}} \\
        = & ~\frac{\delta}{p_\text{b}}\sum_{m_1=1}^\infty (m_1 t_\text{m}\!+\!t_\text{cut}) q_\text{m}^{2m_1} q_\text{b}^{\ceil{(m_1 t_\text{m}+t_\text{cut})/t_\text{b}}}  \\
        = & ~\frac{\delta}{p_\text{b}} \Big( t_\text{m}\Theta_c\big(q_\text{m}^2,q_\text{b},\frac{t_\text{m}}{t_\text{b}},-\frac{t_\text{cut}}{t_\text{b}},1,\infty\big) + t_\text{cut}\Pi_c\big(q_\text{m}^2,q_\text{b},\frac{t_\text{m}}{t_\text{b}},-\frac{t_\text{cut}}{t_\text{b}},1,\infty\big) \Big) \\
        \overset{(\text{i})}{=} & ~\frac{\delta}{p_\text{b}} \bigg( t_\text{m}\Big( \frac{\Theta_c(q_\text{m}^2,q_\text{b},{t_\text{m}}/{t_\text{b}},-{t_\text{cut}}/{t_\text{b}},1,m^*)}{1\!-\!\big(q_\text{m}^2q_\text{b}^{t_\text{m}/t_\text{b}} \big)^{m^*}} + m^* \big(q_\text{m}^2q_\text{b}^{t_\text{m}/t_\text{b}}\big)^{m^*} \frac{\Pi_c(q_\text{m}^2,q_\text{b},{t_\text{m}}/{t_\text{b}},-{t_\text{cut}}/{t_\text{b}},1,m^*)}{\big(1\!-\!\big(q_\text{m}^2q_\text{b}^{t_\text{m}/t_\text{b}} \big)^{m^*}\big)^2} \Big) \nonumber \\
        & \quad + t_\text{cut}\Pi_c(q_\text{m}^2,q_\text{b},{t_\text{m}}/{t_\text{b}},-{t_\text{cut}}/{t_\text{b}},1,m^*) \bigg)~,
        \label{eq:finite_sum_A_b1^-A_b2^-}
    \end{align}
\endgroup
where, in (i), we recall that ${z^*(t_\text{m}/t_\text{b})\!=\!m^*}$.
Similarly, on $A_1^- A_2^-$, 
\begingroup
    \setlength{\abovedisplayskip}{6pt}   
    \setlength{\belowdisplayskip}{6pt}   
    \setlength{\abovedisplayshortskip}{2pt}
    \setlength{\belowdisplayshortskip}{2pt}
\begin{align}
        & ~\mathbb{E}\big(Z \mathds{1}_{A_1^- A_2^-}\big) \nonumber \\
        = & ~\mathbb{E}\big((X_{\text{min}}\!+\!t_\text{cut}) \mathds{1}_{A_{1}^-A_{2}^-}\big) \\
        = & ~\underbrace{\bigg(\frac{p_\text{m}^2 p_\text{b}}{(1\!-\!p_\text{m})^2(1\!-\!p_\text{b})}\bigg)}_{\delta} \sum_{m_\text{b}=1}^{\infty} (m_\text{b} t_\text{b} \!+\! t_\text{cut}) \big(\underbrace{1\!-\!p_\text{b}}_{q_\text{b}} \big)^{m_\text{b}} \!\!\!\sum_{m_1=\left\lceil (m_\text{b} t_\text{b} + t_\text{cut})/t_\text{m} \right\rceil}^{\infty} \!\!\! \big(\underbrace{1\!-\!p_\text{m}}_{q_\text{m}} \big)^{2m_1} \\
        = & ~\frac{\delta}{1\!-\!q_\text{m}^2}\sum_{m_\text{b}=1}^{\infty} (m_\text{b} t_\text{b} + t_\text{cut}) q_\text{b}^{m_\text{b}} q_\text{m}^{2 \left\lceil (m_\text{b} t_\text{b} + t_\text{cut})/t_\text{m} \right\rceil} 
    \end{align}
\endgroup
\begingroup
    \setlength{\abovedisplayskip}{6pt}   
    \setlength{\belowdisplayskip}{6pt}   
    \setlength{\abovedisplayshortskip}{2pt}
    \setlength{\belowdisplayshortskip}{2pt}
    \begin{align}
        = & ~\frac{\delta}{1\!-\!q_\text{m}^2} \Big(t_\text{b}\Theta_c\big(q_\text{b},q_\text{m}^2,\frac{t_\text{b}}{t_\text{m}},-\frac{t_\text{cut}}{t_\text{m}},1,\infty\big) + t_\text{cut}\Pi_c\big(q_\text{b},q_\text{m}^2,\frac{t_\text{b}}{t_\text{m}},-\frac{t_\text{cut}}{t_\text{m}},1,\infty\big) \Big) \qquad~ \\
        \overset{(\text{i})}{=} & ~\frac{\delta}{1\!-\!q_\text{m}^2} \bigg(t_\text{b} \Big( \frac{\Theta_c(q_\text{b},q_\text{m}^2,{t_\text{b}}/{t_\text{m}},-{t_\text{cut}}/{t_\text{m}},1,{m^*t_\text{m}}/{t_\text{b}})}{1\!-\!\big(q_\text{b} q_\text{m}^{2t_\text{b}/t_\text{m}}\big)^{m^* t_\text{m}/t_\text{b}}} \nonumber \\
        & \quad + \frac{m^*t_\text{m}}{t_\text{b}} \big( q_\text{b} q_\text{m}^{2t_\text{b}/t_\text{m}}\big)^{m^*t_\text{m}/t_\text{b}}  \frac{\Pi_c(q_\text{b},q_\text{m}^2,{t_\text{b}}/{t_\text{m}},-{t_\text{cut}}/{t_\text{m}},1,{m^*t_\text{m}}/{t_\text{b}})}{\big(1\!-\!\big(q_\text{b} q_\text{m}^{2t_\text{b}/t_\text{m}}\big)^{m^* t_\text{m}/t_\text{b}}\big)^2} \nonumber \\
        & \quad + t_\text{cut}\frac{\Pi_c(q_\text{b},q_\text{m}^2,{t_\text{b}}/{t_\text{m}},-{t_\text{cut}}/{t_\text{m}},1,{m^*t_\text{m}}/{t_\text{b}})}{1\!-\!\big(q_\text{b} q_\text{m}^{2t_\text{b}/t_\text{m}}\big)^{m^* t_\text{m}/t_\text{b}}} \bigg) ~,
        \label{eq:finite_sum_A_1^-A_2^-}    
    \end{align}
\endgroup
where, in (i), we recall that ${z^*(t_\text{b}/t_\text{m})\!=\!m^*t_\text{m}/t_\text{b}}$.
Finally, on $A_1^- A_b^-$, we have
\begin{align}
    & ~\mathbb{E}\big(Z \mathds{1}_{A_1^- A_b^-}\big) \nonumber \\
    = & ~\mathbb{E}\big((X_{\text{min}}\!+\!t_\text{cut}) \mathds{1}_{A_{1}^-A_{b}^-}\big) \\
    = & ~\underbrace{\bigg(\frac{p_\text{m}^2 p_\text{b}}{(1\!-\!p_\text{m})^2(1\!-\!p_\text{b})}\bigg)}_{\delta} \sum_{m_2=1}^{\infty} (m_2 t_\text{m} \!+\! t_\text{cut}) (\underbrace{1\!-\!p_\text{m}}_{q_\text{m}})^{m_2} \!\!\!\! \sum_{k=\left\lceil (m_2 t_\text{m} + t_\text{cut})/(m^* t_\text{m}) \right\rceil}^{\infty} \!\!\!\! {(\underbrace{1\!-\!p_\text{m}}_{q_\text{m}})}^{km^*} {(\underbrace{1\!-\!p_\text{b}}_{q_\text{b}})}^{km^*t_\text{m}/t_\text{b}} \\
    = & ~\delta \sum_{m_2=1}^{\infty} (m_2 t_\text{m} \!+\! t_\text{cut}) q_\text{m}^{m_2} \sum_{k=\left\lceil (m_2 t_\text{m} + t_\text{cut})/(m^* t_\text{m}) \right\rceil}^\infty {\big(\underbrace{q_\text{m}^{m^*}q_\text{b}^{m^*t_\text{m}/t_\text{b}}}_{\gamma}\big)}^{k}   \\
    = & ~\frac{\delta}{1\!-\!\gamma}\sum_{m_2=1}^{\infty} (m_2 t_\text{m} \!+\! t_\text{cut}) q_\text{m}^{m_2}\gamma^{\left\lceil (m_2 t_\text{m} + t_\text{cut})/(m^*t_\text{m}) \right\rceil}   \\
    = &~ \frac{\delta}{1\!-\!\gamma}\Big( t_\text{m} \Theta_c\big(q_\text{m},\gamma,\frac{1}{m^*},-\frac{t_\text{cut}}{m^*t_\text{m}},1,\infty\big) + t_\text{cut}\Pi_c\big(q_\text{m},\gamma,\frac{1}{m^*},-\frac{t_\text{cut}}{m^*t_\text{m}},1,\infty\big) \Big) \\
    \overset{(\text{i})}{=} &~ \frac{\delta}{1\!-\!\gamma} \bigg(t_\text{m}\Big(\frac{\Theta_c(q_\text{m},\gamma,{1}/{m^*},-{t_\text{cut}}/{(m^*t_\text{m})},1,m^*)}{1\!-\!\big(q_\text{m}\gamma^{1/m^*}\big)^{m^*}} \nonumber\\
    & ~ + m^* \big(q_\text{m}\gamma^{1/m^*}\big)^{m^*} \frac{\Pi_c(q_\text{m},\gamma,{1}/{m^*},-{t_\text{cut}}/{(m^*t_\text{m})},1,m^*)}{\big(1\!-\!\big(q_\text{m}\gamma^{1/m^*}\big)^{m^*}\big)^2} \Big) \!+\! t_\text{cut}\frac{\Pi_c(q_\text{m},\gamma,{1}/{m^*},-{t_\text{cut}}/{(m^*t_\text{m})},1,m^*)}{1\!-\!\big(q_\text{m}\gamma^{1/m^*}\big)^{m^*}} \bigg) ~,
    \label{eq:finite_sum_A_1^-A_b^-}
\end{align}
where in (i), we note that ${z^*(1/m^*) \!=\! m^*}$ and $\big(q_\text{m}\gamma^{1/m^*}\big)^{m^*} \!\!=\! \beta$.
By plugging~\eqref{eq:finite_sum_A_12^-}, \eqref{eq:finite_sum_A_1b^-}, \eqref{eq:finite_sum_A_b1^-}, \eqref{eq:finite_sum_A_12^-A_1b^-}, \eqref{eq:finite_sum_A_b1^-A_b2^-}, \eqref{eq:finite_sum_A_1^-A_2^-}, and \eqref{eq:finite_sum_A_1^-A_b^-} in~\eqref{eq:inclusion-exclusion_part_2}, we obtain the value of $\mathbb{E}(Z\mathds{1}_{Y=0})$.
Thus, by plugging these in~\eqref{eq:zy1},~\eqref{eq:inclusion-exclusion_part_2}, and~\eqref{eq:ent_gen_time_int} and subsequently in~\eqref{eq:def_rate_teleportation_int}, we obtain the teleportation rate in the IN.

\subsection{Closed-form Expressions for Relevant Infinite Sums} \label{subsec:InfiniteToFiniteSum}

In this appendix, we derive closed-form expressions for infinite sums  in \eqref{eq:Pi_c_finite},\eqref{eq:Pi_f_finite},\eqref{eq:Theta_c_finite},\eqref{eq:Theta_f_finite}, which facilitates exact computations of these quantities.
We first consider the following term with \( x, y \in [0,1) \), \( q \in \mathbb{Q} \), and \( \alpha \in \mathbb{R} \):
\begingroup
    \setlength{\abovedisplayskip}{0pt}   
    \setlength{\belowdisplayskip}{6pt}   
    \setlength{\abovedisplayshortskip}{2pt}
    \setlength{\belowdisplayshortskip}{2pt}
    \begin{align}\label{eq:recallPiC}
        \Pi_c(x,y,q,\alpha,l,\infty) = \sum_{i=l}^{\infty} x^i y^{\lceil iq - \alpha \rceil}~.
    \end{align}
\endgroup
We now define
\begingroup
    \setlength{\abovedisplayskip}{2pt}   
    \setlength{\belowdisplayskip}{6pt}   
    \setlength{\abovedisplayshortskip}{2pt}
    \setlength{\belowdisplayshortskip}{2pt}
    \begin{align}
        z^*(q) = \min\{z \in \mathbb{N} : zq \in \mathbb{N}\}~.
    \end{align}
\endgroup
\looseness=-1
For brevity, we often suppress the argument $q$ of $z^*$ when it is clear from the context.
Observing that
\begingroup
    \setlength{\abovedisplayskip}{2pt}   
    \setlength{\belowdisplayskip}{6pt}   
    \setlength{\abovedisplayshortskip}{2pt}
    \setlength{\belowdisplayshortskip}{2pt}
    \begin{align}\label{eq:recurrencePattern}
        \sum_{i = l+z^*}^{l+2z^*-1} x^i y^{\lceil iq - \alpha \rceil} = (xy^q)^{z^*} \sum_{i=l}^{l+z^*-1} x^i y^{\lceil iq - \alpha \rceil} = (xy^q)^{z^*} \Pi_c(x, y, q, \alpha, l, l\!+\!z^*\!-\!1)~,
    \end{align}
\endgroup
we can rewrite~\eqref{eq:recallPiC} as
\begingroup
    \setlength{\abovedisplayskip}{2pt}   
    \setlength{\belowdisplayskip}{6pt}   
    \setlength{\abovedisplayshortskip}{2pt}
    \setlength{\belowdisplayshortskip}{2pt}
    \begin{align}
        \Pi_c(x,y,q,\alpha,l,\infty) = \sum_{j=0}^\infty (xy^q)^{jz^*} \sum_{i=l}^{l+z^*-1} x^i y^{\lceil iq - \alpha \rceil} = \frac{\Pi_c(x,y,q,\alpha,l,l\!+\!z^*\!-\!1)}{1 - (xy^q)^{z^*}}~.
    \end{align}
\endgroup
\looseness = -1 Similarly, we obtain the closed-form expressions for $\Pi_{f}, \Pi_{cc}, \Pi_{cf}$ and $\Pi_{ff}$, introduced in~\eqref{eq:def_Pi}, ~\eqref{eq:def_Pi_cc, Pi_cf}, and~\eqref{eq:def_Pi_ff}.

Next, we consider
\begingroup
    \setlength{\abovedisplayskip}{-4pt}   
    \setlength{\belowdisplayskip}{6pt}   
    \setlength{\abovedisplayshortskip}{2pt}
    \setlength{\belowdisplayshortskip}{2pt}
    \begin{align} \label{eq:Theta_cInfiniteSum}
        \Theta_c(x,y,q,\alpha,l,\infty) = \sum_{i=l}^{\infty} i x^i y^{\lceil iq - \alpha \rceil}~.
    \end{align}
\endgroup
Similar to~\eqref{eq:recurrencePattern}, we observe that
\begingroup
    \setlength{\abovedisplayskip}{6pt}   
    \setlength{\belowdisplayskip}{6pt}   
    \setlength{\abovedisplayshortskip}{2pt}
    \setlength{\belowdisplayshortskip}{2pt}
    \begin{align}
        \sum_{i=l+z^*}^{l+2z^*-1} i x^i y^{\lceil iq - \alpha \rceil} = (xy^q)^{z^*} \big(\Theta_c(x,y,q,\alpha,l,l\!+\!z^*\!-\!1) + z^*\Pi_c(x,y,q,\alpha,l,l\!+\!z^*\!-\!1)\big).
    \end{align}
\endgroup
Therefore,
\begingroup
    \setlength{\abovedisplayskip}{6pt}   
    \setlength{\belowdisplayskip}{6pt}   
    \setlength{\abovedisplayshortskip}{2pt}
    \setlength{\belowdisplayshortskip}{2pt}
    \begin{align}
        \Theta_c(x,y,q,\alpha,1,\infty) = &~ \sum_{j=0}^\infty (xy^q)^{jz^*}\sum_{i=l}^{l+z^*-1} i x^i y^{\lceil iq-\alpha \rceil} + \sum_{j=1}^\infty jz^* (xy^q)^{jz^*}\sum_{i=l}^{l+z^*-1} x^i y^{\lceil iq-\alpha \rceil} \\
        = &~ \frac{\Theta_c(x,y,q,\alpha,l,l\!+\!z^*\!-\!1)}{1 - (xy^q)^{z^*}} + \frac{z^* (xy^q)^{z^*} \Pi_c(x,y,q,\alpha,l,l\!+\!z^*\!-\!1)}{(1 - (xy^q)^{z^*})^2}~.
    \end{align}
\endgroup
Similarly, we can obtain closed-form expressions for the rest of the terms in~\eqref{eq:def_Theta}, ~\eqref{eq:def_Theta_cc, Theta_cf}, and ~\eqref{eq:def_Theta_ff}.
Finally, for the following term, with $x,y,s\in [0,1)$, $r, q\in\mathbb{Q}$, and $\kappa \in \mathbb{R}$:
\begingroup
    \setlength{\abovedisplayskip}{6pt}   
    \setlength{\belowdisplayskip}{6pt}   
    \setlength{\abovedisplayshortskip}{2pt}
    \setlength{\belowdisplayshortskip}{2pt}
    \begin{align} \label{eq:delta_infity}
        \Delta(x,y,s,r,\kappa,q,l,\infty) = \sum_{i=l}^{\infty} x^i \sum_{j=\ceil{ir-\kappa}}^{\floor{ir}} y^j s^{\ceil{jq}} ~,
    \end{align}
\endgroup
we define
\begingroup
    \setlength{\abovedisplayskip}{-2pt}   
    \setlength{\belowdisplayskip}{6pt}   
    \setlength{\abovedisplayshortskip}{2pt}
    \setlength{\belowdisplayshortskip}{2pt}
    \begin{align}
        \Bar{z}(r,q) = \min \big\{z\in\mathbb{N} : zr, zrq \in \mathbb{N} \big\}
    \end{align}
\endgroup
Suppressing the arguments $q$ and $r$ of $\Bar{z}$ for brevity, we observe that
\begingroup
    \setlength{\abovedisplayskip}{6pt}   
    \setlength{\belowdisplayskip}{2pt}   
    \setlength{\abovedisplayshortskip}{2pt}
    \setlength{\belowdisplayshortskip}{2pt}
    \begin{align}
        \sum_{i=l+\Bar{z}}^{l+2\Bar{z}-1} x^i \sum_{j=\ceil{ir-\kappa}}^{\floor{ir}} y^j s^{\ceil{jq}} & = (x y^r s^{rq})^{\Bar{z}} \sum_{i=l}^{l+\Bar{z}-1} x^i \sum_{j=\ceil{ir-\kappa}}^{\floor{ir}} y^j s^{\ceil{jq}} = (x y^r s^{rq})^{\Bar{z}} \Delta(x,y,s,r,\kappa,q,l,l\!+\!\Bar{z}\!-\!1)~.
    \end{align}
\endgroup
Thus, ~\eqref{eq:delta_infity} can be expressed as 
\begingroup
    \setlength{\abovedisplayskip}{6pt}   
    \setlength{\belowdisplayskip}{6pt}   
    \setlength{\abovedisplayshortskip}{2pt}
    \setlength{\belowdisplayshortskip}{2pt}
    \begin{align}
        \Delta(x,y,s,r,\kappa,q,l,\infty) = \sum_{n=1}^{\infty} (x y^r s^{rq})^{n\Bar{z}} \sum_{i=l}^{l+\Bar{z}-1} x^i \sum_{j=\ceil{ir-\kappa}}^{\floor{ir}} y^j s^{\ceil{jq}} = \frac{\Delta(x,y,s,r,\kappa,q,l,l\!+\!\Bar{z}\!-\!1)}{1\!-\!(x y^r s^{rq})^{\Bar{z}}} ~.
    \end{align}
\endgroup

\section{Calculation of Individual Terms in the Expression of Expected Teleportation Fidelity in the Intercity Network}
\label{sec:cutOff_fidelity}

\subsection{Calculations for individual terms in \texorpdfstring{$U(v,\alpha)$}{U(v,alpha)}}
\label{subsec:terms_U_alpha}

Here, we calculate the individual terms of~\eqref{eq:def_U_alpha}, which facilitates the calculation of $U_1(v)$ and $U_2(v)$ via~\eqref{eq:U_1_andU_2_to_U_alpha}.
On $A_{12}^+$
\begingroup
    \setlength{\abovedisplayskip}{6pt}   
    \setlength{\belowdisplayskip}{6pt}   
    \setlength{\abovedisplayshortskip}{2pt}
    \setlength{\belowdisplayshortskip}{2pt}
    \begin{align}
        & ~\mathbb{E}(e^{-v((\alpha-1)X_1 - X_2)} \mathds{1}_{A_{12}^+}) \nonumber \\
        =& \underbrace{\bigg(\frac{p_\text{m}^2 p_\text{b}}{(1\!-\!p_\text{m})^2(1\!-\!p_\text{b})}\bigg)}_{\delta} \sum_{m_1=\left\lceil t_\text{b} / t_\text{m}\right\rceil}^{\infty} \sum_{\substack{m_2=\text{max}(1,\\ \left\lceil m_1-t_\text{cut}^{\prime}/t_\text{m}\right\rceil)}}^{m_1} \sum_{m_\text{b}=\left\lceil m_2 t_\text{m}/t_\text{b}\right\rceil}^{\left\lfloor m_1 t_\text{m}/t_\text{b}\right\rfloor} \!\!\!\!\!\! e^{-v(\alpha-1) m_1t_\text{m} + vm_2t_\text{m}} (\underbrace{1\!-\!p_\text{m}}_{q_\text{m}})^{m_1} (\underbrace{1\!-\!p_\text{m}}_{q_\text{m}})^{m_2} (\underbrace{1\!-\!p_\text{b}}_{q_\text{b}})^{m_\text{b}}   \\
        =&~ \delta \sum_{m_1=\left\lceil t_\text{b} / t_\text{m}\right\rceil}^{\infty}\big(\underbrace{e^{-v(\alpha-1) t_\text{m}} q_\text{m}}_{=:\thinspace d_s}\big)^{m_1} \!\!\!\!\!\sum_{m_2=\text{max}(1, \left\lceil m_1-t_\text{cut}^{\prime}/t_\text{m}\right\rceil)}^{m_1} \!\!\!\big(\underbrace{e^{v t_\text{m}} q_\text{m}}_{=:\thinspace d_m}\big)^{m_2} \sum_{m_\text{b}=\left\lceil m_2 t_\text{m}/t_\text{b}\right\rceil}^{\left\lfloor m_1 t_\text{m}/t_\text{b}\right\rfloor} q_\text{b}^{m_\text{b}} \\
        =&~ \frac{\delta}{1\!-\!q_\text{b}} \sum_{m_1=\left\lceil t_\text{b}/t_\text{m}\right\rceil}^{\infty} \!\! d_s^{m_1} \!\!\!\!\!\! \sum_{m_2=\max \left(1,\left\lceil m_1-t_\text{cut}^\prime / t_\text{m}\right\rceil\right)}^{m_1} \!\!\!\!\!\!\!\!\! d_m^{m_2} \big( q_\text{b}^{\ceil{m_2 t_\text{m}/t_\text{b}}} - q_\text{b}^{\floor{m_1t_\text{m}/t_\text{b}}+1} \big) \\
        =&~ \frac{\delta}{p_\text{b}} \sum_{m_1=\left\lceil t_\text{b}/t_\text{m}\right\rceil}^{\infty} \!\! d_s^{m_1} \!\!\!\!\!\! \sum_{m_2=\max \left(1,\left\lceil m_1-t_\text{cut}^\prime / t_\text{m}\right\rceil\right)}^{m_1} \!\!\!\!\!\!\!\!\!\!\!\!\!\! d_m^{m_2} q_\text{b}^{\ceil{m_2 t_\text{m}/t_\text{b}}} - \frac{\delta q_\text{b}}{p_\text{b}(1\!-\!d_m)} \!\!\sum_{m_1=\left\lceil t_\text{b}/t_\text{m}\right\rceil}^\infty \!\!\!\!\! d_s^{m_1} q_\text{b}^{\left\lfloor m_1 t_\text{m}/t_\text{b}\right\rfloor}
        (d_m^{\max(1,\ceil{m_1-t_\text{cut}^\prime/t_\text{m}})} \!-\! d_m^{m_1+1}) \\
        =&~ \frac{\delta}{p_\text{b}} \bigg( \underbrace{\sum_{m_1=\left\lceil t_\text{b}/t_\text{m}\right\rceil}^{1+\floor{t_\text{cut}^\prime/t_\text{m}}} \!\! d_s^{m_1} \sum_{m_2=1}^{m_1} d_m^{m_2} q_\text{b}^{\ceil{m_2 t_\text{m}/t_\text{b}}}}_{=:\thinspace S_{{12}^+}^{(1)}} + \sum_{m_1=2+\floor{t_\text{cut}^\prime/t_\text{m}}}^{\infty} \!\!\! d_s^{m_1} \!\!\! \sum_{m_2=\ceil{m_1-t_\text{cut}^\prime/t_\text{m}}}^{m_1} \!\!\! d_m^{m_2} q_\text{b}^{\ceil{m_2 t_\text{m}/t_\text{b}}} \bigg)   \nonumber \\
        & \quad - \frac{\delta q_\text{b}}{p_\text{b}(1\!-\!d_m)} \bigg(d_m\underbrace{\!\!\!\!\sum_{m_1=\ceil{t_\text{b}/t_\text{m}}}^{1+\floor{t_\text{cut}^\prime/t_\text{m}}} \!\!\! d_s^{m_1} q_\text{b}^{\floor{m_1t_\text{m}/t_\text{b}}}} _{=:\thinspace S_{{12}^+}^{(2)}} + \!\! \sum_{m_1=2+\floor{t_\text{cut}^\prime/t_\text{m}}}^\infty \!\!\!\!\!\!\!d_s^{m_1} q_\text{b}^{\floor{m_1t_\text{m}/t_\text{b}}} d_m^{\ceil{m_1-t_\text{cut}^\prime/t_\text{m}}} \nonumber \\
        & \quad - d_m\!\!\! \sum_{m_1=\ceil{t_\text{b}/t_\text{m}}}^\infty \!\!\! (d_s d_m)^{m_1} q_\text{b}^{\floor{m_1 t_\text{m}/t_\text{b}}} \bigg) \\
        =&~ \frac{\delta}{p_\text{b}} \Big(S_{{12}^+}^{(1)} - \frac{q_\text{b} d_m}{1\!-\!d_m} S_{{12}^+}^{(2)} + \sum_{m_1=2+\floor{t_\text{cut}^\prime/t_\text{m}}}^{\infty} d_s^{m_1} \!\!\!\!\!\!\sum_{m_2=\ceil{m_1-t_\text{cut}^\prime/t_\text{m}}}^{m_1} \!\!\!\!\!\!d_m^{m_2} q_\text{b}^{\ceil{m_2 t_\text{m}/t_\text{b}}} \nonumber \\
        & \quad - \frac{ q_\text{b} d_m^{\thinspace -\floor{t_\text{cut}^\prime/t_\text{m}}}}{1\!-\!d_m} \sum_{m_1=2+\floor{t_\text{cut}^\prime/t_\text{m}}}^{\infty} \!\! (d_s d_m)^{m_1} q_\text{b}^{\floor{m_1t_\text{m}/t_\text{b}}} + \frac{ q_\text{b} d_m}{1\!-\!d_m} \sum_{m_1=\ceil{t_\text{b}/t_\text{m}}}^{\infty} \!\! (d_s d_m)^{m_1} q_\text{b}^{\floor{m_1 t_\text{m}/t_\text{b}}} \Big) \\
        =&~ \frac{\delta}{p_\text{b}} \Big(S_{{12}^+}^{(1)} - \frac{q_\text{b} d_m}{1\!-\!d_m} S_{{12}^+}^{(2)} + \Delta\big(d_s,d_m,q_\text{b},1,\frac{t_\text{cut}^\prime}{t_\text{m}},\frac{t_\text{m}}{t_\text{b}}, 2\!+\!\floor{\frac{t_\text{cut}^\prime}{t_\text{m}}}\!,\infty \big) \nonumber \\
        & \quad -\! \frac{q_\text{b} d_m^{\thinspace -\floor{t_\text{cut}^\prime/t_\text{m}}}}{1\!-\!d_m}\Pi_f\big(d_sd_m,q_\text{b},\frac{t_\text{m}}{t_\text{b}},0, 2\!+\!\floor{\frac{t_\text{cut}^\prime}{t_\text{m}}}\!,\infty\big) + \frac{q_\text{b}d_m}{1\!-\!d_m} \Pi_f\big(d_sd_m,q_\text{b},\frac{t_\text{m}}{t_\text{b}},0, \ceil{\frac{t_\text{b}}{t_\text{m}}}\!,\infty \big) \! \Big) 
    \end{align}
\endgroup
\begin{align}
    \overset{(\text{i})}{=} &~ \frac{\delta}{p_\text{b}} \Big(S_{{12}^+}^{(1)} - \frac{q_\text{b} d_m}{1\!-\!d_m} S_{{12}^+}^{(2)} + 
    \frac{1}{1\!-\!(d_sd_mq_\text{b}^{t_\text{m}/t_\text{b}})^{m^*} } \Delta\big(d_s,d_m,q_\text{b},1,\frac{t_\text{cut}^\prime}{t_\text{m}},\frac{t_\text{m}}{t_\text{b}}, 2\!+\!\floor{\frac{t_\text{cut}^\prime}{t_\text{m}}}\!, 1\!+\!\floor{\frac{t_\text{cut}^\prime}{t_\text{m}}} \!+\! m^* \big) \nonumber \\
    & \quad -\! \frac{q_\text{b} d_m^{\thinspace -\floor{t_\text{cut}^\prime/t_\text{m}}}}{1\!-\!d_m} \frac{1}{1\!-\!(d_sd_mq_\text{b}^{t_\text{m}/t_\text{b}})^{m^*}} \Pi_f\big(d_sd_m,q_\text{b},\frac{t_\text{m}}{t_\text{b}},0, 2\!+\!\floor{\frac{t_\text{cut}^\prime}{t_\text{m}}}\!, 1 \!+\! \floor{\frac{t_\text{cut}^\prime}{t_\text{m}}} \!+\! m^*\big) \nonumber \\
    & \quad + \frac{q_\text{b}d_m}{1\!-\!d_m} \frac{1}{1\!-\!(d_sd_mq_\text{b}^{t_\text{m}/t_\text{b}})^{m^*} } \Pi_f\big(d_sd_m,q_\text{b},\frac{t_\text{m}}{t_\text{b}},0, \ceil{\frac{t_\text{b}}{t_\text{m}}}\!, \ceil{\frac{t_\text{b}}{t_\text{m}}} \!+\! m^* \!-\! 1 \big)\! \Big) ~,
    \label{eq:indiv_term_U_alpha_A_12+}
\end{align}
where we used the definitions of $\Delta$ and $\Pi_f$ from~\eqref{eq:def_Delta} and~\eqref{eq:def_Pi}.
In (i), we used ${\Bar{z}(1,t_\text{m}/t_\text{b}) \!=\! m^*}$ and ${z^*(t_\text{m}/t_\text{b}) \!=\! m^*}$.

On $A_{{1b}^+}$, we have
\begingroup
    \setlength{\abovedisplayskip}{6pt}   
    \setlength{\belowdisplayskip}{6pt}   
    \setlength{\abovedisplayshortskip}{2pt}
    \setlength{\belowdisplayshortskip}{2pt}
    \begin{align}
        & ~\mathbb{E}(e^{-v((\alpha-1)X_1 - X_b)} \mathds{1}_{A_{1b}^+}) \nonumber \\
        =& \underbrace{\bigg(\frac{p_\text{m}^2 p_\text{b}}{(1\!-\!p_\text{m})^2(1\!-\!p_\text{b})}\bigg)}_{\delta} \sum_{m_1=\left\lceil t_\text{b}/t_\text{m}\right\rceil}^{\infty} \sum_{\substack{m_\text{b}=\text{max}(1,\\ \left\lceil(m_1 t_\text{m}-t_\text{cut}^{\prime}) / t_\text{b}\right\rceil)}}^{\left\lfloor m_1 t_\text{m} / t_\text{b}\right\rfloor}\sum_{m_2=\left\lceil m_\text{b} t_\text{b}/t_\text{m}\right\rceil}^{m_1} \!\!\!\!\!\!\!\! e^{-v(\alpha-1) m_1t_\text{m} + vm_\text{b}t_\text{b}} q_\text{m} (\underbrace{1\!-\!p_\text{m}}_{q_\text{m}})^{m_1} (\underbrace{1\!-\!p_\text{m}}_{q_\text{m}})^{m_2} (\underbrace{1\!-\!p_\text{b}}_{q_\text{b}})^{m_\text{b}} \\
        =&~ \delta \!\! \sum_{m_1=\left\lceil t_\text{b}/t_\text{m}\right\rceil}^{\infty} \!\!\big(\underbrace{e^{-v(\alpha-1) t_\text{m}} q_\text{m}}_{d_s}\big)^{m_1} \sum_{\substack{m_\text{b}=\text{max}(1, \\\left\lceil(m_1 t_\text{m}-t_\text{cut}^{\prime}) / t_\text{b}\right\rceil)}}^{\left\lfloor m_1 t_\text{m} / t_\text{b}\right\rfloor} \!\! \big(\underbrace{e^{v t_\text{b}} q_\text{b}}_{=:\thinspace d_b}\big)^{m_\text{b}} \sum_{m_2=\left\lceil m_\text{b} t_\text{b}/t_\text{m}\right\rceil}^{m_1} q_\text{m}^{m_2} \\
        =&~ \frac{\delta}{1\!-\!q_\text{m}} \sum_{m_1=\left\lceil t_\text{b} / t_\text{m}\right\rceil}^{\infty} d_s^{m_1} \sum_{\substack{m_\text{b}=\max (1, \\\left\lceil(m_1 t_\text{m}-t_\text{cut}^{\prime}) / t_\text{b}\right\rceil)}}^{\left\lfloor m_1 t_\text{m} / t_\text{b}\right\rfloor} d_b^{m_\text{b}} \big(q_\text{m}^{\left\lceil m_\text{b} t_\text{b} / t_\text{m}\right\rceil} - q_\text{m}^{m_1+1}  \big)  \\
        =&~ \frac{\delta}{p_\text{m}} \sum_{m_1=\left\lceil t_\text{b} / t_\text{m}\right\rceil}^{\infty} d_s^{m_1} \sum_{\substack{m_\text{b}=\max (1,\\ \left\lceil(m_1 t_\text{m}-t_\text{cut}^{\prime}) / t_\text{b}\right\rceil)}}^{\left\lfloor m_1 t_\text{m} / t_\text{b}\right\rfloor} \!\!\!\!\!\! d_b^{m_\text{b}} q_\text{m}^{\left\lceil m_\text{b} t_\text{b} / t_\text{m}\right\rceil}  \nonumber \\
        & \quad - \frac{\delta q_\text{m}}{p_\text{m}(1\!-\!d_b)} \sum_{m_1=\left\lceil t_\text{b} / t_\text{m}\right\rceil}^{\infty} \!\!\!\!\!\! (d_s q_\text{m})^{m_1}\Big(d_b^{~\max (1, \ceil{(m_1 t_\text{m}-t_\text{cut}^{\prime}) / t_\text{b}})} - d_b^{\left\lfloor m_1 t_\text{m} / t_\text{b}\right\rfloor + 1}\Big) \\
        =&~ \frac{\delta}{p_\text{m}} \Big( \underbrace{\sum_{m_1=\left\lceil t_\text{b} / t_\text{m}\right\rceil}^{\floor{(t_\text{b}+t_\text{cut}^\prime)/t_\text{m}}} d_s^{m_1} \sum_{m_\text{b}=1}^{\floor{m_1t_\text{m}/t_\text{b}}} d_b^{m_\text{b}} q_\text{m}^{\ceil{m_\text{b}t_\text{b}/t_\text{m}}}}_{=:\thinspace S_{{1b}^+} } + \sum_{m_1=1+\floor{(t_\text{b}+t_\text{cut}^\prime)/t_\text{m}}}^\infty \!\!\! d_s^{m_1} \!\!\! \sum_{m_\text{b}=\ceil{(m_1t_\text{m}-t_\text{cut}^\prime)/t_\text{b}}}^{\floor{m_1t_\text{m}/t_\text{b}}} \!\!\!\!\!\! d_b^{m_\text{b}} q_\text{m}^{\ceil{m_\text{b}t_\text{b}/t_\text{m}}} \Big) \nonumber \\
        & \quad - \frac{\delta q_\text{m}}{p_\text{m}(1\!-\!d_b)} \Big(d_b \sum_{m_1=\ceil{t_\text{b}/t_\text{m}}}^{\floor{(t_\text{b}+t_\text{cut}^\prime)/t_\text{m}}} \!\!\!(d_sq_\text{m})^{m_1} + \sum_{m_1=1+\floor{(t_\text{b}+t_\text{cut}^\prime)/t_\text{m}}}^{\infty} \!\!\!\!\!\! (d_sq_\text{m})^{m_1} d_b^{\ceil{(m_1t_\text{m}-t_\text{cut}^\prime)/t_\text{b}}} \nonumber \\
        & \quad - d_b \!\!\! \sum_{m_1=\ceil{t_\text{b}/t_\text{m}}}^\infty \!\!\!\!\! (d_sq_\text{m})^{m_1} d_b^{\floor{m_1t_\text{m}/t_\text{b}}} \Big)  \\
        =&~ \frac{\delta}{p_\text{m}}\Big( S_{{1b}^+} + \Delta\big(d_s,d_b,q_\text{m},\frac{t_\text{m}}{t_\text{b}},\frac{t_\text{cut}^\prime}{t_\text{b}},\frac{t_\text{b}}{t_\text{m}}, 1\!+\!\floor{\frac{t_\text{b}\!+\!t_\text{cut}^\prime}{t_\text{m}}}\!,\infty \big) \Big) \nonumber \\
        & \quad - \frac{\delta q_\text{m}}{p_\text{m}(1\!-\!d_b)} \bigg(d_b\frac{(d_sq_\text{m})^{\ceil{t_\text{b}/t_\text{m}}} - (d_sq_\text{m})^{\floor{(t_\text{b}+t_\text{cut}^\prime)/t_\text{m}}+1}}{1\!-\!d_sq_\text{m}} \nonumber \\
        & \quad + \Pi_c\big(d_sq_\text{m},d_b,\frac{t_\text{m}}{t_\text{b}},\frac{t_\text{cut}^\prime}{t_\text{b}}, 1\!+\!\floor{\frac{t_\text{b}\!+\!t_\text{cut}^\prime}{t_\text{m}}}\!,\infty \big) - d_b \Pi_f\big(d_sq_\text{m},d_b,\frac{t_\text{m}}{t_\text{b}},0,\ceil{\frac{t_\text{b}}{t_\text{m}}},\infty \big) \bigg)
    \end{align}
\endgroup
\begin{align}
    \overset{(\text{i})}{=} &~ \frac{\delta}{p_\text{m}}\Big( S_{{1b}^+} + \frac{1}{1\!-\!\big(d_s d_b^{\thinspace t_\text{m}/t_\text{b}} q_\text{m}\big)^{m^*}} \Delta\big(d_s,d_b,q_\text{m},\frac{t_\text{m}}{t_\text{b}},\frac{t_\text{cut}^\prime}{t_\text{b}},\frac{t_\text{b}}{t_\text{m}}, 1\!+\!\floor{\frac{t_\text{b}\!+\!t_\text{cut}^\prime}{t_\text{m}}}\!, \floor{\frac{t_\text{b}\!+\!t_\text{cut}^\prime}{t_\text{m}}} \!+\! m^* \big) \Big) \nonumber \\
    & \quad - \frac{\delta q_\text{m}}{p_\text{m}(1\!-\!d_b)} \bigg(d_b\frac{(d_sq_\text{m})^{\ceil{t_\text{b}/t_\text{m}}} - (d_sq_\text{m})^{\floor{(t_\text{b}+t_\text{cut}^\prime)/t_\text{m}}+1}}{1\!-\!d_sq_\text{m}} \nonumber \\
    & \quad + \frac{1}{1\!-\!\big(d_s q_\text{m} d_b^{\thinspace t_\text{m}/t_\text{b}} \big)^{m^*}} \Pi_c\big(d_sq_\text{m},d_b,\frac{t_\text{m}}{t_\text{b}},\frac{t_\text{cut}^\prime}{t_\text{b}}, 1\!+\!\floor{\frac{t_\text{b}\!+\!t_\text{cut}^\prime}{t_\text{m}}}\!,\floor{\frac{t_\text{b}\!+\!t_\text{cut}^\prime}{t_\text{m}}} \!+\! m^* \big) \nonumber \\
    & \quad - d_b \frac{1}{1\!-\!\big(d_s q_\text{m} d_b^{\thinspace t_\text{m}/t_\text{b}} \big)^{m^*}}\Pi_f\big(d_sq_\text{m},d_b,\frac{t_\text{m}}{t_\text{b}},0,\ceil{\frac{t_\text{b}}{t_\text{m}}}\!,\ceil{\frac{t_\text{b}}{t_\text{m}}} \!+\! m^* \!-\! 1 \big) \bigg) ~,
    \label{eq:indiv_term_U_alpha_A_1b+}
\end{align}
where, in (i), we used ${\Bar{z}(t_\text{m}/t_\text{b},t_\text{b}/t_\text{m}) \!=\! m^*}$ and ${z^*(t_\text{m}/t_\text{b}) \!=\! m^*}$.

On $A_{12}^+A_{1b}^+$, we have
\begingroup
    \setlength{\abovedisplayskip}{6pt}   
    \setlength{\belowdisplayskip}{10pt}   
    \setlength{\abovedisplayshortskip}{2pt}
    \setlength{\belowdisplayshortskip}{2pt}
    \begin{align}
        & ~\mathbb{E}(e^{-v((\alpha-1)X_1 - X_\text{min})} \mathds{1}_{A_{12}^+A_{1b}^+}) \nonumber \\
        =& \underbrace{\bigg(\frac{p_\text{m}^2 p_\text{b}}{(1\!-\!p_\text{m})^2(1\!-\!p_\text{b})}\bigg)}_{\delta} \sum_{k=1}^{\infty} \sum_{m_1=k m^*}^{k m^*+\left\lfloor t_\text{cut}^{\prime}/t_\text{m}\right\rfloor} \!\!\!\!\!\! e^{-v(\alpha-1) m_1t_\text{m} + v km^*t_\text{m}} (\underbrace{1\!-\!p_\text{m}}_{q_\text{m}})^{m_1} (\underbrace{1\!-\!p_\text{m}}_{q_\text{m}})^{km^*} (\underbrace{1\!-\!p_\text{b}}_{q_\text{b}})^{km^*t_\text{m}/t_\text{b}} \\
        =&~ \delta \sum_{k=1}^{\infty}\big(\underbrace{e^{v t_\text{m}} q_\text{m} q_\text{b}^{t_\text{m}/t_\text{b}}}_{=:\thinspace d_{ms}}\big)^{k m^*} \sum_{m_1=k m^*}^{k m^*+\left\lfloor t_\text{cut}^{\prime}/t_\text{m}\right\rfloor}\big(\underbrace{e^{-v(\alpha-1) t_\text{m}} q_\text{m}}_{d_s}\big)^{m_1} \\
        =&~ \frac{\delta}{1\!-\!d_s}  \sum_{k=1}^{\infty} d_{ms}^{~k m^*} \big(d_s^{~km^*} \!-\! d_s^{~k m^* + \floor{t_\text{cut}^\prime/t_\text{m}}+1}\big)  \\
        =&~ \frac{\delta \big(1 \!-\! d_s^{\left\lfloor t_\text{cut}'/t_\text{m}\right\rfloor+1}\big)}{1\!-\!d_s}  \sum_{k=1}^{\infty} (d_{ms} d_s)^{k m^*}  \\
        =&~ \frac{\delta \big(1 \!-\! d_s^{\left\lfloor t_\text{cut}'/t_\text{m}\right\rfloor+1}\big) }{ 1\!-\!d_s } \frac{(d_{ms} d_s)^{m^*}}{ 1\!-\!(d_{ms} d_s)^{m^*}} ~.
        \label{eq:indiv_term_U_alpha_A_12+A_1b+}
    \end{align}
\endgroup

On $A_b^+$, we have
\begin{align}
    & ~\mathbb{E}(e^{-v(\alpha X_b - X_1 - X_2)} \mathds{1}_{A_{b}^+}) \nonumber \\ 
    =& \underbrace{\bigg(\frac{p_\text{m}^2 p_\text{b}}{(1\!-\!p_\text{m})^2(1\!-\!p_\text{b})}\bigg)}_{\delta} \sum_{m_\text{b}=1}^\infty \!\!\sum_{\substack{m_1=\max(1,\\ \ceil{(m_\text{b}t_\text{b}-t_\text{cut}^\prime)/t_\text{m}})}}^{\floor{m_\text{b}t_\text{b}/t_\text{m}}} \sum_{\substack{m_2=\max(1, \\ \ceil{(m_\text{b}t_\text{b}-t_\text{cut}^\prime)/t_\text{m}}}}^{\floor{m_\text{b}t_\text{b}/t_\text{m}}} \!\!\!\! e^{-v\alpha m_\text{b}t_\text{b} + v m_1t_\text{m} + vm_2t_\text{m}} (\underbrace{1\!-\!p_\text{m}}_{q_\text{m}})^{m_1} (\underbrace{1\!-\!p_\text{m}}_{q_\text{m}})^{m_2} (\underbrace{1\!-\!p_\text{b}}_{q_\text{b}})^{m_\text{b}} \\
    = &~\delta \sum_{m_\text{b}=1}^{\infty}\big(\underbrace{e^{-v \alpha t_\text{b}} q_\text{b}}_{=:\thinspace{d_{\alpha b}}}\big)^{m_\text{b}}\Big(\sum_{m_1=\text{max}(1, \left\lceil(m_\text{b} t_\text{b}-t_\text{cut}^{\prime}) / t_\text{m}\right\rceil)}^{\left\lfloor m_\text{b} t_\text{b} / t_\text{m}\right\rfloor}\big(\underbrace{e^{v t_\text{m}} q_\text{m}}_{=:\thinspace d_m}\big)^{m_1}\Big)^2 \\
    =&~ \frac{\delta}{(1\!-\!d_m)^2} \sum_{m_\text{b}=1}^\infty d_{\alpha b}^{\thinspace m_\text{b}} \big( d_m^{\thinspace 2\max(1,\ceil{(m_\text{b}t_\text{b}-t_\text{cut}^\prime)/t_\text{m}})} - 2d_m^{\thinspace \max(1,\ceil{(m_\text{b}t_\text{b}-t_\text{cut}^\prime)/t_\text{m}})+\floor{m_\text{b}t_\text{b}/t_\text{m}}+1} + d_m^{\thinspace 2\floor{m_\text{b}t_\text{b}/t_\text{m}}+2} \big)
\end{align}
\begin{align}
    =&~ \frac{\delta}{(1\!-\!d_m)^2} \bigg( \underbrace{\sum_{m_\text{b}=1}^{\floor{(t_\text{m}+t_\text{cut}^\prime)/t_\text{b}}} \!\!\! d_{\alpha b}^{\thinspace m_\text{b}}\big(d_m^{\thinspace 2} - 2 d_m^{\thinspace\floor{m_\text{b}t_\text{b}/t_\text{m}}+2} \big)}_{=:\thinspace S_{b^+} }  \!\!\!\! \nonumber \\
    & \quad + \!\!\!\!\!\! \sum_{m_\text{b}=1+\floor{(t_\text{m}+t_\text{cut}^\prime)/t_\text{b}}}^\infty\!\!\!\!\!\!\!\!\!\!\!\! d_{\alpha b}^{\thinspace m_\text{b}} \big(d_m^{\thinspace 2\ceil{(m_\text{b}t_\text{b}-t_\text{cut}^\prime)/t_\text{m}}} \!-\! 2d_m^{\thinspace \ceil{(m_\text{b}t_\text{b}-t_\text{cut}^\prime)/t_\text{m}}+\floor{m_\text{b}t_\text{b}/t_\text{m}}+1} \big) \!+\!\! \sum_{m_\text{b}=1}^\infty d_{\alpha b}^{\thinspace m_\text{b}} d_m^{\thinspace 2\floor{m_\text{b}t_\text{b}/t_\text{m}}+2}  \bigg) \qquad  \\
    =&~ \frac{\delta}{(1\!-\!d_m)^2} \bigg( S_{b^+} + \!\!\!\!\!\! \sum_{m_\text{b}=1+\floor{(t_\text{m}+t_\text{cut}^\prime)/t_\text{b}}}^\infty \!\!\!\!\!\!\!\!\!\!\!\! d_{\alpha b}^{\thinspace m_\text{b}} d_m^{\thinspace 2\ceil{(m_\text{b}t_\text{b}-t_\text{cut}^\prime)/t_\text{m}}} - 2 d_m \!\!\!\!\!\!\!\!\!\!\!\! \sum_{m_\text{b}=1+\floor{(t_\text{m}+t_\text{cut}^\prime)/t_\text{b}}}^\infty \!\!\!\!\!\!\!\!\!\!\!\! d_{\alpha b}^{\thinspace m_\text{b}} d_m^{\thinspace \ceil{(m_\text{b}t_\text{b}-t_\text{cut}^\prime)/t_\text{m}}+\floor{m_\text{b}t_\text{b}/t_\text{m}}} \big) \nonumber \\
    & \quad + d_m^{\thinspace 2}\sum_{m_\text{b}=1}^\infty d_{\alpha b}^{\thinspace m_\text{b}} d_m^{\thinspace 2\floor{m_\text{b}t_\text{b}/t_\text{m}}} \bigg) 
\end{align}
\begin{align}
    =&~ \frac{\delta}{(1\!-\!d_m)^2} \Big( S_{b^+} + \Pi_c\big(d_{\alpha b},d_m^{\thinspace 2},\frac{t_\text{b}}{t_\text{m}},\frac{t_\text{cut}^\prime}{t_\text{m}}, 1\!+\!\floor{\frac{t_\text{m}\!+\!t_\text{cut}^\prime}{t_\text{b}}}\!,\infty \big) \nonumber \\
    & \quad - 2d_m \Pi_{cf}\big(d_{\alpha b},d_m,\frac{t_\text{b}}{t_\text{m}},\frac{t_\text{cut}^\prime}{t_\text{m}},0,1\!+\!\floor{\frac{t_\text{m}\!+\!t_\text{cut}^\prime}{t_\text{b}}}\!,\infty\big) + d_m^{\thinspace 2} \Pi_f\big(d_{\alpha b},d_m^{\thinspace 2},\frac{t_\text{b}}{t_\text{m}},0,1,\infty \big) \Big) \\
    \overset{(\text{i})}{=} &~ \frac{\delta}{(1\!-\!d_m)^2} \Big( S_{b^+}  + \frac{1}{ 1\!-\!\big(d_{\alpha b} d_m^{\thinspace 2t_\text{b}/t_\text{m}}\big)^{m^* t_\text{m}/t_\text{b}} } \Pi_c\big(d_{\alpha b},d_m^{\thinspace2},\frac{t_\text{b}}{t_\text{m}},\frac{t_\text{cut}^\prime}{t_\text{m}}, 1\!+\!\floor{\frac{t_\text{m}\!+\!t_\text{cut}^\prime}{t_\text{b}}}\!, \floor{\frac{t_\text{m}\!+\!t_\text{cut}^\prime}{t_\text{b}}} \!+\! \frac{m^*t_\text{m}}{t_\text{b}} \big) \nonumber \\
    & \quad - 2d_m \frac{1}{ 1\!-\!\big(d_{\alpha b} d_m^{\thinspace 2t_\text{b}/t_\text{m}}\big)^{m^* t_\text{m}/t_\text{b}} } \Pi_{cf}\big(d_{\alpha b},d_m,\frac{t_\text{b}}{t_\text{m}},\frac{t_\text{cut}^\prime}{t_\text{m}},0, 1\!+\!\floor{\frac{t_\text{m}\!+\!t_\text{cut}^\prime}{t_\text{b}}}\!, \floor{\frac{t_\text{m}\!+\!t_\text{cut}^\prime}{t_\text{b}}} \!+\!\frac{m^*t_\text{m}}{t_\text{b}}\big) \nonumber \\
    & \quad + d_m^{\thinspace 2} \frac{1}{ 1\!-\!\big(d_{\alpha b} d_m^{\thinspace 2t_\text{b}/t_\text{m}}\big)^{m^* t_\text{m}/t_\text{b}} } \Pi_f\big(d_{\alpha b},d_m^{\thinspace 2},\frac{t_\text{b}}{t_\text{m}},0,1,\frac{m^*t_\text{m}}{t_\text{b}} \big) \Big) ~,
    \label{eq:indiv_term_U_alpha_A_b+}
\end{align}
where we used ${z^*(t_\text{b}/t_\text{m}) \!=\! m^*t_\text{m}/t_\text{b}}$ in (i).

On $A_1^+ A_2^+$, we have
\begin{align}
    &~\mathbb{E}(e^{-v((\alpha-1) X_1 - X_b)} \mathds{1}_{A_{1}^+ A_{2}^+}) \nonumber \\
    =& \underbrace{\bigg(\frac{p_\text{m}^2 p_\text{b}}{(1\!-\!p_\text{m})^2(1\!-\!p_\text{b})}\bigg)}_{\delta} \sum_{m_1=\ceil{t_\text{b}/t_\text{m}}}^\infty \sum_{\substack{m_\text{b}=\max(1,\\ \ceil{m_1t_\text{m}-t_\text{cut}^\prime)/t_\text{b}}}}^{\floor{m_1t_\text{m}/t_\text{b}}} e^{-v(\alpha-1)m_1t_\text{m} + v m_\text{b}t_\text{b}}  (\underbrace{1\!-\!p_\text{m}}_{q_\text{m}})^{2m_1} (\underbrace{1\!-\!p_\text{b}}_{q_\text{b}})^{m_\text{b}} \\
    = &~\delta \sum_{m_1=\left\lceil t_\text{b}/t_\text{m}\right\rceil}^{\infty}\big(\underbrace{e^{-v(\alpha-1) t_\text{m}} q_\text{m}^2}_{=:\thinspace d_{\alpha m}}\big)^{m_1} \sum_{m_\text{b}=\text{max}(1, \left\lceil(m_1 t_\text{m}-t_\text{cut}^{\prime})/t_\text{b}\right\rceil)}^{\left\lfloor m_1 t_\text{m} / t_\text{b}\right\rfloor}\big(\underbrace{e^{v t_\text{b}} q_\text{b}}_{d_b}\big)^{m_\text{b}}  \\
    =&~ \frac{\delta}{1\!-\!d_b} \sum_{m_1=\ceil{t_\text{b}/t_\text{m}}}^{\infty} d_{\alpha m}^{\thinspace m_1} \Big(d_b^{\thinspace \max(1,\ceil{(m_1t_\text{m}-t_\text{cut}^\prime)/t_\text{b}})} -d_b^{\thinspace \floor{m_1t_\text{m}/t_\text{b}}+1}\Big) \\
    = & ~ \frac{\delta}{1\!-\!d_b} \Big( {\sum_{m_1=\ceil{t_\text{b}/t_\text{m}}}^{\floor{(t_\text{b}+t_\text{cut}^\prime)/t_\text{m}}} d_{\alpha m}^{\thinspace m_1} d_b} ~+ \!\!\!\!\! \sum_{m_1=1+\floor{(t_\text{b}+t_\text{cut}^\prime)/t_\text{m}}}^{\infty}\!\!\!\!\!\!\!\! d_{\alpha m}^{\thinspace m_1} d_b^{\ceil{(m_1t_\text{m}-t_\text{cut}^\prime)/t_\text{b}}}  - d_b \!\!\!\! \sum_{m_1=\ceil{t_\text{b}/t_\text{m}}}^{\infty} d_{\alpha m}^{\thinspace m_1} d_b^{\floor{m_1t_\text{m}/t_\text{b}}} \Big)
\end{align}
\begin{align}
    =&~ \frac{\delta}{1\!-\!d_b} \bigg( \underbrace{\frac{d_b \big(d_{\alpha m}^{\ceil{t_\text{b}/t_\text{m}}} - d_{\alpha m}^{\floor{(t_\text{b}+t_\text{cut}^\prime)/t_\text{m}}+1}\big)}{1\!-\!d_{\alpha m}}}_{=:\thinspace S_{1^+2^+} } + \thinspace \Pi_c\big(d_{\alpha m},d_b,\frac{t_\text{m}}{t_\text{b}},\frac{t_\text{cut}^\prime}{t_\text{b}}, 1\!+\!\floor{\frac{t_\text{b}\!+\!t_\text{cut}^\prime}{t_\text{m}}}\!,\infty \big) \qquad \qquad ~ \nonumber \\
    &\quad - d_b \Pi_f\big( d_{\alpha m},d_b,\frac{t_\text{m}}{t_\text{b}},0, \ceil{\frac{t_\text{b}}{t_\text{m}}}\!,\infty\big)  \bigg)  \\
    \overset{(\text{i})}{=}&~ \frac{\delta}{1\!-\!d_b} \bigg(S_{1^+2^+} + \frac{1}{1\!-\!\big(d_{\alpha m}d_b^{\thinspace t_\text{m}/t_\text{b}}\big)^{m^*}} \Pi_c\big(d_{\alpha m},d_b,\frac{t_\text{m}}{t_\text{b}},\frac{t_\text{cut}^\prime}{t_\text{b}}, 1\!+\!\floor{\frac{t_\text{b}\!+\!t_\text{cut}^\prime}{t_\text{m}}}\!, \floor{\frac{t_\text{b}\!+\!t_\text{cut}^\prime}{t_\text{m}}} \!+\! m^* \big) \nonumber \\
    & \quad - d_b \frac{1}{1\!-\!\big(d_{\alpha m}d_b^{\thinspace t_\text{m}/t_\text{b}}\big)^{m^*}} \Pi_f\big(d_{\alpha m},d_b,\frac{t_\text{m}}{t_\text{b}},0, \ceil{\frac{t_\text{b}}{t_\text{m}}}\!, \ceil{\frac{t_\text{b}}{t_\text{m}}} \!+\!m^* \!-\!1 \big) \bigg)~,
    \label{eq:indiv_term_U_alpha_A_1+A_2+}
\end{align}
where we recall that ${z^*(t_\text{m}/t_\text{b}) \!=\! m^*}$ in (i).

On $A_1^+ A_b^+$, we have
\begin{align}
    & ~\mathbb{E}(e^{-v((\alpha-1) X_1 - X_2)} \mathds{1}_{A_{1}^+ A_{b}^+}) \nonumber \\
    =& \underbrace{\bigg(\frac{p_\text{m}^2 p_\text{b}}{(1\!-\!p_\text{m})^2(1\!-\!p_\text{b})}\bigg)}_{\delta} \sum_{k=1}^\infty \sum_{\substack{m_2=\text{max}(1,\\ \left\lceil k m^* - t_\text{cut}^{\prime}/t_\text{m}\right\rceil)}}^{k m^*} \!\!\!\! e^{-v(\alpha-1)km^*t_\text{m} + vm_2t_\text{m}} (\underbrace{1\!-\!p_\text{m}}_{q_\text{m}})^{km^*}(\underbrace{1\!-\!p_\text{m}}_{q_\text{m}})^{m_2} (\underbrace{1\!-\!p_\text{b}}_{q_\text{b}})^{km^*t_\text{m}/t_\text{b}} \\
    = & ~\delta \thinspace\sum_{k=1}^{\infty}\big(\underbrace{e^{-v(\alpha-1) t_\text{m}} q_\text{m} q_\text{b}^{t_\text{m}/t_\text{b}}}_{=:\thinspace d_{\alpha mb}}\big)^{k m^*} \!\!\!\! \sum_{\substack{m_2=\text{max}(1,\\ \left\lceil k m^* - t_\text{cut}^{\prime}/t_\text{m}\right\rceil)}}^{k m^*} \!\!\! \big(\underbrace{e^{v t_\text{m}} q_\text{m}}_{\thinspace d_m}\big)^{m_2} \\
    =&~ \frac{\delta}{1\!-\!d_m} \Big( \sum_{k=1}^{\floor{(t_\text{m}+t_\text{cut}^\prime)/m^*t_\text{m}}} \!\!\! d_{\alpha mb}^{\thinspace km^*} d_m + \!\!\!\!\!\! \sum_{k=1+\floor{(t_\text{m}+t_\text{cut}^\prime)/(m^*t_\text{m})}}^\infty \!\!\!\!\!\!\!\! d_{\alpha mb}^{\thinspace km^*} d_m^{\thinspace \ceil{km^*-t_\text{cut}^\prime/t_\text{m}}} - \sum_{k=1}^\infty d_{\alpha mb}^{\thinspace km^*} d_m^{\thinspace km^*+1} \Big) \\
    =&~ \frac{\delta}{1\!-\!d_m} \bigg(\frac{d_m\big(d_{\alpha mb}^{\thinspace m^*} - d_{\alpha mb}^{\thinspace m^*+m^*\floor{(t_\text{m}+t_\text{cut}^\prime)/m^*t_\text{m}}} \big)}{1\!-\!d_{\alpha mb}^{\thinspace m^*}} \nonumber \\
    & \quad + \frac{d_m^{\thinspace -\floor{t_\text{cut}^\prime/t_\text{m}}}(d_{\alpha mb}d_m)^{m^*+m^*\floor{(t_\text{m}+t_\text{cut}^\prime)/m^*t_\text{m}}} - d_m (d_{\alpha mb}d_m)^{m^*}}{1\!-\!(d_{\alpha mb}d_m)^{m^*}}  \bigg) ~.
    \label{eq:indiv_term_U_alpha_A_1+A_b+}
\end{align}

On $A_1^+ A_2^+ A_b^+$, we have
\begin{align}
    &~\mathbb{E}(e^{-v(\alpha-2) X_{\text{max}}} \mathds{1}_{A_1^+ A_2^+ A_b^+}) \nonumber \\ 
    =& \underbrace{\bigg(\frac{p_\text{m}^2 p_\text{b}}{(1\!-\!p_\text{m})^2(1\!-\!p_\text{b})}\bigg)}_{\delta} \thinspace\sum_{k=1}^{\infty} e^{-v(\alpha-2)km^*t_\text{m}} {(\underbrace{1\!-\!p_\text{m}}_{q_\text{m}})}^{m^*k} {(\underbrace{1\!-\!p_\text{m}}_{q_\text{m}})}^{m^*k} {(\underbrace{1\!-\!p_\text{b}}_{q_\text{b}})}^{m^*kt_\text{m}/t_\text{b}}  \\
    = &~ \delta \sum_{k=1}^{\infty} e^{-v(\alpha-2)k m^* t_\text{m}} \big(\underbrace{q_\text{m}^{2 m^*} q_\text{b}^{m^* t_\text{m}/t_\text{b}}}_{\beta}\big)^k \\ 
    = &~ \frac{\delta \beta e^{-v(\alpha-2)m^* t_\text{m}}}{1-\beta e^{-v(\alpha-2) m^* t_\text{m}}} ~.
    \label{eq:indiv_term_U_alpha_A_1+A_2+A_b+}
\end{align}

\subsection{Calculations of individual terms in \texorpdfstring{$U_3(v)$}{U3(v)}}
\label{sec:calculation_indiv_terms_U_3}

In this section, we calculate the individual terms of~\eqref{eq:term_expansion_U_3}.
The first expected value on $A_{12}^-$ is 
\begin{align}
    &~\mathbb{E}\big(e^{-X_{2}/t_{\text{coh}}} \thinspace\mathds{1}_{A_{12}^-}\big) \nonumber \\
    = &~ \underbrace{\bigg(\frac{p_\text{m}^2 p_\text{b}}{(1\!-\!p_\text{m})^2(1\!-\!p_\text{b})}\bigg)}_{\delta} \thinspace \sum_{m_2=1}^{\infty} \thinspace \sum_{\substack{m_1=(m_2\\+\left\lceil t_\text{cut} / t_\text{m}\right\rceil)}}^{\infty} \sum_{m_\text{b}=\left\lceil m_2 t_\text{m} / t_\text{b}\right\rceil}^{\left\lfloor m_1 t_\text{m} / t_\text{b}\right\rfloor} \!\!\!\! e^{-m_2 t_\text{m} / t_{\text{coh}}} {(\underbrace{1\!-\!p_\text{m}}_{q_\text{m}})}^{m_1} {(\underbrace{1\!-\!p_\text{m}}_{q_\text{m}})}^{m_2} {(\underbrace{1\!-\!p_\text{b}}_{q_\text{b}})}^{m_\text{b}} \\
    = &~ \delta \sum_{m_2=1}^{\infty} \big(\underbrace{q_\text{m} e^{-t_\text{m} / t_{\text{coh}}}}_{=:\thinspace \kappa_\text{m}}\big)^{m_2} \!\! \sum_{m_1=m_2+\left\lceil t_\text{cut} / t_\text{m}\right\rceil}^{\infty} \!\!\!\! q_\text{m}^{m_1} \sum_{m_\text{b}=\left\lceil m_2 t_\text{m} / t_\text{b}\right\rceil}^{\left\lfloor m_1 t_\text{m} / t_\text{b}\right\rfloor} q_\text{b}^{m_\text{b}} \\
    = &~ \frac{\delta}{p_\text{b}} \sum_{m_2=1}^\infty \kappa_\text{m}^{m_2} \sum_{m_1=m_2+\ceil{t_\text{cut}/t_\text{m}}}^\infty q_\text{m}^{m_1} q_\text{b}^{\ceil{m_2t_\text{m}/t_\text{b}}} - \frac{\delta q_\text{b}}{p_\text{b}} \sum_{m_2=1}^\infty \kappa_\text{m}^{m_2} \sum_{m_1=m_2+\ceil{t_\text{cut}/t_\text{m}}}^\infty q_\text{m}^{m_1} q_\text{b}^{\floor{m_1t_\text{m}/t_\text{b}}} 
\end{align}
\begin{align}
    =&~ \frac{\delta q_\text{m}^{\ceil{t_\text{cut}/t_\text{m}}}}{p_\text{m} p_\text{b}} \sum_{m_2=1}^\infty (\kappa_\text{m} q_\text{m})^{m_2} q_\text{b}^{\ceil{m_2t_\text{m}/t_\text{b}}} - \frac{\delta q_\text{b}}{p_\text{b}} \sum_{m_1=1+\ceil{t_\text{cut}/t_\text{m}}}^\infty q_\text{m}^{m_1} q_\text{b}^{\floor{m_1t_\text{m}/t_\text{b}}} \sum_{m_2=1}^{m_1-\ceil{t_\text{cut}/t_\text{m}}} \kappa_\text{m}^{m_2}   \\
    =&~ \frac{\delta q_\text{m}^{\ceil{t_\text{cut}/t_\text{m}}}}{p_\text{m} p_\text{b}} \Pi_c\big(\kappa_\text{m}q_\text{m},q_\text{b},\frac{t_\text{m}}{t_\text{b}},0,1,\infty \big)  - \frac{\delta q_\text{b} \kappa_\text{m}}{p_\text{b} (1\!-\!\kappa_\text{m})} \Big( \sum_{m_1=1+\ceil{t_\text{cut}/t_\text{m}}}^\infty q_\text{m}^{m_1} q_\text{b}^{\floor{m_1t_\text{m}/t_\text{b}}} \nonumber \\
    & \quad -\kappa_\text{m}^{\thinspace -\ceil{t_\text{cut}/t_\text{m}}} \!\!\!\!\!\!\!\! \sum_{m_1=1+\ceil{t_\text{cut}/t_\text{m}}}^\infty (\kappa_\text{m} q_\text{m})^{m_1} q_\text{b}^{\floor{m_1t_\text{m}/t_\text{b}}} \big)  \Big)  \\
    =&~ \frac{\delta q_\text{m}^{\ceil{t_\text{cut}/t_\text{m}}}}{p_\text{m} p_\text{b}} \Pi_c\big(\kappa_\text{m}q_\text{m},q_\text{b},\frac{t_\text{m}}{t_\text{b}},0,1,\infty \big)  - \frac{\delta q_\text{b} \kappa_\text{m}}{p_\text{b} (1\!-\!\kappa_\text{m})} \Big(\Pi_f\big(q_\text{m},q_\text{b},\frac{t_\text{m}}{t_\text{b}},0,1\!+\!\ceil{\frac{t_\text{cut}}{t_\text{m}}}\!,\infty\big) \nonumber \\
    & \quad - \kappa_\text{m}^{-\ceil{t_\text{cut}/t_\text{m}}} \Pi_f(\kappa_\text{m}q_\text{m},q_\text{b},\frac{t_\text{m}}{t_\text{b}},0,1\!+\!\ceil{\frac{t_\text{cut}}{t_\text{m}}}\!,\infty) \Big) \\
    \overset{(\text{i})}{=}&~ \frac{\delta q_\text{m}^{\ceil{t_\text{cut}/t_\text{m}}}}{p_\text{m} p_\text{b}} \frac{1}{1\!-\!\big(\kappa_\text{m} q_\text{m} q_\text{b}^{t_\text{m}/t_\text{b}}\big)^{m^*}} \Pi_c\big(\kappa_\text{m}q_\text{m},q_\text{b},\frac{t_\text{m}}{t_\text{b}},0,1,m^*\big) \nonumber \\
    & \quad - \frac{\delta q_\text{b} \kappa_\text{m}}{p_\text{b} (1\!-\!\kappa_\text{m})} \Big(\frac{1}{1\!-\!\big(q_\text{m} q_\text{b}^{t_\text{m}/t_\text{b}}\big)^{m^*}}\Pi_f\big(q_\text{m},q_\text{b},\frac{t_\text{m}}{t_\text{b}},0, 1\!+\!\ceil{\frac{t_\text{cut}}{t_\text{m}}}\!,\ceil{\frac{t_\text{cut}}{t_\text{m}}} \!+\!m^*\big) \nonumber \\
    & \quad - \frac{\kappa_\text{m}^{-\ceil{t_\text{cut}/t_\text{m}}}}{1\!-\!\big(\kappa_\text{m} q_\text{m} q_\text{b}^{t_\text{m}/t_\text{b}}\big)^{m^*}} \Pi_f\big(\kappa_\text{m}q_\text{m},q_\text{b},\frac{t_\text{m}}{t_\text{b}},0, 1\!+\!\ceil{\frac{t_\text{cut}}{t_\text{m}}}\!,\ceil{\frac{t_\text{cut}}{t_\text{m}}} \!+\! m^*\big) \Big) ~,
    \label{eq:indiv_term_U_3_A_12-}
\end{align}
where in (i), we used ${z^*(t_\text{m}/t_\text{b}) \!=\!m^*}$.
 
On $A_{1b}^-$, we have
\begin{align}
    &~\mathbb{E}\big(e^{-X_{b}/t_{\text{coh}}} \thinspace\mathds{1}_{A_{1b}^-}\big) \nonumber \\
    =& \underbrace{\bigg(\frac{p_\text{m}^2 p_\text{b}}{(1\!-\!p_\text{m})^2(1\!-\!p_\text{b})}\bigg)}_{\delta} \sum_{m_\text{b}=1}^{\infty} \sum_{m_1=\ceil{(m_\text{b} t_\text{b}+t_\text{cut})/t_\text{m}}}^{\infty} \sum_{m_2=\ceil{m_\text{b} t_\text{b}/t_\text{m}}}^{m_1} \!\!\!\! e^{-m_\text{b} t_\text{b} / t_{\text{coh}}} {(\underbrace{1\!-\!p_\text{m}}_{q_\text{m}})}^{m_1} {(\underbrace{1\!-\!p_\text{m}}_{q_\text{m}})}^{m_2} {(\underbrace{1\!-\!p_\text{b}}_{q_\text{b}})}^{m_\text{b}} \\
    = &~ \delta \sum_{m_\text{b}=1}^{\infty} \big(\underbrace{{q_\text{b} e^{-t_\text{b}/t_{\text{coh}}}}}_{=:\thinspace \kappa_\text{b}} \big)^{m_\text{b}} \sum_{m_1=\ceil{(m_\text{b} t_\text{b}+t_\text{cut})/t_\text{m}}}^{\infty} q_\text{m}^{m_1} \sum_{m_2=\ceil{m_\text{b} t_\text{b}/t_\text{m}}}^{m_1} q_\text{m}^{m_2}
\end{align}
\begin{align}
    =&~ \frac{\delta}{p_\text{m}} \sum_{m_\text{b}=1}^\infty \kappa_\text{b}^{m_\text{b}} q_\text{m}^{\ceil{m_\text{b}t_\text{b}/t_\text{m}}} \!\!\!\!\!\!\!\! \sum_{m_1=\ceil{(m_\text{b}t_\text{b}+t_\text{cut})/t_\text{m}}}^\infty \!\!\!\!\!\!\!\!\!\!\!\!q_\text{m}^{m_1} - \frac{\delta q_\text{m}}{p_\text{m}} \sum_{m_\text{b}=1}^\infty \kappa_\text{b}^{m_\text{b}} \!\!\!\!\!\!\!\! \sum_{m_1=\ceil{(m_\text{b}t_\text{b}+t_\text{cut})/t_\text{m}}}^\infty \!\!\!\!\!\!\!\! q_\text{m}^{2m_1} \\
    =&~ \frac{\delta}{p_\text{m}^2} \sum_{m_\text{b}=1}^\infty \kappa_\text{b}^{m_\text{b}} q_\text{m}^{\ceil{m_\text{b}t_\text{b}/t_\text{m}} + \ceil{(m_\text{b}t_\text{b}+t_\text{cut})/t_\text{m}}} - \frac{\delta q_\text{m}}{p_\text{m} (1\!-\!q_\text{m}^2)} \sum_{m_\text{b}=1}^\infty \kappa_\text{b}^{m_\text{b}} q_\text{m}^{2\ceil{(m_\text{b}t_\text{b}+t_\text{cut})/t_\text{m}}}    \\
    =&~ \frac{\delta}{p_\text{m}^2} \Pi_{cc}\big(\kappa_\text{b},q_\text{m},\frac{t_\text{b}}{t_\text{m}},-\frac{t_\text{cut}}{t_\text{m}},0,1,\infty \big) - \frac{\delta q_\text{m}}{p_\text{m} (1\!-\!q_\text{m}^2)} \Pi_{c}\big(\kappa_\text{b},q_\text{m}^2,\frac{t_\text{b}}{t_\text{m}},-\frac{t_\text{cut}}{t_\text{m}},1,\infty\big)    \\
    \overset{(\text{i})}{=} &~ \frac{\delta}{p_\text{m}^2} \frac{\Pi_{cc}\big(\kappa_\text{b},q_\text{m},{t_\text{b}}/{t_\text{m}},-{t_\text{cut}}/{t_\text{m}},0,1,{m^*t_\text{m}}/{t_\text{b}}\big)}{1\!-\!\big(\kappa_\text{b} q_\text{m}^{2t_\text{b}/t_\text{m}}\big)^{m^*t_\text{m}/t_\text{b}}} - \frac{\delta q_\text{m}}{p_\text{m} (1\!-\!q_\text{m}^2)} \frac{\Pi_{c}\big(\kappa_\text{b},q_\text{m}^2,{t_\text{b}}/{t_\text{m}},-{t_\text{cut}}/{t_\text{m}},1,{m^*t_\text{m}}/{t_\text{b}}\big)}{1\!-\!\big(\kappa_\text{b} q_\text{m}^{2t_\text{b}/t_\text{m}}\big)^{m^*t_\text{m}/t_\text{b}}} ~,
    \label{eq:indiv_term_U_3_A_1b-}
\end{align}
where, in (i), we use ${z^*(t_\text{b}/t_\text{m}) \!=\! m^*t_\text{m}/t_\text{b}}$.

On $A_{b1}^-$, we have
\begin{align}
    &~\mathbb{E}\big(e^{-X_{1}/t_{\text{coh}}} \thinspace\mathds{1}_{A_{b1}^-}\big) \nonumber \\
    =& \underbrace{\bigg(\frac{p_\text{m}^2 p_\text{b}}{(1\!-\!p_\text{m})^2(1\!-\!p_\text{b})}\bigg)}_{\delta} \thinspace \sum_{m_1=1}^\infty \thinspace \sum_{m_\text{b}=\ceil{(m_1 t_\text{m}+t_\text{cut})/t_\text{b}}}^{\infty} \!\!\sum_{m_2=m_1}^{\floor{m_\text{b} t_\text{b}/t_\text{m}}} \!\!\!\! e^{-m_1t_\text{m}/t_{\text{coh}}} {(\underbrace{1\!-\!p_\text{m}}_{q_\text{m}})}^{m_1} {(\underbrace{1\!-\!p_\text{m}}_{q_\text{m}})}^{m_2} {(\underbrace{1\!-\!p_\text{b}}_{q_\text{b}})}^{m_\text{b}}  \\
    = &~ \delta \sum_{m_1=1}^{\infty} \big(\underbrace{q_\text{m} e^{-t_\text{m}/t_{\text{coh}}}}_{\kappa_\text{m}} \big)^{m_1} \!\!\!\!\!\! \sum_{m_\text{b}=\ceil{(m_1 t_\text{m}+t_\text{cut})/t_\text{b}}}^{\infty} \!\!\!\!\!\!\!\! q_\text{b}^{m_\text{b}} \sum_{m_2=m_1}^{\floor{m_\text{b} t_\text{b}/t_\text{m}}} q_\text{m}^{m_2}  \\
    =&~ \frac{\delta}{p_\text{m}} \sum_{m_1=1}^\infty (\kappa_\text{m} q_\text{m})^{m_1} \!\!\!\!\!\!\!\! \sum_{m_\text{b}=\ceil{(m_1t_\text{m}+t_\text{cut})/t_\text{b}}}^\infty \!\!\!\!\!\!\!\! q_\text{b}^{m_\text{b}} - ~ \frac{\delta q_\text{m}}{p_\text{m}} \sum_{m_1=1}^\infty \kappa_\text{m}^{m_1} \!\!\!\!\!\!\sum_{m_\text{b}=\ceil{(m_1t_\text{m}+t_\text{cut})/t_\text{b}}}^\infty \!\!\!\!\!\!\!\! q_\text{b}^{m_\text{b}} q_\text{m}^{\floor{m_\text{b}t_\text{b}/t_\text{m}}} \\
    \overset{(\text{i})}{=} &~ \frac{\delta}{p_\text{m} p_\text{b}} \sum_{m_1=1}^\infty (\kappa_\text{m} q_\text{m})^{m_1} q_\text{b}^{\ceil{(m_1t_\text{m}+t_\text{cut})/t_\text{b}}} - ~\frac{\delta q_\text{m}}{p_\text{m}} \sum_{m_\text{b}=\ceil{(t_\text{m}+t_\text{cut})/t_\text{b}}}^\infty \!\!\!\!\!\!\!\! q_\text{b}^{m_\text{b}} q_\text{m}^{\floor{m_\text{b}t_\text{b}/t_\text{m}}} \!\! \sum_{m_1=1}^{\floor{(m_\text{b}t_\text{b}-t_\text{cut})/t_\text{m}}} \!\!\!\!\!\!\!\! \kappa_\text{m}^{m_1}  \\
    =&~ \frac{\delta}{p_\text{m} p_\text{b}} \Pi_c\big(\kappa_\text{m}q_\text{m},q_\text{b},\frac{t_\text{m}}{t_\text{b}},-\frac{t_\text{cut}}{t_\text{b}},1,\infty \big) - \frac{\delta q_\text{m}}{p_\text{m}} \!\sum_{m_\text{b}=\ceil{(t_\text{m}+t_\text{cut})/t_\text{b}}}^\infty \!\! \!\!\!\!\!\!\!\! q_\text{b}^{m_\text{b}} q_\text{m}^{\floor{m_\text{b}t_\text{b}/t_\text{m}}} \! \Big( \frac{\kappa_\text{m} \!-\! \kappa_\text{m}^{\floor{(m_\text{b}t_\text{b}-t_\text{cut})/t_\text{m}}+1}}{1\!-\!\kappa_\text{m}} \Big) \\
    =&~ \frac{\delta}{p_\text{m} p_\text{b}} \Pi_c\big(\kappa_\text{m}q_\text{m},q_\text{b},\frac{t_\text{m}}{t_\text{b}},-\frac{t_\text{cut}}{t_\text{b}},1,\infty \big) - \frac{\delta q_\text{m}\kappa_\text{m}}{p_\text{m}(1\!-\!\kappa_\text{m})} \Big( \!\sum_{m_\text{b}=\ceil{(t_\text{m}+t_\text{cut})/t_\text{b}}}^\infty \!\! \!\!\!\!\!\!\!\! q_\text{b}^{m_\text{b}} q_\text{m}^{\floor{m_\text{b}t_\text{b}/t_\text{m}}} \nonumber \\
    & \quad - \!\!\!\!\!\! \sum_{m_\text{b}=\ceil{(t_\text{m}+t_\text{cut})/t_\text{b}}}^\infty \!\! \!\!\!\!\!\!\!\! q_\text{b}^{m_\text{b}} q_\text{m}^{\floor{m_\text{b}t_\text{b}/t_\text{m}}} \kappa_\text{m}^{\floor{(m_\text{b}t_\text{b}-t_\text{cut})/t_\text{m}}}\Big) \\
    =&~ \frac{\delta}{p_\text{m} p_\text{b}} \Pi_c\big(\kappa_\text{m}q_\text{m},q_\text{b},\frac{t_\text{m}}{t_\text{b}},-\frac{t_\text{cut}}{t_\text{b}},1,\infty \big) - \frac{\delta q_\text{m}\kappa_\text{m}}{p_\text{m}(1\!-\!\kappa_\text{m})} \Big( \Pi_f\big(q_\text{b},q_\text{m},\frac{t_\text{b}}{t_\text{m}},0, \ceil{\frac{t_\text{m}\!+\!t_\text{cut}}{t_\text{b}}}\!,\infty\big) \nonumber \\
    & \quad - \Gamma\big(q_\text{b},q_\text{m},\kappa_\text{m},\frac{t_\text{b}}{t_\text{m}},\frac{t_\text{cut}}{t_\text{m}},\ceil{\frac{t_\text{m}\!+\!t_\text{cut}}{t_\text{b}}}\!,\infty \big)\Big)
\end{align}
\begingroup
    \setlength{\abovedisplayskip}{4pt}   
    \setlength{\belowdisplayskip}{4pt}   
    \setlength{\abovedisplayshortskip}{2pt}
    \setlength{\belowdisplayshortskip}{2pt}
    \begin{align}
        \overset{(\text{ii})}{=}&~ \frac{\delta}{p_\text{m} p_\text{b}} \frac{1}{1\!-\!\big(\kappa_\text{m} q_\text{m} q_\text{b}^{t_\text{m}/t_\text{b}}\big)^{m^*}} \Pi_c\big(\kappa_\text{m}q_\text{m},q_\text{b},\frac{t_\text{m}}{t_\text{b}},-\frac{t_\text{cut}}{t_\text{b}},1,m^*\big) \nonumber \\
        & \quad - \frac{\delta q_\text{m} \kappa_\text{m}}{p_\text{m}(1\!-\!\kappa_\text{m})} \Big( \frac{1}{1\!-\! \big(q_\text{b} q_\text{m}^{t_\text{b}/t_\text{m}} \big)^{m^*t_\text{m}/t_\text{b}}} \Pi_f\big(q_\text{b},q_\text{m},\frac{t_\text{b}}{t_\text{m}},0,\ceil{\frac{t_\text{m}\!+\!t_\text{cut}}{t_\text{b}}}\!,\ceil{\frac{t_\text{m}\!+\!t_\text{cut}}{t_\text{b}}} \!+\!\frac{m^*t_\text{m}}{t_\text{b}} \!-\! 1\big) \nonumber \\
        & \quad - \frac{1}{1\!-\! \big(q_\text{b} q_\text{m}^{t_\text{b}/t_\text{m}} \kappa_\text{m}^{t_\text{b}/t_\text{m}} \big)^{m^*t_\text{m}/t_\text{b}}} \Gamma\big(q_\text{b},q_\text{m},\kappa_\text{m},\frac{t_\text{b}}{t_\text{m}},\frac{t_\text{cut}}{t_\text{m}}, \ceil{\frac{t_\text{m}\!+\!t_\text{cut}}{t_\text{b}}}\!, \ceil{\frac{t_\text{m}\!+\!t_\text{cut}}{t_\text{b}}} \!+\! \frac{m^*t_\text{m}}{t_\text{b}} \!-\! 1\big) \Big) ~,
        \label{eq:indiv_term_U_3_A_b1-}
    \end{align}
\endgroup
where we have changed the order of summation in (i).
Note that, in (ii), we recall that ${z^*(t_\text{m}/t_\text{b}) \!=\! m^*}$ and ${z^*(t_\text{b}/t_\text{m}) \!=\! m^*t_\text{m}/t_\text{b}}$, respectively.

On $A_{12}^- A_{1b}^-$, we have
\begin{align}
    &~\mathbb{E}\big(e^{-X_{2}/t_{\text{coh}}} \thinspace\mathds{1}_{A_{12}^- A_{1b}^-}\big) \nonumber \\
    =&~ \underbrace{\bigg(\frac{p_\text{m}^2 p_\text{b}}{(1\!-\!p_\text{m})^2(1\!-\!p_\text{b})}\bigg)}_{\delta} \thinspace \sum_{i=1}^\infty \thinspace \sum_{m_1=im^*+\ceil{t_\text{cut}/t_\text{m}}}^{\infty} e^{-im^*t_\text{m}/t_{\text{coh}}} {(\underbrace{1\!-\!p_\text{m}}_{q_\text{m}})}^{m_1} {(\underbrace{1\!-\!p_\text{m}}_{q_\text{m}})}^{im^*} {(\underbrace{1\!-\!p_\text{b}}_{q_\text{b}})}^{im^*t_\text{m}/t_\text{b}}  \\
    = &~ \delta \sum_{i=1}^{\infty} \big(e^{-t_\text{m}/t_{\text{coh}}} q_\text{m} q_\text{b}^{t_\text{m}/t_\text{b}}\big)^{im^*} \sum_{m_1=im^*+\ceil{t_\text{cut}/t_\text{m}}}^{\infty} q_\text{m}^{m_1} \\
    = &~ \frac{\delta q_\text{m}^{\ceil{t_\text{cut}/t_\text{m}}}}{p_\text{m}} \sum_{i=1}^{\infty} \big( \underbrace{e^{-m^*t_\text{m}/t_{\text{coh}}} q_\text{m}^{2m^*} q_\text{b}^{m^*t_\text{m}/t_\text{b}}}_{=: \thinspace \tilde{\beta}}\big)^{i} \\
    = &~ \frac{\delta q_\text{m}^{\ceil{t_\text{cut}/t_\text{m}}} \tilde{\beta} }{p_\text{m}(1\!-\!\tilde{\beta} )} ~.
    \label{eq:indiv_term_U_3_A_12-A_1b-}
\end{align}
On $A_{b1}^- A_{b2}^-$, we have
\begin{align}
    &~\mathbb{E}\big(e^{-X_{1}/t_{\text{coh}}} \thinspace\mathds{1}_{A_{b1}^- A_{b2}^-}\big) \nonumber \\
    =&~ \underbrace{\bigg(\frac{p_\text{m}^2 p_\text{b}}{(1\!-\!p_\text{m})^2(1\!-\!p_\text{b})}\bigg)}_{\delta} \thinspace \sum_{m_1=1}^\infty \thinspace \sum_{m_\text{b}=\ceil{(m_1 t_\text{m}+t_\text{cut})/t_\text{b}}}^{\infty} e^{-m_1t_\text{m}/t_{\text{coh}}} {(\underbrace{1\!-\!p_\text{m}}_{q_\text{m}})}^{2m_1}{(\underbrace{1\!-\!p_\text{b}}_{q_\text{b}})}^{m_\text{b}}  \\
    = &~ \delta \sum_{m_1=1}^{\infty} \big(q_\text{m} \underbrace{q_\text{m} e^{-t_\text{m}/t_{\text{coh}}}}_{\kappa_\text{m}} \big)^{m_1} \!\! \sum_{m_\text{b}=\ceil{(m_1 t_\text{m}+t_\text{cut})/t_\text{b}}}^{\infty} q_\text{b}^{m_\text{b}} \nonumber \\
    = &~ \frac{\delta}{p_\text{b}} \sum_{m_1=1}^{\infty} (\kappa_\text{m} q_\text{m})^{m_1} q_\text{b}^{\ceil{(m_1 t_\text{m}+t_\text{cut})/t_\text{b}}} \\
    = &~ \frac{\delta}{p_\text{b}} \Pi_c\big(\kappa_\text{m} q_\text{m}, q_\text{b},\frac{t_\text{m}}{t_\text{b}},-\frac{t_\text{cut}}{t_\text{b}},1,\infty \big) \\
    \overset{(\text{i})}{=} &~ \frac{\delta}{p_\text{b}} \frac{\Pi_c\big(\kappa_\text{m} q_\text{m}, q_\text{b},{t_\text{m}}/{t_\text{b}},-{t_\text{cut}}/{t_\text{b}},1,m^* \big)}{1\!-\!\big(\kappa_\text{m}q_\text{m}q_\text{b}^{t_\text{m}/t_\text{b}}\big)^{m^*}} ~,
    \label{eq:indiv_term_U_3_A_b1-A_b2-}
\end{align}
where, in (i), we recall that ${z^*(t_\text{m}/t_\text{b}) \!=\! m^*}$.

On $A_{1}^- A_{2}^-$, we have
\begin{align}
    &~\mathbb{E}\big(e^{-X_{b}/t_{\text{coh}}} \thinspace\mathds{1}_{A_{1}^- A_{2}^-}\big) \nonumber \\
    =& \underbrace{\bigg(\frac{p_\text{m}^2 p_\text{b}}{(1\!-\!p_\text{m})^2(1\!-\!p_\text{b})}\bigg)}_{\delta} \thinspace \sum_{m_\text{b}=1}^\infty \thinspace \sum_{m_1=\ceil{(m_\text{b} t_\text{b}+t_\text{cut})/t_\text{m}}}^{\infty} \!\!\!\!\!\!\!\! e^{-m_\text{b}t_\text{b}/t_{\text{coh}}} {(\underbrace{1\!-\!p_\text{m}}_{q_\text{m}})}^{2m_1} {(\underbrace{1\!-\!p_\text{b}}_{q_\text{b}})}^{m_\text{b}}   \\
    = &~ \delta \sum_{m_\text{b}=1}^{\infty} \big(\underbrace{q_\text{b} e^{-t_\text{b}/t_{\text{coh}}}}_{\kappa_\text{b}} \big)^{m_\text{b}} \!\!\! \sum_{m_1=\ceil{(m_\text{b} t_\text{b}+t_\text{cut})/t_\text{m}}}^{\infty} \!\!\! q_\text{m}^{2 m_1} \nonumber \\
    = &~ \frac{\delta}{1\!-\! q_\text{m}^2} \sum_{m_\text{b}=1}^{\infty} \kappa_\text{b}^{m_\text{b}} q_\text{m}^{2\ceil{(m_\text{b} t_\text{b}+t_\text{cut})/t_\text{m})}} \\
    = &~ \frac{\delta}{1\!-\! q_\text{m}^2} \Pi_c\big(\kappa_\text{b},q_\text{m}^2,\frac{t_\text{b}}{t_\text{m}},-\frac{t_\text{cut}}{t_\text{m}},1,\infty\big) \\
    \overset{(\text{i})}{=} &~ \frac{\delta}{(1\!-\! q_\text{m}^2)} \frac{\Pi_c\big(\kappa_\text{b},q_\text{m}^2,{t_\text{b}}/{t_\text{m}},-{t_\text{cut}}/{t_\text{m}},1,m^*t_\text{m}/t_\text{b}\big)}{1\!-\!\big(\kappa_\text{b}q_\text{m}^{2t_\text{b}/t_\text{m}}\big)^{m^*t_\text{m}/t_\text{b}} } ~,
    \label{eq:indiv_term_U_3_A_1-A_2-}
\end{align}
where, in (i), we recall that ${z^*(t_\text{b}/t_\text{m}) \!=\! m^*t_\text{m}/t_\text{b}}$.

On $A_{1}^- A_{b}^-$, we have
\begin{align}
    &~\mathbb{E}\big(e^{-X_{2}/t_{\text{coh}}} \thinspace\mathds{1}_{A_{1}^- A_{b}^-}\big) \nonumber \\
    =& \underbrace{\bigg(\frac{p_\text{m}^2 p_\text{b}}{(1\!-\!p_\text{m})^2(1\!-\!p_\text{b})}\bigg)}_{\delta} \sum_{m_2=1}^{\infty} \thinspace \sum_{k=\ceil{(m_2 t_\text{m}+t_\text{cut})/(m^* t_\text{m})}}^{\infty} \!\!\!\!\!\!\!\! e^{-m_2t_\text{m}/t_{\text{coh}}} {(\underbrace{1\!-\!p_\text{m}}_{q_\text{m}})}^{km^*} {(\underbrace{1\!-\!p_\text{m}}_{q_\text{m}})}^{m_2} {(\underbrace{1\!-\!p_\text{b}}_{q_\text{b}})}^{km^*t_\text{m}/t_\text{b}} \\
    = &~ \delta \sum_{m_2=1}^{\infty} \big(\underbrace{q_\text{m} e^{-t_\text{m}/t_{\text{coh}}}}_{\kappa_\text{m}}\big)^{m_2} \!\!\!\! \sum_{k=\ceil{(m_2 t_\text{m}+t_\text{cut})/(m^* t_\text{m})}}^{\infty} (\underbrace{q_\text{m}^{m^*} q_\text{b}^{m^* t_\text{m}/t_\text{b}}}_{\gamma})^k  \\
    = &~ \frac{\delta}{1\!-\!\gamma} \sum_{m_2=1}^{\infty} \kappa_\text{m}^{m_2} \gamma^{\ceil{(m_2 t_\text{m}+t_\text{cut})/(m^*t_\text{m})}} \\
    = &~ \frac{\delta}{1\!-\!\gamma} \Pi_c\big(\kappa_\text{m},\gamma,\frac{1}{m^*},-\frac{t_\text{cut}}{m^*t_\text{m}},1,\infty\big)  \\
    \overset{(\text{i})}{=} &~ \frac{\delta}{1\!-\!\gamma} \frac{\Pi_c\big(\kappa_\text{m},\gamma,{1}/{m^*},-{t_\text{cut}}/{(m^*t_\text{m})},1,m^*\big)}{1\!-\!(\kappa_\text{m} \gamma^{1/m^*})^{m^*}} ~,
    \label{eq:indiv_term_U_3_A_1-A_b-}
\end{align}
where, in (i), we recall that ${z^*(1/m^*) \!=\! m^*}$.

\section{Teleporting a mixed state using a Werner state}
\label{sec:teleportation_fidelity}

\begin{lemma}
Consider the teleportation of a data qubit in a mixed state 
\begingroup
    \setlength{\abovedisplayskip}{6pt}   
    \setlength{\belowdisplayskip}{6pt}   
    \setlength{\abovedisplayshortskip}{2pt}
    \setlength{\belowdisplayshortskip}{2pt}
    \begin{align}
    \label{eq:data-qubit-density-matrix}
        \rho_{\mathrm{data}} = p_d \ket{\phi}\!\bra{\phi} + (1-p_d)\frac{\mathbb{I}_2}{2},
    \end{align}
\endgroup
where $p_d \in [0,1]$ is the probability weight of the pure state $\ket{\phi}$ and $\mathbb{I}_2$ is the $2 \times 2$ identity matrix.
Suppose Alice and Bob share an entangled resource described by the Werner state
\begingroup
    \setlength{\abovedisplayskip}{6pt}   
    \setlength{\belowdisplayskip}{6pt}   
    \setlength{\abovedisplayshortskip}{2pt}
    \setlength{\belowdisplayshortskip}{2pt}
    \begin{align}
    \label{eq:werner-state}
        \rho_w = w \ket{\Phi^+}\!\bra{\Phi^+} + (1-w) \frac{\mathbb{I}_4}{4},
    \end{align}
\endgroup
where $\ket{\Phi^+}$ is the maximally entangled Bell state, $w \in [0,1]$ in the corresponding Werner parameter and $\mathbb{I}_2$ is the $2 \times 2$ identity matrix.
Assume that each qubit is subject to decoherence modelled by a depolarising channel $\mathcal{E}_t$ acting as
\begingroup
    \setlength{\abovedisplayskip}{6pt}   
    \setlength{\belowdisplayskip}{6pt}   
    \setlength{\abovedisplayshortskip}{2pt}
    \setlength{\belowdisplayshortskip}{2pt}
    \begin{align}
        \mathcal{E}_t(\ket{\psi}\!\bra{\psi}) = e^{-t/t_{\mathrm{coh}}} \ket{\psi}\!\bra{\psi} + \bigl(1 - e^{-t/t_{\mathrm{coh}}}\bigr) \frac{\mathbb{I}_2}{2},
    \end{align}
\endgroup
with coherence time $t_{\mathrm{coh}}$ and storage duration $t$.
Further, assume that local operations are instantaneous compared to classical communication and entanglement generation times, so that the time for Bell measurement and application of Pauli corrections is negligible.
Then, the fidelity of the teleported qubit with respect to $\ket{\phi}$ is given by
\begingroup
    \setlength{\abovedisplayskip}{6pt}   
    \setlength{\belowdisplayskip}{6pt}   
    \setlength{\abovedisplayshortskip}{2pt}
    \setlength{\belowdisplayshortskip}{2pt}
    \begin{align}
    \label{eq:teleportation_fidelity_general_expression}
        F_{\mathrm{tel}} = \frac{1 + w p_d e^{-t_{\mathrm{class}}/t_{\mathrm{coh}}}}{2}~,
    \end{align}
\endgroup
where $t_{\text{class}}$ denotes the classical communication time required to transmit the Pauli correction message.
\end{lemma}

\begin{proof}
We begin with the density matrix representation of the joint state between Alice and Bob throughout the stages of the standard teleportation protocol. 
The data qubit is represented as ${\ket{\phi} = a\ket{0}+b\ket{1}}$, where $a,b$ are complex numbers.
The entanglement resource in~\eqref{eq:werner-state} can be written as a sum of the Bell states  ${\ket{\Phi^+}}$, ${\ket{\Phi^-}}$, ${\ket{\Psi^+}}$, and ${\ket{\Psi^-}}$ as follows
\begingroup
    \setlength{\abovedisplayskip}{6pt}   
    \setlength{\belowdisplayskip}{6pt}   
    \setlength{\abovedisplayshortskip}{2pt}
    \setlength{\belowdisplayshortskip}{2pt}
    \begin{align}
    \label{eq:werner_state_maximally_mixed}
        \rho_w = p_1 \ket{\Phi^+}\!\!\bra{\Phi^+} + p_2\ket{\Phi^-}\!\!\bra{\Phi^-} + p_3\ket{\Psi^+}\!\!\bra{\Psi^+} + p_4\ket{\Psi^-}\!\!\bra{\Psi^-}~,
    \end{align}
\endgroup
where ${p_1=(1+3w)/4}$ and ${p_2 = p_3 = p_4 = (1 - w)/4}$.
In the following section, we investigate teleportation using these maximally entangled states separately.
We consider that Alice holds the data qubit $\rho_{\text{data}}$, as given in~\eqref{eq:data-qubit-density-matrix}. 
We begin with the resource state ${\ket{\Psi^-}=(\ket{01}-\ket{10})/\sqrt{2}}$ being shared between Alice and Bob. 
In this context, according to the standard convention, the first and second qubits belong to Alice: the first qubit is the data qubit, and the second is half of the shared entanglement resource. 
Bob holds the third qubit.
For clarity, we explicitly write the qubit indices in the first line, while omitting them in subsequent expressions for brevity. 
Following the stages of the standard teleportation protocol with the state $\ket{\Psi^-}$, the state immediately before Alice’s measurement is written as:
\begingroup
    \setlength{\abovedisplayskip}{6pt}   
    \setlength{\belowdisplayskip}{6pt}   
    \setlength{\abovedisplayshortskip}{2pt}
    \setlength{\belowdisplayshortskip}{2pt}
    \begin{align}
        & (\rho_{\text{data}})_1\otimes (\ket{\Psi^-}\!\bra{\Psi^-} )_{23}  \nonumber \\
        =& \Big( \big(p_d\abs{a}^2\!+\!\frac{1\!-\!p_d}{2}\big)\ket{0}\!\bra{0} \!+\! p_dab^*\ket{0}\!\bra{1} \!+\! p_da^*b\ket{1}\!\bra{0} \!+\! \big(p_d\abs{b}^2+\frac{1\!-\!p_d}{2}\big)\ket{1}\!\bra{1}\big)_{1} \nonumber \\
        & ~ \otimes \frac{1}{2}\Big(\ket{01}\!\bra{01} \!-\! \ket{01}\!\bra{10} \!-\! \ket{10}\!\bra{01} \!+\! \ket{10}\!\bra{10}\Big)_{23} \\
        \overset{O_1}{\longrightarrow} & \frac{1}{2}\big(p_d\abs{a}^2\!+\!\frac{1\!-\!p_d}{2} \big) \big(\ket{001}\!\bra{001} \!-\! \ket{001}\!\bra{010} \!-\! \ket{010}\!\bra{001} \!+\!\ket{010}\!\bra{010} \big) \nonumber \\
        & ~ + \frac{p_dab^*}{2}\big(\ket{001}\!\bra{111} \!-\! \ket{001}\!\bra{100} \!-\! \ket{010}\!\bra{111} \!+\!\ket{010}\!\bra{100} \big) \nonumber \\
        & ~ + \frac{p_da^*b}{2}\big(\ket{111}\!\bra{001} \!-\! \ket{111}\!\bra{010} \!-\! \ket{100}\!\bra{001} \!+\!\ket{100}\!\bra{010} \big) \nonumber \\
        & ~ + \frac{1}{2}\big(p_d\abs{b}^2\!+\!\frac{1\!-\!p_d}{2} \big) \big(\ket{111}\!\bra{111} \!-\! \ket{111}\!\bra{100} \!-\! \ket{100}\!\bra{111} \!+\!\ket{100}\!\bra{100} \big)
        \end{align}
    \endgroup
\begingroup
    \setlength{\abovedisplayskip}{6pt}   
    \setlength{\belowdisplayskip}{6pt}   
    \setlength{\abovedisplayshortskip}{2pt}
    \setlength{\belowdisplayshortskip}{2pt}
    \begin{align}
        \overset{O_2}{\longrightarrow} & \frac{1}{4}\big(p_d\abs{a}^2\!+\!\frac{1\!-\!p_d}{2} \big) \big((\ket{001}\!+\!\ket{101})(\bra{001}\!+\!\bra{101}) \!-\! (\ket{001}\!+\!\ket{101})(\bra{010}\!+\!\bra{110}) \nonumber \\
        & ~ - (\ket{010}\!+\!\ket{110})(\bra{001}\!+\!\bra{101}) \!+\! (\ket{010}\!+\!\ket{110})(\bra{010}\!+\!\bra{110}) \big) \nonumber \\
        & ~ + \frac{p_dab^*}{4} \big((\ket{001}\!+\!\ket{101})(\bra{011}\!-\!\bra{111}) \!-\! (\ket{001}\!+\!\ket{101})(\bra{000}\!-\!\bra{100}) \nonumber \\
        & ~ - (\ket{010}\!+\!\ket{110})(\bra{011}\!-\!\bra{111}) \!+\! (\ket{010}\!+\!\ket{110})(\bra{000}\!-\!\bra{100}) \big) \nonumber \\
        & ~ + \frac{p_da^*b}{4} \big((\ket{011}\!-\!\ket{111})(\bra{001}\!+\!\bra{101}) \!-\! (\ket{011}\!-\!\ket{111})(\bra{010}\!+\!\bra{110}) \nonumber \\
        & ~ - (\ket{000}\!-\!\ket{100})(\bra{001}\!+\!\bra{101}) \!+\! (\ket{000}\!-\!\ket{100})(\bra{010}\!-\!\bra{110}) \big) \nonumber \\
        & ~ + \frac{1}{4}\big(p_d\abs{b}^2\!+\!\frac{1\!-\!p_d}{2} \big) \big((\ket{011}\!-\!\ket{111})(\bra{011}\!-\!\bra{111}) \!-\! (\ket{011}\!-\!\ket{111})(\bra{000}\!-\!\bra{100}) \nonumber \\
        & ~ - (\ket{000}\!-\!\ket{100})(\bra{011}\!-\!\bra{111}) \!+\! (\ket{000}\!-\!\ket{100})(\bra{000}\!-\!\bra{100}) \big) \\
        = & ~\frac{1}{4} \ket{00}\!\bra{00} \Big( \big(p_d\abs{a}^2 \!+\!\frac{1\!-\!p_d}{2}\big)\ket{1}\!\bra{1} \!-\! p_dab^* \ket{1}\!\bra{0} \!-\! p_da^*b \ket{0}\!\bra{1} \!+\! \big(p_d\abs{b}^2 \!+\! \frac{1\!-\!p_d}{2}\big)\ket{0}\!\bra{0} \Big) \nonumber \\
        & ~ + \frac{1}{4} \ket{01}\!\bra{01} \Big( \big(p_d\abs{a}^2 \!+\!\frac{1\!-\!p_d}{2}\big)\ket{0}\!\bra{0} \!-\! p_dab^* \ket{0}\!\bra{1} \!-\! p_da^*b \ket{1}\!\bra{0} \!+\! \big(p_d\abs{b}^2 \!+\! \frac{1\!-\!p_d}{2}\big)\ket{1}\!\bra{1} \Big) \nonumber \\
        & ~ + \frac{1}{4} \ket{10}\!\bra{10} \Big( \big(p_d\abs{a}^2 \!+\!\frac{1\!-\!p_d}{2}\big)\ket{1}\!\bra{1} \!+\! p_dab^* \ket{1}\!\bra{0} \!+\! p_da^*b \ket{0}\!\bra{1} \!+\! \big(p_d\abs{b}^2 \!+\! \frac{1\!-\!p_d}{2}\big)\ket{0}\!\bra{0} \Big) \nonumber \\
        & ~ + \frac{1}{4} \ket{11}\!\bra{11} \Big( \big(p_d\abs{a}^2 \!+\!\frac{1\!-\!p_d}{2}\big)\ket{0}\!\bra{0} \!+\! p_dab^* \ket{0}\!\bra{1} \!+\! p_da^*b \ket{1}\!\bra{0} \!+\! \big(p_d\abs{b}^2 \!+\! \frac{1\!-\!p_d}{2}\big)\ket{1}\!\bra{1} \Big) \nonumber \\
        & ~ + \rho_{\text{res}} \\
        =& ~\frac{1}{4} \ket{00}\!\bra{00} \otimes \underbrace{ZX\rho_\text{data}XZ}_{\rho_{zx}} + \frac{1}{4}\ket{01}\!\bra{01}\otimes \underbrace{Z\rho_\text{data}Z}_{\rho_{z}} + \frac{1}{4}\ket{10}\!\bra{10}\otimes \underbrace{X\rho_\text{data}X}_{\rho_{x}} + \frac{1}{4}\ket{11}\!\bra{11}\otimes \rho_\text{data}  \nonumber \\
        & ~ + \rho_\text{res}~,
    \end{align}
\endgroup
where $\rho_\text{res}$ denotes the terms not contributing to the final state when measured in the standard basis.
The operations $O_1$ and $O_2$ correspond to ${CNOT}_{12}\otimes {I}_3$ and ${H}_1\otimes {I}_2\otimes {I}_3$, respectively, with ${I}$ denoting the identity operator and the subscripts indicating the qubits on which the gates act.
We can perform similar calculations for teleportation using the remaining three maximally entangled states $\ket{\Phi^+}$, $\ket{\Phi^-}$, and $\ket{\Psi^+}$.
Therefore, starting the teleportation using the state $\rho_w$ from~\eqref{eq:werner_state_maximally_mixed} as the entanglement resource, to teleport the data qubit $\rho_\text{data}$, the state of the system immediately before Alice's measurement is given by:
\begin{align} \label{eq:teleported-state-density-matrix}
    &\frac{1}{4}\ket{00}\!\bra{00} \underbrace{\big(p_1\rho_\text{data} + p_2\rho_{z} + p_3\rho_{x} + p_4\rho_{zx}\big)}_{=:\rho_{{\Phi}^+}} + \frac{1}{4}\ket{01}\!\bra{01} \underbrace{\big(p_1\rho_{x} + p_2\rho_{zx} + p_3\rho_{\text{data}} + p_4\rho_{z}\big)}_{=:\rho_{{\Phi}^-}} \nonumber \\
    & + \frac{1}{4}\ket{10}\!\bra{10} \underbrace{\big(p_1\rho_{z} + p_2\rho_{\text{data}} + p_3\rho_{zx} + p_4\rho_{x}\big)}_{=:\rho_{{\Psi}^+}} + \frac{1}{4}\ket{11}\!\bra{11} \underbrace{\big(p_1\rho_{zx} + p_2\rho_{x} + p_3\rho_{z} + p_4\rho_{\text{data}}\big)}_{=:\rho_{{\Psi}^-}} \nonumber \\
    & + \rho_{\text{res}}^\prime~,
\end{align}
where $\rho_{\text{res}}^\prime$ collects all terms not contributing to the measurement result. 
After Alice performs a Bell-state measurement on the first two qubits, she communicates the outcome to Bob over a classical channel.
Based on the measurement outcome, Bob applies the appropriate Pauli correction to recover the desired output state $\rho_{\Phi^+}$ as
\begin{align}
    \rho_{\Phi^+} = X\rho_{\Phi^-}X = Z\rho_{\Psi^+}Z = ZX\rho_{\Psi^-}ZX~.
\end{align}
Substituting the specific probabilities for the Werner state, given by $p_1 \!=\! ({1\!+\!3w}/{4})$ and ${p_2 \!=\! p_3 \!=\! p_4 \!=\! ({1\!-\!w}/{4})}$ into~\eqref{eq:teleported-state-density-matrix}, and using the identity ${(\rho_\text{data}\!+\!\rho_z\!+\!\rho_x\!+\!\rho_{zx}) \!=\! 2\thinspace\mathbb{I}}$, we can express ${\rho_{\Phi^+} \!=\! w\rho_\text{data} \!+\!(1\!-\!w)\mathbb{I}/2}$.
However, the classical communication of the measurement outcome from Alice to Bob takes a finite time $t_{\text{class}}$, during which the stored qubit in Bob's memory undergoes decoherence.
Thus, the resulting teleported state is
\begingroup
    \setlength{\abovedisplayskip}{6pt}   
    \setlength{\belowdisplayskip}{6pt}   
    \setlength{\abovedisplayshortskip}{2pt}
    \setlength{\belowdisplayshortskip}{2pt}
    \begin{align}
        \rho_\text{tel} &= \mathcal{E}_{t_\text{class}}\Big(w\rho_\text{data} + (1\!-\!w)\frac{\mathbb{I}}{2}\Big) \nonumber \\
        &= wp_d e^{-t_\text{class}/t_\text{coh}} \ket{\phi}\!\bra{\phi} + (1\!-\!wp_d e^{-t_\text{class}/t_\text{coh}})\frac{\mathds{1}}{2}~.
    \end{align}
\endgroup
Consequently, the teleportation fidelity is given by
\begingroup
    \setlength{\abovedisplayskip}{6pt}   
    \setlength{\belowdisplayskip}{6pt}   
    \setlength{\abovedisplayshortskip}{2pt}
    \setlength{\belowdisplayshortskip}{2pt}
    \begin{align}
        F_\text{tel} = \text{Tr}(\ket{\phi}\!\!\bra{\phi}\rho_\text{tel}) = \frac{1\!+\!w p_d e^{-t_\text{class}/t_\text{coh}}}{2}~. \nonumber
    \end{align}
\endgroup
Thus, we have established~\eqref{eq:teleportation_fidelity_general_expression}.
\end{proof}

\end{document}